\documentclass[11pt]{article}

\usepackage[margin=1in]{geometry}
\usepackage{graphicx,xcolor,tikz}
\usepackage{hyperref}
\usepackage{amssymb}
\usepackage{amsthm}
\usepackage{enumitem}
\usepackage{caption}
\usepackage{subcaption}

\newtheorem{thm}{Theorem}
\newtheorem{lem}[thm]{Lemma}
\newtheorem*{lem*}{Lemma}

\newtheorem{cor}[thm]{Corollary}

\newtheorem{defn}[thm]{Definition}

\theoremstyle{remark}
\newtheorem{claim}[thm]{Claim}

\newcommand{\Gtri}{G_{\Delta}}

\newcommand{\ci}{\mathcal{C}_i}
\newcommand{\cli}{\mathcal{C}_{<i}}
\newcommand{\clei}{\mathcal{C}_{\leq i}}
\newcommand{\cgi}{\mathcal{C}_{>i}}
\newcommand{\cgei}{\mathcal{C}_{\geq i}}
\newcommand{\cc}{\mathcal{C}}

\definecolor{applegreen}{rgb}{0.55, 0.71, 0.0}
\definecolor{teal}{rgb}{0.0, 0.5, 0.5}
\definecolor{alizarin}{rgb}{0.82, 0.1, 0.26}
\definecolor{slategray}{rgb}{0.44, 0.5, 0.56}
\definecolor{amber}{rgb}{1.0, 0.75, 0.0}
\definecolor{mikadoyellow}{rgb}{1.0, 0.77, 0.05}
\definecolor{cadmiumgreen}{rgb}{0.0, 0.42, 0.24}
\definecolor{forestgreen}{rgb}{0.13, 0.55, 0.13}
\definecolor{lust}{rgb}{0.9, 0.13, 0.13}
\definecolor{denim}{rgb}{0.08, 0.38, 0.74}
\definecolor{purpleheart}{rgb}{0.41, 0.21, 0.61}
\definecolor{cherryblossompink}{rgb}{1.0, 0.72, 0.77}
\definecolor{darktangerine}{rgb}{1.0, 0.66, 0.07}
\definecolor{bananayellow}{rgb}{1.0, 0.88, 0.21}
\definecolor{lightblue}{rgb}{0.55,0.82,0.77}
\definecolor{lightgray}{rgb}{0.83, 0.83, 0.83}
\definecolor{languidlavender}{rgb}{0.84, 0.79, 0.87}
\definecolor{jasper}{rgb}{0.84, 0.23, 0.24}
\definecolor{tangelo}{rgb}{0.98, 0.3, 0.0}
\definecolor{tearose}{rgb}{0.97, 0.51, 0.47}
\definecolor{royalfuchsia}{rgb}{0.79, 0.17, 0.57}
\definecolor{cinnabar}{rgb}{0.89, 0.26, 0.2}
\definecolor{cinnamon}{rgb}{0.82, 0.41, 0.12}

\newcommand{\todoo}[1]{}

\begin{document}

\title{Irreducibility of Recombination Markov Chains in the Triangular Lattice}
\author{Sarah Cannon\footnote{Corresponding author. scannon@cmc.edu. 850 Columbia Ave, Claremont, CA 91711.}\\ Claremont McKenna College}
\date{}
\maketitle
\pagestyle{empty}

{\bf Abstract.} 
In the United States, regions (such as states or counties) are frequently divided into districts for the purpose of electing representatives.  How the districts are drawn can have a profound effect on who's elected, and drawing the districts to give an advantage to a certain group is known as gerrymandering.  It can be surprisingly difficult to detect when gerrymandering is occurring, but one algorithmic method is to compare a current districting plan to a large number of randomly sampled plans to see whether it is an outlier. Recombination Markov chains are often used to do this random sampling: randomly choose two districts, consider their union, and split this union up in a new way. This approach works well in practice and has been widely used, including in litigation, but the theory behind it remains underdeveloped. For example, it's not known if recombination Markov chains are irreducible, that is, if recombination moves suffice to move from any districting plan to any other. 

Irreducibility of recombination Markov chains can be formulated as a graph problem: for a planar graph $G$, is the space of all partitions of $G$ into $k$ connected subgraphs ($k$ districts) connected by recombination moves? While the answer is yes when districts can be as small as one vertex, this is not realistic in real-world settings where districts must have approximately balanced populations. Here we fix district sizes to be $k_1 \pm 1$ vertices, $k_2 \pm 1$ vertices, $\ldots$ for fixed $k_1$, $k_2$, $\ldots$, a more realistic setting. We prove for arbitrarily large triangular regions in the triangular lattice, when there are three simply connected districts, recombination Markov chains are irreducible. This is the first proof of irreducibility under tight district size constraints for recombination Markov chains beyond small or trivial examples. The triangular lattice is the most natural setting in which to first consider such a question, as graphs representing states/regions are frequently triangulated. The proof uses a sweep-line argument, and there is hope it will generalize to more districts, triangulations satisfying mild additional conditions, and other redistricting Markov chains.

\vspace{5mm}

\vspace{2mm}

{\bf Declarations of interest.} None
\vspace{2mm}

{\bf Funding.} S. Cannon is supported in part by NSF grants CCF-2104795 and DMS-1803325. 
\vspace{2mm}

{\bf Previous Versions.} A 10 page conference version of this paper, not including any proofs, appeared in the Proceedings of the SIAM Conference on Applied and Computational Discrete Algorithms (ACDA23), May 31-June 2, 2023. This version is available at \url{https://epubs.siam.org/doi/10.1137/1.9781611977714.9}~\cite{irred-conf}.

\newpage

\pagestyle{plain}
\setcounter{page}{1}


\section{Introduction}

	
In the United States, regions (such as cities, counties, or states)  are frequently divided into {\it districts} for the purpose of electing officials to positions ranging from a local school board to the U.S. House of Representatives. The way these districts are drawn can have a large effect on who is elected. Drawing the lines of these districts so as to give an advantage to a certain individual, group, or political party is known as {\it gerrymandering}.  It can be surprisingly difficult to detect if gerrymandering is occurring or whether outcomes considered `unfair' are a result of other factors, such as the spatial distribution of voters (for example, see~\cite{pg-intro}).

One method to detect gerrymandering is to see where the current or proposed districting plan lies within the context of all possible districting plans for a region. Districting plans are generally expected to have contiguous, compact districts, even when not explicitly required by law, though there are a variety of competing notion of compactness (see, for example, \cite{pg-compactness}).
Other legal requirements, such as respecting communities of interest,  avoiding county splits, or incorporating incumbency, are considered in certain jurisdictions as well.  However, no matter the restrictions placed on a districting plan, the number of possible districting plans for any state or region is far too large to be studied in its entirety. For example, the number of ways to divide a $9 \times 9$ square grid into 9 contiguous, equally-sized districts is more than 700 trillion \cite{mggg-table}. 
	 
Because of this, random sampling of political districting plans has  become an important tool to help understand the space of possible  plans, beginning with the work of Chen and Rodden~\cite{chen-rodden-unintentional, chen-rodden-thicket}. A collection of randomly sampled districting plans has come to be called an {\it ensemble}. If a current or proposed districting plan is an outlier with respect to the ensemble, this may be evidence it is~gerrymandered.

A variety of methods for creating ensembles of randomly sampled districting plans exist (see Related Work, below). In this paper, the focus is on recombination Markov chains.  
These chains create random districting plans by repeatedly choosing two random districts; merging these two districts  together; and splitting this union up in a new way so that population balance and other constraints are still satisfied.  
Recombination Markov chains have been used in a variety of academic papers \cite{DukeReCom1, DukeReCom2,  VRA-ELJ, revrecom, shortbursts,  charikar2022complexity,chenstephanopoulos, clelland2021compactness, competitiveness, recom,  MRL,  georgia}
, technical reports \cite{ Alaska, VA-criteria,  VA-report, santaclara, lowell, chicago}, and court cases \cite{hirsch2022brief, chen2022brief, duchin2019brief}, including in 2021--2022 litigation in Pennsylvania, South Carolina, and Texas.  


One common problem with all recombination Markov chains is that it's not known whether they're irreducible. That is, it's not known whether recombination moves suffice to reach all districting plans, or if there are some plans that cannot be created by the repeated merging and splitting process. This could potentially create problems if ensembles are created by sampling from only the reachable subset of plans rather than the entire space of plans. While there is no evidence so far that any real-world examples of recombination Markov chains are not irreducible, proofs have largely remained elusive. 

\subsection{Dual Graphs, balanced partitions, and recombination}
	
Recombination Markov chains work with a discretization of the real-world political districting problem by considering {\it dual graphs}: graphs whose nodes are small geographic units (such as census blocks, census block groups, or voting precincts) and whose edges show geographic adjacency~\cite{DuchinTenner}. While political districts may occasionally split census blocks or voting precincts, this is rare, and it is generally agreed that considering districting plans built out of only whole geographic units is both a good approximation to considering all districting plans as well as necessary to make the problem tractable. This means districting plans are partitions of the nodes of the dual graphs into $k$ connected sugraphs, where $k$ is the number of districts.  In real-world examples, nodes also have attached populations and districts must be population-balanced. This is usually operationalized by ensuring the population of each district doesn't differ by more than an $\varepsilon$ multiplicative factor from its ideal population, which is the total population divided by the number of districts.  In practice, $\varepsilon$ is typically about 1-2\%; if tighter population balance is required, this can be achieved by making targeted local changes after the fact.  
	

Recombination Markov chains are a family of algorithms for randomly sampling connected, approximately population balanced partitions of a dual graph. Variants include {\sf ReCom}~\cite{recom}, {\sf ForestReCom}~\cite{DukeReCom1}, and {\sf ReversibleReCom}~\cite{revrecom}. 
At a high level, all these Markov chains choose two random adjacent districts; merge the districts together; find a way of splitting the vertices in these two district up in a new way; and accept the new partition with some probability. One can ask both about the {\it probability} of these moves, which vary between the different versions of recombination, or about the {\it valid steps}, the moves that occur with non-zero probability, which are the same between all these variants. These chains all put positive probability on every possible recombination step, that is, every way of recombining two adjacent districts to produce two new districts that are connected and satisfy the population constraints.
As we will later use, one consequence is if a move from partition $\sigma$ to partition $\tau$ is valid, so is the reverse move from $\tau$ to $\sigma$.  
 We focus only on these valid moves, and our results will apply to any recombination Markov chain that puts non-zero probability on every possible recombination step (as {\sf ReCom}, {\sf ForestReCom}, and {\sf ReversibleReCom} do).

\subsection{Irreducibility, Aperiodicity, and Ergodicity}

A Markov chain is {\it irreducible} if the transitions of the chain suffice to go from any state of the Markov chain to any other state (here, each state is a partition of the dual graph into $k$ connected subgraphs with some population balance constraint). The set of states of a Markov chain is known as the {\it state space}, and if the chain is irreducible the state space is said to be {\it connected}. 
Determining whether a Markov chain is irreducible only requires looking at the valid moves (the moves with non-zero probability) rather than at the probabilities of those moves, which is why our results apply to the entire family of recombination Markov chains at once.  
 
A Markov chain is {\it ergodic} if it is both irreducible and {\it aperiodic}, meaning there is no periodicity in the way the chain moves around the state space. A chain is aperiodic if there is at least one {\it self-loop}, that is, at least one state that has a non-zero probability of remaining in the state after one step. All of {\sf ReCom}, {\sf ForestReCom}, and {\sf ReversibleReCom} have self-loops, so they are aperiodic. If a chain is not aperiodic, it is easy to make it aperiodic as follows: At each step, with probability $0.01$ remain in the current state and with probability $0.99$ do a transition of the Markov chain (the values of $0.01$ and $0.99$ were chosen arbitrarily, and any values summing to one suffice). As it is easy to achieve aperiodicity, to know whether a Markov chain is ergodic the key question is whether it is irreducible, the focus of this paper. 

We care a great deal about ergodicity because every ergodic Markov chain eventually converges to a unique {\it stationary distribution} over the states of the chain. 
This is a necessary requirement for using a Markov chain to draw samples from this stationary distribution. 
We show recombination Markov chains are irreducible on the triangular lattice, and we know they can easily be made aperiodic (if they are not already), meaning they are ergodic. 
Once we know Recombination Markov chains converge to a unique stationary distribution, we can begin to develop the theory behind them and consider questions such as how long this convergence takes: a long-term goal of this research community is to be able to say something rigorous about the mixing and/or relaxation times of recombination Markov chains, and knowing the chains are irreducible is a necessary first step.
Gaining a rigorous understand of these Markov chains and their behavior is essential so that we can have confidence in the conclusions about gerrymandering they are used to produce. 



\subsection{Irreducibility of Recombination}

In beginning to study the irreducibility of recombination Markov chains, researchers have introduced a simplification: assume each node has the same population (see, for example,~\cite{akitaya2022recom}). This means one can determine if a districting plan is population-balanced by looking only at the number of nodes in each district, rather than the populations at those nodes. In this setting, the {\it ideal size} of a district is $n/k$, the total number of vertices divided by the number of districts.  In~\cite{akitaya2022recom}, authors show (1) when districts sizes can get arbitrarily large or small, recombination Markov chains are irreducible; (2) when the underlying graph is Hamiltonian, recombination Markov chains are irreducible when districts are allowed to get as small as one vertex or as large as twice their ideal size; and (3) there exist Hamiltonian planar graphs on which recombination is not irreducible when district sizes are constrained to be at least $(2/3) n/k$ and at most $(4/3) n/k$. This last result shows how imposing tight size constraints makes it much harder to reach all possible districting plans. The only known positive irreducibility results under tight size constraints are for double-cycle graphs and grid-with-a-hole graphs, where the large amount of structure present makes the proofs nearly trivial~\cite{charikar2022complexity}. 
No positive irreducibility results are known for recombination Markov chains beyond these trivial examples when district sizes are constrained any tighter than in result (2) above. 


	
The result of~\cite{akitaya2022recom} is related to one approach to irreducibility sometimes taken in practice: Initially allow districts to get arbitrarily small, and then gradually tighten population constraints (`cool' the system) until the districts are as balanced in size or population as desired. However, if the state space is disconnected when restricted to partitions whose districts are close to their ideal sizes, this process may not end in each connected component of the state space with the correct relative probabilities. This means certain parts of the state space may be oversampled or undersampled as a result, though approaches such as parallel tempering can address this.  However, this tempering process is computationally expensive and is often skipped in practice. On the other hand, new work suggests under certain stationary distributions population-balanced partitions are polynomially-likely~\cite{cannon2023balancedforest} and thus rejection sampling can be used to obtain balanced partitions as proposed in~\cite{charikar2022complexity}, but this certainly does not hold for all distributions one might want to use a recombination Markov chain to sample from.

Because of this, irreducibility results under tight size constraints are very desireable. The negative result of~\cite{akitaya2022recom} mentioned above implies general results for planar graphs are impossible. 
However, their examples are far from the types of planar graphs that might be encountered in real-world applications, which don't usually have faces with long boundary cycles. However, there are also some negative irreducibility results even for grids when district sizes are constrained to take on an exact value. For example, see Figure~\ref{fig:non-irred-ex}(a) for an example that is rigid under recombination moves when districts are restricted to be size 3: any attempt to merge two districts and split the resulting union into two equal-sized pieces will always produce this exact partition~\cite{tuckerfoltz2023locked}. It's worth noting such examples also exist in the triangular lattice, which will be our focus. For example, in a triangle with three vertices along each side (six vertices total), recombination is not irreducible when districts are constrained to be size exactly two; see Figure~\ref{fig:non-irred-ex}(b). 
This indicates that even in simple graphs such as grids, requiring exact sizes for districts is likely too much to ask for.

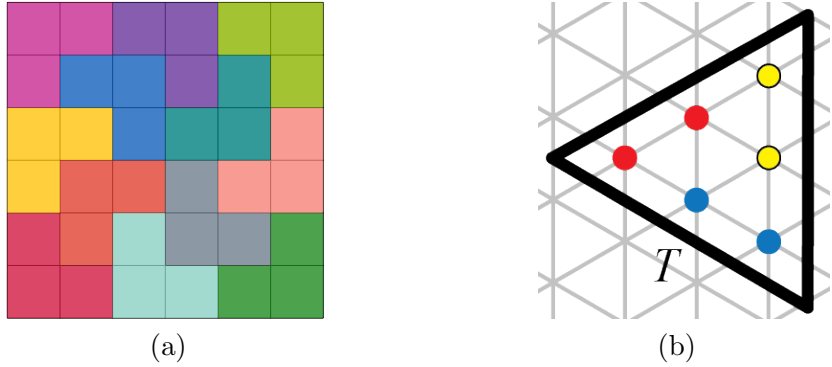
\begin{figure}
\centering
\hfill 
\begin{tikzpicture}[scale=0.7]
	\begin{scope}[xshift=.5cm,yshift=.5cm]
		\draw[step=1.0,black,thin] (0,0) grid (6,6);
	\end{scope}
	\draw [line width=0.7cm,alizarin,opacity=.8] (1,2.5)--(1,1)--(2.5,1);
	\draw [line width=0.7cm,lightblue,opacity=.8] (3,2.5)--(3,1)--(4.5,1);
	\draw [line width=0.7cm,forestgreen,opacity=.8] (6,2.5)--(6,1)--(4.5,1);
	\draw [line width=0.7cm,mikadoyellow,opacity=.8] (1,2.5)--(1,4)--(2.5,4);
	\draw [line width=0.7cm,cinnabar,opacity=.8] (2,1.5)--(2,3)--(3.5,3);
	\draw [line width=0.7cm,slategray,opacity=.8] (4,3.5)--(4,2)--(5.5,2);
	\draw [line width=0.7cm,tearose,opacity=.8] (6,4.5)--(6,3)--(4.5,3);
	\draw [line width=0.7cm,purpleheart,opacity=.8] (4,4.5)--(4,6)--(2.5,6);
	\draw [line width=0.7cm,teal,opacity=.8] (5,5.5)--(5,4)--(3.5,4);
	\draw [line width=0.7cm,royalfuchsia,opacity=.8] (1,4.5)--(1,6)--(2.5,6);
	\draw [line width=0.7cm,denim,opacity=.8] (3,3.5)--(3,5)--(1.5,5);
	\draw [line width=0.7cm,applegreen,opacity=.8] (6,4.5)--(6,6)--(4.5,6);
	
\end{tikzpicture}
\hfill
\includegraphics[scale = 1.5]{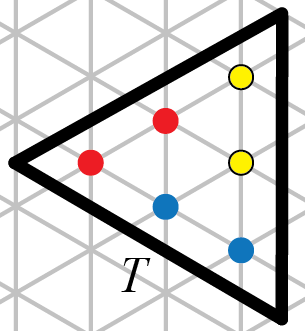} 
\hfill\hfill

(a) \hspace{6cm} (b)
\caption{(a) From~\cite{tuckerfoltz2023locked}, a partition of a $6 \times 6$ grid into 12 districts of size exactly 3 that is rigid under recombination moves. (b) A partition of a triangular subgraph $T$ of the triangular lattice into 3 districts of size exactly 2 that is rigid under recombination moves. 
	}\label{fig:non-irred-ex}
\end{figure}
	
\subsection{Results}
	
In both examples in Figure~\ref{fig:non-irred-ex}, irreducibility can easily be achieved by allowing districts sizes to get one larger or one smaller than their prescribed ideal sizes. This motivates the following definitions used in this work. We assume districts have prescribed sizes $k_1$, $k_2$, $\ldots$.  A partition is {\it balanced} if each district $i$ has exactly $k_i$ vertices, and a partition is {\it nearly balanced} if the partition is not balanced but the number of vertices in district $i$ is at least $k_i - 1$ and at most $k_i + 1$. While the case $k_1 = k_2 = \ldots$ is most interesting, equality of district sizes in not required for our results. 
	
This paper provides the first irreducibility results for an infinite class of graphs for the recombination Markov chain whose state space is all balanced and nearly balanced partitions. The graphs we consider, chosen to strike a balance between having enough structure for proofs to be possible while still being motivated by real world examples, are subsets of the triangular lattice. The triangular lattice was chosen instead of the more frequently used square lattice because dual graphs derived from geography frequently have mostly triangular faces: it's rare to have four or more geographic domains (such as census blocks or voting precincts) meet at a single point, making faces with a boundary cycle of length four or more rare in the corresponding dual graphs. Triangulations also have a nice feature that is incredibly important for our proofs: The neighbors of any vertex form a cycle. 
	
Our main result is that for arbitrarily large triangular subgraphs of the triangular lattice (e.g., Figure~\ref{fig:non-irred-ex}(b)), recombination Markov chains for three districts are irreducible on the state space of balanced and nearly balanced partitions. There is one small caveat: we require districts to be simply connected rather than just connected, and this is crucially used at several points in our proofs. However, this is not an unreasonable restriction: it's rare in practice to see one district completely encircling another. We also require the minor technical condition that the districts are not too small: if the triangular subgraph has side length $n$, then $k_i$ must be at least $n$ for all $i = 1,2,3$. We also assume $n \geq 5$.  We summarize this result in the following theorem.




\begin{thm}\label{thm:main}
	Let $T$ be a triangular subset of the triangular lattice with side length $n\geq 5$. Let $k_1$, $k_2$, and $k_3$ be integers satisfying $k_1 + k_2 + k_3 = n(n+1)/2 = |V(T)|$ and each $k_i \geq n$.  Let $\Omega$ be all partitions of $T$ into three simply  connected pieces $P_1$, $P_2$, and $P_3$ where $|P_i| \in [k_i-1, k_i+1]$ for $i = 1,2,3$. Recombination Markov chains on $\Omega$ are irreducible. 
\end{thm}

This is a significant step beyond previous irreducibility results. We go beyond small, computationally-verified examples to an infinite class of graphs, and do so under extremely tight conditions on the district size: districts never get more than one vertex larger or smaller than their prescribed size. 
It is surprising that relaxing district sizes by just one vertex is sufficient for proofs to be successful. 
%
It is hoped this first step showing irreducibility for arbitrarily large graphs, and the ideas and approaches contained in this paper, can be used as a springboard for further irreducibility results. Future work includes generalizing this result to more than three districts, other subsets of the triangular lattice, and perhaps eventually all planar triangulations. These all now seem plausibly within reach.

Our proof is constructive, that is, we give a sequence of moves that can transform any partition into any other partition.  As a consequence, analyzing the number of recombination steps in these paths gives us a bound on the diameter of the state space. 

\begin{thm}\label{thm:dia_intro}
		Let $T$ be a triangular subset of the triangular lattice with side length $n\geq 5$. Let $k_1$, $k_2$, and $k_3$ be integers satisfying $k_1 + k_2 + k_3 = n(n+1)/2 = |V(T)|$ and each $k_i \geq n$.  Let $\Omega$ be all partitions of $T$ into three simply  connected pieces $P_1$, $P_2$, and $P_3$ where $|P_i| \in [k_i-1, k_i+1]$ for $i = 1,2,3$.  At most $O(n^3)$ recombination steps are required to transform any partition of $\Omega$ into any other partition of $\Omega$. 
\end{thm}

If we let $N = n(n+1)/2$ be the number of vertices in the region $T$ we consider, this bound is $O(N^{3/2})$, which is much smaller than expected. 
Because we do not claim to have found the most efficient way of moving between any two partitions, it is possible the true diameter of the state space is even smaller.




\subsection{Related Work}

{\bf Extensions of Recombination Markov Chains:} While recombination Markov chains have been widely used in practice to create ensembles and evaluate potential gerrymandering, they have also formed the basis for further explorations.  For example,~\cite{DukeReCom2} gives a multi-level version of the recombination Markov chain that, in addition to computational speed-ups, can help preserve communities of interest, such as counties.  Additionally,~\cite{shortbursts} uses the recombination Markov chain to find districting plans that have many majority-minority districts by repeatedly running short `bursts' of the recombination chain from carefully chosen starting points. Since new algorithms are being created with recombination Markov chains as a key underlying process, it's essential we continue to develop a rigorous understanding of recombination Markov chains, as we do here.

{\bf Flip Markov Chains}: Another type of Markov chain that has been used for sampling districting plans is flip Markov chains~\cite{DukeNC, ImaiFlip, frieze2022subexponential}.  In these chains, only a single vertex is reassigned to a new district in each step.
 Flip moves are a subset of recombination moves: any flip move can be achieved by a recombination step that merges two districts and splits them up such that only one vertex has changed its district assignment.  
It is known that flip moves sometimes connect the state space and sometimes do not. For example, in~\cite{akitaya2023flip}, authors show in 2-connected graphs when districts can get arbitrarily large or small, flip moves connect all districting plans. In~\cite{frieze2022subexponential}, the authors note their flip Markov chain may not be irreducible (and give an example of this) and instead restrict their attention to the connected pieces of the state space. However, in many instances, whether or not flip moves connects the state space remains an open question. 

In our work, because flip moves are a subset of recombination moves and are simpler, we frequently focus on flip moves. Most results apply to both recombination moves and flip moves. The main exception is the {\it Cycle Recombination Lemma} (Lemma~\ref{lem:cycle-recom}) which requires a recombination step rather than a flip step. Recombination steps are also used at the end of our sweep line process to ultimately reach a ground state (Lemma~\ref{lem:groundstate}), and to move between ground states. 
If alternate proofs of these lemmas using flip steps rather than recombination moves is found, our results would also imply flip Markov chains are irreducible in the same settings. 
We believe such a result is likely possible, but we have not pursued it because our focus is on recombination moves. 


While Markov chains have been most often used to create random samples, requiring running the chains for many steps, flip Markov chains have also been used to detect `careful crafting' of districting plans by detecting whether a plan is an outlier with respect to the stationary distribution~\cite{CFP, CFMP}. These methods do not require the Markov chains to be irreducible, but provide a single significance value rather than a more robust understanding of the space of possible districting~plans.

	
	

	
	

{\bf Other Methods for Generating Ensembles}: It should be noted Markov chains are not the only methods employed to generate ensembles of districting plans. For example, the first papers introducing the idea of ensembles created them by randomly merging precincts to form the correct number of districts, and then exchanging precincts between districts to achieve population balance~\cite{chen-rodden-unintentional, chen-rodden-thicket}; it is challenging to know which distribution of districting plans this samples from. A technique known as Sequential Monte Carlo~\cite{ImaiReCom} generates random districting plans by iteratively sampling one district at a time using spanning tree methods and reweighting at each step to ensure convergence to a desired target distribution. The authors of~\cite{GS21,GGRS22} use a two-stage method to generate districting plans that allows the incorporation of a notion of fairness into the district selection process. 
	
		
{\bf Proving Irreducibility}: The proof of irreducibility we give here has some features in common with the irreducibility proofs in~\cite{cdrr16,oh_foraging}. The first shows a Markov chain on simply connected subgraphs of the triangular lattice is irreducible, and the second does the same in the presence of a fixed vertex that is constrained to always be in the subgraph. As in~\cite{cdrr16}, the main idea we use is a sweep-line procedure, adjusting the districting in a left-to-right fashion; sweep line approaches are common throughout the field of computational geometry.  As in~\cite{cdrr16, oh_foraging}, the idea of {\it towers} we use is inspired by the towers of~\cite{lrs}. 


\subsection{Discussion and Next Steps}

Before providing both a proof outline and a complete proof, which occupies the rest of this paper, we discuss the significance of our results and some potential next steps. 
First, we believe there are ample opportunities to simplify and shorten this proof. The focus of this paper was getting a complete proof rather than getting the most concise, elegant proof, so improvements can likely be made. This should be a first step before attempting to extend these results to new settings. 

The constraints placed on the problem (such as $n \geq 5$, $k_i \geq n$) were done so to simplify certain parts of the proofs; it's likely these conditions are not required for this result to be true, and additional work could weaken or entirely remove these constraints. Similarly, one could also hope to remove the simple connectivity constraint on each district. If one could show that from any districting plan where each district is connected but one or more districts are not simply connected, there exists a sequence of valid moves producing a simply connected districting plan, this would imply that Recombination is irreducible for all (not necessarily simply) connected districting plans. Such a result is likely possible using the approaches and lemmas of this paper, but this has not yet been pursued.

  
More significant next steps include generalizing the proof to more than three districts or other subregions of the triangular lattice. 
A main challenge in extending to more that three districts is case explosion: while a sweep-line argument like we use is likely possible, each step in this sweep line process can result in a nearly balanced partition, and we must rebalance the partition before proceeding with the next step of the sweep line process. We accomplish this rebalancing step by considering four cases, based on which districts are adjacent and whether these adjacencies occur along the boundary or not (see Figure~\ref{fig:4cases}). For even just four districts, the number of cases here would be much, much larger, and the corresponding rebalancing step much more difficult. Once the first district has been handled by the sweep-line algorithm, however, handling the remaining three districts should be straightforward using our results. Alternately, using an inductive approach by fixing one district and considering only the other three is also challenging, as the resulting region is not a triangle and can be extremely irregularly shaped. Extending our results beyond triangles to other convex shapes such as hexagons and parallelograms is likely possible, but non-convex regions - especially those with narrow bottlenecks - seem much more difficult. While we believe extending to more than three districts and the related problem of non-convex regions is possible, significant additional work will likely be required.  

A major next step would be to prove a similar result beyond the triangular lattice. The class of bounded-degree Hamiltonian planar triangulations seems the most likely candidate for success. Restricting our attention to triangulations is helpful because the neighbors of any given vertex form a cycle, meaning we can easily understand when a vertex can be removed from one district and added to another. Triangulations are also relevant to real-world redistricting applications, as the dual graphs representing states or regions are frequently triangulated or nearly triangulated. Hamiltonicity makes the definition of a ground state straightforward, and was also used in some of the results of ~\cite{akitaya2023flip, akitaya2022recom}, suggesting its usefulness. Several of the arguments included here would break down in the presence of large degree vertices, which is why we propose degree restrictions. However, there will be significant challenges in moving beyond the triangular lattice, as the assumption that the underlying graph is a regular lattice pervades nearly all of the proof.

The fundamental challenges in these kinds of results are finite-scale and non-local: in order to adjust a partition near a particular vertex, one may be constrained by the partition in the neighborhood of that vertex, necessitating first considering and adjusting the partition far away. While some recent results have considered infinite limits of graph partitions as partitions of the plane~\cite{akitaya2023polygonal,cannon2023balancedforest}, these resolve such problems by assuming one can always look at a finer scale.  For the redistricting application, it is interesting to understand how problems can be resolved without resorting to refinement, which is not always practicable for graphs arising from real-world geography. 

Finally, one might wonder about the more general case of population-weighted nodes. If one implements a multiplicative population tolerance $\varepsilon$, each district has an ideal population $P$, and each vertex has population less than $(\varepsilon/2)   P$, this means from a partition balanced within a multiplicative tolerance $\varepsilon/2$, any single vertex can be reassigned while staying within the overall tolerance $\varepsilon$.
In this case, there is hope an approach similar to ours might work, where we reassign single vertices in one step and rebalance (to within $\varepsilon/2$) when necessary. However, if node populations are larger, this presents additional difficulties that seem hard to resolve using our approaches.

Overall, this result is a major advancement that holds promise for inspiring future results.  
Knowing that recombination Markov chains are irreducible is a necessary first step to developing the theory behind them. A rigorous understanding of these Markov chains and their behavior is needed so that we can have confidence in the conclusions about gerrymandering they are used to produce.

\section{Proof Overview}

\label{sec:pfoverview}

The proof proceeds largely from first principles. The most complicated mathematics used are breadth-first search trees and some facts about boundaries of planar sets. Despite this, an incredible attention to detail and extensive careful constructions are still required to account for the intricacies that are possible in the partitions we consider. 

Let $\Gtri$ be the infinite graph whose edges and vertices are those of the infinite triangular lattice.  Let $T$ be a triangular subgraph of $\Gtri$.  We let $n$ denote the number of vertices along one side of (equilateral) triangle $T$, meaning $T$ contains $n(n+1)/2$ vertices total; Figure~\ref{fig:non-irred-ex}(b) gives an example when $n = 3$ and Figure~\ref{fig:ordering} gives an example when $n = 8$. Our proof will apply when $n \geq 5$. For simplicity, we assume $T$ is always oriented so that it has a vertical edge on its right side, as in both figures.  

Let $k_1$, $k_2$, and $k_3$ be such that each $k_i \geq n$ and $k_1 + k_2 + k_3 = n(n+1)/2$. 
We are interested in partitions of the vertices of $T$ into three simply connected subgraphs $P_1$, $P_2$, and $P_3$, of sizes $k_1\pm 1$, $k_2\pm1$, and $k_3\pm1$, respectively. The three sets $P_1$, $P_2$, and $P_3$ must be disjoint and their union must be $T$. In analogy to the redistricting motivation, we will call each of $P_1$, $P_2$, and $P_3$ a {\it district} of this partition. In an abuse of notation, we will let $P_i$ represent both the vertices in district $i$ as well as the induced subgraph of $T$ on this vertex set. Because it is a partition, the three sets $P_1$, $P_2$, and $P_3$ must be disjoint and their union must be $T$.  To avoid cumbersome language, throughout this paper we will say `partition' to mean a partition of $T$ into three simply connected districts. If $|P_1| = k_1$, $|P_2| = k_2$, and $|P_3| = k_3$, the partition is {\it balanced}. If the partition is such that $k_i - 1 \leq |P_i| \leq k_i + 1$ for $i = 1,2,3$ but it is not balanced, we say the partition is {\it nearly balanced}.

We will consider the state space $\Omega$ consisting of the balanced and nearly-balanced partitions of $T$ into three districts. We will examine the graph $G_\Omega$ whose vertices are the partitions in $\Omega$ where an (undirected) edge exists between $\sigma$ and $\tau$ if one district of $\sigma$ is identical to one district of $\tau$. These are exactly the transitions allowed by recombination Markov chains, which recombine two districts but leave a third untouched. If $G_\Omega$ is connected, this implies recombination Markov chains are irreducible, and our main theorem is equivalent to showing $G_\Omega$ is connected.  We now outline our approach for proving this theorem, in order to give the reader an idea what to expect,  before proceeding with the details in later sections.

\subsection{Ground States and Sweep Line}

We show, for every balanced or nearly balanced partition, there exists a sequence of moves transforming it into one of six ground states. Throughout this paper we consider a left-to-right, top-to bottom ordering of the vertices of $T$, where the leftmost vertex of $T$ is first; followed by the vertices in the next column, ordered from top to bottom; followed by the vertices in the third column, ordered from top to bottom; etc. This ordering when $n = 8$ is shown in Figure~\ref{fig:ordering}(a). Using this ordering, ground state $\sigma_{123}$ has the first $k_1$ vertices in $P_1$, the next $k_2$ vertices in $P_2$, and the final $k_3$ vertices in $P_3$; see Figure~\ref{fig:ordering}(b) for an example when $n = 8$ and $k_1 = k_2 = k_3 = 12$. Other ground states $\sigma_{132}$, $\sigma_{213}$, $\sigma_{231}$, $\sigma_{312}$, and $\sigma_{321}$ are defined similarly. Because each $k_i \geq n$, this is always a valid partition. It is straightforward to see the six ground states are connected to each other by recombination moves: any transposition of two adjacent indices in the ground state can be accomplished with one recombination step. 
Because for every recombination step, the reverse step is also valid, for irreducibility it suffices to show every balanced or nearly balanced partition can be transformed into one of these ground states.  

\begin{figure}
	\centering
	
	\includegraphics[scale = 0.85]{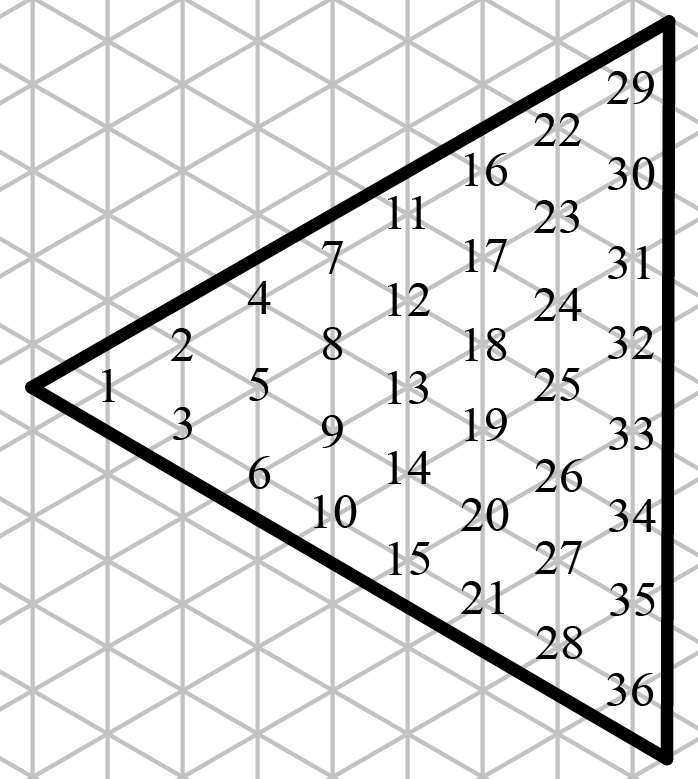} \hspace{10mm} 
	\includegraphics[scale = 0.85]{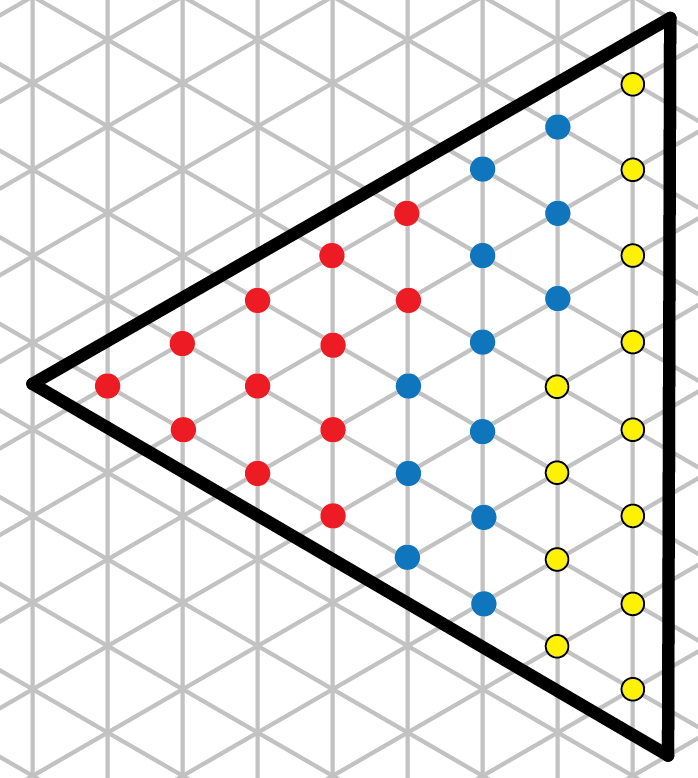} \ \ 
	
	(a) \hspace{6cm}(b) 
	\caption{(a) The left-to-right, top-to-bottom ordering we consider throughout this paper, shown for triangle $T$ with side length $n = 8$. (b) The ground state $\sigma_{123}$ when $k_1 = k_2 = k_3 = 12$, with district 1 shown in red, district 2 shown in blue, and district 3 shown in yellow. The condition $k_i \geq n$ ensures this is a valid partition. }
	\label{fig:ordering}
\end{figure}

Without loss of generality, we suppose $T$'s single leftmost vertex, which we call $\cc_1$, is in $P_1$, and we are trying to reach the ground state $\sigma_{123}$. We let $\cc_i$ be the first column which contains vertices not in $P_1$; see Figure~\ref{fig:sweepline_ex}(a) for an example. We will show how to (1) increase the number of vertices in $\cc_i \cap P_1$ and (2) if necessary, transform the result from a nearly balanced to a balanced partition without decreasing $\cc_i \cap P_1$. In both (1) and (2), any vertices left of $\cc_i$ remain unaffected.  Figure~\ref{fig:sweepline_ex} (b,c) gives an example of what this process looks like. After repeating this process for gradually increasing $i$, we eventually reach a state where there are no vertices in $P_1$ right of $\cc_i$.  In this case, a small number of recombination steps suffice to reach $\sigma_{123}$. 

We begin by outlining some key definitions and lemmas, and then give a high-level overview of how (1) and (2) are achieved.  The lemmas presented here lack some formality for ease of presentation; formal statements and complete proofs of these lemmas can be found in later sections.

\begin{figure}
	\centering
	\includegraphics[scale = 0.85]{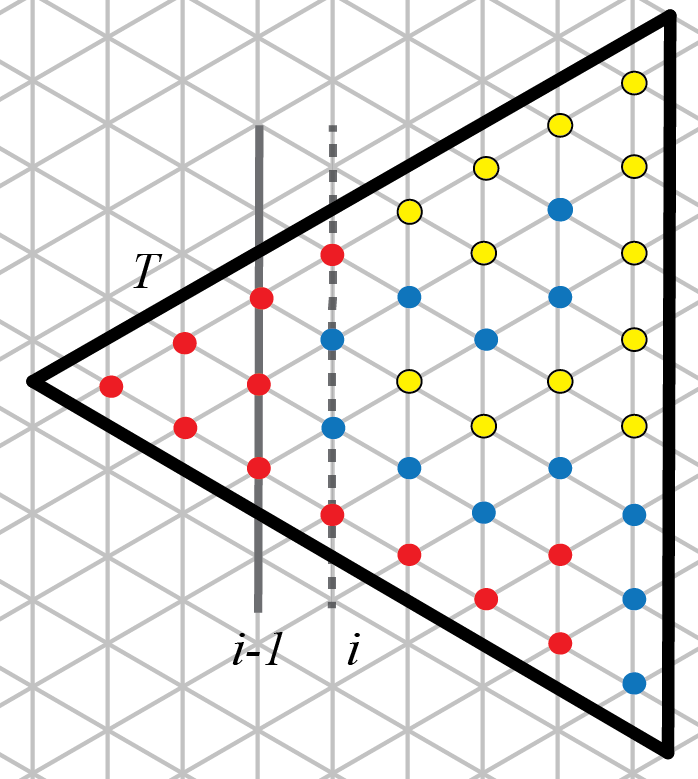} \ \ 
	\includegraphics[scale = 0.85]{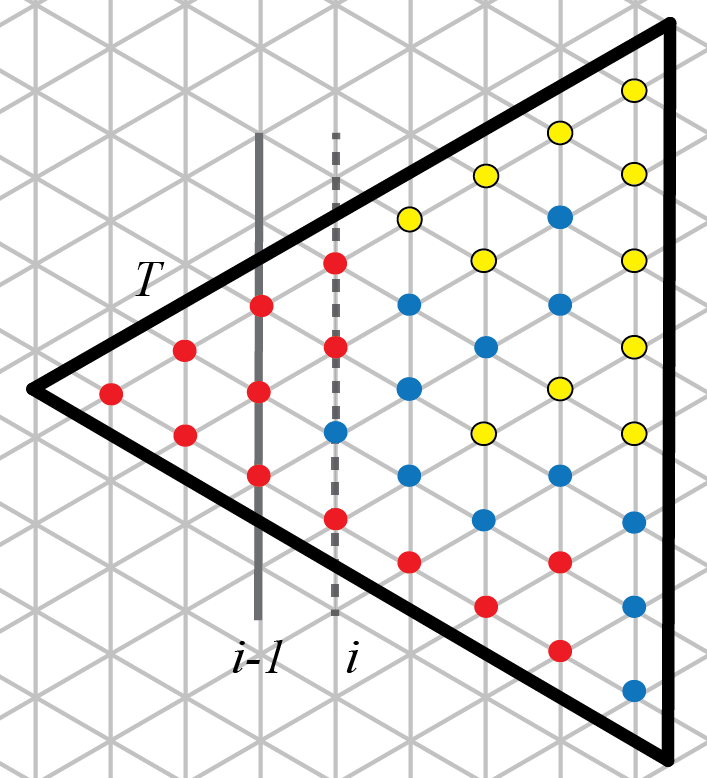} \ \ 
	\includegraphics[scale = 0.85]{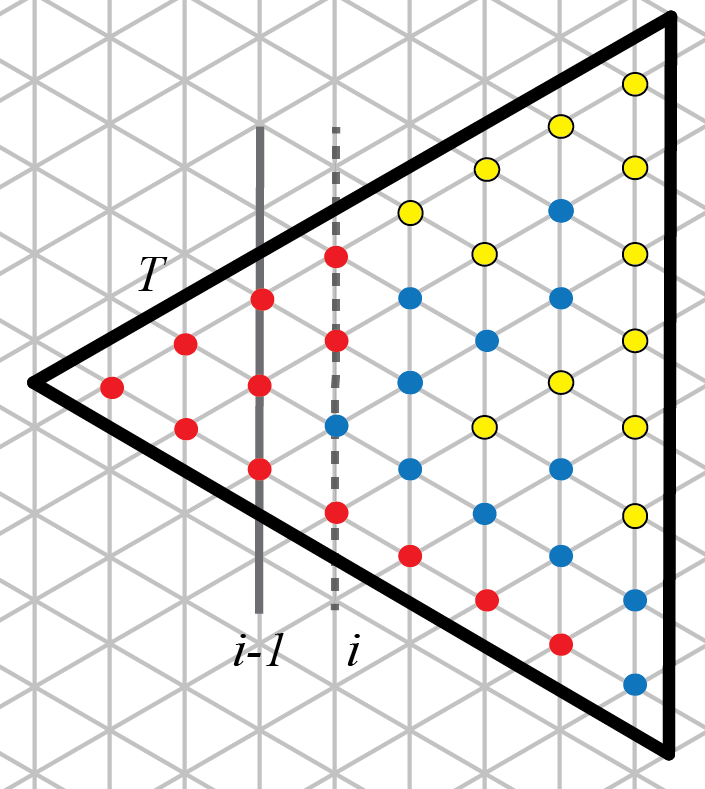}
	
	(a) \hspace{5cm}(b) \hspace{5cm}(c)
	\caption{An example of our sweep line process on a partition where $k_1 = k_2 = k_3 = 12$. $P_1$ is red, $P_2$ is blue, and $P_3$ is yellow.  (a) A balanced partition where the first $i-1$ columns are in $P_1$ but the $i^{th}$ column is not entirely within $P_1$ yet. (b) After applying a tower move, the number of vertices of $P_1$ in the $i^{th}$ column has increased, but the partition is now nearly balanced instead of balanced, with $|P_1| = k_1 + 1$ and $|P_3| = k_3 - 1$. (c) After applying Case A of our rebalancing procedure by making changes near where $P_2$ and $P_3$ are adjacent in $bd(T)$, we reach a balanced partition.}
	\label{fig:sweepline_ex}
\end{figure}

\subsection{Key facts and lemmas}

Let $bd(T)$ be the vertices in $T$ that are adjacent to a vertex outside of $T$. 
We let $N(v)$ be all neighbors of vertex $v$ in $\Gtri$, and note $N(v)$ is always a cycle of length 6.  
For $i \in \{1,2,3\}$, we say $v$'s $i$-neighborhood is $N(v) \cap P_i$, that is, all vertices in $P_i$ that are adjacent to $v$. This $i$-neighborhood is connected if $N(v) \cap P_i$ has only one connected component. While our overall proof is about recombination Markov chains, flip moves (where one vertex is assigned to a new district) are a subset of recombination moves, and we will often focus on flip moves because it makes our arguments easier. The following flip lemma describes when flip moves are possible. The simplicity of this flip lemma is a large reason why it is convenient to be working in the triangular lattice.  

\begin{lem*}[Flip Lemma, informal; Formally stated as Lemma~\ref{lem:remove-add}]
	If $P$ is a partition of $T$ into three simply connected districts, and $v \in P_i$ has a connected $i$-neighborhood and a connected, nonempty $j$-neighborhood for $j \in \{1,2,3\}$, $j \neq i$, then removing $v$ from $P_i$ and adding it to $P_j$ produces another partition of $T$ into three simply connected districts. 
\end{lem*}

\noindent The following lemma suggests in most cases, when a vertex can be removed from $P_i$, it can always be added to one of the two other districts, $P_j$ or $P_l$ with $j,l \in \{1,2,3\}$ and $i,j,l$ all distinct.  

\begin{lem*}[Alternation Lemma, informal; Formally stated as Lemma~\ref{lem:alternation}]
	Let $P$ be a partition of $T$ into three simply connected districts, and suppose $v \in P_i$ has a connected $i$-neighborhood, is not in $bd(T)$, and is adjacent to at least one vertex in a different district. Then $v$'s $j$-neighborhood or $l$-neighborhood is connected and nonempty, for $j \neq l$. 
\end{lem*}
\noindent This is called the Alternation Lemma because its proof involves showing districts $j$ and $l$ cannot alternate too much in $N(v)$: having an ordered sequence of four vertices $a$, $b$, $c$, $d$ in $N(v)$ with $a,c \in P_j$ and $b,d \in P_l$ is impossible because then $P_j$ and $P_l$ cannot both be connected. 

Not every vertex can be removed from one district and added to another while maintaining simple connectivity of all districts.  We say a simply connected subgraph $S \subset P_i$ is {\it shrinkable} if it contains a vertex that can be removed from $P_i$ and added to a different district. 
Most sets will be shrinkable, but there are two notable exceptions: If $S$ does not contain any vertices adjacent to other districts, or if $S$ is a path ending with a single vertex in $bd(T)$. Figure~\ref{fig:non-shrinkable} gives an example of each.

\begin{figure}
	\centering
	\includegraphics[scale = 0.85]{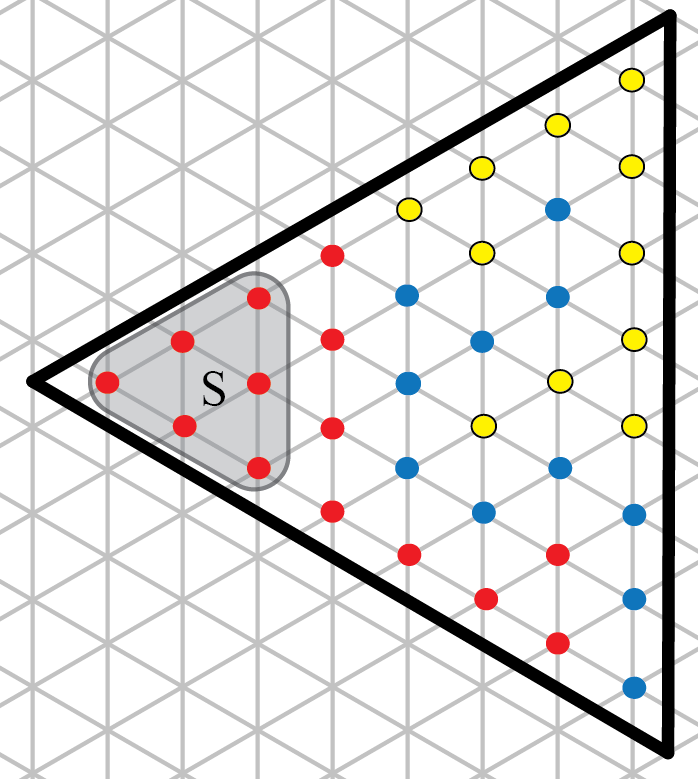}
	\hspace{1cm}
	\includegraphics[scale = 0.85]{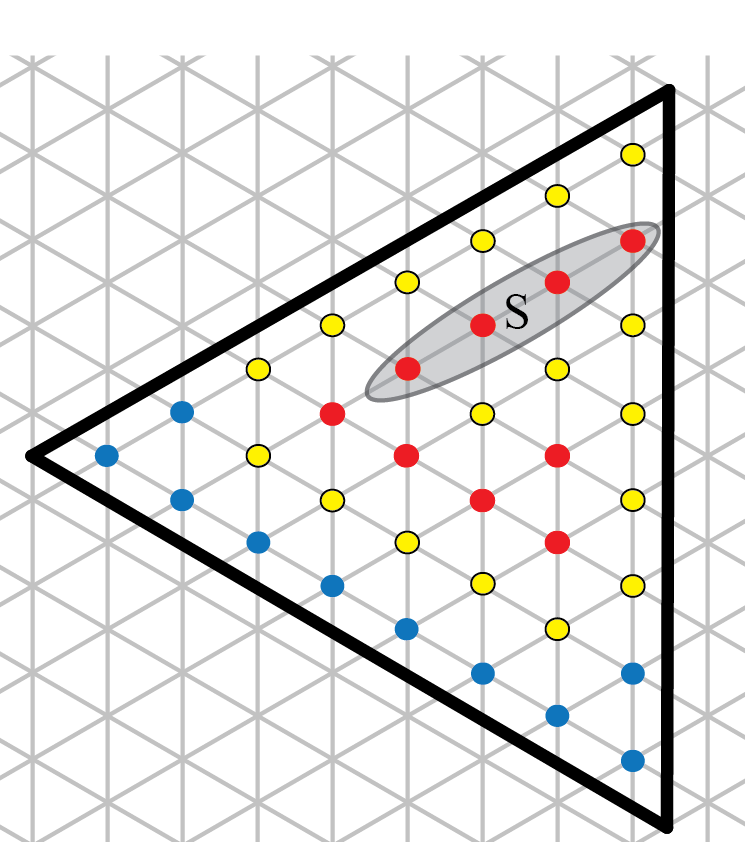}
	
	(a) \hspace{5.5cm} (b)
	\caption{Two examples of simply connected subsets $S$ (grey) of $P_1$ (red) that are not shrinkable. In (a), $S$ is not shrinkable because no vertex of $S$ is adjacent to any district besides $P_i$. In (b), $S$ is not shrinkable because the removal of any vertex except the rightmost will disconnect $P_1$, and while the rightmost vertex of $S$ can be removed from $P_1$, adding it to $P_2$ (blue) produces something not connected and adding it to $P_3$ (yellow) produces something that is not simply connected. }
	\label{fig:non-shrinkable}
\end{figure}

The following lemma gives sufficient conditions for $S \subseteq P_i$ to be shrinkable, and was crafted exactly to avoid the two non-shrinkable examples of Figure~\ref{fig:non-shrinkable}.  Note we only consider the $S$ that can be produced by removing a simply connected subset of $P_i$; this ensures, for example, that $S$ is not entirely contained in the interior of $P_i$. 

\begin{lem*}[Shrinkability Lemma, informal; Formally stated as Lemma~\ref{lem:shrinkable}] If $S \subseteq P_i$ is simply connected and $P_i \setminus S$ is simply connected, the following two conditions are each sufficient for $S$ to be shrinkable: 
\begin{enumerate}[label=(\Roman*)]
	\item $S \cap bd(T) = \emptyset$.
	\item $S$ is adjacent to a different district, and $P_i$ contains at least two vertices in $bd(T)$.  
\end{enumerate}
\end{lem*}

While the Shrinkability Lemma allows us to find a single vertex to remove, we cannot repeatedly remove vertices from the same district because this will produce partitions that are not balanced or nearly balanced.  Instead, if we wish to remove multiple vertices from a particular district, we must alternate with adding new vertices to that district somewhere else.  It is important the vertices we are adding are not adjacent to the vertices we are removing, otherwise we can't know any real progress is being made. The following Unwinding Lemma gets at this idea, where we have $S_1\subseteq P_1$ that we want to add to $P_2$ and $S_2 \subseteq P_2$ that we want to add to $P_1$. This lemma is only applied in the case where $|P_3| = k_3 - 1$, so adding a vertex to $P_3$ to bring it up to its ideal size is also considered a successful outcome. It is called the Unwinding Lemma because $S_1$ and $S_2$ are frequently long, winding arms of $P_1$ and $P_2$, respectively, that we wish to contract so our partition is less~intertwined.

\begin{lem*}
	[Unwinding Lemma, informal; Formally stated as Lemma~\ref{lem:s1s2}] Let $S_1 \subseteq P_1$ and $S_2 \subseteq P_2$ be shrinkable and not adjacent. There exists a sequence of moves after which (1) a vertex has been added to $P_3$, (2) all vertices in $S_1$ have been added to $P_2$, or (3) all vertices in $S_2$ have been added to $P_1$. 
\end{lem*}

Finally, at times we will need to work with $S_1 \subseteq P_1$ and $S_2 \subseteq P_2$ that are adjacent.  This arises when $S_1$ and $S_2$ are both inside some cycle $C$, where all vertices of $C$ are in $P_1$ except for one, $x$, which is in $P_2$. If $y$ is one of the vertices in $C$ adjacent to $x$ and $y$ is a cut vertex of $P_1$, the case where other approaches fail is when one component of $P_1 \setminus y$ is inside $C$. We would like $y$'s 1-neighborhood to be connected but the component $S_1$ of $P_1 \setminus y$ that is inside $C$ prevents that from happening; see Figure~\ref{fig:cycle-recom}(a) for an example. Instead of removing vertices from $S_1$ one at a time, we rearrange the entire interior of $C$ with one recombination step. 

\begin{lem*}[Cycle Recombination Lemma, informal; Formally stated as Lemma~\ref{lem:cycle-recom}]
	Let $C$ be a cycle of vertices in $P_1$ with one vertex, $x$, in $P_2$. Suppose no vertices of $P_3$ are inside $C$.  There exists one recombination step, changing only district assignments of vertices enclosed by $C$, after which $x$'s neighbor $y$ in $C$ has a connected 1-neighborhood (at least when looking in or inside $C$). 
\end{lem*}

\begin{figure}\centering
	\includegraphics[scale = 0.87]{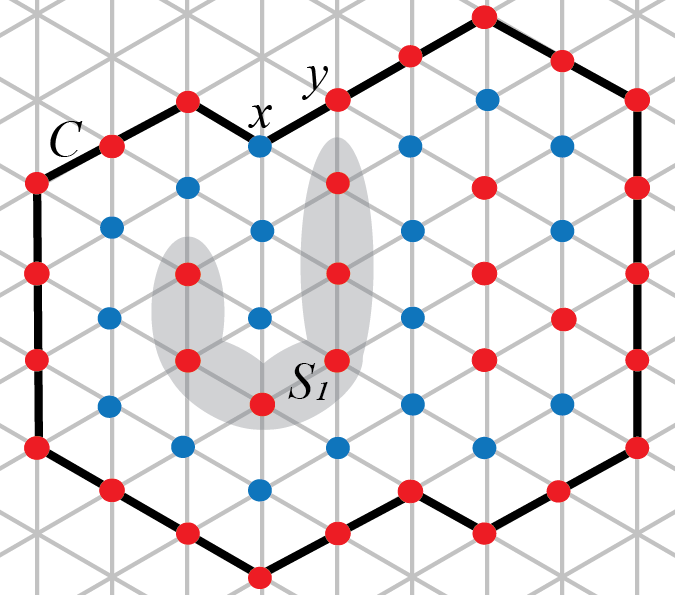}\hfill
	\includegraphics[scale = 0.87]{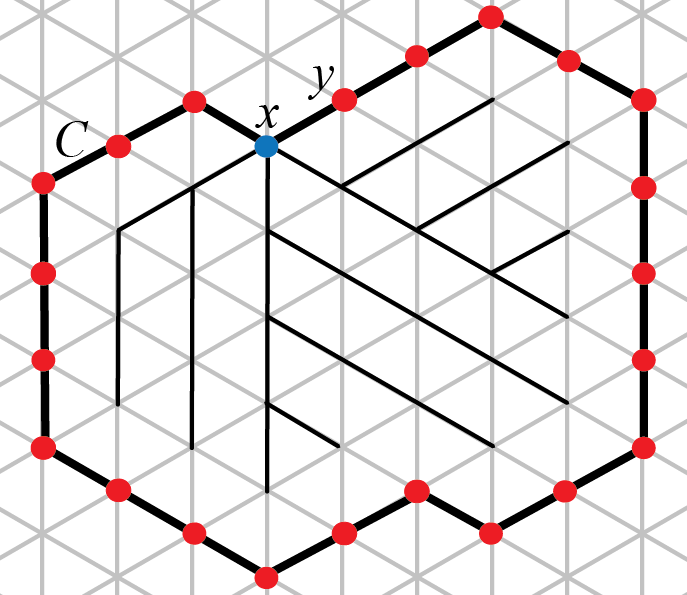}\hfill
	\includegraphics[scale = 0.87]{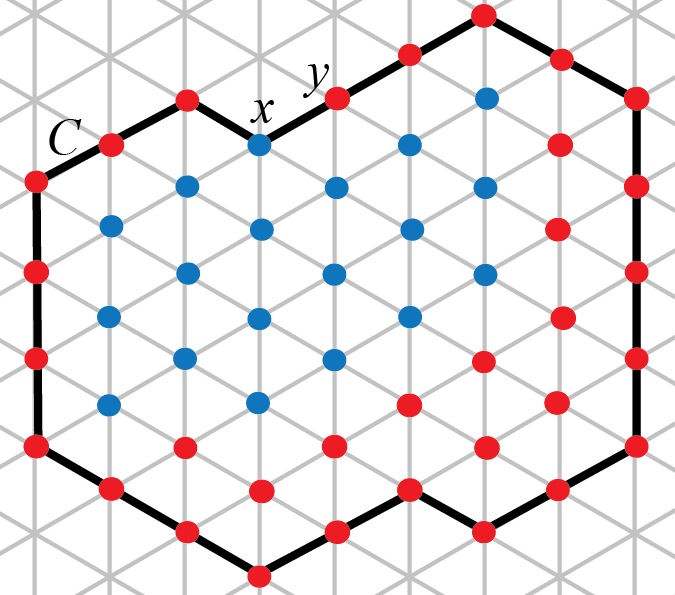}
	
	(a) \hspace{4.5cm}(b) \hspace{4.5cm}(c)
	\caption{An example application of the Cycle Recombination Lemma. (a) An example satisfying the hypotheses of the lemma: A cycle $C$ of vertices in $P_1$ (red) plus one vertex $x$ in $P_2$ (blue), such that $x$'s neighbor $y$ in $C$ is a cut vertex of $P_1$ and $P_1 \setminus y$ has a component $S_1$ (grey) inside $C$. The district assignments of vertices outside $C$ are not shown. (b) All district assignments inside $C$ are erased, and we build a breadth first search tree of the vertices inside $C$. (c) If initially there were $m$ vertices of $P_1$ inside $C$, the last $m$ vertices added to the BFS tree are added to $P_1$ while the remaining vertices are added to $P_2$. After this process $y$ will have a connected 1-neighborhood.}
	\label{fig:cycle-recom}
\end{figure}

\noindent The main idea of this recombination step is to erase all district assignments of vertices enclosed by $C$ and build a breadth first search tree of the interior of $C$. If initially there were $m$ vertices of $P_1$ enclosed by $C$, the last $m$ vertices added to the BFS tree are added to $P_1$ while the remaining vertices are added to $P_2$. 
An example of this process is shown in Figure~\ref{fig:cycle-recom}. Because the vertices in $N(y)$ that are inside $C$ monotonically increase in their distance from $x$, after this recombination step $N(y)$ will consist of $x$, followed by some vertices in $P_2$, followed by some vertices in $P_1$, followed by $y$'s other neighbor in $C$. While the lemma does not say anything about the parts of $N(y)$ that are outside $C$, we will apply it in cases where knowing $y$'s 1-neighborhood in or inside $C$ is connected implies $y$'s entire 1-neighborhood is connected.

\subsection{Advancing toward Ground State: Towers}

In our sweep line procedure, the two main steps are (1) increase the number of vertices in $\cc_i \cap P_1$ and (2) if necessary, transform the result from a nearly balanced to a balanced partition without decreasing $\cc_i \cap P_1$. The way we achieve (1) is using {\it towers}.  For a particular vertex in $\cc_i$, we may want to add it to $P_1$ but be unable to because doing so produces a partition with districts that are not simply connected. Let $v_1$ be a vertex in $\cc_i$ that is not in $P_1$ but adjacent to a vertex of $P_1 \cap \cc_i$, and suppose without loss of generality that $v_1 \in P_2$; see Figure~\ref{fig:tower-ex}(a) for an example. This means $v_1$ has three neighbors in $P_1$, and it's the middle of its other three neighbors that is crucial for determining whether $v_1$ can be added to $P_1$ or not.  If it can't, we then examine this middle neighbor, which must have a similar neighborhood structure to $v_1$.  This process is repeated and must eventually end at a vertex that can be added to the district of the vertex before it in the tower. Flip moves are then made all the way back up the tower, ultimately producing a configuration in which $v_1$ can be added to $P_1$. We do not state our tower lemma, or even formally define a tower, because of the technical details involved, but depict a sample application of the tower procedure in Figure~\ref{fig:tower-ex}. 

\begin{figure}
	\centering
	\hfill
	\includegraphics[scale = 1]{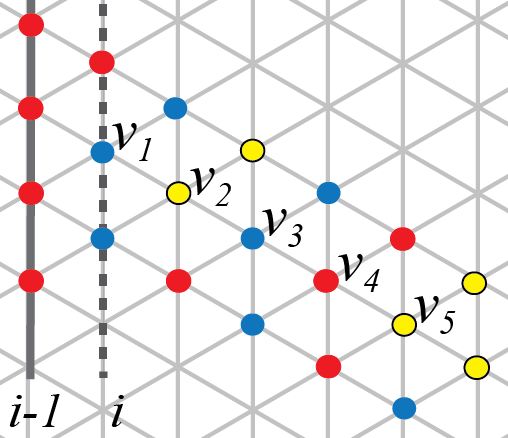}\hfill
	\includegraphics[scale = 1]{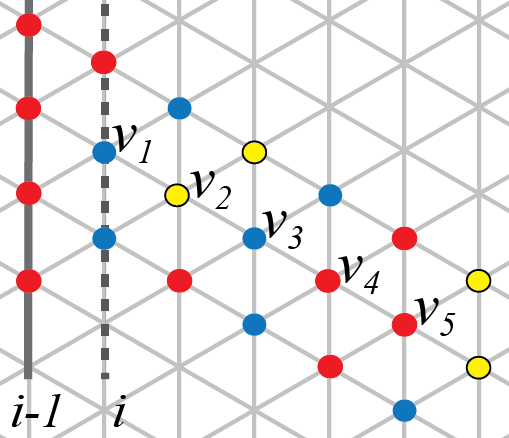}\hfill
	\includegraphics[scale = 1]{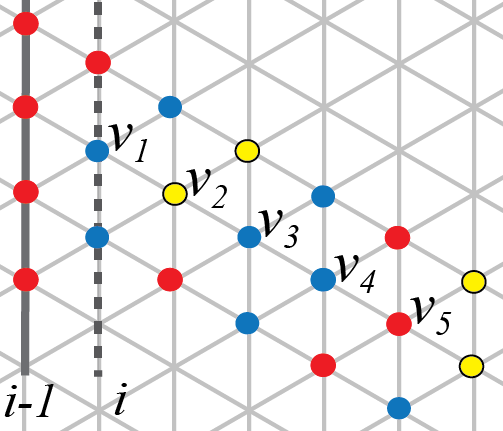} \hfill\textcolor{white}{.} \\
	
	(a) \hspace{4.5cm} (b) \hspace{4.5cm} (c)
	
	\hfill
	\includegraphics[scale = 1]{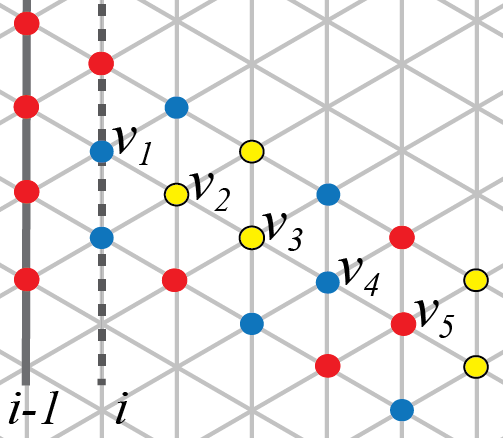}\hfill
	\includegraphics[scale = 1]{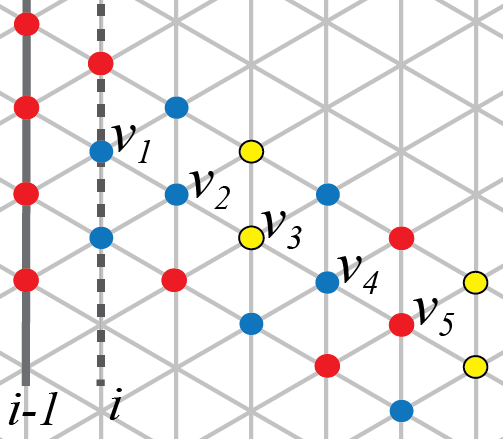}\hfill
	\includegraphics[scale = 1]{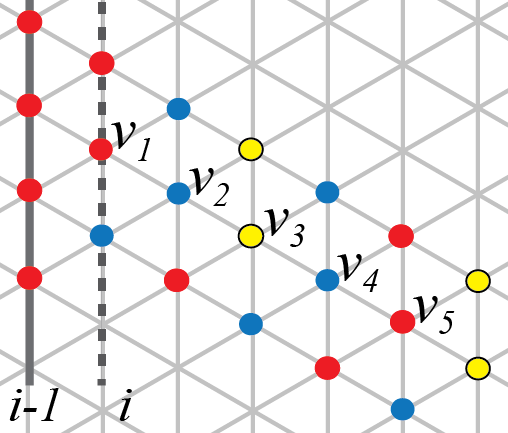} \hfill\textcolor{white}{.} \\
	
	(d) \hspace{4.5cm} (e) \hspace{4.5cm} (f)
	
	\caption{An example where we wish to add $v_1 \in P_2$ (blue) to $P_1$ (red), but cannot because this would disconnect $P_2$. Instead we look at the middle of $v_1$'s three neighbors not in $P_1$, which we call $v_2$, and see if it can be added to $P_2$; in this example, it can't because doing so would create a cycle in $P_2$. We continue along the line spanned by $v_1$ and $v_2$ until we find a vertex that can be added to the district of the vertex before it, which we prove must eventually happen. (a) In this example, vertex $v_5$ can be added to the district of $v_4$, and the result of this move is shown in (b). Now $v_4$ can be added to the district of $v_3$, and the result is (c).  The same procedure for $v_3$, $v_2$, and $v_1$ gives the results shown in (d), (e), and (f), respectively. The end result is that there is one additional vertex of $P_1$ ($v_1$) in column $i$. 
	}
	\label{fig:tower-ex}
\end{figure}

\subsection{Rebalancing: Cases}

Performing a tower move as described in the previous section increases the number of vertices of $P_1$ in $\cc_i$, but can also move us from a balanced partition to a nearly balanced partition, because the number of vertices in $P_1$ has increased by one.  Before proceeding further, we must return to a balanced partition, and do so without decreasing the number of vertices in $P_1 \cap \cc_i$. This is the most challenging part of the proof.  

Without loss of generality, we suppose we have $|P_2| = k_2$ and $|P_3| = k_3 - 1$. We can use the Shrinkability Lemma to always find a vertex of $P_1$ that can be removed from $P_1$ and added to a different district; however, it may be that all such vertices can only be added to $P_2$, not $P_3$, and additionally all vertices that can be removed from $P_2$ can only be added to $P_1$.  In these cases we need to do some rearranging with $P_1$ (in columns $i+1$ and greater) and $P_2$ before finding a vertex that can be added to $P_3$ to reach a balanced partition.  It is in this rearranging that the Unwinding Lemma and the Cycle Recombination Lemma play crucial roles. 


Our proof considers four main cases for the rebalancing process, depending on the type of adjacency between $P_2$ and $P_3$: (A) There exists $a \in P_2 \cap bd(T)$ and $b \in P_3 \cap bd(T)$ that are adjacent; (B) $P_2 \cap bd(T) = \emptyset$; (C) $P_3 \cap bd(T) = \emptyset$; and (D) No vertex of $P_2$ is adjacent to any vertex of $P_3$.  
See Figure~\ref{fig:4cases} for cartoonish examples of the four cases. We prove these four cases are disjoint and cover all possibilities, and consider each separately, though there are certainly common elements between their proofs.  Cases (A) and (D) are the most straightforward because we only need to consider reassigning vertices near $bd(T)$, while (B) and especially (C) are more challenging because we must work in the interior of $T$, far from $bd(T)$. 

\subsection{Reaching a Ground State}

After performing our sweep line procedure, for some $i$ we have $P_1$ occupying all of the first $i-1$ columns, some of $\cc_i$, and none of columns $i+1$ or greater. At this point we describe a sequence of steps alternating recombining $P_2$ and $P_3$ with moving a vertex of $P_1$ higher in column $i$. The end result is a partition where the vertices of $P_1$ in $\cc_i$ occupy all of the topmost positions in $\cc_i$, as they must in the ground state $\sigma_{123}$. One final recombination step for $P_2$ and $P_3$ reaches the ground state $\sigma_{123}$.  Because the ground states are all easily connected by recombination moves, this proves there exists a sequence of moves from any balanced partition to any other balanced partition, moving through balanced and nearly balanced partitions. 

The proof outlined so far assumes we begin at a balanced partition. Some additional work is required to show any nearly balanced partition can be transformed into a balanced partition (Lemma~\ref{lem:nearlybalanced}).  Similar lemmas and approaches to the rebalancing step described above are used to do so, completing the~proof.









\section{Preliminary Lemmas and Towers}

 %
%




We now begin to formalize some of the notions describe above in the proof overview.  We will first show that for every balanced partition, there is a sequence of steps in $G_\Omega$ (possibly passing through some nearly balanced partitions) leading to a ground state, which is the bulk of our proof. We then show that for every nearly balanced partition, there exists a sequence of steps producing a balanced partition.

First, we present some definitions and lemmas we will use throughout this paper. Though recombination moves can change the district assignments of many nodes at once, our sweep-line argument will focus on changing the district assignment of one vertex at a time.  Because of this, moves reassigning one vertex to a new district (which are just one type of simple recombination step) will play an important role. 



\subsection{$i$-neighborhoods and the Flip Lemma} 

Central to our arguments will be the notion of a vertex's $i$-neighborhood.

\begin{defn}
	Let $P$ be a partition of $T$. For a vertex $v\in T$, its {\em $i$-neighborhood} is all neighbors of $v$ in $T$ that are in $P_i$.
\end{defn}

\noindent At times we will concretely refer to a vertex's $1$-neighborhood, $2$-neighborhood, or $3$-neighborhood; this will mean all neighbors of the vertex that are in $P_1$, $P_2$, or $P_3$, respectively, not the vertices at distance 1, 2, or 3 away. 

 \begin{defn}
 	For a set of vertices $Q$ in $T$, the {\em neighborhood} of $Q$, $N(Q)$, is all vertices in $\Gtri$ that are not in $Q$ but are adjacent to a vertex of $Q$. 
 \end{defn}

Note that we include in $N(Q)$ any vertices that may be adjacent to a vertex of $Q$ but not in $T$; if we want the neighbors of $Q$ that are also in $T$, we will explicitly clarify $N(Q) \cap T$. For a single vertex $v$, $N(v)$ is always a cycle of length 6.  

These next two lemmas give conditions under which one can remove a vertex from a district or add a vertex to a district while maintaining simple connectivity. 

\begin{lem}\label{lem:remove}
	Let $P$ be a partition of $T$. For a vertex $v \in P_i$, if $v$'s $i$-neighborhood is connected and of size at most 5, then $P_i \setminus \{v\}$ is simply connected. 
\end{lem}
\begin{proof} Recall that by a partition of $T$, we mean a partition of $T$ into three simply connected districts, so $P_i$ is simply connected.
	
	First we show $P_i \setminus \{v\}$ is connected. Because $P_i$ is connected, this means for any pair of vertices $x$ and $y$ in $P_i$, there exists a path between them consisting entirely of vertices in $P_i$. If that path passes through $v$, let $n_1$ be the neighbor of $v$ that is before $v$ on this path, and let $n_2$ be the neighbor of $v$ that is after $v$ on this path. Because the $i$-neighborhood of $v$ is connected, there exists a path in $P_i$ from $n_1$ to $n_2$ through $N(v) \cap P_i$.  Replacing the path segment $n_1 - v - n_2$ with this path from $n_1$ to $n_2$ through $N(v) \cap P_i$ results in a walk from $x$ to $y$ that does not pass through $v$.  Because all pairs of vertices $x$ and $y$ are connected by walks that do not pass through $v$, $P_i \setminus \{v\}$ is connected. 
	
	Because $v$ has at most 5 neighbors in $P_i$, removing it cannot possibly create a hole in $P_i$ that was not there before.  We conclude $P_i \setminus \{v\}$ is simply connected.
\end{proof}

\noindent The inverse of this lemmas is also true. 

\begin{lem}
	\label{lem:cutvertex}
	Let $P$ be a partition of $T$. For $v \in P_i$, if $v$'s $i$-neighborhood is not connected then $P_i \setminus \{v\}$ is not connected. 
\end{lem}
\begin{proof}
	Let $v \in P_i$ be such that $v$'s $i$-neighborhood is not connected. Let $a$ and $b$ be two neighbors of $v$ in $P_i$ that are not connected by a path in $N(v) \cap P_i$. This means each of the two paths from $a$ to $b$ in $N(v)$ must either contain a vertex not in $T$ or a vertex in $T$ that is in a different district of $P_i$.  Suppose for the sake of contradiction that $P_i \setminus \{v\}$ is connected. Then, there would need to exist a path from $a$ to $b$ in $P_i \setminus \{v\}$.  We have already seen such a path must leave the neighborhood of $v$. 
	Consider the cycle formed by this path from $a$ to $b$ together with $v$. This cycle is entirely contained in $P_i$, and necessarily encircles one of the two paths from $a$ to $b$ in $N(v)$. This means the cycle encircles a vertex not in $P_i$, a contradiction as $P$ is a valid partition and so $P_i$ must be simply connected. We conclude that $P_i \setminus \{v\}$ cannot be connected, and $v$ must be a cut vertex of~$P_i$. 	
\end{proof}

The two previous lemmas show us that $v$ is a cut vertex of $P_i$ if and only if $v$'s $i$-neighborhood is disconnected.

\begin{cor}\label{cor:cutvertex}
	Let $P$ be a partition of $T$.  Vertex $v \in P_i$ is a cut vertex of $P_i$ if and only if its $i$-neighborhood is disconnected. 
\end{cor}
\begin{proof}
	The two directions of this if and only if statement are Lemma~\ref{lem:remove} and Lemma~\ref{lem:cutvertex}.
\end{proof}
For here on, we will interchangeably use the conditions $v$ being a cut vertex of $P_i$ and $v$ having a disconnected $i$-neighborhood. 
Having considered the remove of a vertex from district $P_i$, we now consider the addition of a vertex to~$P_i$. 

\begin{lem}\label{lem:add}
	Let $P$ be a partition of $T$. For a vertex $v$ of $T$ that is not in $P_i$, if $v$'s $i$-neighborhood is nonempty and connected, then $P_i \cup \{v\}$ is simply connected. 
\end{lem}
\begin{proof}
	Because $P$ is a partition, we know that $P_i$ is simply connected. Because $v$ is adjacent to at least one vertex in $P_i$, adding it to $P_i$ cannot disconnect $P_i$, so $P_i \cup \{v\}$ is connected.  It only remains to show that $P_i \cup \{v\}$ does not surround any vertices that are not in $P_i \cup \{v\}$.  
	
	Suppose $P_i \cup \{v\}$ surrounds some vertex $x \notin P_i \cup \{v\}$.  This means there exists a cycle $C$ consisting of vertices in $P_i \cup \{v\}$ that encircles $x$. Because $P_i$ is simply connected, $C \subseteq P_i$ is not possible, so $v$ must be one of the vertices of $C$.  Let $a$ be the vertex before $v$ in $C$, and let $b$ be the vertex after $v$ in $C$.  Because $a$ and $b$ are both in $v$'s $i$-neighborhood and $v$'s $i$-neighborhood is connected, there must be a path $Q$ from $a$ to $b$ in $N(v)\cap P_i$.  There is then a cycle $C'$ which is the same as $C$ except $v$ is replaced by this path $Q$ from $a$ to $b$ in $N(v) \cap P_i$. Because $x \notin P_i$ and therefore $x$ is not one of the vertices in $Q$, $C'$ also surrounds $x$.  Because $C'$ is a cycle in $P_i$ that surrounds a vertex $x \notin P_i$, this implies $P_i$ is not simply connected, a contradiction.  It follows that $P_i \cup \{v\}$ does not have any holes and therefore is simply connected. 	
\end{proof}

Lemmas~\ref{lem:remove} and~\ref{lem:add} together show us that if $x \in P_i$ has a connected $i$-neighborhood and a connected $j$-neighborhood, then reassigning $x$ from $P_i$ to $P_j$ results in a valid partition (this partition may not be balanced or nearly balanced); this is formalized as follows, which is the formal version of the Flip Lemma presented in the proof overview.

\begin{lem}[Flip Lemma]\label{lem:remove-add}
	Let $P$ be a partition of $T$, and suppose $v \in P_i$ and $v$ is adjacent to a vertex in $P_j$ for $j \neq i$.  If $v$'s $i$-neighborhood and $j$-neighborhood are both connected, then removing $v$ from $P_i$ and adding it to $P_j$ produces a valid partition.  
\end{lem}
\begin{proof}
	Because $v$ has a neighbor in $P_j$, it's $i$-neighborhood is connected and of size at most 5. By Lemma~\ref{lem:remove}, when $v$ is removed from $P_i$ then $P_i \setminus \{v\}$ is simply connected. Because $v$'s $j$-neighborhood is nonempty and simply connected, by Lemma~\ref{lem:add}, $P_j \cup \{v\}$ is simply connected.  Thus this is a valid move. 
\end{proof}



\subsection{Finding a single vertex to reassign: Alternation and Shrinkability Lemmas} 

In order to know whether a vertex that can be removed from $P_i$ can be added to another district, we have to first know that it's adjacent to another district. 

\begin{defn}
	Let $P$ be a partition of $T$. A vertex in $P_i$ is {\em exposed} if it is adjacent to a vertex in a different district $P_j$, $j \neq i$. 
\end{defn}

We begin by showing that if a connected component of $P_i$ has an exposed vertex, then it has an exposed vertex whose $i$-neighborhood is connected. 

\begin{lem}\label{lem:remove-W}
Let $P$ be a partition of $T$, and let $P_i$ be one district of $P$. Let $W \subseteq P_i$ be a simply connected subset of $P_i$. Let $S$ be a connected component of $P_i \setminus W$ that has at least one exposed vertex. Then there exists an exposed vertex $x \in S$ such that $P_i \setminus x$ is simply connected.  
\end{lem}
\begin{proof}
	This proof proceeds by strong induction on $|S|$. 
	
	First, suppose $|S| = 1$.  This means $S$ consists of exactly one vertex, $x$. Because $S$ contains an exposed vertex, $x$ must be exposed. We claim that $N(x) \cap W$ is connected. To see this, suppose that $N(x) \cap W$ is not connected, and let $a$ and $b$ be two vertices in different connected components of $N(x) \cap W$.  Because $W$ is connected and $x \notin W$, there must be a path $Q$ in $W$ from $a$ to $b$.  Together, $x$ and $Q$ form a cycle that is entirely within $P_i$, and encircles some vertex of $N(x)$ that is not in $W$. Note because $|S| = 1$, $x$ cannot have any neighbors in $P_i$ that are not in $W$, as these vertices would then also be in connected component $S$ and we would have $|S| > 1$. This means the vertex of $N(x)$ that is not in $W$ and is encircled by $Q \cup x$ is not in $P_i$.  This contradicts the simple connectivity of $P_i$. Thus it must be that $N(x) \cap W$ is connected. Because $W$ is simply connected, $|N(x) \cap W| \neq 6$, otherwise $x$ would be a hole in $W$.  It follows from Lemma~\ref{lem:remove} that $P_i \setminus \{x\}$ is simply connected. This concludes the proof when $|S| = 1$. 
	
	Next, suppose $|S| > 1$.  Consider all exposed vertices in $S$, which by assumption is a nonempty set. If $S$ has an exposed vertex that is not a cut vertex of $P_i$, by Lemma~\ref{lem:remove} and Corollary~\ref{cor:cutvertex}, $P_i \setminus x$ is simply connected and we are done.  Otherwise, let $w$ be any exposed cut vertex of $S$. Consider $P_i \setminus \{w\}$, and look at any connected component $S'$ of $P_i \setminus \{w\}$ that does not contain $W$.  Note it must be that $S' \subseteq S$, and because $w \in S \setminus S'$, $|S'| < |S|$. Furthermore, look at $N(w) \cap S'$, which must be connected.  Look at the two vertices in $N(w)$ that are not in $S'$ but are adjacent to $S'$; these vertices cannot be in $P_i$, because if they were they would be in $S'$. While one of these vertices could be outside of $T$, because of the shape of $T$ it is not possible that both are.  Thus at least one vertex in $S' \cap N(w)$ is adjacent to a vertex in a different district, so $S'$ has an exposed vertex.  By the induction hypothesis, using $W' = w$ and $S'$ in place of $W$ and $S$, $S'$ has an exposed vertex $x$ such that $P_i \setminus x$ is simply connected.  Because $S' \subseteq S$, this vertex $x$ proves the lemma for $S$ as well.
\end{proof}

When $W = \{w\}$ is a cut vertex of $P_i$, it is not necessary to check the extra condition that $S$ contains an exposed vertex as we can show it always will, just as we did in the proof of the previous Lemma.  

\begin{lem}\label{lem:remove-W1}
	Let $P$ be a partition of $T$, and let $P_i$ be one district of $P$. Let $w \in P_i$ be a cut vertex of $P_i$. For any connected component $S$ of $P_i \setminus w$, there exists an exposed vertex $x \in S$ such that $P_i \setminus x$ remains simply connected. 
\end{lem}
\begin{proof}
	Because $w$ is a cut vertex of $P_i$, $N(w)\cap P_i$ must have at least two connected components. Becuase of this, $w$ can't be a corner vertex of $T$, and so $w$ has at most two neighbors outside $T$, and any such neighbors must be adjacent. 	
	We also note that $N(w)$, as defined, is always a cycle of length 6, though this cycle may include some vertices outside $T$.
	
	Let $S$ be any connected component of $P_i \setminus w$, and let $s$ be any vertex in $S \cap N(w)$. Let $S'$ be any other component of $P_i \setminus w$, and let $s' \in S' \cap N(w)$.  Examine both paths from $s$ to $s'$ in $N(w)$: at most one can leave $T$, so at least one must be entirely contained in $T$. Let $Q$ be a path from $s$ to $s'$ in $N(w)$ that is entirely contained in $T$. Because $s \in S$ and $s' \in S'$ are in different connected components of $P_i \setminus w$, path $Q$ must contain at least one vertex $q \notin P_i$. The last vertex along $Q$ before $q$ is in $S$ and is adjacent to a vertex of $T$ that is not in $P_i$, meaning it is an exposed vertex. Because $S$ contains an exposed vertex, we can apply Lemma~\ref{lem:remove-W} with $W = \{w\}$, and we conclude there exists an exposed vertex $x \in S$ such that $P_i \setminus x$ is simply connected. 
\end{proof}

The previous lemmas identify a vertex $x$ that can be removed from $P_i$ while keeping $P_i$ simply connected; the next lemma will help with understanding when that vertex $x$ can be added to another district. Now that we will be regularly discussing multiple districts at the same time, without loss of generality we focus on removing a vertex from $P_1$ and adding it to $P_2$ or $P_3$. Recall a vertex's $1$-neighborhood is all of its neighbors that are in set $P_1$, not all vertices at distance 1 away, and similarly for a vertex's 2-neighborhood and 3-neighborhood.  We say a vertex of $T$ is a {\em boundary vertex} if it has neighbors in $\Gtri$ that are outside of $T$. We use $bd(T)$ to denote all boundary vertices of $T$. 
The following is the formal statement of the Alternation Lemma.

\begin{lem}[Alternation Lemma]\label{lem:alternation}
	Let $P$ be a partition of $T$, and let $x$ be a vertex of $P_1$ that is not in $bd(T)$, has a connected $1$-neighborhood, and is exposed. Then $x$'s $2$-neighborhood or $3$-neighborhood is connected and nonempty. 
\end{lem}
\begin{proof}
	First, suppose $x$ is adjacent to vertices of only one of $P_2$ and $P_3$, and without loss of generality suppose it is $P_2$. Because $x$ is not adjacent to the boundary of $T$, all of $x$'s neighbors that are not in $P_1$ must be in $P_2$.  As $x$'s $1$-neighborhood is connected, $P_2 \cap N(x)$ must also be a connected set. This implies $x$'s $2$-neighborhood is connected and nonempty, proving the lemma. 
	
	Next, suppose that $x$ is adjacent to vertices of both $P_2$ and $P_3$. Suppose, for the sake of contradiction, that $x$'s 2-neighborhood is disconnected and $x$'s $3$-neighborhood is disconnected. Because $x$'s $2$-neighborhood is disconnected, this means $x$ must have two neighbors, $a$ and $b$, that are both in $P_2$ but are not connected by a path in $N(x) \cap P_2$.  There are two paths from $a$ to $b$ in $N(x)$, and because $x$'s $1$-neighborhood is connected, at most one of these paths can contain a vertex $y \in P_1$.  The other path must contain at least one vertex $c$ that is not in $P_1$ and not in $P_2$.  This means $c \in P_3$, because  it is impossible for $N(x)$ to contain a vertex outside of $T$ as $x \notin bd(T)$. See Figure~\ref{fig:alternation}(a). 
	
	\begin{figure}
		\centering
		\includegraphics{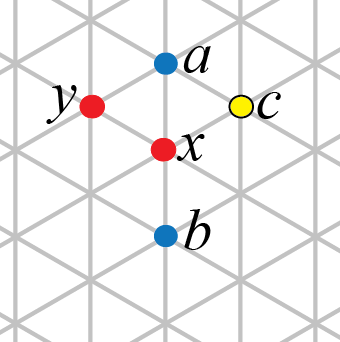} \hspace{1cm}
		\includegraphics{"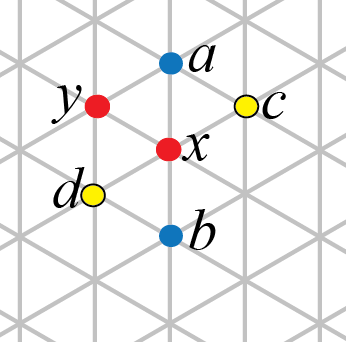"} \hspace{1cm}
		\includegraphics{"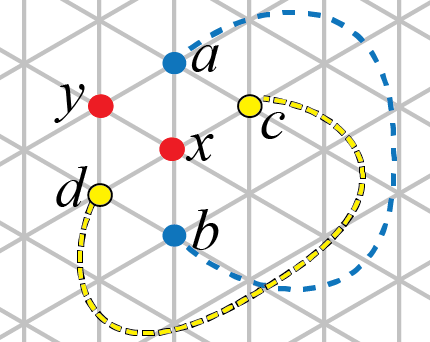"}\\
		(a) \hspace{3.5cm} (b) \hspace{4cm} (c)  
		\caption{Examples from the proof of Lemma~\ref{lem:alternation}. (a) Vertex $x \in P_1$ (red), with neighbor $y$ in $P_1$  and neighbors $a$ and $b$ in different components of $x$'s 2-neighborhood (blue). Because $x$'s 1-neighborhood is connected and $x \notin bd(T)$, one of the two paths from $a$ to $b$ in $N(x)$ must contain a vertex $c \in P_3$ (yellow). (b) If $x$'s 3-neighborhood is also disconnected, then there must be a another vertex $d \in P_3 \cap N(x)$, and each path from $c$ to $d$ in $N(x)$ must contain a vertex of $P_1$ or $P_2$.  (c) Because $P_2$ and $P_3$ must `alternate' around $x$, it's impossible for $a$ and $b$ to be connected by a path in $P_2$ while $c$ and $d$ are simultaneously connected by a path in $P_3$.  }
		\label{fig:alternation}
	\end{figure}
	
 Because $x$'s $3$-neighborhood is disconnected, this means $x$ must have another neighbor $d \in P_3$, and at least one of the paths from $c$ to $d$ in $N(x)$ must not contain any vertices of $P_1$.  For the same reasons as above, this means this path, which we will call $Q$, must contain a vertex of $P_2$. However, $Q$ cannot contain both $a$ and $b$ because $c$ is on the path from $a$ to $b$ in $N(x)$ that avoids $P_1$. This means there exists at least one vertex of $P_2$ in $Q$ and at least one vertex of $P_2$ in $N(x)$ but not in $Q$. This implies $P_2$ and $P_3$ alternate around $N(x)$, though not necessarily consecutively, as shown in the example in Figure~\ref{fig:alternation}(b). 
 
  Consider the cycle formed by a path from $c$ to $d$ in $P_3$ (which exists because $P_3$ is connected) together with $x$; this path from $c$ to $d$ in $P_3$ is shown as a yellow dashed line in Figure~\ref{fig:alternation}(c). This cycle contains no vertices of $P_2$, but has at least one vertex of $P_2$ inside it and at least one vertex of $P_2$ outside it.  This implies $P_2$ is disconnected, a contradiction.  We conclude $x$'s $2$-neighborhood and $3$-neighborhood can't both be disconnected, proving the lemma.
\end{proof}



We now show how one can find a single vertex that can be removed from one district and added to another district.  We informally defined what it means for a component to be {\it shrinkable} in the proof overview, but here we include the formal definition. 

\begin{defn}\label{defn:shrinkable}
	Let $P$ be a partition of $T$, and let $P_i$ be one district of $P$. Let $W \subseteq P_i$ be a simply connected subset of $P_i$. Let $S$ be a connected component of $P_i \setminus W$.  We say $S$ is {\em shrinkable} if it contains a vertex that can be removed from $P_i$ and added to another district, producing a valid partition. 
\end{defn}


Recall Figure~\ref{fig:non-shrinkable} gives two different examples of sets that are not shrinkable: one because it contains no exposed vertices, and the other because $S$ is a path ending at $bd(T)$. The informal Shrinkability Lemma stated in the proof overview gave two sufficient conditions for $S$ to be shrinkable: (I) $S \cap bd(T) = \emptyset$ or (II) $S$ contains an exposed vertex and $|P_i \cap bd(T)| \geq 2$.  While these remain the two main conditions we will use, we also present Conditions~\ref{item:exp_corner}-\ref{item:cut_corner} which are sufficient to imply Condition~\ref{item:exp_2bd} and included here to make future applications of this lemma more straightforward. 

\todoo{It doesn't seem like we need the hypothesis that $S$ is simply connected? Tentatively it seems alright to remove this hypothesis but will leave for now. Check and don't seem to need this hypothesis in Lemma~\ref{lem:remove-W}.  We do need the hypothesis that $W$ is simply connected for Condition~\ref{item:nobd}, but not for any of the other conditions. }

\begin{lem}[Shrinkability Lemma]\label{lem:shrinkable}
	Let $P$ be a partition of $T$, and let $P_i$ be one district of $P$. Let $W \subseteq P_i$ be a simply connected subset of $P_i$. Let $S$ be a simply connected component of $P_i \setminus W$.  Each of the following is a sufficient condition for $S$ to be shrinkable: 
	\begin{enumerate}[label=(\Roman*)]
		\item \label{item:nobd} $S \cap bd(T) = \emptyset$
		\item \label{item:exp_2bd} $S$ contains an exposed vertex and $P_i$ contains at least two vertices in $bd(T)$
		\item \label{item:exp_corner} $S$ contains an exposed vertex and $P_i$ contains a corner of $T$
		\item \label{item:cut_2bd} $W$ is a cut vertex and $P_i$ contains at least two vertices in $bd(T)$
		\item \label{item:cut_corner} $W$ is a cut vertex and $P_i$ contains a corner of $T$.
	\end{enumerate}
\end{lem}
\begin{proof}
First, consider Condition~\ref{item:nobd}.  Consider $N(S)$, all vertices of $\Gtri$ not in $S$ but adjacent to a vertex of $S$. Because $W$ is simply connected, $N(S) \setminus W$ is nonempty; if all vertices of $N(S)$ were in $W$, $W$ would necessarily contain a cycle encircling $S$, a contradiction. We also note that because $S \cap bd(T) = \emptyset$, $N(S) \subseteq T$.  This means $N(S)$ must contain at least one vertex in a district that is not $P_i$. This means $S$ has at least one exposed vertex, and so by Lemma~\ref{lem:remove-W}, this means there is an exposed vertex $x \in S$ such that $P_i \setminus x$ is simply connected. By Lemma~\ref{lem:alternation}, appropriately interchanging the roles of districts 1, 2, and 3, then for $j,k \neq i$, $x$'s $j$-neighborhood or $k$-neighborhood is connected and nonempty.  By Lemma~\ref{lem:remove-add}, this means $x$ can be removed from $P_i$ and added to $P_j$ or $P_k$, meaning $S$ is shrinkable, as desired. 

Next, consider Condition~\ref{item:exp_2bd}. By Lemma~\ref{lem:remove-W}, there exists an exposed vertex $x \in S$ such that $P_i \setminus x$ is simply connected.  If $x \notin bd(T)$ then $S$ is shrinkable by Lemma~\ref{lem:alternation} as above.  If $x \in bd(T)$, let $P_j$ and $P_k$ be the other two districts, with $j \neq i \neq k$. 
Because $x$ is exposed, $x$ has a neighbor in $P_j$ or $P_k$.  If $x$'s $j$-neighborhood or $k$-neighborhood is nonempty and connected, then we are done by Lemma~\ref{lem:remove-add}. If $x$ is adjacent to both $P_j$ and $P_k$, then because $x$ must have at least one neighbor in $P_i$ and $x$ has at most four neighbors in $T$, in at least one of $P_j$ and $P_k$ it must have exactly one neighbor. This neighborhood is thus connected and we are done. 

All that remains is to prove Condition~\ref{item:exp_2bd} is sufficient for shrinkability in the case where $x$ has neighbors in exactly one of $P_j$ and $P_k$; without loss of generality suppose it is $P_j$. Because we are done if $x$'s $j$-neighborhood is connected, we assume $x$ has two neighbors, $a$ and $b$, that are both in $P_j$ but are not connected by a path in $N(x) \cap P_j$.  There are two paths from $a$ to $b$ in $N(x)$, and because $x$'s $1$-neighborhood is connected, at most one of these paths can contain vertices in $P_1$. The other path must contain a vertex outside of $T$. If this were to be the case, $x$ could not be added to $P_j$ while maintaining simple connectivity; for an example, see Figure~\ref{fig:bdryx}, where adding $x$ to $P_j$ results in a cycle in $P_j$ encircling vertices not in $P_j$. 
\begin{figure}
	\centering
	\includegraphics{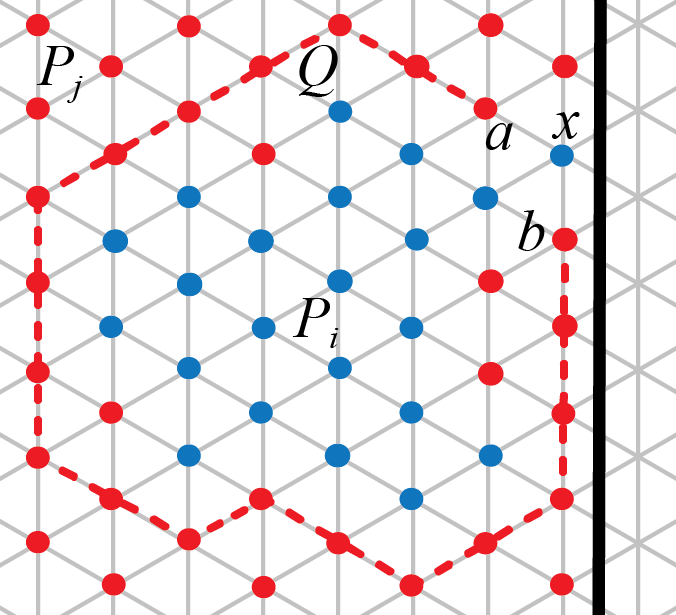}
	\caption{An example of an exposed vertex $x$ in $P_i$ (blue) that, despite having a connected $i$-neighborhood, cannot be added to a different district while maintaining simple connectivity: adding it to $P_j$ (red) causes $P_j$ to no longer be simply connected because $a$ and $b$ must be connected by some path $Q$ in $P_j$ (dashed red).  Vertex $x$ also can't be added to $P_k$ for $k \neq i,j$ because it has no neighbors in $P_k$. }\label{fig:bdryx}
\end{figure}
However, because $a$ and $b$ are both in $P_j$, they must be connected via a path $Q$ in $P_j$, and $Q$ must be in $T$ (one example of such a $Q$ is shown as a dashed red line in Figure~\ref{fig:bdryx}). 
This path $Q$, together with $x$, creates a cycle $C$ whose only vertex in $P_i$ is $x$, while all other vertices of $C$ are in $P_j$.  Because $x$'s $i$-neighborhood is connected, $x$ cannot have neighbors in $P_i$ both inside and outside $C$. Because $a$ and $b$ are in different components of $N(x) \cap P_j$, the path from $a$ to $b$ in $N(x)$ that goes inside $C$ must contain a vertex not in $P_j$, and because $x$ has no neighbors in $P_k$, this vertex must be in $P_i$. Therefore $x$ has a neighbor inside $C$ in $P_i$ and therefore no neighbors in $P_i$ outside $C$. This means all of  $P_i$ must be inside $C$, except for $x \in C$. 
This means $x$ is the only vertex in $P_i \cap bd(T)$, a contradiction as Condition~\ref{item:exp_2bd} assumes $|P_i \cap bd(T) | \geq 2$.  Therefore this case is impossible under the assumptions of Condition~\ref{item:exp_2bd}, and in fact $x$ must fall into one of the previous cases in which $x$ can be successfully added to $P_j$ or $P_k$.

The remaining conditions all imply Condition~\ref{item:exp_2bd} holds in a fairly straightforward way, and are meant to make later applications of this lemma easier. 
Consider Condition~\ref{item:exp_corner}. If $P_i$ contains a corner of $T$, then because each district has more than one vertex, this corner vertex cannot be the only vertex of $P_i$. Any corner vertex has two neighbors in $T$, and both are in $bd(T)$. At least one of these must be in $P_i$, meaning $P_i$ has at least two boundary vertices, the corner and its neighbor. Therefore $P_i$ containing a corner vertex of $P_i$ implies $|P_i \cap bd(T)| \geq 2$, so Condition~\ref{item:exp_corner} implies Condition~\ref{item:exp_2bd}. 

Next consider Conditions~\ref{item:cut_2bd} and \ref{item:cut_corner}. By Lemma~\ref{lem:remove-W1}, when $W$ is a single cut vertex of $P_i$, then each component of $W \setminus P_i$ contains an exposed vertex.  Therefore Conditions~\ref{item:cut_2bd} and \ref{item:cut_corner} imply Conditions~\ref{item:exp_2bd} and \ref{item:exp_corner}, respectively. \end{proof}

\subsection{Tower Moves and the Tower Lemma}

There will be certain points in our sweep-line proof where we will want to add a particular vertex to a different district, but may not be able to in one move. Instead, a whole sequence of moves might be required.  This motivates the following definition. 



\begin{defn}
	A {\em tower} is a sequence of at least two adjacent vertices $v_1$, $v_2$, $\ldots$, $v_t$ lying on a straight line such that: 
	\begin{itemize}
		\item $v_1$ is in the same district as its two common neighbors with $v_2$. 
		\item No two adjacent vertices $v_i$ and $v_{i+1}$ are in the same district
		\item For any $l = 2, \ldots, t$, vertex $v_l$ cannot be removed from its current district and reassigned to the district of $v_{l-1}$. 
		\item The vertex on this line after $v_t$ cannot be added to the tower. 
	\end{itemize}
	The {\em height} of this tower is $t$, the number of vertices in it. We call $v_1$ the {\em top} of the tower and $v_t$ the {\em bottom} of the tower.  
\end{defn}

\begin{figure}
	\centering
	\includegraphics{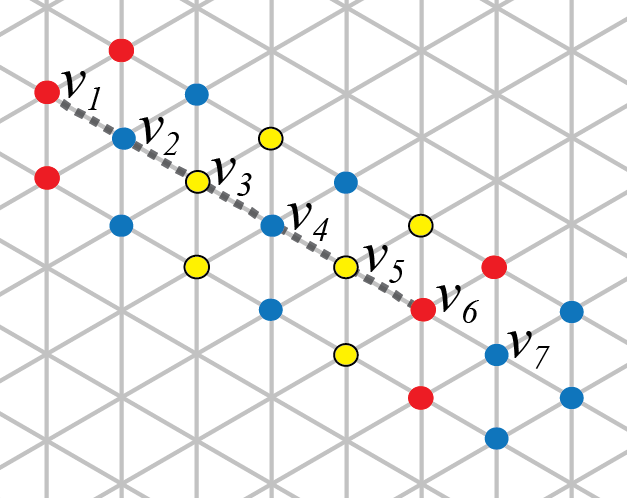}\hspace{1.2cm}
	\includegraphics{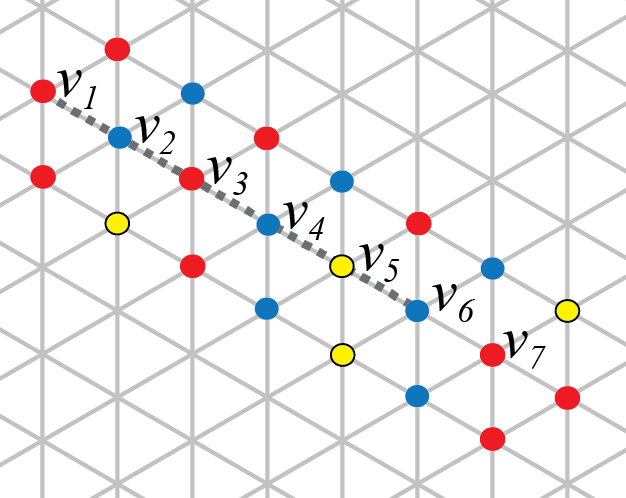}\\
	(a) \hspace{6cm} (b)
	\caption{(a) Vertices $v_1$, $v_2$, $\ldots$, $v_6$ form a tower of height $6$ where each $v_i$'s common neighbors with $v_{i+1}$ are in the same district as $v_i$. Note $v_7$ is not part of the tower, because it can be added to the district of $v_6$. (b) Vertices $v_1$, $v_2$, $\ldots$, $v_6$ form a tower of height $6$ where it is not true that each $v_i$'s common neighbors with $v_{i+1}$ are in the same district as $v_i$. Note $v_7$ is not part of the tower, because it can be added to the district of $v_6$.}
	\label{fig:tower}
\end{figure}

This definition of a tower is similar to that of~\cite{lrs} and~\cite{cdrr16}. 
The simplest example of a tower has $v_i$ in the same district as its common neighbors with $v_{i+1}$, as seen in Figure~\ref{fig:tower}(a).  However, towers do not exclusively have this property.  Figure~\ref{fig:tower}(b) gives a tower whose second vertex does not satisfy this property; this vertex $v_2$ cannot be added to the district containing $v_1$ because its 1-neighborhood is not connected. Despite this, towers must have a fairly rigid structure. 
For the following lemma which restricts what a tower can look like, you may find it helpful to reference Figure~\ref{fig:tower-forbid}.

\begin{figure}
\centering
	\includegraphics{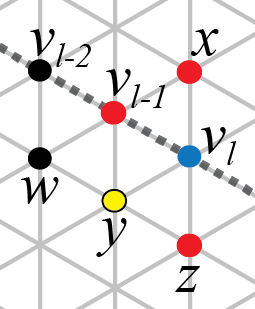}\hspace{2cm}
	\includegraphics{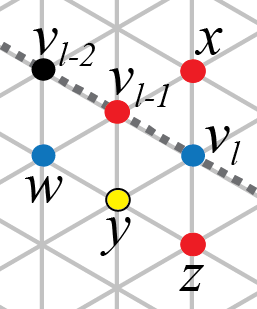}\hspace{2cm}
	\includegraphics{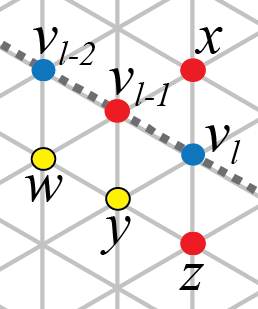}\\
(a) \hspace{3.5cm} (b) \hspace{3.5cm} (c)
\caption{Images from the proof of Lemma~\ref{lem:tower-forbid}, with $P_1$ in red, $P_2$ in blue, and $P_3$ in yellow.  Vertices whose district has not been specified are shown in black. (a) The partition (up to permuting districts, rotation, and reflection) in the neighborhood of a tower that is proved to be impossible in Lemma~\ref{lem:tower-forbid}.  (b) When $w \in P_2$, the cycle formed by any path from $w$ to $v_l$ in $P_2$ together with $y$ separates $v_{l-1} \in P_1$ from $z \in P_1$, a contradiction.  (c) When $w \in P_3$ and $v_{l-2} \in P_2$, the cycle formed by any path from $v_{l-2}$ to $v_l$ in $P_2$ together with $w$ and $y$ separates $v_{l-1} \in P_1$ from $z \in P_1$, a contradiction.} \label{fig:tower-forbid}	
\end{figure}

\begin{lem}\label{lem:tower-forbid}
	Let $T$ be a tower, and let $v_{l-2}$, $v_{l-1}$, and $v_{l}$ be three consecutive vertices in the tower for $3 \leq l \leq t$. Let $x$ and $y$ be the two common neighbors of $v_{l-1}$ and $v_l$, let $z$ be $y$ and $v_l$'s common neighbor that is not $v_{l-1}$, and let $w$ be $y$ and $v_{l-1}$'s common neighbor that is not $v_l$.	
	Without loss of generality, suppose $v_{l-1} \in P_1$ and $v_{l} \in P_2$.  
	The following district assignments are impossible: $x \in P_1$, $y \in P_3$, $z \in P_1$, and $w \notin P_1$. 
\end{lem}
\begin{proof}
	This proof will have two cases, for $w \in P_2$ and $w \in P_3$, finding a contradiction in both. First, suppose $w \in P_2$; see Figure~\ref{fig:tower-forbid}(b). This means that sequentially in $N(y)$, we have $w \in P_2$, $v_{l-1} \in P_1$, $v_{l} \in P_2$, and $z \in P_1$.  Because $w$ and $v_{l}$ are both in $P_2$, there must exist a path between them in $P_2$.  This path, together with $y$, forms a cycle that encircles exactly one of $v_{l-1}$ and $z$. This cycle contains no vertices in $P_1$ but separates two vertices of $P_1$, namely $v_{l-1}$ and $z$, implying that $P_1$ is not connected. This is a contradiction. Therefore if $w \in P_2$, this partition in the neighborhood of $v_{l-1}$ and $v_{l}$ is impossible. 
	
	Next, suppose $w \in P_3$. In this case we will need to look at $v_{l-2}$.  This vertex cannot be in $P_1$ because $v_{l-1} \in P_1$ and sequential tower vertices must be in different districts.  If $v_{l-2} \in P_3$, then $v_{l-1}$'s $1$-neighborhood and $3$-neighborhood would both be connected; this contradicts that $v_{l-1}$ cannot be added to $P_3$, which is part of the definition of a tower. Thus, $v_{l-2}$ must be in $P_2$.  However, in this case, clockwise around $w \cup y$, we have $v_{l-2} \in P_2$, $v_{l-1} \in P_1$, $v_{l} \in P_2$ and $z \in P_1$.  Because $v_{l-2}$ and $v_{l}$ are both in $P_2$, there must exist a path between them in $P_2$.  This path, together with $y$ and $w$, forms a cycle that encircles exactly one of $v_{l-1}$ and $z$. This cycle contains no vertices in $P_1$ but separates two vertices of $P_1$, implying that $P_1$ is not connected, a contradiction.  Thus when $w \in P_3$, this is also not an allowable partition.  We conclude the lemma is true. 
\end{proof}

In an abuse of notation, for the bottom vertex $v_t$ in a tower of height $t$, $v_{t+1}$ will refer to the next vertex along the line defining the tower past $v_t$, despite the fact that this vertex is not in the tower.  For a tower of height $t$, we begin by proving some facts about $v_{t+1}$, under assumptions that we will later show must always hold.

\begin{lem}
	\label{lem:tower-end}
	Consider a tower of height $t$.  Assume of the two common neighbors of $v_t$ and $v_{t-1}$, both are in $T$; at least one is in the same district as $v_{t-1}$; and neither is in the same district as $v_t$. Furthermore, if $t \geq 3$, assume the same is true for $t-1$: the two common neighbors of $v_{t-1}$ and $v_{t-2}$ are in $T$, at least one is in the same district as $v_{t-2}$, and neither is the same district as $v_{t-1}$. Then  $v_{t+1}$ is in $T$ and is not in the same district as  $v_t$.
\end{lem}
\begin{proof}
	Without loss of generality, assume  $v_{t-1} \in  P_1$ and $v_t \in P_2$. Suppose, for the sake of contradiction, that $v_{t+1} \notin T$.   Then $v_t \in bd(T)$ and $|N(v_t) \cap T| \leq 4$. By the assumptions of the lemma, $v_t$ and $v_{t-1}$'s common neighbors $x$ and $y$ are in $T$, at least one is in $P_1$, and neither is in $P_2$; see Figure~\ref{fig:tower-end}(a). Three of $v_t$'s neighbors must be $v_{t-1}$, $x$, and $y$, none of which are in $P_2$, so $v_t$'s one remaining neighbor must be in $P_2$, otherwise $P_2$ would be disconnected.  Because $v_t$ has only one neighbor in $P_2$, its $2$-neighborhood is connected. Because of the presence of the boundary of $T$, $v_t$'s $1$-neighborhood consists of $v_{t-1}$ and at least one of its common neighbors with $v_{t-1}$; in any case, this $1$-neighborhood is connected.  Because $v_t$'s $1$-neighborhood and $2$-neighborhood are both connected, by Lemma~\ref{lem:remove-add}, $v_t$ can be removed from $P_2$ and added to $P_1$.  This contradicts that $v_t$ is in the tower.  Because of this contradiction, it must be that $v_{t+1}\in T$.

	\begin{figure}
		\centering
		\includegraphics{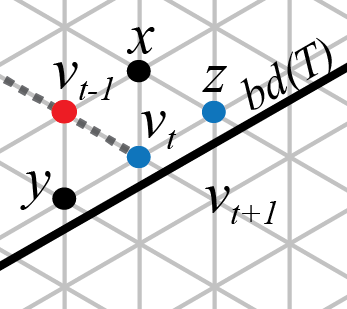}\hspace{1cm}
		\includegraphics{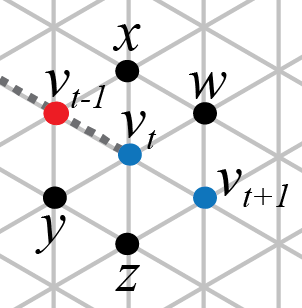}\hspace{1cm}
		\includegraphics{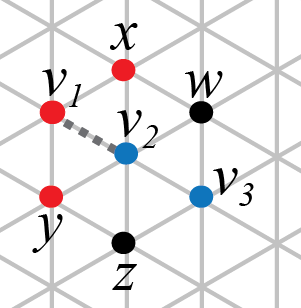}\hspace{1cm}
		\includegraphics{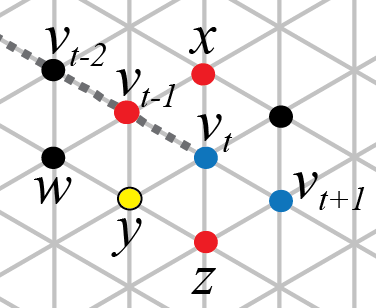}\\
		(a) \hspace{3cm} (b) \hspace{3cm} (c) \hspace{3cm} (d)
		\caption{Images from the proof of Lemma~\ref{lem:tower-end}. (a) The case where $v_{t+1} \notin T$. (b) An example where $v_{t}$ and $v_{t+1}$ are in the same district. (c) When $t = 2$, $v_2$'s 1-neighborhood must be connected. (d) When $t = 2$ and $v_t$ and $v_{t+1}$ are in the same district, $v_t$'s 1-neighborhood being disconnected implies the partition around $v_t$ must be the impossible configuration described in Lemma~\ref{lem:tower-forbid}. } \label{fig:tower-end}	
	\end{figure}

	Next, suppose for the sake of contradiction that $v_{t+1}$ and $v_t$  are in the same district, so $v_{t+1} \in P_2$ because we have assumed (without loss of generality) that $v_t$ is. See Figure~\ref{fig:tower-end}(b).  By the Lemma's assumptions, $v_t$'s common neighbors with $v_{t-1}$, which we call $x$ and $y$, are not in $P_2$. Let $w$ and $z$ be $v_t$'s two common neighbors with $v_{t+1}$. Vertices $w$ or $z$ may or may not be in $P_2$, but regardless of whether they are or not, because they are adjacent to $v_{t+1} \in P_2$, $v_t$'s 2-neighborhood will always be connected.

	If $t = 2$, see Figure~\ref{fig:tower-end}(c).  In this case $v_2$'s $1$-neighborhood also must be connected, because $v_1$ and $v_2$'s common neighbors $x$ and $y$ must both be in $P_1$ by the definition of a tower. This contradicts that $v_2$ cannot be added to $P_1$, so we conclude that $v_{t+1} = v_3$ must be in a different district than $v_t = v_2$.
	
	If $t \geq 3$, then $v_{t-2}$ must be in the tower. Because $v_t$ is part of the tower, it cannot be  removed from $P_2$ and added to $P_1$, so by Lemma~\ref{lem:remove-add}, it must be that $v_t$'s $1$-neighborhood is not connected. The only way for this to occur is if one of $v_t$ and $v_{t+1}$'s common neighbors $z$ is in $P_1$, and its neighbor $y$ in $N(v_{t-1} ) \cap N(v_t)$ is not in $P_1$; by our assumptions, $y \notin P_2$, so it must be that $y \in P_3$. See Figure~\ref{fig:tower-end}(d).  Because at least one of $v_{t-1}$ and $v_{t}$'s common neighbors must be in $P_1$, their other common neighbor $x$ must be in $P_1$, as shown.  By Lemma~\ref{lem:tower-forbid}, $y$'s common neighbor $w \neq v_l$ with $v_{l-1}$ cannot be in $P_2$ or $P_3$; the Lemma's assumptions also imply $w$ must be in $T$ and cannot be in $P_1$. This gives a contradiction, as $w$ must be in one of the three districts. We conclude $v_{t+1}$ and $v_t$ cannot be in the same district. 
\end{proof}

\noindent This enables us to prove the following lemma about vertices $v_l$ in a tower, even when $v_l = v_t$ is the bottom vertex of the tower. 

\begin{lem}\label{lem:tower-nbrs}
Consider a tower of height $t$. For any $l = 1, \ldots, t$, of the two common neighbors of $v_l$ and $v_{l+1}$, both are in $T$; at least one is in the same district as $v_l$; and neither is in the same district as $v_{l+1}$.  
\end{lem}
\begin{proof}
	We will prove this by induction. First, by the definition of a tower, $v_1$'s two common neighbors with $v_2$ are both in $T$ and in the same district as $v_1$, which is a different district than $v_2$, and the lemma is satisfied. 
	
	Now, for some $2 \leq l \leq t$, consider $v_l$ and $v_{l+1}$, and assume the lemma is true for all pairs of consecutive vertices earlier in the tower. If $l < t$, $v_{l+1} \in T$ because it is a tower vertex.  If $l = t$, then $v_{l+1} \in T$ by Lemma~\ref{lem:tower-end}, where the hypotheses of that lemma hold because of the induction hypothesis here. In either case, it must be true that $v_{l+1} \in T$. 
	
\begin{figure}
	\centering
	\includegraphics{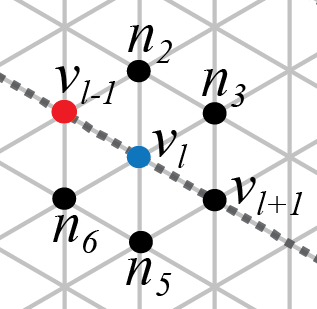}\hspace{1cm}
	\includegraphics{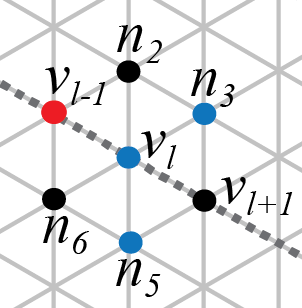}\hspace{1cm}
	\includegraphics{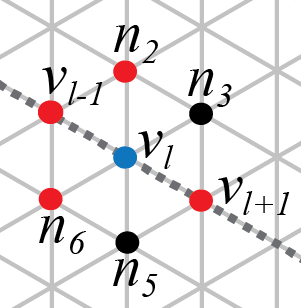}\hspace{1cm}
	\includegraphics{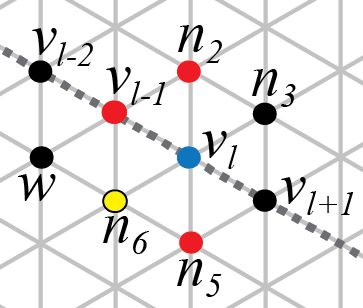}\\
	(a) \hspace{3cm} (b) \hspace{3cm} (c) \hspace{3cm} (d)
	\caption{Images from the proof of Lemma~\ref{lem:tower-nbrs}. (a) The labeling of the vertices in $N(v_l)$. (b) When $v_l$ has a disconnected 2-neighborhood, it must be that $n_3 \in P_2$ and $n_5 \in P_2$, and the claims of the lemma are satisfied. (c) When $v_l$ has a disconnected 1-neighborhood and $n_2, n_6 \in P_1$, then $v_{l+1} \in P_1$ and the claims of the lemma are satisfied. (d) When $v_l$ has a disconnected 1-neighborhood and one of $n_2$ and $n_6$ is not in $P_1$ -- without loss of generality, $n_6$ -- and $n_5 \in P-1$, we find a contradiction via Lemma~\ref{lem:tower-forbid}.  
	 } \label{fig:tower-nbrs}	
\end{figure}

	Without loss of generality, suppose that $v_{l-1} \in P_1$ and $v_l \in P_2$. Label the six neighbors of $v_l$ in $\Gtri$ in clockwise order as $n_1 = v_{l-1}$, $n_2$, $n_3$, $n_4 = v_{l+1}$, $n_5$, and $n_6$; see Figure~\ref{fig:tower-nbrs}(a). It is $n_3$ and $n_5$ that are $v_l$'s common neighbors with $v_{l+1}$. We know that $n_1=v_{l-1}$, $n_2$, $n_4 = v_{l+1}$, and $n_6$ are all in $T$.  It follows that $n_3$ and $n_5$ must also be in $T$, because $T$ is convex.  This proves the first claim. 	
		
	By induction, we know that $v_{l-1} = n_1$, $n_2$, and $n_6$ are not in $P_2$, and at least one of $n_2$ and $n_6$ is in $P_1$. If $l < t$, then $v_{l+1}$ is not in $P_2$ by the definition of a tower. If $l = t$, then by Lemma~\ref{lem:tower-end}, whose hypotheses follow from the induction hypothesis, $v_{l+1}$ is not in $P_2$. 
	This leaves $n_3$ and $n_5$ as $v_l$'s only potential neighbors in $P_2$. Because $P_2$ is connected, at least one of these vertices must be in $P_2$, the same district as $v_l$.  This proves the next part of the lemma.  
	
	Recall we have assumed without loss of generality that $v_{l-1} \in P_1$ and $v_l \in P_2$. Because $v_l$ cannot be removed from $P_2$ and added to $P_1$, it must be that $v_l$'s $2$-neighborhood is disconnected or $v_l$'s $1$-neighborhood is disconnected. If $v_l$'s $2$-neighborhood is disconnected, it must be that $n_3$ and $n_5$ are both in $P_2$ while we know $n_4 = v_{l+1}$ is not in $P_2$; see Figure~\ref{fig:tower-nbrs}(b). In this case $v_l$ and $v_{l+1}$'s two common neighbors are both in the same district as $v_l$ and neither is in the same district as $v_{l+1}$, proving the lemma. 
	
	If $v_l$'s $2$-neighborhood is connected but its $1$-neighborhood is not connected, there are more cases to consider.  First, we will do the easier case, when both $n_2$ and $n_6$ (along with $n_1 = v_{l-1}$) are in $P_1$; see Figure~\ref{fig:tower-nbrs}(c).  Then, for $v_l$'s $1$-neighborhood to be disconnected, it must be that $v_{l+1}\in P_1$, while $n_3, n_5 \notin P_1$, as required.  
	
	Next, suppose exactly one of $n_2$ and $n_6$ is in $P_1$. Note when $l = 2$, $n_2$ and $n_6$ are necessarily in the same district because of the definition of a tower, so if we are in this case it must be that $l \geq 3$.  This is important because we will need to look at $v_{l-2}$, and we now know that $l-2 \geq 1$.   Because we know by induction neither $n_2$ nor $n_6$ can be in $P_2$, without loss of generality assume $n_2 \in P_1$ and $n_6\in P_3$; see Figure~\ref{fig:tower-nbrs}(d). For $v_l$'s 1-neighborhood to not be connected, one or both of $v_{l+1} = n_4$ and $n_5$ must be in $P_1$. 	Assume for the sake of contradiction that $n_5 \in P_1$. Applying Lemma~\ref{lem:tower-forbid}, where $x = n_2$, $y = n_6$, $z = n_5$, we see that the vertex $w$, which is $n_6$ and $v_{l-1}$'s common neighbor that is not $v_l$, cannot be in $P_2$ or $P_3$.  It also cannot be in $P_1$, by the induction hypothesis for $v_{l-1}$ and $v_{l-2}$. The only remaining option is that vertex $w$ is not in $T$. By convexity, this would necessarily imply that $v_{l-2}$ is also not in $T$, a contradiction as $l \geq 3$. Thus it cannot be the case that $n_5 \in P_1$. 
	
	If $v_l$'s $1$-neighborhood is to be disconnected, it must be that $v_{l+1} = n_4$ is in $P_1$ while $n_3$ and $n_5$ are not.  Thus $v_{l+1}$ is in a different district than both of its common neighbors with $v_{l}$, as desired.   
	This completes the proof for $v_l$ and $v_{l+1}$'s common neighbors, and by induction we conclude it is true for all $l$.
\end{proof}

\begin{lem}\label{lem:tower}
	Consider a tower with top vertex $v_1\in P_i$, bottom vertex $v_t$, and $v_{t+1}\in P_j$ the next vertex along the line of the tower past $v_t$. There exists a sequence of recombination steps, where only vertices $v_2$, $v_3$ $\ldots$, $v_{t+1}$ change to different districts, after which $v_2$ has been added to $P_i$. After this sequence of moves $|P_i|$ has increased by one and $|P_j|$ has decreased by one; if $i = j$, the sizes of the districts remain unchanged after this sequence of moves. At every intermediate step, the size of $P_j$ is equal to or one less than its initial size, while the size of the other two partitions are equal to or one greater than their initial sizes. 
\end{lem}
\begin{proof}
	We prove this by induction on the number of vertices in the tower.  First, suppose the tower is height $2$, with $v_1 \in P_i$, $v_2 \in P_l$, and $v_3 \in P_j$. It must be that $i \neq l$ and $l \neq j$, but possibly $i = j$.  Because $v_3$ is not in the tower, it can be added to to $P_l$, causing $|P_j|$ to decrease by one and causing $|P_l|$ to increase by one.  Now that $v_3$ has been added to $P_l$, $v_2$ and $v_3$ are in the same district.  This means $v_2$ cannot be the bottom of a tower, because if it were, $v_2$ and $v_3$ would need to be in different districts. It follows that because $v_1$ and $v_2$ are still in different districts, that it must be possible to add $v_2$ to $P_i$. This causes $|P_i|$ to increase by one and $|P_l|$ to decrease by one. If $i \neq j$, the net result of this process is that $|P_i|$ has increased by one and $|P_j|$ has decreased by one; if $i = j$, the net result is that all parts stayed the same size.  During these two steps, $|P_j|$ was equal to or one less than its original size, and all other parts were the size or one greater than their original size.  This proves the lemma when $t = 2$.

	Now, consider a tower of some height $t > 2$. Let $v_{t-1} \in P_m$ and $v_t \in P_l$, where $m \neq l$ and $l \neq j$. 
	Because vertex $v_{t+1}$ isn't in the tower, and by Lemmas~\ref{lem:tower-end} and Lemma~\ref{lem:tower-nbrs} it is in a different district than $v_t$, $v_{t+1}$ can be removed from its current district and reassigned to $P_l$. After this step, we claim that $v_1$, $\ldots$, $v_{t-1}$ is still a tower, albeit one of a slightly smaller height.   The first three properties of a tower are still satisfied by $v_1$, $\ldots$, $v_{t-1}$ because $v_{t+1}$ is the only vertex that has changed districts. We will now see why $v_t$ cannot be included in this tower, making it a tower of height $t-1$. 
	Note that after adding $v_{t+1}$ to $P_l$, $v_t$'s $l$-neighborhood must be connected. If $v_t$'s $m$-neighborhood isn't connected, then we are in the same situation as in Lemma~\ref{lem:tower-forbid}, and we get a contradiction.  Thus $v_t$'s $m$-neighborhood must also be connected.  This means $v_t$ could be removed from its current district and added to the district of $v_{t-1}$. Thus $v_t$ does not satisfy the requirements to be in the tower.  For all other vertices in the tower, their neighborhood has remained unchanged, and thus they are still part of the tower.  This is now a tower of height $t-1$, from $v_1$ to $v_{t-1}$.  By the induction hypothesis, there exists a sequence of moves for this tower after which $v_2$ has been added to $P_i$. These moves had increased the size of $P_i$ by one and decreased the size of $P_l$ by one. When this process started, $P_j$ was one smaller than its original size and $P_l$ was one larger than its original size (before $v_{t+1}$ was reassigned).  During this process, the size of $P_j$ remained its original size or one less, $P_l$ remained its original size or one more, and the other district remained its original size or one more. Together with the first move that decreased the size of the partition $P_j$ of $v_{t+1}$ by one and increased the size of $P_l$, this is a sequence of moves after which $v_2$ has been added to $P_i$. The net result is that $|P_j|$ has decreased by one and $|P_i|$ has increased by one. If $P_i = P_j$, there is no net change in the district sizes. 
		
\end{proof}

\section{Sweep Line Procedure for Balanced Partitions} 

Let $P$ be a balanced partition of $T$.  Without loss of generality, suppose that $T$'s single leftmost vertex is in $P_1$. We will give a sequence of steps that transforms $P$ into the ground state $\sigma_{123}$. We will proceed column-by-column, left to right, modifying the partition so that all vertices in one column are in $P_1$ before continuing on the the next column.  Eventually we will reach a configuration where for some integer $c$, all vertices in the first $c$ columns are in $P_1$; there may be some vertices in $P_1$ in the $(c+1)^{st}$ column; and there are no vertices in $P_1$ in any other columns. For here it is then easy to reach a ground state using recombination steps, as we describe at the end of this section. 
 
For simplicity throughout this section, in images vertices in $P_1$ will be red, vertices in $P_2$ will be blue, and vertices in $P_3$ will be yellow.

Let $\ci \subseteq T$ be vertices that are in the $i^{th}$ column of $T$, where $T$'s single leftmost corner vertex is $\cc_1$ and $T$'s rightmost column is $\cc_n$. We further define $\cli = \mathcal{C}_1 \cup \ldots \cup \mathcal{C}_{i-1}$, $\clei = \mathcal{C}_1 \cup \ldots \cup \mathcal{C}_i$, $\cgei = \mathcal{C}_{i} \cup \ldots \cup \mathcal{C}_n$, and $\cgi = \mathcal{C}_{i+1} \cup \ldots \cup \mathcal{C}_n$.  

We begin with a simple lemma we will need regarding the pattern of sequential vertices in $bd(T)$. 

\begin{lem}\label{lem:bdry-seq}
	Let $a$, $b$, $c$, and $d$ be four consecutive vertices in $bd(T)$. In any partition $P$, it is impossible to have $a,c \in P_i$ and $b,d \in P_j$ for $i \neq j$. 
\end{lem}
\begin{proof}
	Because $a$ and $c$ are both in $P_i$ and $P_i$ is connected, there exists a path in $P_i$ from $a$ to $c$.  There also exists a path from $a$ to $c$ with two intermediate vertices, namely $b$'s two neighbors in $\Gtri$ that are outside of $T$.  Putting these two paths together creates a cycle in $\Gtri$ that contains no vertices of $P_2$, has $b \in P_2$ on its interior, and has $d \in P_2$ on its exterior.  This is a contradiction, as $b$ and $d$ must be connected by a path in $P_2$.  We conclude this partition is impossible. 
\end{proof}

\noindent In fact this lemma is true even if $a$, $b$, $c$ and $d$ are not consecutive on the boundary of $T$, though we will not need this fact. 

We now show, if we are at an intermediate step of our sweep-line procedure, that one can always find a sequence of moves through balanced or nearly balanced partitions resulting in an additional vertex of~$P_1$~in~$\cc_i$.  
For this lemma and throughout the rest of this section, we will assume we are at a stage where: The first $i-1$ columns of $T$ are already in $P_1$; we have not yet added all vertices in the $i^{th}$ column to $P_1$; and there remain vertices in $P_1$ in columns with indices greater than $i$.

\begin{lem}\label{lem:increase-Ci}
	Let $P$ be a balanced partition of $T$ where $\cli \subseteq P_1$, $\ci \cap P_1 \neq \ci$,  and $\cgi \cap P_1 \neq \emptyset$. There exists a sequence of moves, though balanced and nearly balanced partitions, that maintains  $\cli \subseteq P_1$ and increases $|\ci \cap P_1|$.  The resulting partition is balanced or has $|P_1| = k_1 + 1$. 
\end{lem}
\begin{proof}
	Let $v$ and $w$ be any adjacent vertices in $\ci$ where $v \notin P_1$ and $w \in P_1$.  Without loss of generality, suppose $v \in P_2$ and that $w$ is above $v$.  
	
	First, suppose that $v$ is a boundary vertex of $T$. This means $v$ must have exactly one neighbor in $\mathcal{C}_{i-1}$, which must be in $P_1$; $v$ only has one neighbor ($w$) in $\ci$, and this neighbor is in $P_1$; and $v$ has two neighbors  in $\mathcal{C}_{i+1}$.  Vertex $v$ must have at least one neighbor in $P_2$ because $P_2$ is connected, and this neighbor must be in $\mathcal{C}_{i+1}$ because $v$'s neighbors in $\mathcal{C}_{i-1}$ and $\ci$ are all in $P_1$.  At most $v$ can have two neighbors in $P_2$, specifically its two neighbors in $\mathcal{C}_{i+1}$.  Regardless, $v$'s $2$-neighborhood must be connected. If $v$'s $1$-neighborhood is also connected, then we can remove $v$ from $P_2$ and add it to $P_1$ (giving $|P_1| = k_1 + 1$), satisfying the lemma.  So, we assume $v$'s $1$-neighborhood is not connected. Let $x$ and $y$ be $v$'s two neighbors in $\cc_{i+1}$, where $x$ is a boundary vertex of $T$ and $y$ is the common neighbor of $v$ and $w$.  The only way for $v$'s $1$-neighborhood to not be connected is if $x \in P_1$ and $y \in P_2$; see Figure~\ref{fig:v-bdry-a}.  
	\begin{figure}
		\centering
		\begin{subfigure}[b]{0.4\textwidth}
			\centering
			\includegraphics[scale = 0.5]{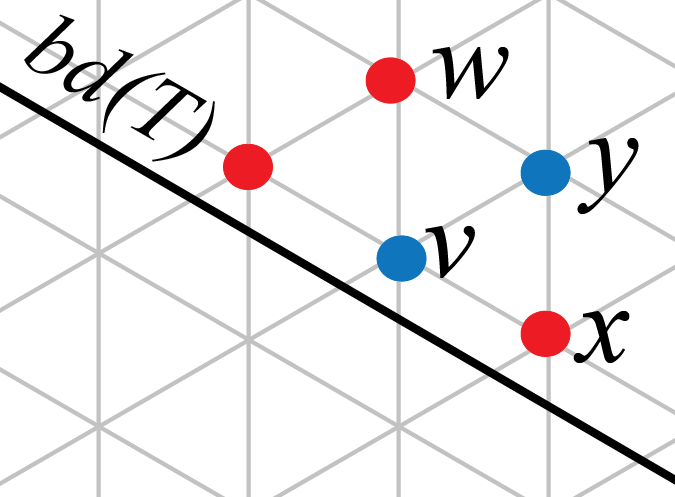}
			\caption{}
			\label{fig:v-bdry-a}
		\end{subfigure}
		\begin{subfigure}[b]{0.4\textwidth}
			\centering
			\includegraphics[scale = 0.5]{"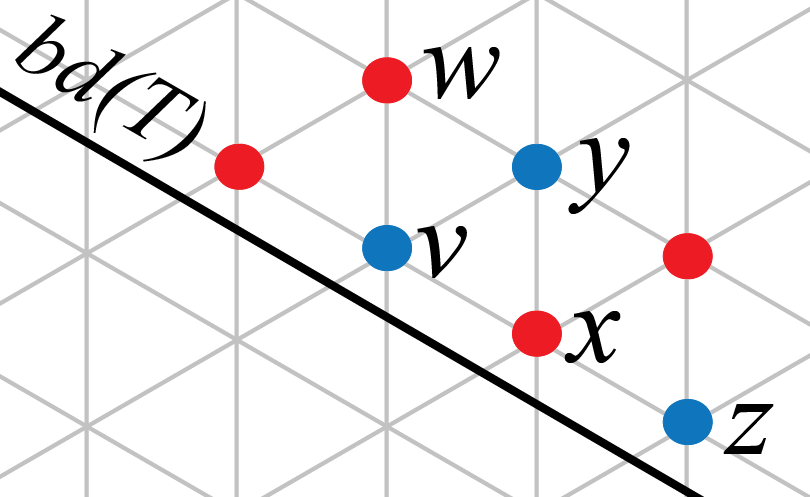"}
			\caption{}
			\label{fig:v-bdry-b}
		\end{subfigure}
		\caption{Cases from the proof of Lemma~\ref{lem:increase-Ci} when $v \in bd(T)$. (a) If $v \in P_2$ (blue) cannot be added to $P_1$ (red), its neighborhood must consist of this configuration. (b) If $v$'s neighbor $x \in P_1\cap \cc_{i+1}$ cannot be added to $P_2$, its neighborhood must consist of this configuration. Note there four vertices in alternating districts along the boundary, which Lemma~\ref{lem:bdry-seq} shows is impossible.  } 
		\label{fig:v-bdry}
	\end{figure}
	In this case, we first examine $x$.  Note $x$ has at most four neighbors in $T$, and two of them ($v$ and $y$) are in $P_2$.  At least one and at most two of $x$'s two remaining neighbors must be in $P_1$, and because these neighbors are adjacent, $x$'s $1$-neighborhood must be connected.  The only way it is possible for $x$'s $2$-neighborhood to be disconnected is if $x$'s boundary neighbor $z$ in column $\cc_{i+2}$ is in $P_2$ and $x$'s other neighbor in $\cc_{i+2}$ is in $P_1$; see Figure~\ref{fig:v-bdry-b}.  However, this gives us four sequential vertices along the boundary of $T$ where the first and third are in $P_1$ while the second and fourth are in $P_2$; by Lemma~\ref{lem:bdry-seq}, this configuration is impossible.  We conclude $x$'s $2$-neighborhood must be connected. Because $x$'s $1$-neighborhood and $2$-neighborhood are both connected, $x$ can be removed from $P_1$ and added to $P_2$.  Now $v$'s $1$-neighborhood is also connected, so $v$ can be removed from $P_2$ and added to $P_1$.  This results in a balanced partition satisfying the conditions of the lemma. 
	
	Next, suppose that $v$ is not a boundary vertex of $T$. Label $v$'s two neighbors in $\cc_{i-1}$ as $x$ and $y$, where $x$ is adjacent to $w$. If $v$ can be immediately removed from $P_2$ and added to $P_1$, such a move results in a configuration satisfying the conditions of the lemma with $|P_1| = k_1 + 1$. If not, then note that $x$ and $v$ form the start of a tower: their common neighbors are in the same district as $x$, $x$ and $v$ are in different districts, and $v$ cannot be added to $P_1$. Because $v$ cannot be added to $P_1$, this begins a tower.  Examine this tower, whose first vertex is $x$ and whose second vertex is $v$.  By Lemma~\ref{lem:tower}, there exists a sequence of recombination steps, where only vertices in the tower are reassigned, after which $v$ has been added to $P_1$. All other vertices in the tower are in columns $\cc_{i+1}$ or greater, so no other vertices in $\ci$ have changed to a different district and it is still true that $\cli \subseteq P_1$. The resulting partition is balanced or has $|P_1 | = k_1 + 1$, and $|P_1 \cap \ci|$ has increased by one because $v$ has been added to $P_1$. This proves the lemma.    
\end{proof}

If applying Lemma~\ref{lem:increase-Ci} results in a balanced partition, then we have successfully made progress in our sweep line algorithm.  However, if applying Lemma~\ref{lem:increase-Ci} results in a partition with $|P_1| = k_1 + 1$, we must first correct this imbalance before continuing with our sweep line procedure. While correcting this imbalance, we must be sure not to decrease the number of vertices in $P_1 \cap \clei$. This rebalancing process is the most challenging part of our proof. 

\subsection{Rebalancing: Lemmas}

We will need to following lemmas so that we are able to reach a balanced partition without decreasing $|\clei \cap P_1|$. 
Throughout this section, without loss of generality we assume it is the case that  $|P_2| = k_2$ and $|P_3| = k_3 - 1$.
We begin with two preliminary results about finding vertices in $P_1 \cap \cgi$ that can be removed from $P_1$.  

\begin{lem}\label{lem:1-to-change}
	Let $P$ be a nearly balanced partition of $T$ where $\cli \subseteq P_1$, $\ci \cap P_1 \neq \ci$,  and $\cgi \cap P_1 \neq \emptyset$. There exists a vertex in $P_1 \cap \cgi$ that can be removed from $P_1$ and added to another district. 
\end{lem}
\begin{proof}
	Let $v$ be any vertex in $P_1 \cap \cc_{i+1}$, which exists because $\cgi \cap P_1 \neq \emptyset$ and $P_1$ is connected. Consider the connected component $S$ of $P_1 \cap \cgi$ containing $v$. To begin, look specifically at the component of $S \cap \cc_{i+1}$ containing $v$. This component of $S \cap \cc_{i+1}$ cannot extend for all of $\cc_{i+1}$:  $\cc_i \setminus P_1$ is nonempty, so $\cc_i$ contains some vertices not in $P_1$.  Each district has at least $n$ vertices, $\cc_i$ has fewer than $n$ vertices, and $\cli \subseteq P_1$, so because $P_2$ and $P_3$ are connected, there must be some vertices of $P_2$ or $P_3$ in $\cc_{i+1}$. Because the component of $S \cap \cc_{i+1}$ containing $v$ doesn't extend for all of $\cc_{i+1}$, 
	this means a vertex above or below this component is not in $P_1$, meaning its adjacent vertex in $\cc_{i+1} \cap S$ is exposed and therefore $S$ contains an exposed vertex. 
	
	Let $Q$ be any path from $v$ to $\cli$ in $P_1$. Let $w$ be the first vertex in $\ci \cap P_1$ along this path, and let $W$ be the connected component of $\ci \cap P_1$ containing $w$. We know $W$ has both a neighbor in $P_1 \cap \cc_{i-1}$ and a neighbor in $P_1 \cap \cc_{i+1}$, and $W$'s removal separates $P_1$ into at least two parts, one of which is $S$. Because $S$ contains an exposed vertex and $P_1$ contains corner vertex $\cc_1$, by Condition~\ref{item:exp_corner} of Lemma~\ref{lem:shrinkable}, $S$ is shrinkable. This means it contains a vertex that can be removed from $P_1$ and added to a different district, completing the proof. 
\end{proof}

The following lemma can be used in conjunction with Lemma~\ref{lem:shrinkable} to find a vertex that can be removed from $P_1$ and added to another district in a different way, by looking at cut vertices. 
\begin{lem}\label{lem:s1_c1}
	Let $P$ be a partition such that $\cli \subseteq P_1$ and $P_1 \cap \cgi \neq \emptyset$.  Suppose $x$ is a cut vertex of $P_1$ and $S_1$ is a component of $P_1 \setminus \{x\}$ that does not contain $\cc_1$.   Then $S_1 \subseteq \cgi$. 
\end{lem}
\begin{proof}
	Suppose, for the sake of contradiction, that there is a vertex $s \in S_1$ where $s \in \clei$. 
	Let $y$ be a vertex in the component of $N(x) \cap P_1$ containing $\cc_1$, and let $z$ be a vertex in $N(x) \cap S_1$. Consider the closed walk $C$ in $P_1$ formed by any path from $\cc_0$ to $y$; vertex $x$; any path in $S_1$ from $z$ to $s$; and any path from $s$ to $\cc_0$ within $\clei$. One path from $y$ to $z$ in $N(x)$ must be enclosed by $C$. This path must have at least one vertex not in $P_1$ because $y$ and $z$ are in separate components of $P_1 \cap N(x)$, and must be in $T$ because it's inside $C$, which means $C$ - a closed walk entirely contained in $P_1$ - surrounds a vertex not in $P_1$, contradicting that $P_1$ is simple connected.  Because we have found a contradiction, it is impossible for $S_1$ to contain a vertex of $\clei$, so $S_1 \subseteq \cgi$, as desired. 
\end{proof}

If there is a vertex in $P_1 \cap \cgi$ that can be removed from $P_1$ and added to $P_3$ which has $|P_3| = k_3 - 1$, then we can perform this step to obtain a valid balanced partition. However, while the previous lemma guarantees there is a vertex in $P_1 \cap \cgi$ that can be removed, there is no guarantee it can be added to $P_3$: for example, maybe every vertex that can be removed from $P_1 \cap \cgi$ can only be added to $P_2$.  Doing so is not sufficient to produce a balanced partition.  In such situations significantly more work is needed to reach a balanced partition.

\subsubsection{Disconnected 3-neighborhoods}

There are some cases that are easy to resolve and will not require extensive case analysis to discover how to reach a balanced partition. These revolve around vertices in $P_1$ or $P_2$ that have disconnected 3-neighborhoods.  We present these results here, first as a general lemma for $P_j$ where $j = 1$ or $2$ and then as corollaries for $P_1$ and $P_2$ specifically.

\begin{lem}\label{lem:pj_3nbhd_discon}
	Let $P$ be a partition. Suppose there exists $x \in P_j$ for $j = 1,2$ such that $x$'s 3-neighborhood is disconnected. Furthermore, suppose that if $C$ is a cycle formed by any path in $P_3$ connecting different components of $P_3 \cap N(x)$ together with $x$, there is at least one vertex of $P_\ell$ outside $C$ for $\ell = 1,2$, $\ell \neq j$. Then there exists a move assigning a vertex of $P_j$ inside $C$ to~$P_3$.
\end{lem}
\begin{proof}
	Consider vertex $x$, and let $y, z \in N(x) \cap P_3$ be in two different components of $N(x) \cap P_3$, possible because we know $x$'s 3-neighborhood is disconnected; see Figure~\ref{fig:pj_3nbhd_discon} where $P_j = P_2$ is blue, though the picture is identical when $j = 1$ with red vertices replacing blue. Because $y$ an $z$ are both in $P_3$, they must be connected by a path in $P_3$, shown in dashed yellow in Figure~\ref{fig:pj_3nbhd_discon}. 	
	Consider the cycle $C$ formed by this path from $y$ to $z$ in $P_3$ together with $x$, shown as dashed in Figure~\ref{fig:pj_3nbhd_discon}. By assumption, there is at least one vertex of $P_\ell$ outside of $C$. Because cycle $C$ is entirely comprised of vertices in $P_j$ and $P_3$, it must be that all of $P_\ell$ is outside of $C$. It follows that all vertices enclosed by $C$ are in $P_j$ or $P_3$. 
	
	\begin{figure}
		\centering
		\includegraphics[scale = 0.7]{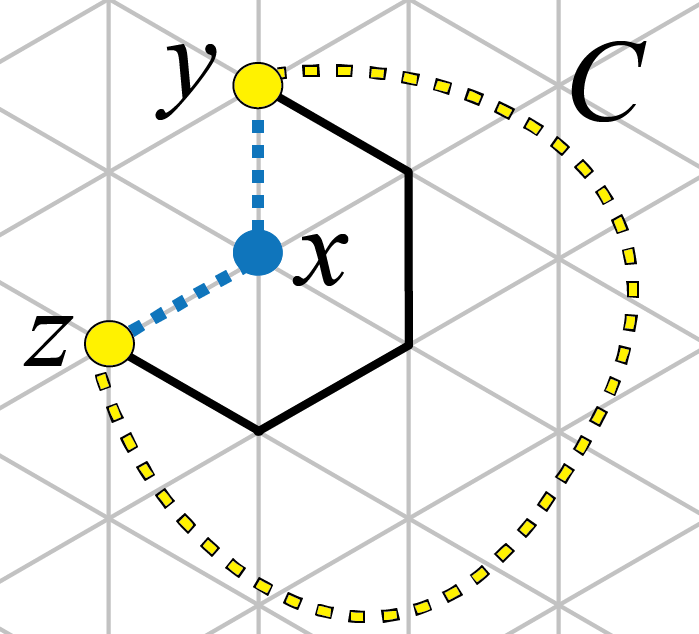}
		
		\caption{From the proof of Lemma~\ref{lem:p2_3nbhd_discon} when $P_j = P_2$; the same argument applies when $P_j = P_1$.  Vertex $x \in P_2$ has neighbors $y$ and $z$ in different connected components of $N(x) \cap P_3$. Vertices $y$ and $z$ must be connected by some path in $P_3$ (dashed yellow). When looking at the path from $y$ to $z$ in $N(x)$ enclosed by $C$ (black lines), some vertex on this path must be in $P_2$. } 
		\label{fig:pj_3nbhd_discon}
	\end{figure}
	
	We next consider $N(x)$ in more detail, and look at the paths from $y$ to $z$ in $N(x)$: there are two, one of which goes outside $C$ and the other of which goes inside $C$, and we focus on the path that goes inside $C$, shown as bold black in Figure~\ref{fig:pj_3nbhd_discon}. Because it is surrounded by $C$, all vertices on this path must be in $T$.  Because $y$ and $z$ are in different connected components of $N(x) \cap P_3$, this path must contain a vertex that is not in $P_3$.  It cannot contain a vertex of $P_\ell$ because $P_\ell$ is outside of $C$, so it must contain a vertex of $P_j$.  We let $S$ be the component of $P_j \setminus x$ that is inside $C$, and note $S$ must be nonempty and must be simply connected. It follows that $S \cap bd(T) = \emptyset$. 
	Therefore by Condition~\ref{item:nobd} of Lemma~\ref{lem:shrinkable}, $S$ is shrinkable and there exists a vertex $v \in S $ that can be removed from $P_j$ and added to another district.  Because all vertices enclosed by $C$ are in $P_j$ or $P_3$, $v$ can be removed from $P_j$ and added to $P_3$, proving the lemma. 
\end{proof}

The following corollary shows how to rebalance when a vertex of $P_2$ has a disconnected $3$-neighborhood. 

\begin{lem}\label{lem:p2_3nbhd_discon}
	Let $P$ be a partition such that $\cli \subseteq P_1$, $\ci \cap P_1 \neq \ci$, $\cgi \cap P_1 \neq \emptyset$, $|P_1| = k_1 + 1$, and $|P_3| = k_3 - 1$. Suppose vertex $x \in P_2$ has a disconnected 3-neighborhood. Then there exists a sequence of one or two moves resulting in a balanced partition that does not reassign any vertices in $P_1 \cap \clei$. 
\end{lem}
\begin{proof}
	By Lemma~\ref{lem:1-to-change}, there exists a vertex $v \in P_1 \setminus \clei$ that can be removed from $P_1$ and added to another district. If $v$ can be added to $P_3$, we do so and reach a balanced partition. We therefore assume $v$ can only be added to $P_2$.
	
	First, we note that because $\cli \subseteq P_1$, this means $P_1$ contains a vertex of $bd(T)$. If $C$ is a cycle formed by any path in $P_3$ connecting different components of $P_3 \cap N(x)$ together with $x$, then $\cc_1$ must be outside of $C$, so $P_1$ contains a vertex outside of $C$. By Lemma~\ref{lem:pj_3nbhd_discon}, we know there exists a vertex $v'$ of $P_2$ inside $C$ that can be added to $P_3$.  Because $v \in P_1$ is outside $C$ and $v' \in P_2$ is inside $C$, $v$ and $v'$ are not adjacent and so adding $v$ to $P_2$ doesn't affect whether $v'$ can be added to $P_3$.   We add $v$ to $P_2$ and subsequently add $v'$ to $P_3$, reaching a balanced partition. 

\end{proof}

\noindent We can prove a similar lemma about vertices in $P_1$ with a disconnected 3-neighborhood.


\begin{lem}\label{lem:p1_3nbhd_discon}
	Let $P$ be a partition such that $\cli \subseteq P_1$, $\cgi \cap P_1 \neq \emptyset$, $|P_1| = k_1 + 1$, and $|P_3| = k_3 - 1$.	
	Suppose there exists $x \in P_1$ such that $x$'s 3-neighborhood is disconnected. Furthermore, suppose that if $C$ is a cycle formed by any path in $P_3$ connecting different components of $P_3 \cap N(x)$ together with $x$, there is at least one vertex of $P_2$ outside $C$. Then there exists a move resulting in a balanced partition that does not reassign any vertices in $P_1 \cap \clei$. 
\end{lem}
\begin{proof}
	First, by Lemma~\ref{lem:pj_3nbhd_discon} with $j = 1$ and $\ell = 2$, there exists a move reassigning a vertex $v \in P_1$ inside $C$ to $P_3$.  That $v \notin P_1 \cap \clei$ follow from Lemma~\ref{lem:s1_c1}. 
\end{proof}

The following lemma provides an alternate condition on $P_2$ that can be used, that will make future applications of this lemma more straightforward. 

\begin{lem}\label{lem:p1_3nbhd_discon_bdry}
	Let $P$ be a partition such that $\cli \subseteq P_1$, $P_1 \cap \cgi \neq \emptyset$, $|P_1| = k_1 + 1$, and $|P_3| = k_3 - 1$. Suppose there exists $x \in P_1$ such that $x$'s 3-neighborhood is disconnected and $P_2 \cap bd(T) \neq \emptyset$. Then there exist a move resulting in a balanced partition that does not reassign any vertices in $P_1 \cap \clei$. 
\end{lem}
\begin{proof}
	Consider vertex $x$, and let $y, z \in N(x) \cap P_3$ be in two different components of $N(x) \cap P_3$, possible because we know $x$'s 3-neighborhood is disconnected. Consider the cycle $C$ formed by any path from $y$ to $x$ in $P_3$ together with $x$. Because this cycle is in $T$, every vertex inside $C$ is in $T \setminus bd(T)$.  It follows that because $P_2 \cap bd(T)$ is nonempty, there exists a vertex of $P_2$ outside $C$.  The conclusion follows from Lemma~\ref{lem:p1_3nbhd_discon}. 
\end{proof}

\subsubsection{Unwinding Lemmas} 

In this section we give some rebalancing lemmas which we call unwinding lemmas.
While attempting to rebalance a partition, we cannot repeatedly remove vertices from the same district because this will produce partitions that are not balanced or nearly balanced.  Instead, if we wish to remove multiple vertices from a particular district, we must alternate with adding new vertices to that district somewhere else.  It is important the vertices we are adding are not adjacent to the vertices we are removing, otherwise we can't know any real progress is being made. 

The following Unwinding Lemma gets at this idea, where we have $S_1\subseteq P_1$ that we want to add to $P_2$ and $S_2 \subseteq P_2$ that we want to add to $P_1$. This lemma is only applied in the case where $|P_3| = k_3 - 1$, so adding a vertex to $P_3$ to bring it up to its ideal size is also considered a successful outcome. It is called the Unwinding Lemma because $S_1$ and $S_2$ are frequently long, winding arms of $P_1$ and $P_2$, respectively, that we wish to contract so that our partition is less~intertwined. 

In order for this lemma to be true, we must know that both $S_1$ and $S_2$ are shrinkable. Additionally, we need both to remain shrinkable even as vertices are removed from $S_1$ and added to $P_2 \setminus S_2$ and as vertices are removed from $S_2$ and added to $P_1 \setminus S_1$. We accomplish this by only considering $S_1$ where one of three conditions hold: $\cc_1 \in P_1 \setminus S_1$; $S_1 \cap bd(T) = \emptyset$; or $(P_1 \setminus S_1) \cap bd(T) \neq \emptyset$. 
Because $S_1$ only shrinks as vertices are removed and added to $P_2$, and $P_1 \setminus S_1$ only grows as vertices from $S_2$ are added to it, if these conditions are initially true they will remain true throughout the unwinding process. 
Similar conditions are used for $S_2$, though we omit the $\cc_1 \in P_2 \setminus S_2$ condition because we will never need it. 

For $S_1$, the first condition $\cc_1 \in P_1 \setminus S_1$ implies the last condition $(P_1 \setminus S_1) \cap bd(T) \neq \emptyset$ also holds, but we include both to ease future applications of this lemma.  While we will always know that $\cc_1 \in P_1$ during our rebalancing process described in this section, when we consider the steps necessary to turn an arbitrary nearly balanced partition into a balanced partition in Section~\ref{sec:getbalanced}, we will not be able to rely on this assumption and instead will use one of the other two hypotheses for $S_1$.

\begin{lem}[Unwinding Lemma]\label{lem:s1s2}
	Consider a partition where $|P_1| = k_1 + 1$, $|P_2| = k_2$, and $|P_3| = k_3 - 1$. 
	Suppose $w_1 \in P_1$ is a cut vertex of $P_1$ and  $w_2 \in P_2$ is a cut vertex of $P_2$.  Suppose $S_1 \subseteq P_1$ is a component of $P_1 \setminus w_1$ and $S_2 \subseteq P_2$ is a component of $P_2 \setminus w_2$, where no vertex of $S_1$ is adjacent to any vertex of $S_2$. Suppose one of the following is true for $S_1$: $\cc_1 \in P_1 \setminus S_1$, $S_1 \cap bd(T) = \emptyset$, or $(P_1 \setminus S_1) \cap bd(T) \neq \emptyset$. Suppose one of the following is true for $S_2$: $S_2 \cap bd(T) = \emptyset$ or $(P_2 \setminus S_2) \cap bd(T) \neq \emptyset$. 
	
	 There exists a sequence of moves through balanced or nearly balanced partitions after which (1) the partition is balanced, (2) all vertices in $S_1$ have been added to $P_2$, or (3) all vertices in $S_2$ have been added to $P_1$. In these moves only vertices in $S_1$ and $S_2$ have been reassigned, and in outcomes (2) and (3) the resulting partition has $|P_1| = k_1 + 1$ and $|P_3| = k_3 - 1$.
\end{lem}
\begin{proof}	
	First, we note that the hypothesis $\cc_1 \in P_1 \setminus S_1$ implies the hypothesis  $(P_1 \setminus S_1) \cap bd(T) \neq \emptyset$, so we only work with the latter (the former is included to ease later applications of this lemma).

	We will show that we reach outcome (1), (2), or (3), or we do two moves and find new sets $\overline{S_1} \subseteq S_1$ and $\overline{S_2} \subseteq S_2$ that still satisfy the hypotheses of the lemmas but have $|\overline{S_1}| + |\overline{S_2}| < |S_1| + |S_2|$. In particular, we will see $\overline{S_1}$ is $S_1$ with a single vertex removed and added to $P_2$, and $\overline{S_2}$ is $S_2$ with a single vertex removed and added to $P_1$. Because set sizes must be positive integers, we cannot indefinitely continue to find smaller and smaller sets satisfying the hypotheses of the lemma. This means that, unless we reach outcome (1) along the way, we continue until $|\overline{S_1}|$ or $|\overline{S_2}|$ is empty and we have reached outcome (2) or (3), respectively.  
	
First, if $S_1 \cap bd(T) = \emptyset$, then $S_1$ is shrinkable by Condition~\ref{item:nobd} of Lemma~\ref{lem:shrinkable}. If $S_1 \cap bd(T) \neq \emptyset$, then because $S_1$ satisfies at least one of the hypotheses given in the lemma, it must be true that $(P_1 \setminus S_1) \cap bd(T)  \neq \emptyset$ as well. This implies $|bd(T) \cap P_1| \geq 2$.  
Because $S_1$ was created by removing a single cut vertex from $P_1$, it follows that $S_1$ is shrinkable by Condition~\ref{item:cut_2bd} of Lemma~\ref{lem:shrinkable}. In either case, there must exist a vertex $v_1 \in S_1$ that can be removed from $S_1$ and added to another district. Following the same reasoning for $P_2$ and $S_2$, there also exists $v_2 \in S_2$ that can be removed from $S_2$ and added to another district.	


	
If $v_1$ can be added to $P_3$, we do so and the result is a balanced partition.  If $v_1$ can't be added to $P_3$, it can be added to $P_2$. If $v_2$ can be added to $P_3$, we do so and also add $v_1$ to $P_2$; because $S_1$ and $S_2$ are not adjacent, adding $v_1$ to $P_2$ will not affect whether $v_2$ can be added to $P_3$. This results in a balanced partition. 
	
If neither $v_1$ nor $v_2$ can be added to $P_3$, we add $v_1$ to $P_2$ and add $v_2$ to $P_1$, both of which must be valid: because $S_1$ and $S_2$ are not adjacent, adding $v_1$ to $P_2$ cannot affect whether $v_2$ can be added to $P_1$. This results in a partition that still has $|P_1| = k_1 + 1$, $|P_2| = k_2$, and $|P_3| = k_3 - 1$.
Let $\overline{S_1}$ be $S_1 \setminus v_1$, and let $\overline{S_2}$ be $S_2 \setminus v_2$, and note $|\overline{S_1}| + |\overline{S_2}| = |S_1| + |S_2| - 2$. If $|\overline{S_1}| = 0$, then all vertices of $S_1$ have been added to $S_2$, and we have reached outcome (2).  If $|\overline{S_2}| = 0$, then all vertices of $S_2$ have been added to $S_1$, and we have reached outcome (3).  Otherwise, now that $v_1$ has been added to $P_2$, $\overline{S_1}$ is a component of $P_1 \setminus w_1$ and so $w_1$ is still a cut vertex of $P_1$. 
Similarly, $\overline{S_2}$ is a component of $P_2 \setminus w_2$ and $w_2$ is still a cut vertex of $P_2$. $\overline{S_1}$ and $\overline{S_2}$ are not adjacent because $S_1$ and $S_2$ were not.  
Whichever of  $S_1 \cap bd(T) = \emptyset$ or $(P_1 \setminus S_1) \cap bd(T) \neq \emptyset$ was true for $S_1$ must still be true for $\overline{S_1}$: $\overline{S_1} \cap bd(T) = \emptyset$ or $(P_1 \setminus \overline{S_1} ) \cap bd(T)) \neq \emptyset$, because the former set only shrinks and the latter set only grows. 
Similarly, whichever of  $S_2 \cap bd(T) = \emptyset$ or $(P_2 \setminus S_2) \cap bd(T) \neq \emptyset$ was true for $S_2$ must still be true for $\overline{S_2}$.
Therefore sets $\overline{S_1}$ and $\overline{S_2}$, smaller than $S_1$ and $S_2$, satisfy the hypotheses of the lemma. We repeat this process until one of the three outcomes occurs, which must eventually happen because discrete sets cannot get arbitrarily small. 
	
Throughout the entire process outlined above, only vertices originally in $S_1$ and $S_2$ have been reassigned. 
	
\end{proof}

We will also need a modified version of the lemma, which allow us to specify a vertex of $S_2$ adjacent to $w_2$ that is removed from $S_2$ last. This will be needed in certain cases to ensure progress is always made.  The proof largely proceeds similarly to the above proof, but with extra care taken around $S_2$ and $w_2$. While the previous lemma and the next lemma could be integrated into one lemma, we've kept them separate for the sake of readability. Because the following lemma is only applied when $\cc_1 \in P_1 \setminus S_1$ and $S_2 \cap bd(T) = \emptyset$, this lemma includes fewer possible hypotheses for $S_1$ and $S_2$ than the previous lemma.



\begin{lem}\label{lem:s1s2x}
	Consider a partition where $|P_1| = k_1 + 1$, $|P_2| = k_2$, and $|P_3| = k_3 - 1$. 
	Suppose $w_1 \in P_1$ is a cut vertex of $P_1$ and $w_2 \in P_2$ is a cut vertex of $P_2$.  Suppose $S_1 \subseteq P_1$ is a component of $P_1 \setminus w_1$  and $S_2 \subseteq P_2$ is a component of $P_2 \setminus w_2$, where no vertex of $S_1$ is adjacent to any vertex of $S_2$. Suppose $\cc_1 \in P_1 \setminus S_1$ and $S_2 \cap bd(T) = \emptyset$. 
	Let $x$ be a vertex in $S_2$ that is adjacent to $w_2$, and suppose 
	$S_2$ contains at least one additional vertex in addition to $x$. 
	
	There exists a sequence of moves through balanced or nearly balanced partitions after which (1) the partition is balanced, (2) all vertices in $S_1$ have been added to $P_2$, or (3) all vertices in $S_2$ except $x$ have been added to $P_1$. In these moves only vertices in $S_1$ and $S_2\setminus x$ have been reassigned, and in outcomes (2) and (3) the resulting partition has $|P_1| = k_1 + 1$ and $|P_3| = k_3 - 1$.
\end{lem}
\begin{proof}
	We will show that we reach outcome (1), (2), or (3), or we do two moves and find new sets $\overline{S_1}$ and $\overline{S_2}$ that still satisfy the hypotheses of the lemmas but have $|\overline{S_1}| + |\overline{S_2}| < |S_1| + |S_2|$. Because set sizes must be positive integers, we cannot indefinitely continue to find smaller and smaller sets satisfying the hypotheses of the lemma. This means that, unless we reach outcome (1) along the way, we continue until $\overline{S_1} = \emptyset$ or $\overline{S_2} = x$ and we have reached outcome (2) or (3), respectively.  
	
	First, because $\cc_1 \in P_1$ and $S_1$ was created by removing a single cut vertex from $P_1$, by Condition\ref{item:cut_corner} of Lemma~\ref{lem:shrinkable}, $S_1$ is shrinkable. Therefore there exists $v_1 \in S_1$ that can be removed from $S_1$ and added to another district.	
	We note that $P_2 \setminus \{w_2, x\}$ has at least one connected component contained in $S_2$; we will call this component $S_2'$. We note the hypotheses of the lemma imply $|S_2'| \geq 1$. As $S_2 \cap bd(T) = \emptyset$, then $S_2' \cap bd(T) = \emptyset$ as well and $S_2'$ is shrinkable by Condition~\ref{item:nobd} of Lemma~\ref{lem:shrinkable}.
	Therefore we have shown there exists a vertex  $v_2 \in S_2' \subseteq S_2$, $v_2 \neq x$, that can be removed from $S_2$ and added to another district.
	
	If $v_1$ can be added to $P_3$, we do so and the result is a balanced partition.  If $v_1$ can't be added to $P_3$, it can be added to $P_2$. If $v_2$ can be added to $P_3$, we do so and also add $v_1$ to $P_2$; because $S_1$ and $S_2$ are not adjacent, adding $v_1$ to $P_1$ cannot affect whether $v_2$ can be added to $P_3$. This results in a balanced partition. 
	
	If neither $v_1$ nor $v_2$ can be added to $P_3$, we add $v_1$ to $P_2$ and add $v_2$ to $P_1$, both of which must be valid: because $S_1$ and $S_2$ are not adjacent, adding $v_1$ to $P_2$ cannot affect whether $v_2$ can be added to $P_1$. This results in a partition that still has $|P_1| = k_1 + 1$, $|P_2| = k_2$, and $|P_3| = k_3 - 1$. 	
	Let $\overline{S_1}$ be $S_1 \setminus v_1$ and let $\overline{S_2}$ be $S_2 \setminus v_2$, and note $|\overline{S_1}| + |\overline{S_2}| = |S_1| + |S_2| - 2$ and $x \in \overline{S_2}$ because $v_2 \neq x$. If $|\overline{S_1}| = 0$, then all vertices of $S_1$ have been added to $P_2$, and we ave reached outcome (2).  If $|\overline{S_2}| = 1$, then because $x \in \overline{S_2}$, all vertices of $S_2$ except $x$ have been added to $P_1$, and we have reached outcome (3). 
	Otherwise, after $v_1$ has been added to $P_2$, $\overline{S_1}$ is still a component of $P_1 \setminus w_1$ and so $w_1$ is still a cut vertex of $P_1$.  Because $\cc_1 \in P_1 \setminus S_1$, it remains true that $\cc_1 \in P_1 \setminus \overline{S_1}$. 
	Similarly, $\overline{S_2}$ is a component of $P_2 \setminus w_2$, $w_2$ is still a cut vertex of $P_2$, and because $S_2 \cap bd(T) = \emptyset$ it remains true that $\overline{S}_2 \cap bd(T) = \emptyset$.  
 %
	 It also remains true that $x \in \overline{S_2}$ is adjacent to $w_2$ and, because we did not reach outcome (3), $\overline{S_2}$ contains at least one vertex in addition to $x$. 	Furthermore,  $\overline{S_1}$ and $\overline{S_2}$ are not adjacent because $S_1$ and $S_2$ were not. Therefore sets $\overline{S_1}$ and $\overline{S_2}$, smaller than $S_1$ and $S_2$, satisfy the hypotheses of the lemma. We repeat this process until one of the three outcomes occurs, which must eventually happen because discrete sets cannot get arbitrarily small. 
	


	Throughout the entire process outlined above, only vertices originally in $S_1$ and $S_2 \setminus x$ have been reassigned. 
\end{proof}

\subsubsection{Cycle Recombination Lemmas} 

At times we will need to rearrange the partition inside a cycle beyond just reassigning a single vertex, as we did in the previous two lemmas.  We will use a breadth-first search tree to determine how we want to rearrange districts inside this cycle, so begin with the following lemma. 

\begin{lem}\label{lem:BFS-remove}
	Let $U$ be a simply connected subset of the triangular lattice and let $x$ be any vertex of $U$.  In any breadth-first search tree of $U$ rooted at $x$, let $v$ be the last vertex added to this tree.  Then $N(v) \cap U$ is connected and of size less than six.  
\end{lem}
\begin{proof}
	Let $l$ be the distance from $x$ to $v$ in $U$, which means $v$ is in level $l$ of the BFS tree. Because $v$ was the last vertex added to the BFS tree, there are no vertices in $U$ at distance larger than $l$ from $v$.  
	
	First, we show that $|N(v) \cap U| < 6$.  This follows immediately from the regular structure of the triangular lattice, as it is impossible for all six neighbors of $v$ to be at distance less than or equal to $l$ from $v$. 
	
	Next, we suppose that $N(v) \cap U$ is disconnected. Let $Q$ be any path from $v$ to $x$ in $U$, and let $N_1$ be the component of $N(v) \cap U$ containing the first vertex along this path.  Suppose, for the sake of contradiction, that there exists a second component of $N(v) \cap U$ containing a vertex $w$. Let $Q'$ be a shortest path from $w$ to $x$ in $S$. If $Q'$ passes through $v$, then because the distance from $v$ to $x$ in $U$ is $l$, $w$ must be at distance at least $l+1$ from $v$.  This is a contradiction, as there are no vertices in $U$ at distance larger than $l$ from $v$.  Suppose instead that $Q'$ does not contain $v$. Combining $Q$ and $Q'$, together with the edge from $v$ to $w$, forms a closed walk in $U$.  However, because this closed walk passes through two different components of $N(v) \cap U$, it necessarily encircles some vertex of $N(v)$ that is not in $S$, namely at least one vertex separating component $N_1$ of $N(v) \cap U$ from component $N_2$ of $N(v) \cap U$.  This is also a contradiction, as we know $S$ is simply connected.  We conclude $N(v) \cap U$ must be connected. 
\end{proof}
 
Here we give the situation in which we rearrange the entire partition enclosed by a cycle, as well as what we can guarantee as a result of this recombination. This is the formal Cycle Recombination Lemma, which was stated informally in the Proof Overview in Section~\ref{sec:pfoverview}. We note it is  one of the two main lemmas in this paper which requires a non-flip recombination step, reassigning multiple vertices to new districts at the same time (the other is Lemma~\ref{lem:groundstate}).

\begin{lem}[Cycle Recombination Lemma]\label{lem:cycle-recom}
	Let $P$ be a partition of $T$.  Let $C$ be a cycle in $T$ consisting entirely of vertices in $P_1$ except for one vertex $x \in P_2$. Suppose all vertices enclosed by $C$ are in $P_1$ or $P_2$. Let $y$ be one of the two vertices in $C$ adjacent to $x$. Let $N_C(y)$ be all vertices in $N(y)$ that are in or inside $C$, and suppose $N_C(y) \cap P_1$ is disconnected. Then there exists a recombination step for $P_1$ and $P_2$, only changing the district assignment of vertices enclosed by $C$ and leaving $|P_1|$ and $|P_2|$ unchanged, after which $N_C(y) \cap P_1$ is connected. 
\end{lem}
\begin{proof}
	Let $C$, $x$, and $y$, be as described, with $N_C(y) \cap P_1$ disconnected. See Figure~\ref{fig:cycle-recom-exs}(a,b) for two examples, which are not meant to be exhaustive. 

	\begin{figure}
		\centering
	\begin{subfigure}[b]{0.3\textwidth}
		\centering
		\includegraphics[scale = 0.8]{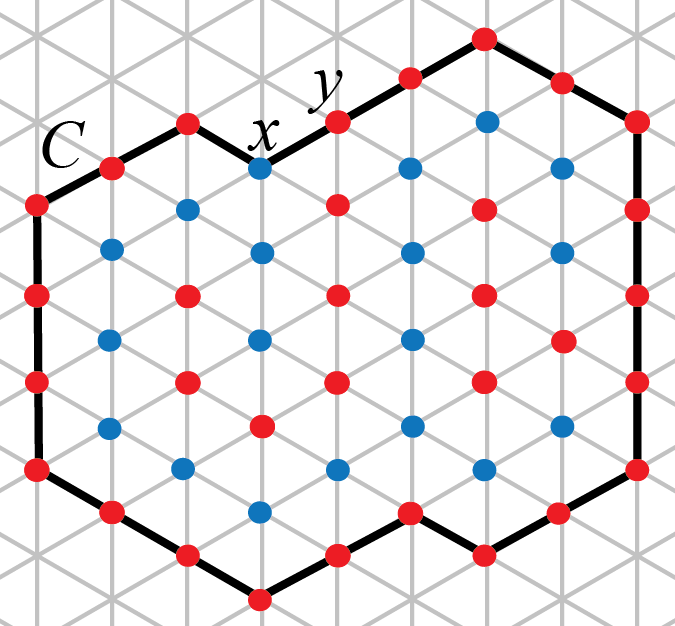}
		\caption{}
		\label{fig:cycle-recom-exs-onepart}
	\end{subfigure}
\hfill
	\begin{subfigure}[b]{0.3\textwidth}
		\centering
		\includegraphics[scale = 0.8]{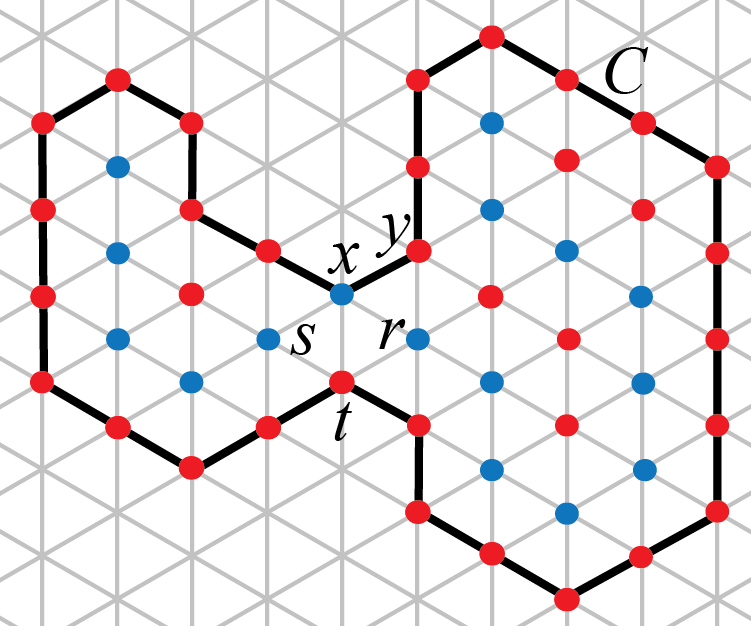}
		\caption{}
		\label{fig:cycle-recom-exs-twoparts}
	\end{subfigure}
\hfill
	\begin{subfigure}[b]{0.3\textwidth}
		\centering
		\includegraphics[scale = 0.8]{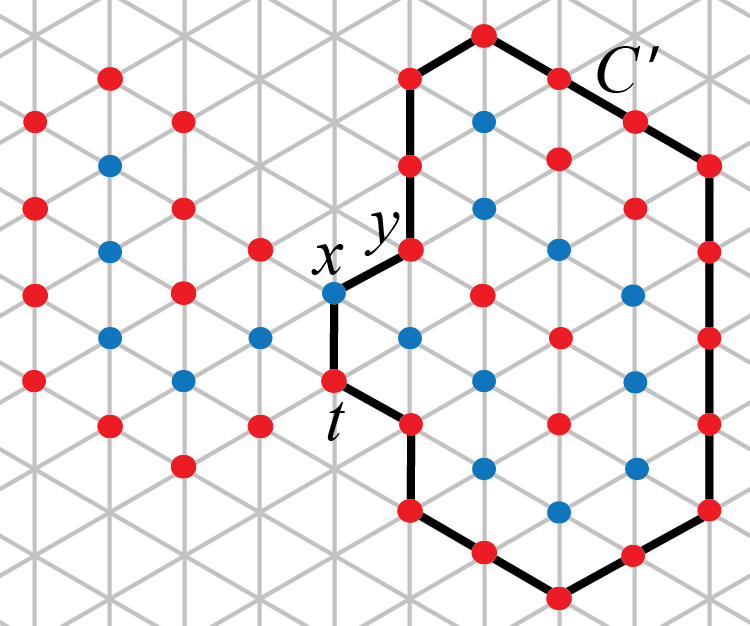}
		\caption{}
		\label{fig:cycle-recom-exs-newC}
	\end{subfigure}
		\caption{From the proof of Lemma~\ref{lem:cycle-recom}. (a) An example of a cycle $C$ and vertices $x$ and $y$ satisfying the hypotheses of the lemma, where $T \setminus C$ has one connected component inside $C$ that is adjacent to $x$. (c) 	 An example of a cycle $C$ and vertices $x$ and $y$ satisfying the hypotheses of the lemma, where $T \setminus C$ has two connected components inside $C$ that are adjacent to $x$.	(c) When $T \setminus C$ has two connected components inside $C$ that are adjacent to $x$, we can replace $C$ with a smaller cycle $C'$ that only has one connected component of $T \setminus C$ inside $C$ and adjacent to $x$. 	}
		\label{fig:cycle-recom-exs}
	\end{figure}
			
Consider the connected components of $T \setminus C$, of which some are inside $C$ and some are outside $C$. Of the ones inside $C$, at least one and at most two can be adjacent to $x$. In Figure~\ref{fig:cycle-recom-exs}(a), $T \setminus C$ has a single component inside $C$, and in Figure~\ref{fig:cycle-recom-exs}(b), $T \setminus C$ has two connected components inside $C$, both of which are adjacent to $x$.  If there are two components of $T \setminus C$ inside $C$ that are adjacent to $x$, then $x$'s neighborhood must include vertices $r$ and $s$ that are in $P_2$ and inside $C$ but in different components of $x$'s 2-neighborhood; the path between them in $N(x)$ that does not include $y$ must include some $t \in P_1 \cap C$. These vertices are labeled in the example in Figure~\ref{fig:cycle-recom-exs}(b). Because $C$ includes a path from $y$ to $t$, we simply replace $C$ with $C'$ consisting of the path from $y$ to $t$ in $C$ together with $x \in P_2$; see Figure~\ref{fig:cycle-recom-exs}(c).  This cycle $C'$ still satisfies the hypotheses of the lemma, as $N_{C'}(y) \cap P_1 = N_{C'}(y) \cap P_1$, but now there is only one internal component of $T \setminus C'$ adjacent to $x$. Therefore, we assume without loss of generality that there is at most one component of of $T \setminus C$ that is inside $C$ and adjacent to $x$. 
	
We now show there must be at least one component of $T \setminus C$ inside $C$ adjacent to $x$.  Because $N_C(y) \cap P_1$ is disconnected, $N_C(y)$ contains only vertices in or inside $C$, and all vertices enclosed by $C$ are in $P_1$ or $P_2$, this means $y$ must have a neighbor inside $C$ that is in $P_2$. There must be a path from this vertex to $x \in P_2$, and this path must be entirely enclosed by $C$ (except for its last vertex $x$) because $C$ (except $x$) is in $P_1$. Therefore this path must be part of a component of $T \setminus C$ inside $C$ and adjacent to $x$, showing such a component must exist. 
We label the unique component of $T \setminus C$ inside $C$ and adjacent to $x$ as $S$, and note $S$ initially includes vertices of both $P_1$ and $P_2$. 	The same reasoning also shows any vertices of $P_2$ in $N_C(y)$ must be in $S$. 
	
We will construct a new partition $P'$ from $P$ using a single recombination step.  $P$ and $P'$ will agree on all vertices of $T$ that are not in $S$. Note $S$ must be simply connected, because it is connected and cannot possibly have any holes. Additionally, every vertex that is not in $S$ but adjacent to a vertex of $S$ must be in $C$ and therefore, unless it is $x$ itself, must be in $P_1$.
	
We now explain why $P_1 \setminus S$ must be simply connected. For any two vertices in $P_1$ that are in or outside $C$, they can be connected by a path in $P_1$ that does not go inside $C$: any path going inside $C$ can be replaced by a path going along $C$ instead. Any vertices inside $C$ but not in $S$ can be connected by a path to $C$, and therefor to all other vertices in $P_1$.  It follows that $P_1 \setminus S$ must be connected.  To see that it is simply connected, note there could not have been a cycle in $P_1$ surrounding $S$ because $S$ contains some vertices of $P_2$, so removing $S$ cannot create a hole in $P_1$. Therefore $P_1 \setminus S$ is simply connected. 
	
We also note that $P_2 \setminus S$ must be simply connected as well. It must be connected because no shortest paths in $P_2$ connecting vertices of $P_2 \setminus S$ will go inside $C$ and $S$ is entirely inside $C$, meaning all vertices of  $P_2 \setminus S$ are connected by paths in this set. It is simply connected because $P_2$ cannot contain a cycle encircling $S$ as $S$ contains vertices not in $P_2$, contradicting that $P_2$ itself is simply connected. 
	
We now explore how to recombine the vertices of $S$ to satisfy the lemma. While we describe the below process in terms of erasing all district assignments in $S$ and then adding vertices to $P_1$ one at a time, the resulting configuration can in fact be reached in one recombination step, recombining $P_1$ and $P_2$ to reach the final partition.

\begin{figure}\centering
	\begin{subfigure}[b]{0.3\textwidth}
	\centering
	\includegraphics[scale = 0.8]{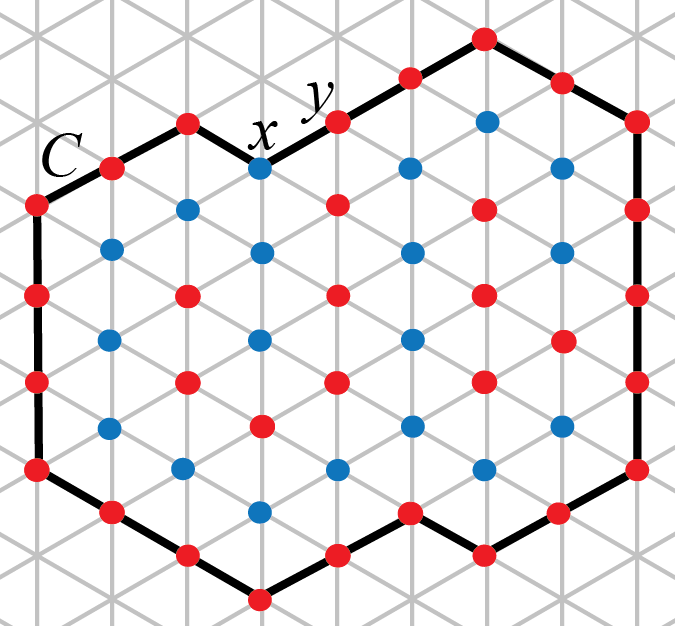}
	\caption{}
	\label{fig:cycle-recom-pf-ex}
\end{subfigure}
\hfill
\begin{subfigure}[b]{0.3\textwidth}
	\centering
	\includegraphics[scale = 0.8]{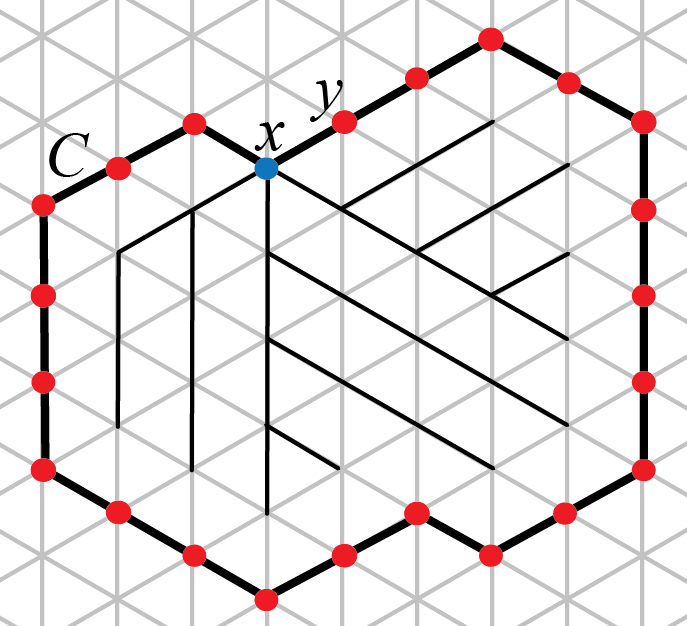}
	\caption{}
	\label{fig:cycle-recom-pf-tree}
\end{subfigure}
\hfill
\begin{subfigure}[b]{0.3\textwidth}
	\centering
	\includegraphics[scale = 0.8]{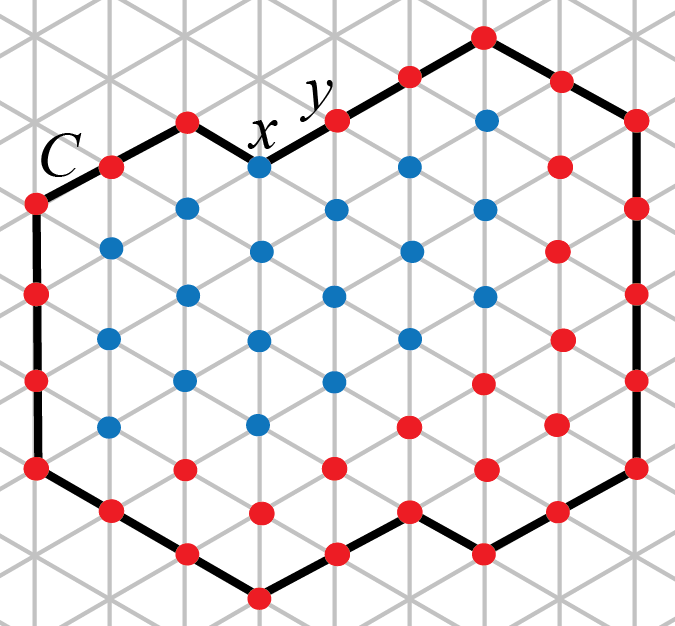}
	\caption{}
	\label{fig:cycle-recom-pf-result}
\end{subfigure}
		
		\caption{An example of the recombination step described in Lemma~\ref{lem:cycle-recom}. For the example shown in (a), $S$ is all $28$ vertices inside $C$, $m = 10$ of which are in $P_1$ (red) while the remaining $18$ are in $P_2$ (blue). 
			 (b) A breadth-first search tree for $S \cup x$ rooted at $x$. (b) The result of the described recombination step inside $C$ which leaves $|P_1|$ and $|P_2|$ unchanged while ensuring $N_C(y) \cap P_1$ is now connected: the last $10$ vertices added to the breadth first search tree are in $P_1$, while all remaining vertices are in~$P_2$. }
		\label{fig:cycle-recom-pf}
	\end{figure}

	
	
Let $m$ be the number of vertices of $P_1$ that are initially in $S$ before any recombination. We initially label all vertices in $S$ as 'unassigned' and put them in a set $U$.  We will sequentially find $m$ vertices that can be removed from $U$ and assigned to $P_1$ such that $P_1$ remains simply connected and, when the remaining unassigned vertices in $U$ are added to $P_2$, $P_2$ is simply connected as well. The way this process is done will ensure $N_C(y) \cap P_1$ is connected.
	
Because $U$ is connected, we create a breadth-first search tree of $U$ rooted at $x$, and let $v$ be the last vertex added to this tree. By Lemma~\ref{lem:BFS-remove} applied to $U \cup x$, $N(v) \cap (U \cup x)$ is connected and of size at most 5. This means $N(v) \setminus (U \cup x)$ must also be connected and have size at least 1.  The neighbors of $v$ that are not in $U \cup x$ must be in $P_1$, so this means $v$'s 1-neighborhood is connected and nonempty. 
%
%
%
%
We add $v$ to $P_1$ and remove it from $U$: $P_1$ remains simply connected by Lemma~\ref{lem:add} and $U \cup \{x\}$ remains simply connected by Lemma~\ref{lem:remove}. We repeat this process on the now-smaller set $U \subseteq S$ until $m$ vertices have been added to $P_1$.  At this point $P_1$ is simply connected and the correct size.  We add all remaining unassigned vertices in $U$ to $P_2$, and $P_2$ remains simply connected because $U \cup \{x\}$ is a tree. Looked at in its totality, this process assigns the last $m$ vertices added to the original breadth first search tree for $S$ to $P_1$, while all remaining vertices in $S$ are in $P_2$. See Figure~\ref{fig:cycle-recom-pf} for an example of this process.


It only remains to check that $N_C(y) \cap P_1$ is now connected.  Note $N_C(y)$ consists of $x$, some sequence of vertices that are inside $C$, and ends with $y$'s other neighbor in $C$, which we will call $z$. Note the vertices of $N(y)$ inside $C$ have their distances from $x$ in $S$ monotonically increasing from $x$ to $z$.  Because vertices are assigned to $P_1$ in order of decreasing distance from $x$, this means it is impossible for a vertex closer to $x$ to be assigned to $P_1$ while a vertex farther from $x$ is assigned to $P_2$. Thus $N_C(y) \cap P_1$ consists of $z$ possibly with some sequential vertices along the path from $z$ back to $x$ in $N(y)$. Regardless, $N_C(y) \cap P_1$ is connected, as desired. 	
\end{proof}

This lemma and Lemma~\ref{lem:groundstate} are the main ones in this paper to use a recombination move rather than a flip move.  We expect, without too much additional work, it could be modified to show the same resulting configuration is achievable with flip moves, but we have not pursued that option as our focus for this paper is irreducibility of recombination chains. 

On one occasion, 
we will need to ensure that a particular neighbor of $x$ inside $C$ remains in $P_2$. This is a straightforward corollary of the previous lemma. 
\begin{lem}\label{lem:cycle-recom-vtx}
	Let $P$ be a partition of $T$.  Let $C$ be a cycle in $T$ consisting entirely of vertices in $P_1$ except for one vertex $x \in P_2$. Suppose all vertices enclosed by $C$ are in $P_1$ or $P_2$. Let $y$ be one of the two vertices in $C$ adjacent to $x$. Let $N_C(y)$ be all vertices in $N(y)$ that are in or inside $C$, and suppose $N_C(y) \cap P_1$ is disconnected. Let $z$ be any neighbor of $x$ inside $C$. Then there exists a recombination step for $P_1$ and $P_2$, only changing the district assignment of vertices enclosed by $C$ and leaving $|P_1|$ and $|P_2|$ unchanged, after which $N_C(y) \cap P_1$ is connected and $z \in P_2$.
\end{lem}
\begin{proof}
	Because $N_C(y)$ is disconnected and $P_3$ is outside $C$, $y$ must have a neighbor inside $C$ that is in $P_2$, meaning there is at least one vertex of $P_1$ inside $C$. The proof of Lemma~\ref{lem:cycle-recom} works with any breadth first search tree for the interior of $T$, so we choose a breadth first search tree in which $z$ is the very first vertex visited from $x$.  This means that $z$ is guaranteed to be in $P_2$ when the vertices inside $C$ are reassigned: only the last $m$ vertices added to the BFS tree are put in $P_1$, and $m$ is strictly less than the number of vertices inside $C$, so as the first vertex added to the tree $z$ will never be placed in $P_1$. 
\end{proof}

\subsection{Correcting Imbalances: Overview of Cases}


We will suppose throughout the rest of this section that we have a partition with districts $P_1$, $P_2$, and $P_3$ where $|P_1| = k_1 + 1$ and $|P_3| = k_3 - 1$. We wish to correct this imbalance without affecting any of the vertices in $P_1 \cap \clei$, and we assume $P_1$ has at least one vertex in $\cc_{i+1}$, that is, that the sweep line procedure is not yet completed. By Lemma~\ref{lem:1-to-change} there exists at least one $v \in P_1 \cap \cgi$ that can be removed from $P_1$ and added to another district; we assume all such $v$ can be added to $P_2$ but not to $P_3$, as otherwise we are done. 

The four cases we will consider, which will show are disjoint and span all possibilities for a partition $P$ where $\cc_1 \in P_1$, are as follows: 
{\begin{enumerate}[label=(\Alph*)]
	\item There exists $a \in P_2 \cap bd(T)$ and $b \in P_3 \cap bd(T)$ that are adjacent
	\item $P_2 \cap bd(T) = \emptyset$
	\item $P_3 \cap bd(T) = \emptyset$ 
	\item No vertex of $P_2$ is adjacent to any vertex of $P_3$
\end{enumerate}}
Cartoonish depictions of what the relationship between $P_1$, $P_2$, and $P_3$ looks like in each of these four cases can be found in Figure~\ref{fig:4cases}.

\renewcommand{\thesubfigure}{\Alph{subfigure}}

\begin{figure}
	\centering 
	\begin{subfigure}[b]{0.35\textwidth}
	\centering
	\includegraphics[scale = 0.45]{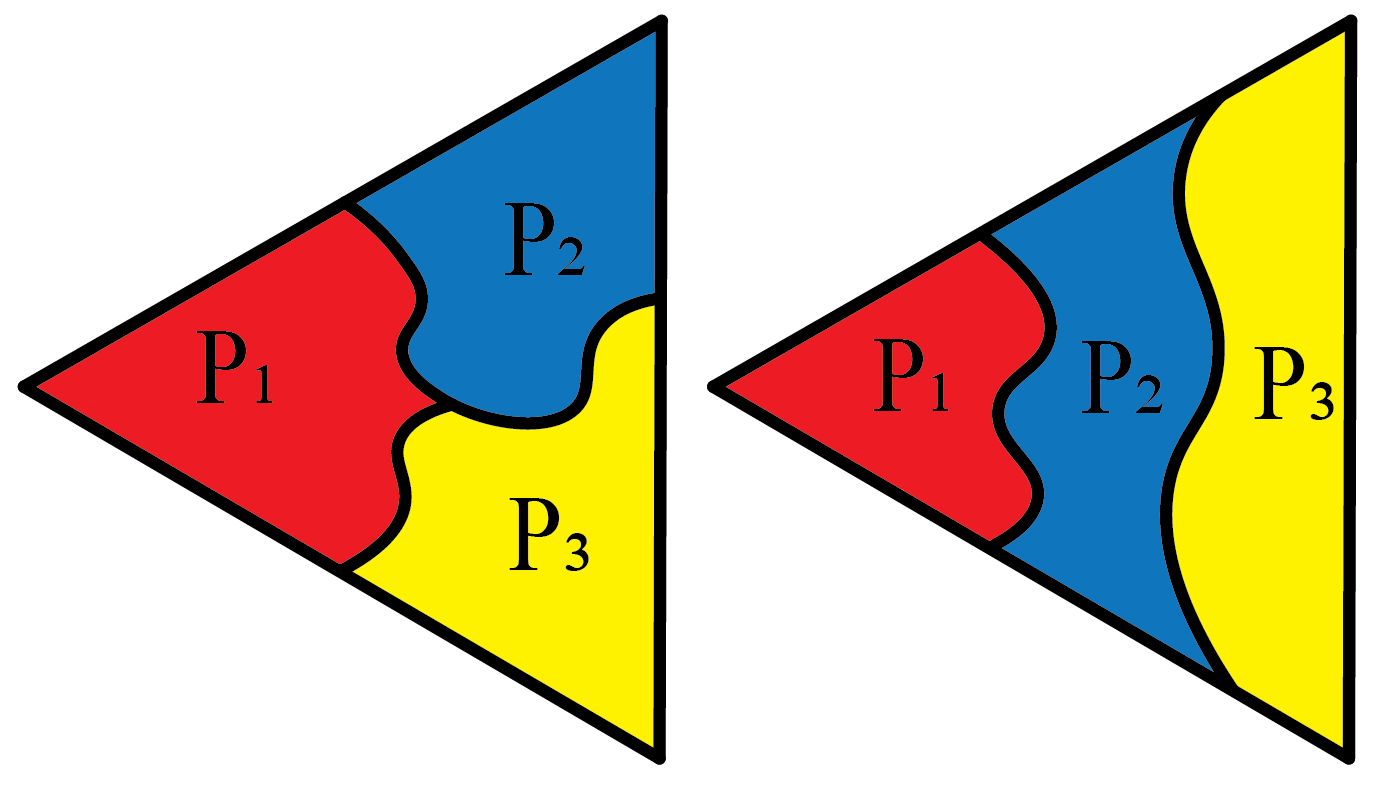}
	\caption{}
	\label{fig:4cases-A}
\end{subfigure}
\hfill
\begin{subfigure}[b]{0.2\textwidth}
	\centering
	\includegraphics[scale = 0.45]{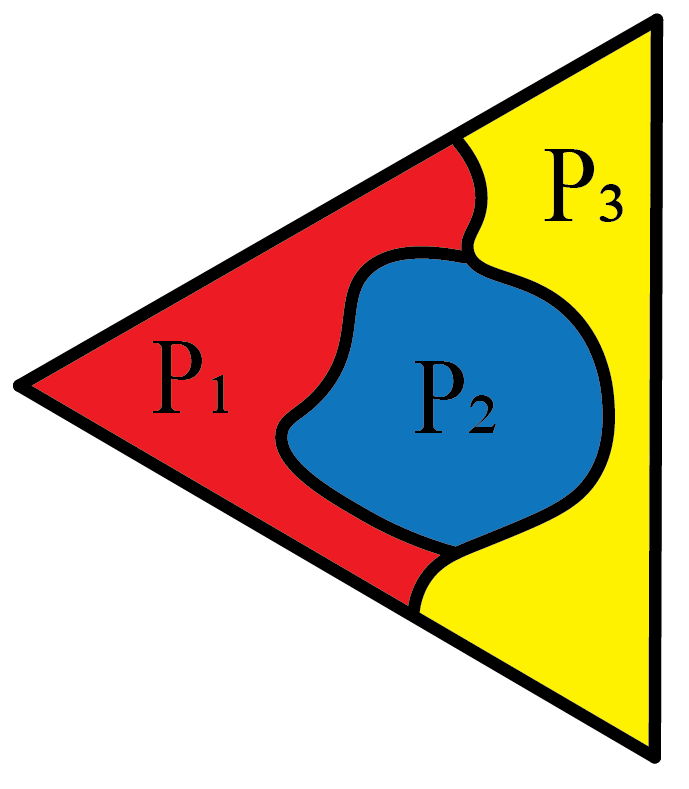}
	\caption{}
	\label{fig:4cases-B}
\end{subfigure}
\hfill
\begin{subfigure}[b]{0.2\textwidth}
	\centering
	\includegraphics[scale = 0.45]{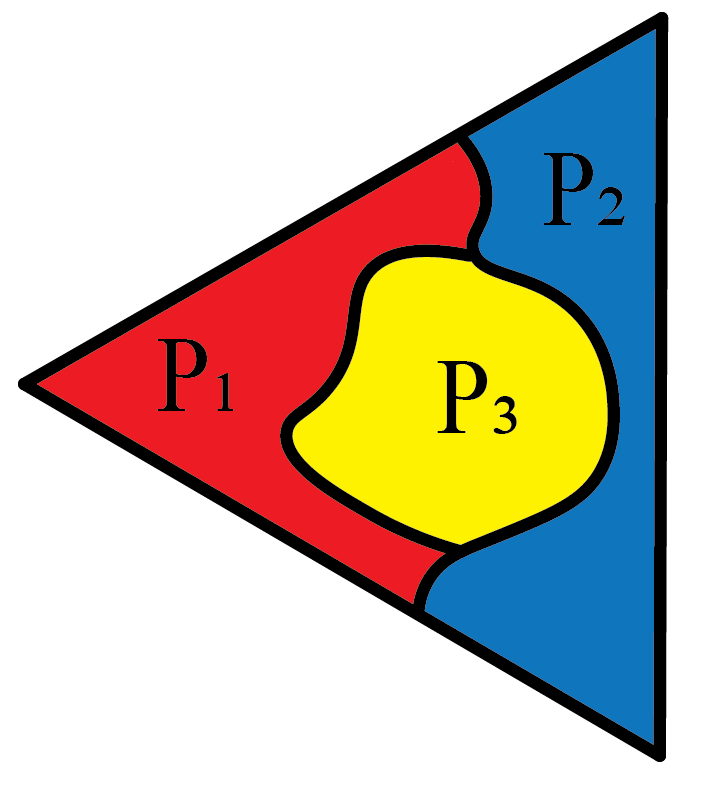}
	\caption{}
	\label{fig:4cases-C}
\end{subfigure}
\hfill
\begin{subfigure}[b]{0.2\textwidth}
	\centering
	\includegraphics[scale = 0.45]{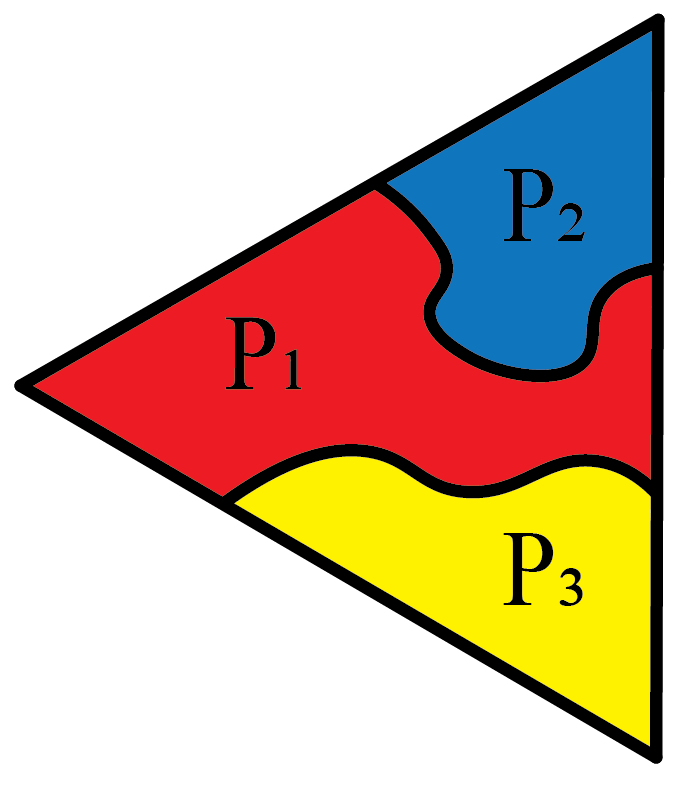}
	\caption{}
	\label{fig:4cases-D}
\end{subfigure}
	\caption{Cartoon representations of the relationship between $P_1$, $P_2$, and $P_3$ in the four rebalancing cases we consider: (A), (B), (C), and (D). There are two possibilities in Case (A), but we handle them simultaneously. }
	\label{fig:4cases}
\end{figure}

\renewcommand{\thesubfigure}{\alph{subfigure}}

We begin with a series of lemmas that show these four cases are disjoint and cover all possible partitions, and some lemmas about additional structures that must exist in the various cases.  We then, in subseqeunt subsections, consider each case individually and show how to reach a balanced partition without changing $P_1 \cap \clei$. 
This lemma first shows Case (B) and Case (C) cannot occur at the same time. 

\begin{lem}\label{lem:p2p3int}
	Let $P$ be a partition. It is not possible to have both $P_3 \cap bd(T) = \emptyset$ and $P_2 \cap bd(T) = \emptyset$. 
\end{lem}
\begin{proof}
	Suppose, for the sake of contradiction, that $P_2 \cap bd(T) = \emptyset$ and $P_3 \cap bd(T) = \emptyset$. Then it must be that every vertex in $bd(T)$ is in $P_1$. The vertices of $bd(T)$ form a cycle. This cycle, consisting entirely of vertices in $P_1$, must encircle all the vertices of $P_2$ and $P_3$.  This implies $P_1$ is not simply connected, a contradiction.  Therefore at least one of $P_2 \cap bd(T) \neq \emptyset$ and $P_3 \cap bd(T) \neq \emptyset$ must be true.
\end{proof}

\noindent It is straightforward to see that if Case (A) occurs, then Cases (B), (C), and (D) cannot occur. This next lemma shows the converse: if Cases (B), (C), and (D) do not occur, then Case (A) must occur. 
\begin{lem}\label{lem:23bd}
	Let $P$ be a partition where $\cc_1 \in P_1$,  $P_2$ and $P_3$ are adjacent, $P_2 \cap bd(T) \neq \emptyset$, and $P_3 \cap bd(T) \neq \emptyset$. Then there exist adjacent vertices $a \in P_2 \cap bd(T)$ and $b \in P_3 \cap bd(T)$. 
\end{lem}
\begin{proof}
	Let $a' \in P_2$ and $b' \in P_3$ be adjacent, and suppose it is not true that $a' \in bd(T)$ and $b' \in bd(T)$. Because $P_2 \cap bd(T) \neq \emptyset$ there exists some vertex $x_2 \in P_2 \cap bd(T)$, and because $P_3 \cap bd(T) \neq \emptyset$ there exists some vertex $x_3 \in P_3 \cap bd(T)$.  Let $Q_2$ be any path in $P_2$ from $a'$ to $x_2$, and let $Q_3$ be any path in $P_3$ from $b'$ to $x_3$. Form a cycle $C$ consisting of $Q_2$, $Q_3$, and any path outside $T$ from $x_2$ to $x_3$.  Because this cycle contains no vertices of $P_1$ and $P_1$ is connected, all of $P_1$ must lie on the same side of this cycle.  Consider the path $Q_{bd}$ in $bd(T)$ from $x_2$ to $x_3$ that does not contain $\cc_1 \in P_1$; this path is entirely on the opposite side of $C$ from $P_1$, so does not contain any vertices of $P_1$.  Because it begins at a vertex of $P_2$ and ends at a vertex in $P_3$, somewhere along path $Q_{bd}$ there must be a vertex of $P_2$ adjacent to a vertex of $P_3$. These vertices $a \in P_2$ and $B \in P_3$ satisfy the conclusions of the lemma as $Q_{bd} \subseteq bd(T)$. 	
\end{proof}

\noindent The next lemma tells us that if Case (B) or (C) occurs, then Case (D) cannot also occur. 
\begin{lem}\label{lem:int_adj}
	Let $P$ be a partition such that $P_i \cap bd(T) = \emptyset$.  Then $P_i$ must be adjacent to both other districts. 
\end{lem}
\begin{proof}
	Suppose, for the sake of contradiction, that $P_i $ is adjacent to $P_j$  but not $P_k$ for distinct $i,j,k$. Let $F$ be the union of all faces in $\Gtri$ that have at least one vertex in $P_i$.	
	Let $E$ consist of all edges $e$ of $\Gtri$ such that one of the faces $e$ is incident on is in $F$ and the other face $e$ is incident on is in the infinite component of $\Gtri\setminus F$. Note no edges in $E$ can have an endpoint in $P_i$. 
	Note that these edges separate $P_i$ from the infinite region $\Gtri \setminus T$, and all are necessary for this separation, so $E$ is a minimal cut set.  It follows that the edges of $E$ must form a cycle $C$ surrounding $P_i$  (see, for example, Proposition 4.6.1 of~\cite{Diestal}). Because $P_i \cap bd(T) = \emptyset$ and each vertex of $C$ is at distance at most one from a vertex of $P_i$, then every vertex of $C$ is in $T$.  Because both other districts must be simply connected, it is impossible for all vertices of $C$ to be in the same district, so $C$ must contain some vertices of each of the other districts.  Because all vertices in $C$ are incident on a face containing a vertex of $P_i$ and thus adjacent to a vertex of $P_i$, this means $P_i$ must be adjacent to both other districts. 
\end{proof}

\noindent We can now use the previous four lemmas to show that exactly one of the four cases (A), (B), (C), or (D) must apply to any partition satisfying our intermediate step of the rebalancing procedure. 

\begin{lem}\label{lem:4cases}
	Let $P$ be a partition of $T$ such that $\cc_1 \in P_1$. The the partition satisfies exactly one of conditions (A), (B), (C), and (D). 
\end{lem}
\begin{proof}
By Lemma~\ref{lem:23bd}, if Conditions (B), (C), and (D) are not met, then condition (A) must be met; this implies that every partition such that $\cc_1$ satisfies at least one of the four conditions. 

It is straightforward to see that if Condition (A) is met, then none of (B), (C), and (D) can be satisfied. Lemma~\ref{lem:p2p3int} implies (B) and (C) cannot occur at the same time. Finally, Lemma~\ref{lem:int_adj} implies that if (B) or (C) is satisfied, then (D) cannot be satisfied.  This means at most one of these four conditions can be satisfied at once, proving the lemma. 
\end{proof}

We now move on to exploring what structures must be present in each of the cases. These lemmas are included here rather than in each case separately because some apply to multiple cases, and additionally there is some overlap in the ideas and techniques used to show each.  This first lemma applies to both Case (B) and Case (C). 


\begin{lem}\label{lem:2tricolortri} 
	Let $P$ be a partition where $P_2 \cap bd(T) = \emptyset$ or $P_3 \cap bd(T) = \emptyset$. 
	Then there exists exactly two distinct triangular faces $F'$ and $F''$ whose three vertices are in three different districts. Exactly one of these triangular faces has its vertices in $P_1$, $P_2$, and $P_3$ in clockwise order, and the other must have its vertices of $P_1$, $P_2$, and $P_3$ in counterclockwise order.
\end{lem}
\begin{proof}
	Without loss of generality, we suppose $P_2 \cap bd(T) = \emptyset$; the same argument applies if $P_3 \cap bd(T) = \emptyset$. 
	Let $F$ be the union of all faces in $\Gtri$ that have at least one vertex in $P_2$.
	For each edge $e$ of $\Gtri$, it is incident on exactly two faces of $\Gtri$. 	
	Let set $E$ contain all edges $e$ of $\Gtri$ incident on both a face of $F$ and a face in the infinite component of $\Gtri \setminus F$.  Note no edges in $E$ can have an endpoint in $P_2$. 
	Note that these edges separate $P_2$ from the infinite region $\Gtri \setminus T$, and all are necessary for this separation, so $E$ is a minimal cut set.  It follows that the edges of $E$ must form a cycle $C$ surrounding $P_2$  (see, for example, Proposition 4.6.1 of~\cite{Diestal}). Because $P_2 \cap bd(T) = \emptyset$ and each vertex of $C$ is at distance at most one from a vertex of $P_2$, then every vertex of $C$ is in $T$.  Because $P_1$ and $P_3$ must both be simply connected, it is impossible for all vertices of $C$ to be in the same district, so $C$ must contain some vertices of $P_1$ and some vertices of $P_3$. In particular, $C$ must contain at least two distinct edges $e$ with one endpoint in $P_1$ and the other endpoint in $P_3$. Each edge is incident on a face of $F$, and so by the definition of $F$ the third vertex in that face of $F$ must be in $P_2$ and inside $C$. Therefore we have found  two triangular faces whose three vertices are in three different districts. It follows easily, because both have their vertex in $P_2$ inside $C$ but have their vertices in $P_1$ and $P_3$ in opposite orders around $C$, that exactly one of these triangular faces has its vertices in $P_1$, $P_2$, and $P_3$ in clockwise order, and the other must have its vertices of $P_1$, $P_2$, and $P_3$ in counterclockwise order. 
\end{proof}


\noindent The following two lemmas give some additional information about what the partition must look like in Case (D).

\begin{lem}\label{lem:23notadj_touchbdry}
	Let $P$ be a partition of $T$ where $P_2$ and $P_3$ are not adjacent. Then $P_2$ and $P_3$ must both touch the boundary of $T$. 
\end{lem}
\begin{proof}
	If $P_2 \cap bd(T) = \emptyset$ then by Lemma~\ref{lem:int_adj}, $P_2$ must be adjacent to $P_3$, a contradiction, so we conclude $P_2 \cap bd(T) \neq \emptyset$. The same argument holds for $P_3$. 	
\end{proof}
\begin{lem}\label{lem:23notadj_bdryconn}
	Let $P$ be a partition of $T$ where $P_2$ and $P_3$ are not adjacent. Then $P_2 \cap bd(T)$ is connected and $P_3 \cap bd(T)$ is connected.
\end{lem}
\begin{proof}
	Suppose, for the sake of contradiction, that $P_2\cap bd(T)$ is not connected. Let $Q$ be one connected component of $P_2 \cap bd(T)$ and let $Q'$ be a different connected component of $bd(T)$. Note these two components must be connected by a path $Q'' \subseteq P_2$. Consider the vertices in $bd(T)$ immediately before and after $Q$: because $P_2$ and $P_3$ are not adjacent, they must both be in $P_1$.  However, this is a contradiction, as path $Q''\subseteq P_2$ separates these two vertices from each other and $P_1$ is connected.  We conclude there can be at most one component in $bd(T) \cap P_2$ and, for the same reasons, in $bd(T) \cap P_3$. 
\end{proof}

We will also use the following lemma throughout multiple cases, so we include it here. It is a striaghforward consequence of our assumption that each $k_i \geq n$.  
\begin{lem}\label{lem:i_n-2}
	Let $P$ be a partition such that $\cli \subseteq P_1$ and $P_1 \cap \cgi \neq \emptyset$. Then $i \leq n-2$. 
\end{lem}
\begin{proof}
	Recall $T$ is a triangle of side length $n$, and $k_1, k_2, k_3 \geq n$. Assume, for the sake of contradiction, that $i = n-1$ or $i = n$. In either case the remaining vertices in $T$ that are not in $P_1$ would be strictly less than $2n-1$.  The number of vertices in $P_2$ and $P_3$ in a balanced or nearly balanced partition is at least $k_2 + k_3 - 1 \geq 2n-1$. This is a contradiction, so $i = n-1$ or $i = n$ is not possible. We conclude that $i \leq n-2$. 
\end{proof}

\subsection{Case A: $P_2$ and $P_3$ adjacent along boundary}\label{sec:23bdryadj}

In this case, there exists $a \in P_2 \cap bd(T)$ and $b \in P_3 \cap bd(T)$ that are adjacent. We split our main result into two cases: when removing $a$ from $P_2$ and adding it to $P_3$ is a valid move, and when it is not. Each of the following two lemmas considers one of these cases. 

\begin{lem}\label{lem:23bdryadj_abvalid}
		Let $P$ be a partition such that $\cli \subseteq P_1$, $P_1 \cap \cgi \neq \emptyset$, $|P_1| = k_1 + 1$, and $|P_3| = k_3 - 1$.  Suppose there exists adjacent vertices $a \in P_2 \cap bd(T)$ and $b \in P_3 \cap bd(T)$ where $a$'s 2-neighborhood and 3-neighborhood are connected. There there exists a sequence of moves resulting in a balanced partition that does not reassign any vertices in $P_1 \cap \clei$. 
\end{lem}
\begin{proof}
	Let $v \in P_1 \cap \cgi$ be a vertex that can be removed from $P_1$ and added to another district; $v$ exists by Lemma~\ref{lem:1-to-change}. If $v$ can be added to $P_3$ we are done, so we assume $v$ can only be added to $P_2$. 
	Note because $a$ has a connected 2-neighborhood and 3-neighborhood, $a$ can be removed from $P_2$ and added to $P_3$.  If $v \notin N(a)$, we add $v$ to $P_2$ and subsequently add $a$ to $P_3$, producing a balanced partition. However, if $v \in N(a)$, adding $v$ to $P_2$ could potentially disconnected $a$'s 2-neighborhood. We consider the case $v \in N(a)$ in more detail.  
	
	We examine $P_1 \cap N(a)$, which must be nonempty if $v \in N(a)$.  
	First, suppose the component of $N(a) \cap P_1$ containing $v$ has only one vertex, $v$. Because this component of $N(a) \cap P_1$ does not contain any vertices other than $v$, both of $v$'s neighbors in $N(a)$ must be in $P_2$, in $P_3$, or outside of $T$. Note $(P_3 \cup (\Gtri \setminus T)) \cap N(a)$ must be connected because $a$ has a connected 3-neighborhood and $b \in P_3 \cap bd(T)$.  Therefore, it cannot be the case that both of $v$'s neighbors in $N(a)$ are in $P_3 \cup (\Gtri \setminus T)$: this would mean $N(a) \subseteq \{v\} \cup P_3 \cup (\Gtri \setminus T)$, impossible because $a$ must have at least one neighbor in $P_2$. We conclude $v$ must have at least one neighbor in $P_2$. This means adding $v$ to $P_2$ cannot result in $a$'s 2-neighborhood becoming disconnected, because it simply grows an existing component of $N(a) \cap P_2$ rather than creating a new one. In this case, we add $v$ to $P_2$ and then add $a$ to $P_3$, which will remain a valid move.

	If the component of $P_1 \cap N(a)$ containing $v$ is of size 2, we must be more careful. Note because $a$ only has four neighbors in $T$, has neighbor $b \in P_3$, and has at least one neighbor in $P_2$ as $|P_2| > 1$, two is the largest possible size for $P_1 \cap N(a)$. Let $W = P_1 \cap N(a)$ be this component of size 2. 
	If $P_1 \setminus W$ has at least two connected components, we let $S$ be one component not containing $\cc_1$. We will show using the Shrinkability Lemma applied to $S$ (Lemma~\ref{lem:shrinkable}) we can find another vertex $v' \in P_1 \cap \cgi$ that can be removed from $P_1$ and added to another district, where $v' \notin N(a)$. 
	 To apply Condition~\ref{item:exp_corner} to this lemma, we first check if $S$ has at least one exposed vertex.  Let $c \in S$ be a vertex adjacent to $W$, and let $d$ be a vertex adjacent to $W$ in the component of $P_1 \setminus W$ containing $\cc_1$. In $N(W)$, which is a cycle of length $8$, there are two paths from $c$ to $d$, one clockwise and one counterclockwise.  Because $c$ and $d$ are in separate components of $P_1 \setminus W$, each of these paths must contain a vertex not in $P_1$. Let $e$ be the first vertex that is not in $P_1$ on one of these paths from $c$ to $d$ in $N(W)$, and let $f$ be the first vertex that is not in $P_1$ on the other of these paths. It is not possible that both $e \notin T$ and $f \notin T$, so it must be that at least one of $e$ and $f$ is in $P_2$ or $P_3$, meaning $S$ has an exposed vertex, the last vertex on the appropriate path before $e$ or $f$.  By Condition~\ref{item:exp_corner} of Lemma~\ref{lem:shrinkable}, there exists a vertex $v' \in S$ that can be removed from $P_1$ and added to another district; such a vertex must be in $\cgi$ by Lemma~\ref{lem:s1_c1}. If $v'$ can be added to $P_3$, we do so and are done.  If $v'$ can only be added to $P_2$, we do so.  Note that $v' \notin N(a)$, so $a$ still has a connected 2-neighborhood and 3-neighborhood.  We then add $a$ to $P_3$ to reach a balanced partition.

	If $W = P_1 \cap N(a)$ is a component of size 2 and $P_1 \setminus W$ only has one component, additional cases are required. Label $a$ and $b$'s common neighbor in $T$ as $c$, and label $a$ and $c$'s other common neighbor as $d$, and label $a$ and $d$'s other common neighbor as $e$, where $e \in bd(T)$; see Figure~\ref{fig:23bdryadj_size2cases}(a). Because $a$ only has four neighbors in $T$, there are only two ways for $P_1 \cap N(a)$ to have a component of size 2: either $c$ and $d$ are in $P_1$ while $e \in P_2$, or $d$ and $e$ are in $P_1$ while $c \in P_2$; see Figure~\ref{fig:23bdryadj_size2cases}(b,c).
	
	\begin{figure}
	\begin{subfigure}[b]{0.3\textwidth}
		\centering
		\includegraphics[scale = 1]{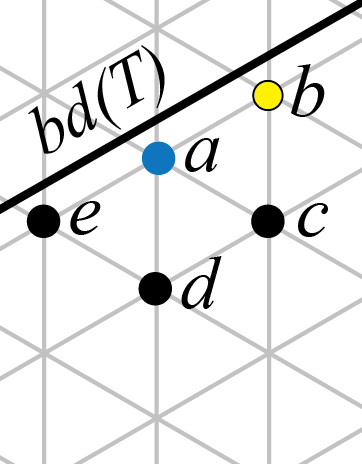}
		\caption{}
		\label{fig:23bdryadj_size2cases_nbhdlabels}
	\end{subfigure}
	\hfill
	\begin{subfigure}[b]{0.3\textwidth}
		\centering
		\includegraphics[scale = 1]{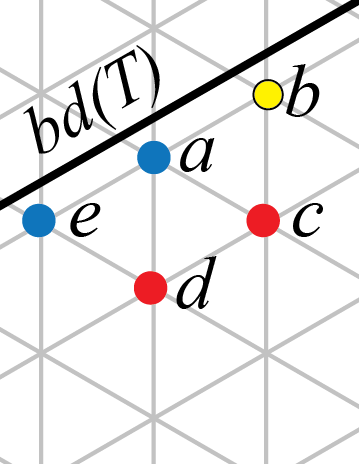}
		\caption{}
		\label{fig:23bdryadj_size2cases-cdP1}
	\end{subfigure}
	\hfill
	\begin{subfigure}[b]{0.3\textwidth}
		\centering
		\includegraphics[scale = 1]{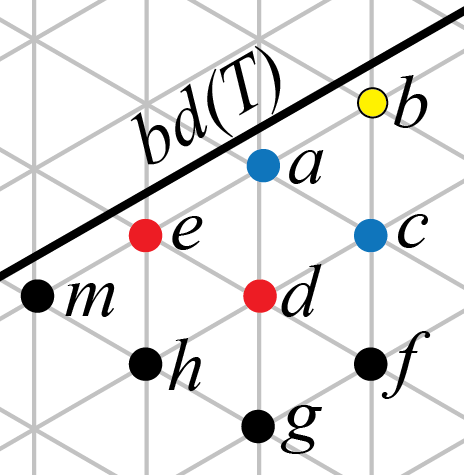}
		\caption{}
		\label{fig:23bdryadj_size2cases-deP1}
	\end{subfigure}
		\caption{In Lemma~\ref{lem:23bdryadj_abvalid}, vertex $a \in P_2 \cap bd(T)$ is adjacent to $b \in P_3 \cap bd(T)$, and $a$ has a connected 2-neighborhood and 3-neighborhood. In this figure $P_2$ is blue, $P_3$ is yellow, and vertices whose district is unknown are black. (a) The labeling used for the vertices in $a$'s neighborhood. (b,c) The two possibilities for when $a$'s 1-neighborhood has a component of size 2.} \label{fig:23bdryadj_size2cases}
	\end{figure}
	
	First, suppose $c,d \in P_1$. If $v \in P_1$ is not in $\{c,d\}$, we can add $v$ to $P_2$ and subsequently add $a$ to $P_3$ to produce a balanced partition.  If $v = d$, adding $v$ to $P_2$ doesn't disconnect $a$'s 2-neighborhood, so we add $v = d$ to $P_2$ and add $a$ to $P_3$. The most challenging case is when $v = c$, as adding it to $P_2$ will disconnect $a$'s 2-neighborhood, meaning we can't subsequently  add $a$ to $P_3$. In this case, because $v = c$, we know $c$'s 1-neighborhood is connected and $c \in \cgi$. We know $c$'s 3-neighborhood is not empty because $c$ is adjacent to $b \in P_3$. If $c$'s 3-neighborhood is not connected, we apply Lemma~\ref{lem:p1_3nbhd_discon_bdry}; if $c$'s 3-neighborhood is connected, we can add $c$ directly to $P_3$. In either case we reach a balanced partition. 	

	Next, suppose $d,e \in P_1$. If $v \neq e$, just as in the previous case we can add $v$ to $P_2$ and subsequently add $a$ to $P_3$. If $v = e$, we will show that in fact $v' = d$ is also a vertex in $\cgi$ that can be removed from $P_1$ and added to $P_2$ and use $v'$ instead. We begin by examining $e$'s other neighbor $m \neq a$ in $bd(T)$; see Figure~\ref{fig:23bdryadj_size2cases}(c). We note $m$ must be in $P_1$: if $m \in P_2$ then the path from $e$ to $\cc_1$ in $P_1$ separates $m$ from $a$, while if $m \in P_3$ then the path from $e$ to $\cc_1$ in $P_1$ separates $m$ from $b$. We also note that as $P_1 \setminus \{d,e\}$ has only one component (because $W = \{d,e\}$ and we are in the case where $P_1 \setminus W$ has only one component), $P_1 \cap N(\{d,e\})$ must be connected: if it was not connected, $P_1$ would have a hole and not be simply connected. Let $h$ be $e$ and $m$'s common neighbor in $T$, let $g$ be $h$ and $d$'s other common neighbor, and let $f$ be $d$ and $g$'s other common neighbor (see Figure~\ref{fig:23bdryadj_size2cases}(c); these labels were chosen for consistency with later lemmas). Because we know $e$'s 1-neighborhood is connected, $h$ must be in $P_1$. It follows that $g$ may or may not be in $P_1$, and if $g$ is in $P_1$, $f$ may or may not be in $P_1$. In any case, we see that $d$'s 1-neighborhood is connected. We know (by the Alternation Lemma, Lemma~\ref{lem:alternation}) that $d$ must have a connected 2-neighborhood or 3-neighborhood, and thus can be added to $P_2$ or $P_3$. If $d$ can be added to $P_3$, do so; if $d$ can be added to $P_2$, we add $d$ to $P_2$ and then add $a$ to $P_3$. It only remains to check that $d \in \cgi$.  If $a$ and $b$ are adjacent to the top or bottom of $T$ with $a$ left of $b$, then $e$ is also left of $d$ and $d \in \cgi$ because $e = v$ is in $\cgi$.  If $a$ and $b$ are adjacent to the top or bottom of $T$ with $b$ left of $a$, then $b \notin P_1$ is left of $d$: because $\cli \subseteq P_1$, this means $b \in \cgei$ and so $d \in \cgi$. If $a$ and $b$ are along the vertical right edge of $T$, then $a,b \in \cc_n$ and $d \in \cc_{n-1}$. By  Lemma~\ref{lem:i_n-2}, $i \leq n-2$, so $d \in \cgi$. This completes the proof. 
\end{proof}

\begin{lem}\label{lem:23bdryadj_abinvalid}
	Let $P$ be a partition such that $\cli \subseteq P_1$, $P_1 \cap \cgi \neq \emptyset$, $|P_1| = k_1 + 1$, and $|P_3| = k_3 - 1$.  Suppose there exists adjacent vertices $a \in P_2 \cap bd(T)$ and $b \in P_3 \cap bd(T)$ where $a$'s 2-neighborhood is not connected or $a$'s 3-neighborhood is not connected. Then there exists a sequence of moves resulting in a balanced partition that does not reassign any vertices in $P_1 \cap \clei$. 
\end{lem}
\begin{proof}
	If $a$'s 3-neighborhood is disconnected, we apply Lemma~\ref{lem:p2_3nbhd_discon} and are done. Therefore we assume $a$'s 3-neighborhood is connected, in which case it must be that $a$'s 2-neighborhood is disconnected.
	
	Let $v \in P_1 \cap \cgi$ be a vertex that can be removed from $P_1$ and added to another district; $v$ exists by Lemma~\ref{lem:1-to-change}. If $v$ can be added to $P_3$ we are done, so we assume $v$ can only be added to $P_2$.
	
	Let $c$ be $a$ and $b$'s common neighbor in $T$; let $d$ be $a$ and $c$'s other common neighbor; and let $e$ be $a$ and $d$'s other common neighbor, where $e \in bd(T)$. The only way for $a$ to have a connected 3-neighborhood but a disconnected 2-neighborhood is to have $c , e \in P_2$ and $d \in P_1$; see Figure~\ref{fig:23bdryadj_a2discon}(a). If $a$ and $b$ are along the top or bottom of $T$, then $d$ has one of $e,c \notin P_1$ to its right, so $d$ cannot be in $\clei$ and so $d \in \cgi$.  If $a$ and $b$ are along the vertical right edge of $T$, then $a,b \in \cc_n$ and $d \in \cc_{n-1}$, and because $i \leq n-2$ (Lemma~\ref{lem:i_n-2}), $d \in \cgi$. In either case $d \in \cgi$, so if $d$ can be removed from $P_1$ and added to $P_2$, we do so, after which we can remove $a$ from $P_2$ and add it to $P_3$, resulting in a balanced partition satisfying the conclusion of the lemma.
	
	If $d$ cannot be removed from $P_1$ and added to $P_2$, it must be that $d$'s 1-neighborhood or 2-neighborhood is disconnected. Let $f \neq a$ be $d$ and $c$'s common neighbor; let $g \neq c$ be $d$ and $f$'s common neighbor; and let $h \neq f$ be $d$ and $g$'s common neighbor; see Figure~\ref{fig:23bdryadj_a2discon}(b).  We do further cases based on $d$'s neighborhood. 
	
	\begin{figure}
		\centering
			\begin{subfigure}[b]{0.22\textwidth}
			\centering
			\includegraphics[scale = 1]{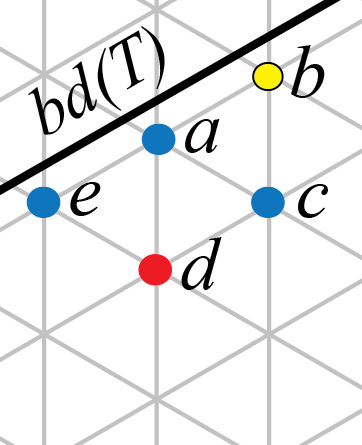}
			\caption{}
			\label{fig:23bdryadj_a2discon_2discon}
		\end{subfigure}
		\hfill
		\begin{subfigure}[b]{0.22\textwidth}
			\centering
			\includegraphics[scale = 1]{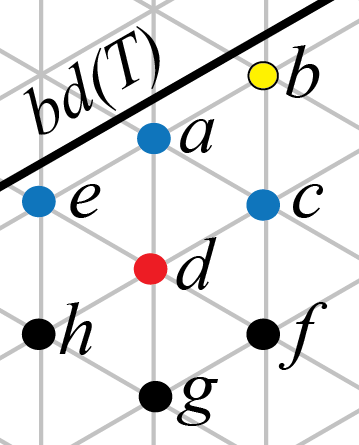}
			\caption{}
			\label{fig:23bdryadj_a2discon-fgh}
		\end{subfigure}
		\hfill
		\begin{subfigure}[b]{0.22\textwidth}
			\centering
			\includegraphics[scale = 1]{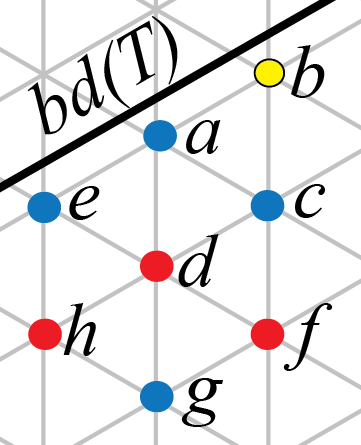}
			\caption{}
			\label{fig:23bdryadj_a2discon-g2}
		\end{subfigure}
		\hfill
			\begin{subfigure}[b]{0.22\textwidth}
			\centering
			\includegraphics[scale = 1]{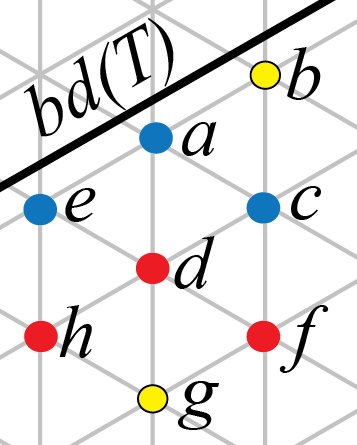}
			\caption{}
			\label{fig:23bdryadj_a2discon-g3}
		\end{subfigure}
		\caption{Figures from the proof of Lemma~\ref{lem:23bdryadj_abinvalid}. Vertices in $P_1$ are red, $P_2$ is blue, $P_3$ is yellow, and vertices whose district is unknown are black.  (a) The districts of $a$'s neighborhood if $a$'s 3-neighborhood is connected but $a$'s 2-neighborhood is disconnected. (b) The labeling of $d$'s neighborhood from the proof of Lemma~\ref{lem:23bdryadj_abinvalid} (c,d) The main configurations that must be considered when $d$' 2-neighborhood is disconnected; both $g \in P_2$ and $g \in P_3$ are possible. } \label{fig:23bdryadj_a2discon}
	\end{figure}
	
	\underline{\it $d$'s 2-neighborhood is disconnected}: First, we suppose $d$'s 2-neighborhood is not connected. In this case, it must be that $g \in P_2$ and $h,f \notin P_2$. First, note that $h \notin P_3$: If $h \in P_3$, the cycle $Q$ formed by any path in $P_3$ from $h$ to $b$ together with $e$ and $a$ separates $d \in P_1$ from $bd(T)$, a contradiction as $P_1$ is connected and includes some vertices of $bd(T)$. It follows that $h \in P_1$.  Next, suppose $f \in P_3$. In this case $d$ will necessarily have a connected 1-neighborhood and a connected 3-neighborhood.  For the same reasons as above $d \in \cgi$, so we remove $d$ from $P_1$ and add $d$ to $P_3$ and are done.  
	
	It remains to consider the case where $f,h \in P_1$; see Figure~\ref{fig:23bdryadj_a2discon}(c), which we will resolve using the Unwinding Lemma (Lemma~\ref{lem:s1s2}). 
	Consider all paths from $g$ to $a$ in $P_2$. Either all of these paths use $c$, or all of these paths use $e$: if there were paths using both, $P_2$ would not be simply connected. First, suppose all paths from $g$ to $a$ use $e$.  Let $Q$ be any such path from $g$ to $a$ in $P_1$. Let $S_1$ be the component of $P_1 \setminus d$ containing $h$, and let $S_2$ be the component of $P_2 \setminus a$ containing $c$.  Note that  $S_1 \subseteq \cgi$ by Lemma~\ref{lem:s1_c1}, so moves reassigning vertices in $S_1$ will not change $P_1 \cap \clei$. Note further that $S_1$ and $S_2$ have no adjacent vertices, as they are separated by the cycle formed by path $Q$ from $g$ to $a$ together with $d$. Because $\cc_1 \in P_1\setminus S_1$ and $(P_2 \setminus S_2 ) \cap bd(T) \neq \emptyset$ because it contains $a$ and $e$, Lemma~\ref{lem:s1s2} (the Unwinding Lemma) applies. Thus there exists a sequence moves after which (1) the partition is balanced, and we are done; (2) all vertices of $S_1$ have been added to $P_2$, or (3) all vertices of $S_2$ have been added to $P_1$. No vertices outside of $S_1$ and $S_2$ have been reassigned, and in outcomes (2) and (3), it remains true after these moves that $|P_1| = k_1 + 1$ and $|P_3| = k_3 - 1$. For outcome (2), if all vertices of $S_1$ have been added to $P_2$, then in particular $h \in S_1$ has been added to $P_2$. Because no vertices outside of $S_1$ and $S_2$ have been reassigned, $a$, $e$, and $g$ remain in $P_2$ while $f$ remains in $P_1$, while $c$ may be in $P_1$ or $P_2$. No matter the district of $c$, $d$ now has a connected 2-neighborhood and 1-neighborhood, an earlier case of this lemma we have already resolved. For outcome (3), if all vertices of $S_2$ have been added to $P_1$, then in particular now $c \in P_2$ while $d$ remains in $P_1$ and $e$ remains in $P_2$. This means $a$ now has  a connected 2-neighborhood and 3-neighborhood, and Lemma~\ref{lem:23bdryadj_abvalid} applies so we know it is possible to reach a balanced partition.

	If all paths from $f$ to $a$ in $P_2$ use $c$, we perform a similar process, with $S_2$ being the component of $P_2 \setminus a$ containing $e$ and $S_1$ being the component of $P_1 \setminus d$ containing $f$. Note  $S_1 \subseteq \cgi$ by Lemma~\ref{lem:s1_c1}. Components $S_1$ and $S_2$ are separated by the cycle $C$ formed by any path $Q$ from $a$ to $g$ in $P_2$ together with $d$, so are not adjacent. Because $\cc_1 \in P_1 \setminus S_1$ and $(P_2 \setminus S_2) \cap bd(T) \neq \emptyset$ because it contains $a$, Lemma~\ref{lem:s1s2} applies.  Thus there exists a sequence moves after which (1) the partition is balanced, and we are done; (2) all vertices of $S_1$ have been added to $P_2$, or (3) all vertices of $S_2$ have been added to $P_1$. No vertices outside of $S_1$ and $S_2$ have been reassigned, and in outcomes (2) and (3), it remains true after these moves that $|P_1| = k_1 + 1$ and $|P_3| = k_3 - 1$. For outcome (2), if all vertices in $S_1$ have been added to $P_2$, then in particular $f$ has been added to $P_2$ and $d$ now has a connected 2-neighborhood and 1-neighborhood, an earlier case of this lemma we have already resolved. For outcome (3), if all vertices of $S_2$ have been added to $P_1$, then in particular $e$ has been added to $P_1$ and $a$ now has a connected 2-neighborhood and 3-neighborhood. This means Lemma~\ref{lem:23bdryadj_abvalid} applies and so we can reach a balanced partition. 
	
	\underline{\it $d$'s 2-neighborhood is connected:} Because we have already resolved the case where $d$ has a connected 1-neighborhood and 2-neighborhood, we assume $d$'s 1-neighborhood is disconnected. It must be that $g \in P_3$ and $f,h \in P_1$; see Figure~\ref{fig:23bdryadj_a2discon}(d). Let $S_1$ be the component of $P_1 \setminus d$ containing $f$, and let $S_2$ be the component of $P_2 \setminus a$ containing $e$.  Consider the cycle $Q$ formed by any path from $b$ to $g$ in $P_3$ together with $a$ and $d$.  This cycle separates $S_1$ from $bd(T)$, meaning $\cc_1 \notin S_1$, and so by Lemma~\ref{lem:s1_c1}, $S_1 \subseteq \cgi$. This cycle also separates $S_1$ from $S_2$, as $S_1$ is inside it and $S_2$ is outside it. Additionally $P_2 \setminus S_2$ contains $a \in bd(T)$.  We again apply the Unwinding Lemma, Lemma~\ref{lem:s1s2}. After this, we have (1) reached a balanced partition without reassigning any vertices in $P_1 \cap \clei$, (2) all vertices of $S_1$ have been added to $S_2$, or (3) all vertices of $S_2$ have been added to $S_1$.  Only vertices in $S_1$ and $S_2$ have been reassigned, and in outcomes (2) and (3), the resulting partition has $|P_1| = k_1 + 1$ and $|P_3| = k_3 - 1$. For outcome (2), $d$'s 1-neighborhood and 2-neighborhood are now connected, an earlier case we've already resolved.  For outcome (3), $a$'s 2-neighborhood and 3-neighborhood are now connected, and Lemma~\ref{lem:23bdryadj_abvalid} applies. 
	
	In all cases, we have shown there exists a sequence of moves resulting in a balanced partition that does not reassign any vertices in $P_1 \cap \clei$. 
\end{proof}

The following corollary resolves Case A, where $P_2$ and $P_3$ have adjacent vertices in $bd(T)$. 

\begin{cor}\label{cor:23bdryadj}
	Let $P$ be a partition such that $\cli \subseteq P_1$, $P_1 \cap \cgi \neq \emptyset$, $|P_1| = k_1 + 1$, and $|P_3| = k_3 - 1$.  Suppose there exists adjacent vertices $a \in P_2 \cap bd(T)$ and $b \in P_3 \cap bd(T)$. Then there exists a sequence of moves resulting in a balanced partition that does not reassign any vertices in $P_1 \cap \clei$. 
\end{cor}
\begin{proof}
If $a$'s 2-neighborhood is connected and $a$'s 3-neighborhood is connected, we apply Lemma~\ref{lem:23bdryadj_abvalid}.  If at least one of $a$'s 2-neighborhood and 3-neighborhood is disconnected, we apply Lemma~\ref{lem:23bdryadj_abinvalid}. 
\end{proof}

\subsection{Case B: $P_2 \cap bd(T) = \emptyset$ }

In this case and the next, $P_2$ and $P_3$ are adjacent but there do not exist adjacent vertices $a \in P_2 \cap bd(T)$ and $b \in P_3 \cap bd(T)$. 

We begin with the following lemma; we note it does not use the condition that $P_2 \cap bd(T) = \emptyset$, so it applies in both this case (Case B) and the next case (Case C). 

\begin{lem}\label{lem:23int_a23con}
	Let $P$ be a partition such that $\cli \subseteq P_1$, $P_1 \cap \cgi \neq \emptyset$, $|P_1| = k_1 + 1$, and $|P_3| = k_3 - 1$. Suppose there exists adjacent vertices $a \in P_2$ and $b \in P_3$ such that $a$ and $b$ are not both in $bd(T)$. If $a$'s 2-neighborhood and 3-neighborhood are both connected, there exists a sequence of moves resulting in a balanced partition that does not reassign any vertices in $P_1 \cap \clei$. 
\end{lem}

\begin{proof}
By Lemma~\ref{lem:1-to-change}, there exists $v \in P_1 \cap \cgi$ that can be removed from $P_1$ and added to another district.  If there exists any such $v\in P_1 \cap \cgi$ that can be removed from $P_1$ and added to $P_3$, we do so and are done, so we assume all such $v$ can only be added to $P_2$. If there are multiple vertices of $P_1 \cap \cgi$ that can be removed and added to $P_2$, we let $v$ be the one whose distance to $\cc_1$ along paths in $P_1$ is longest.

Because $a$'s 2-neighborhood and 3-neighborhood are connected, removing $a$ from $P_2$ and adding it to $P_3$ is a valid move.  If $a$ and $v$ are not adjacent, we add $a$ to $P_3$ and subsequently add $v$ to $P_2$.  However, if $a$ and $v$ are adjacent, adding $a$ to $P_3$ may mean that subsequently adding $v$ to $P_2$ is no longer a valid move (and/or vice versa). Even if $a$ and $v$ are adjacent, it may still be that adding $v$ to $P_2$ remains a valid move even after $a$ has been added to $P_3$, and if so we do that. 

We now consider cases where $v$ and $a$ are adjacent and adding $a$ to $P_3$ means adding $v$ to $P_2$ would no be longer a valid move. Adding $a$ to  $P_3$ does not change $v$'s 1-neighborhood. However, adding $a$ to $P_3$ could result in $v$'s 2-neighborhood becoming disconnected or becoming empty. 

First, suppose adding $a$ to $P_3$ would result in $v$'s 2-neighborhood becoming disconnected.  If this is the case, it must be that $a$'s two common neighbors with $v$ are both in $P_2$; we call these neighbors $c$ and $d$. One path from $c$ to $d$ in $N(a)$ is of length 2 and passes through only $v \in P_1$.  Because $b \in P_3$ is also a neighbor of $a$, the other path from $c$ to $d$ in $N(a)$ must pass through $b$. This means $c$ and $d$ are not connected within $N(a) \cap P_2$, as each of the two possible paths between them in $N(a)$ contains a vertex that is not in $P_2$. This is a contradiction, as we assumed $a$'s 2-neighborhood was connected.  We conclude that adding $a$ to $P_3$ cannot disconnect $v$'s 2-neighborhood. 

The last case to consider is when $a$ is $v$'s only neighbor in $P_2$, and adding $a$ to $P_3$ results in $v$'s 2-neighborhood becoming empty. We will do cases based on $N(v) \cap P_3$ and whether or not $v\in bd(T)$.

\underline{\it Case: $N(v) \cap P_3 = \emptyset$, $v \notin bd(T)$.} In this case, it must be that $v$ has five neighbors in $P_1$ and a single neighbor, $a$, in $P_2$. Label $v$'s neighbors in $P_1$, in (clockwise or counterclockwise) order beginning with $a$, as $c$, $d$, $e$, $f$, and $g$, with $g$ also adjacent to $a$; see Figure~\ref{fig:23int_cbdry-labelling} for an example. We note $b$ may be adjacent to $c$, may be adjacent to $g$, or may be adjacent to neither. We let $Q$ be the shortest path in $P_1$ from $v$ to $\cc_1$.  We note by minimality this path uses at most one of $v$'s neighbors in $P_1$, so in particular it avoids at least one of the sets $\{c,d\}$ or $\{f,g\}$.  Without loss of generality, suppose $Q$ doesn't use $c$ or $d$. 
 We let $N$ be the component of $N(a) \cap P_1$ containing $c$, $v$, and $g$, and let $W = Q \cup \clei \cup (\ci \cap P_1) \cup N$. We note $d \notin W$ because $d \notin Q$, so $P_1 \setminus W$ has at least one nonempty component, and we let $S$ be the component of $P_1 \setminus W$ containing $d$. We now must show $S$ has at least one exposed vertex so that we can apply Lemma~\ref{lem:shrinkable} and find a vertex of $P_1$ outside of $N(a)$ that can be removed from $P_1$ and added to a different district. To show this, we examine $N(c)$ in order, beginning with $v$ followed by $d$. Each vertex in $P_1$ around $c$ is in component $S$, and the last vertex of $N(c)$ in $P_1$ in this order is either (1)  a boundary vertex because the next vertex around $c$ is not in $T$, meaning $c \in bd(T)$, as in Figure~\ref{fig:23int_cbdry}(b,c), or (2) exposed because the next vertex around $c$ is not in $P_1$.

\begin{figure}
	\centering
	\begin{subfigure}[b]{0.3\textwidth}
		\centering
		\includegraphics[scale = 1]{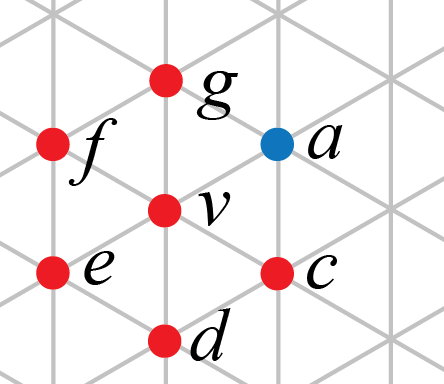}
		\caption{}
		\label{fig:23int_cbdry-labelling}
	\end{subfigure}
	\hfill
	\begin{subfigure}[b]{0.3\textwidth}
		\centering
		\includegraphics[scale = 1]{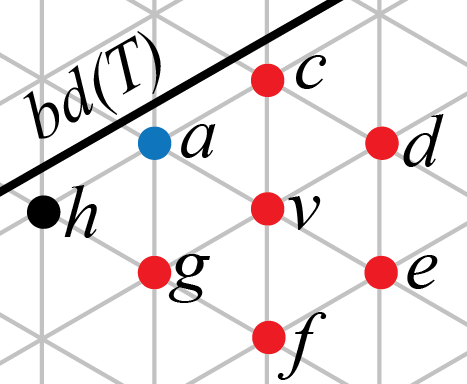}
		\caption{}
		\label{fig:23int_cbdry-abdry}
	\end{subfigure}
	\hfill
	\begin{subfigure}[b]{0.3\textwidth}
		\centering
		\includegraphics[scale = 1]{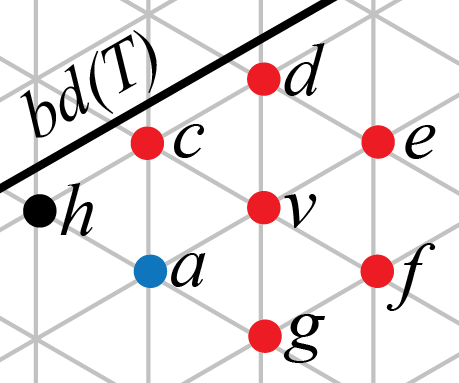}
		\caption{}
		\label{fig:23int_cbdry-aint}
	\end{subfigure}
	
	\caption{Images from the proof of Lemma~\ref{lem:23int_a23con} when $N(v) \cap P_3 = \emptyset$ and $v \notin bd(T)$. (a) The labeling used for $N(v)$. (b) The case where $c \in bd(T)$ and $a \in bd(T)$. (c) The case where $c \in bd(T)$ and $a \notin bd(T)$.  }
	\label{fig:23int_cbdry}
\end{figure}

We consider option (1) first. Because we know $v \notin bd(T)$, if $c \in bd(T)$ then either  $a \in bd(T)$ or $d \in bd(T)$. If $a \in bd(T)$, see Figure~\ref{fig:23int_cbdry-abdry}. Vertex $a$'s two neighbors in $T \setminus bd(T)$ must be $v$ and $g$, leaving only one more neighbor of $a$ in $T$ unaccounted for: this neighbor $h$ must be in $P_2$ because $|P_2| > 1$, but $h$ must also be in $P_3$ because we chose $a$ to be adjacent to $b \in P_3$.  A vertex cannot simultaneously be in $P_2$ and $P_3$, so this contradiction implies $a \notin bd(T)$. Thus it must be that $d \in bd(T)$; see Figure~\ref{fig:23int_cbdry-aint}. Vertex $c$'s two neighbors in $T \setminus bd(T)$ must be $v$ and $a$, and its remaining neighbor in $bd(T)$ that is not $d$ we call $h$. 
If $h \in P_2$ or $h \in P_3$, instead of adding $v$ to another district, we will add $c$ to another district.  Note that $c$, $v$, and $a$ form a triangle, and so exactly two of these must be in the same column. If $c$ and $v$ are in the same column, then $c \in \cgi$ because $v$ is; if $a$ and $c$ are in the same column, then $h$ must be right of $c$ and -- when $h \in P_2$ or $h \in P_3$ -- this means $c \in \cgi$; if $a$ and $v$ are in the same column, then it must be that $c \in \cc_n$. Regardless, $c \notin \clei$. 
If $h \in P_3$, then $c$'s 1-neighborhood and 3-neighborhood are both connected, so (without adding $a$ to $P_3$ first) we add $c$ to $P_3$, immediately resulting in a balanced partition. If $h \in P_2$, then adding $a$ to $P_3$ will result in $c$ having a connected 1-neighborhood and a connected 2-neighborhood, so we add $c$ to $P_2$ (instead of adding $v$ to $P_2$), and the end result after these two valid moves is a balanced partition, as required. Finally, we must consider when $h \in P_1$, meaning $c$ is a cut vertex of $P_1$. In this case we find a contradiction.  Let $S'$ be the component of $P_1 \setminus c$ containing $h$; this component must not contain $\cc_1$, because the shortest path from $v$ to $\cc_1$ doesn't go through $c$. By Condition~\ref{item:cut_corner} of Lemma~\ref{lem:shrinkable}, there exists a $v'$ in $S'$ that can be removed from $P_1$ and added to another district. This contradicts our choice of $v$ as being the farthest such vertex from $\cc_1$ via paths in $P_1$, as $v'$ is necessarily farther from $\cc_1$ than $v$.  Because we have found a contradiction this case, $h \in P_1$ is not possible.

The remaining option to consider is (2), when component $S$ has an exposed vertex. By Condition~\ref{item:exp_corner} of Lemma~\ref{lem:shrinkable}, there exists $v' \in S$ that can be removed from $P_1$ and added to another district; because $v' \notin W$ and $W$ contains all of $P_1 \cap \clei$, $v' \in \cgi$. If $v'$ can be added to $P_3$, we do so and are done.
If $v'$ is not in $N(a)$, we add $v'$ to $P_2$ and subsequently add $a$ to $P_3$, both valid because $v'$ and $a$ are not adjacent, resulting in a balanced partition and proving the lemma. 
Because $W$ contains the component $N$ of $N(a) \cap P_1$ containing $v$, $v'$ cannot be in this component $N$. We wish to claim that $v' \notin N(a)$, however it is still possible for $v'$ to be in a different component of $N(a) \cap P_1$; the only way this can happen is if $a$'s unique neighbors adjacent to $c$ and $g$ are not in $P_1$ (one must be $b \in P_3$ and the other must be in $P_2$), and $a$'s unique neighbor that is not adjacent to $c$ or $g$ is in $P_1$, and this neighbor is $v'$. However, in this case $v'$ is adjacent to $P_3$ and has a connected 1-neighborhood.  If $v'$ has a connected 3-neighborhood, we can remove $v'$ from $P_1$ and add it to $P_3$ (without first adding $a$ to $P_3$), resulting in a balanced partition and completing the proof. If $v'$ has a disconnected $3$-neighborhood, we will carefully apply Lemma~\ref{lem:p1_3nbhd_discon}. Because $W \subseteq P_1$ contains both $\cc_0$ and neighbors of $a$, there is a path from $a \in P_2$ to $bd(T)$ that doesn't use $v'$ or any vertices of $P_3$.  That means if $C$ is a cycle formed by any path in $P_3$ connecting different components of $P_3 \cap N(v')$ together with $v'$, there is at least one vertex of $P_2$ outside of $c$, namely $a$. We conclude by Lemma~\ref{lem:p1_3nbhd_discon} that there exists a move resulting in a balanced partition that does not reassign any vertices in $P_1 \cap \clei$.

\underline{\it Case: $N(v) \cap P_3 = \emptyset$, $v \in bd(T)$, $a \in bd(T)$.} 
Let $v$ and $a$ be adjacent along the boundary, and suppose first that neither is a corner of $T$.  Let $a$ and $v$'s common neighbor be $c$, let $v$ and $c$'s other common neighbor be $d$, and let $v$ and $d$'s other common neighbor (in $bd(T)$) be $e$.  Note $c,d,e \in P_1$ because we assume $N(v) \cap P_3 = \emptyset$ and $N(v) \cap P_2 = \{a\}$. See Figure~\ref{fig:23int_avbdry}(a). 
\begin{figure}
	\centering
	\begin{subfigure}[b]{0.22\textwidth}
		\centering
		\includegraphics[scale = 0.9]{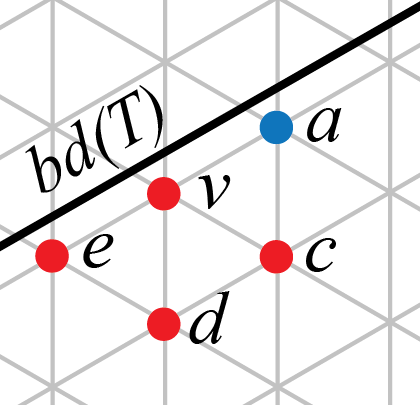}
		\caption{}
		\label{fig:23int_avbdry_labelling}
	\end{subfigure}
	\hfill
	\begin{subfigure}[b]{0.22\textwidth}
		\centering
		\includegraphics[scale = 0.9]{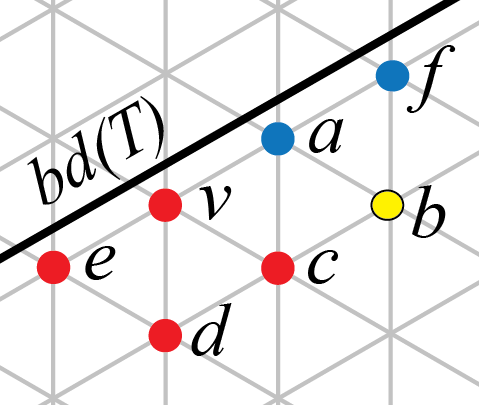}
		\caption{}
		\label{fig:23int_avbdry-bf}
	\end{subfigure}
	\hfill
	\begin{subfigure}[b]{0.22\textwidth}
		\centering
		\includegraphics[scale = 0.9]{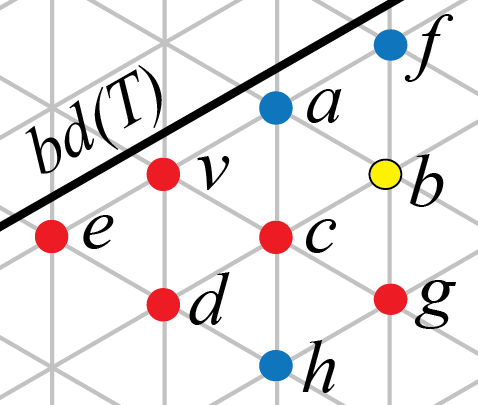}
		\caption{}
		\label{fig:23int_avbdry-gh}
	\end{subfigure}
	\hfill
	\begin{subfigure}[b]{0.22\textwidth}
		\centering
		\includegraphics[scale = 0.9]{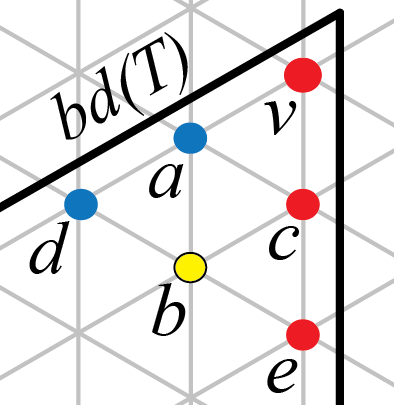}
		\caption{}
		\label{fig:23int_avbdry-corner}
	\end{subfigure}
	\caption{Images from the proof of Lemma~\ref{lem:23int_a23con} when $N(v) \cap P_3 = \emptyset$, $v \in bd(T)$, and $a \in bd(T)$. (a) The labeling and district assignments of $v$'s neighborhood; $a$ is $v$'s only neighbor in $P_2$, so all of $v$'s other neighbors must be in $P_1$.  (b) The labeling and district assignment of $a$'s neighborhood. (c) When $c$ has a disconnected 1-neighborhood, this is the labeling and district assignment of $c$'s neighborhood. (d) The labeling and district assignment near $v$ when $v$ is a corner of $T$.  }
	\label{fig:23int_avbdry}
\end{figure}
Note that $a$ has neighbor $b \in P_3$, and by assumption $a$ and $b$ cannot both be in $bd(T)$, so $b$ must be $a$ and $c$'s other common neighbor (not $v$). Additionally, $a$ must have a neighbor in $P_2$ because $|P_2| > 1$. Thus $a$'s other neighbor in $bd(T)$, which we will call $f$, must be in $P_2$.  See Figure~\ref{fig:23int_avbdry}(b). 
Note this means $c$ has a neighbor in $P_3$.  We assume by Lemma~\ref{lem:p1_3nbhd_discon} that $c$'s 3-neighborhood is connected. 

If $c$'s 1-neighborhood is also connected, then provided $c \notin \clei$ we can add $c$ to $P_3$ and reach a balanced partition.  To verify that $c \notin \clei$, first suppose that $v$ and $a$ are both part of the top or bottom boundary of $T$. In this case, $c$ is in the same column as $v$ or to the right of $v$; in either case, $c \notin \clei$ because $v \notin \clei$. If $a$ and $v$ are both in $T$'s vertical right boundary, then $c \in \cc_{n-1}$. By Lemma~\ref{lem:i_n-2}, $i \leq n-2$ so $c \in \cgi$. Thus if $c$'s 1-neighborhood is connected, we can remove $c$ from $P_1$ and add it to $P_3$, completing the lemma. 

If $c$'s 1-neighborhood is not connected, it must be that $c$ and $b$'s common neighbor $g \in P_1$ and $c$ and $g$'s other common neighbor $h \in P_2$; see Figure~\ref{fig:23int_avbdry}(c). In this case, let $S$ be the component of $P_1 \setminus \{v,c\}$ that doesn't contain $\cc_1$; $S$ necessarily contains an exposed vertex ($d$ or $g$) and $S \subseteq \cgi$ by Lemma~\ref{lem:s1_c1} (applied to cut vertex $c$).  By Condition~\ref{item:exp_corner} of Lemma~\ref{lem:shrinkable}, there exists $v' \in S \subseteq \cgi$ that can be removed from $P_1$ and added to another district. If $v'$ can be added to $P_3$, we do so and are done. Otherwise, the way $S$ was chosen ensures $v' \notin N(a)$.  In this case, we add $v'$ to $P_2$ and add $a$ to $P_3$, resulting in a balanced partition. 

The last remaining case to check is if $a$ or $v$ is a corner of $T$.  Note $a$ cannot be a corner of $T$ because it must have three neighbors: $v \in P_1$, $b \in P_3$, and some neighbor in $P_2$  because $|P_2| > 1$. Suppose $v$ is a corner of $T$.  One of its neighbors must be $a$, and its other neighbor must be $c \in P_1$.  Because we know $b \in P_3$ is adjacent to $a$ but not in $bd(T)$, it must be that $b$ is $a$ and $c$'s common neighbor.  Vertices $b$ and $c$ must have common neighbor $e \in P_1$, and $b$ and $a$'s common neighbor must be $d \in P_2$; see Figure~\ref{fig:23int_avbdry}(d). In this case, clearly $v$ and $c$ are both in $\cc_{\geq n-1}$ and so are are in $\cgi$. We achieve a balanced partition as follows.  First, add $v$ to $P_2$.  Then, apply Condition~\ref{item:exp_2bd} of Lemma~\ref{lem:shrinkable} to the new district $P_2 \cup v$ with $W = \{v\}$ to find $v' \in P_2$ that can be removed from $P_2$ and added to another district. If $v'$ can be added to $P_3$, we have reached a balanced partition and are done.  Otherwise, $v'$ can be added to $P_1$ and we do so. We then subsequently add $c$ to $P_3$, reaching a balanced partition and proving the lemma.

\underline{\it Case: $N(v) \cap P_3 = \emptyset$, $v \in bd(T)$, $a \notin bd(T)$.}  
In this case, note that $v$ cannot be a corner of $T$ because every corner's two neighbors are in $bd(T)$.  Let $c$ be $v$ and $a$'s common neighbor that is not in $bd(T)$, let $d$ be $v$ and $c$'s common neighbor in $bd(T)$, and let $e$ be $v$ and $a$'s common neighbor in $bd(T)$. 
Because we know $a$ is $v$'s only neighbor in $P_2$  and $N(v) \cap P_3 = \emptyset$, it must be that $c,d,e \in P_1$. 
However, note that in this case $v$'s 1-neighborhood is not connected, as it has two components: $\{e\}$ and $\{c,d\}$. This is a contradiction as $v$ was chosen so that it can be removed from $P_1$ and added to another district. Therefore this case cannot occur and we do not need to consider it further.

\underline{\it Case: $N(v) \cap P_3 \neq \emptyset$.} If $v$'s 3-neighborhood is connected, we add $v$ to $P_3$ and are done, so we assume that $v$'s 3-neighborhood is not connected.  Let $d$ and $e$ be two vertices in different connected components of $N(v) \cap P_3$. Let $C$ be the cycle formed by any path from $d$ to $e$ in $P_3$ together with $v$.  If $P_2$ is outside this cycle we are done by Lemma~\ref{lem:p1_3nbhd_discon}, so we assume $P_2$ is entirely inside this cycle. Examine the path from $d$ to $e$ in $N(v)$ that goes inside $C$.  This path cannot contain any vertices of $bd(T)$, and so must contain a vertex in $P_1$ or $P_2$.  First, suppose it contains a vertex $f  \in P_1$.  Note no simple path from $v$ to $\cc_1$ in $P_1$ can go through $f$ as $v$ is the only vertex in $C$ that is also in $P_1$, so any such path would not be able to leave the interior of $C$.  Thus, $v$ must have a neighbor $g \in P_1$ that is outside of $C$, where $g$ is the first vertex on the path from $v$ to $\cc_1$. Note that one path from $f$ to $g$ in $N(v)$ will contain $d$, and the other path from $f$ to $g$ in $N(v)$ will contain $c$.  This means that $f$ and $g$ are in different connected components of $N(v) \cap P_1$.  This means $v$'s 1-neighborhood is disconnected, a contradiction as we know $v$ can be removed from $P_1$ and added to $P_2$. Thus it cannot be the case that the path from $d$ to $e$ in $N(v)$ inside $c$ contains any vertices of $P_1$, and so there are no vertices of $P_1$ inside $C$. This path therefore must contain a vertex in $P_2$. 

Let $h$ be any vertex in $P_2$ on the path from $d$ to $e$ in $N(v)$ that goes inside $C$, and let $H$ be the component of $N(v) \cap P_2$ containing $h$ (possibly $H = \{h\}$). Note that $P_2 \setminus H$ is nonempty, $H$ cannot be adjacent to $P_1$, and $H \cap bd(T) = \emptyset$.  Applying Condition~\ref{item:nobd} of Lemma~\ref{lem:shrinkable} to $P_2$ and $H$, there exists $v' \in P_2$ that can be removed from $P_2$ and added to another district. Because $P_2 \setminus H$ is not adjacent to $P_1$, it must be that $v'$ can be added to $P_3$.  Adding $v$ to $P_2$ and subsequently adding $v'$ to $P_3$ results in a balanced partition, as required. 
\end{proof}

The more difficult case occurs when removing $a$ from $P_2$ and adding it to $P_3$ is not a valid move. To address this case we will need the additional structure we described above, namely that if $P_2 \cap bd(T) = \emptyset$ then $P_3 \cap bd(T) \neq \emptyset$ (Lemma~\ref{lem:p2p3int}) and that there exists at least one tricolor triangle consisting of 
$a \in P_2$, $b \in P_3$, and $c \in P_1$ (Lemma~\ref{lem:2tricolortri}). \footnote{We will not need the existence of a second tricolor triangle in this Case B, but will in the next Case C, when we no longer have $P_2 \cap bd(T) = \emptyset$.}

\begin{lem}\label{lem:23int_a23discon_3bdry}
	Let $P$ be a partition such that $\cli \subseteq P_1$, $P_1 \cap \cgi \neq \emptyset$, $|P_1| = k_1 + 1$, and $|P_3| = k_3 - 1$. Suppose there exists vertices $a \in P_2$, $b \in P_3$, and $c \in P_1$ that are incident on a common triangular face of $T$, where $a \notin bd(T)$ and $P_2 \cap bd(T) = \emptyset$.  If $a$'s 3-neighborhood is connected but $a$'s 2-neighborhood is not connected, there exists a sequence of moves resulting in a balanced partition that does not reassign any vertices in $P_1 \cap \clei$. 
\end{lem}
\begin{proof}
We begin by looking at $N(a)$, which is a cycle of size 6. We note $N(a) \subseteq T$ because $a \in P_2$ and $P_2 \cap bd(T) = \emptyset$.  Note $N(a)$ must contain adjacent vertices $b \in P_3$ and $c \in P_1$, and $N(a) \cap P_2$ must be disconnected. In (clockwise or counterclockwise) order around $a$, this means there there must be: (1) a component $B$ of $P_3 \cap N(a)$ containing $b$; (2) a component $C$ of $P_1 \cap N(a)$ containing $c$; (3) a nonempty component $D$ of $P_2 \cap N(a)$; (4) a nonempty component $E$ of $N(a) \setminus P_2$; and (5) a nonempty component $F$ of $P_2 \cap N(a)$. There then may or may not be an additional sixth component $G$ of $P_1 \cap N(a)$ between $F$ and $B$ (that is, $G$ may be empty or may be nonempty). Note $E$ cannot contain any vertices of $P_3$ as then $a$'s 3-neighborhood would be disconnected and we would be done by Lemma~\ref{lem:p2_3nbhd_discon}, so it must be that $E \subseteq P_1$.  We also note there can be no other unique components in $N(a)$: there cannot be a vertex of $P_2$ between $B$ and $C$ as $b$ and $c$ are adjacent, and there cannot be a vertex of $P_3$ between any of $C$, $D$, $E$, $F$, and $G$ (if $G$ is nonempty). If $G$ is empty, because $N(a)$ has five unique components but six vertices, it must be that exactly one of $B$,$C$, $D$, $E$, and $F$ is of size 2 while the others are of size 1. If $G$ is nonempty, then each component of $N(a)$ is of size $1$, including $G$. This means there are six cases for what $N(a)$ can look like; these are shown in Figure~\ref{fig:23int_anbhd}. 

\begin{figure}
	\centering
	\begin{subfigure}[b]{0.16\textwidth}
		\centering
		\includegraphics[scale = 0.7]{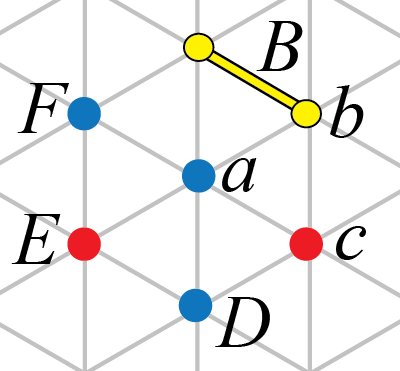}
		\caption{}
		\label{fig:23int_anbhd-B}
	\end{subfigure}
	\hfill
	\begin{subfigure}[b]{0.16\textwidth}
		\centering
		\includegraphics[scale = 0.7]{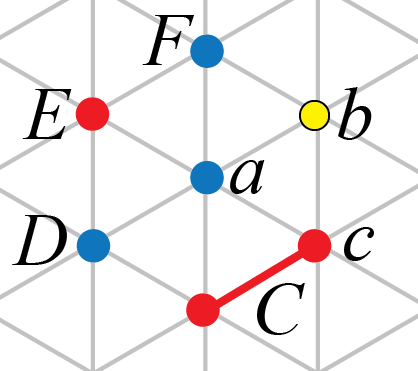}
		\caption{}
		\label{fig:23int_anbhd-C}
	\end{subfigure}
	\hfill
	\begin{subfigure}[b]{0.16\textwidth}
		\centering
		\includegraphics[scale = 0.7]{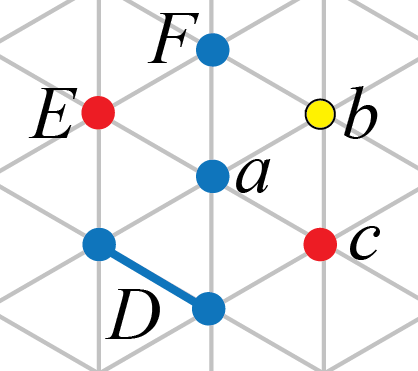}
		\caption{}
		\label{fig:23int_anbhd-D}
	\end{subfigure}
	\hfill
	\begin{subfigure}[b]{0.16\textwidth}
		\centering
		\includegraphics[scale = 0.7]{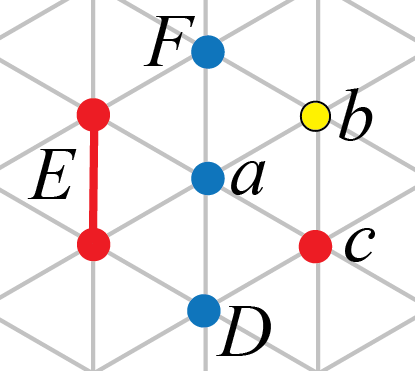}
		\caption{}
		\label{fig:23int_anbhd-E}
	\end{subfigure}
	\hfill
	\begin{subfigure}[b]{0.16\textwidth}
		\centering
		\includegraphics[scale = 0.7]{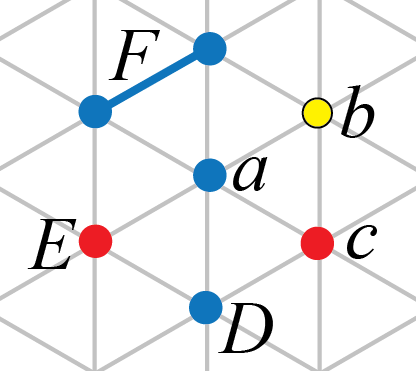}
		\caption{}
		\label{fig:23int_anbhd-F}
	\end{subfigure}
		\hfill
	\begin{subfigure}[b]{0.16\textwidth}
		\centering
		\includegraphics[scale = 0.7]{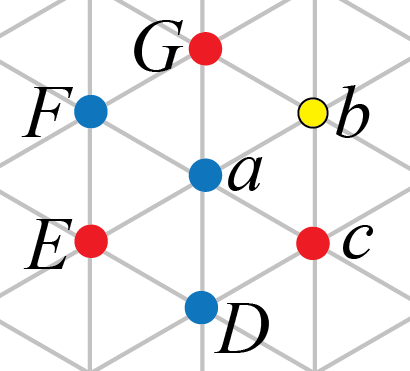}
		\caption{}
		\label{fig:23int_anbhd-G}
	\end{subfigure}
	\caption{In Lemma~\ref{lem:23int_a23discon_3bdry}, the six possible cases for $N(a)$. (a) $|B| = 2$; (b) $|C| = 2$, (c) $|D|= 2$; (d) $|E|=2$; (e) $|F| = 2$; (f) $|G| = 1$. }
	\label{fig:23int_anbhd}
\end{figure}

We already know $b \in B$ and $c \in C$; we will let $d \in D$ be such that $d$ is adjacent to $E$; we let $f \in F$ such that $f$ is adjacent to $E$; and we let $e \in E$ such that $e$ is adjacent to $d$. In all cases for $N(a)$ except when $|E| = 2$, $e$ is also adjacent to $f$.

We now focus on the vertices $e$ and $c$. We will show there exists a sequence moves, not reassigning any vertices of $P_1 \cap \clei$, that results in a balanced partition.  
 We begin with the case where $e \in \clei$. 

\underline{Case: $e \in \clei$}. First, note it is impossible to have both $e \in \clei$ and $c \in \clei$, as $e$ and $c$ have a common neighbor $a$ but are not adjacent, and some of the vertices between them in $N(a)$ are not in $P_1$ and therefore can't be in $\cli$. Therefore in this case it must hold that $c \notin \clei$. If adding $c$ to $P_3$ is valid, we do so and are done. This means $c$'s 3-neighborhood or 1-neighborhood must be disconnected. 

Suppose $c$'s $3$-neighborhood is not connected.  We know $b \in P_3$ is one of $c$'s neighbors, and let $h \in P_3$ be a vertex in a different component of $N(c) \cap P_3$ than $b$. Look at the cycle $C$ formed by any path from $b$ to $h$ in $P_3$ together with $c$. Because this cycle includes no vertices of $P_2$, then $P_2$ must be entirely inside this cycle or entirely outside this cycle. We note $a \in P_2$ is adjacent to $e \in \clei$.  Additionally, $e$ cannot be inside this cycle because it is connected to $\cc_1 \in bd(T)$ by a path not containing $c$, and if $e$ were inside $C$ every path from $e$ to $bd(T)$ would need to include $c$, the only vertex of $P_1$ in the cycle. Because $a \in P_2$ is adjacent to a vertex outside $C$, then all of $P_2$ must be outside of $C$. Therefore, by Lemma~\ref{lem:p1_3nbhd_discon} applied to $P_1$ and $c$, there exists a move resulting in a balanced partition that does not reassign any vertices of $P_1 \cap \clei$.  

Suppose that $c$'s 3-neighborhood is connected. This means $c$'s 1-neighborhood is not connected and thus $c$ is a cut vertex of $P_1$. Note  $c$'s 1-neighborhood can have at most two components, because $c$ already has $a ,b \notin P_1$ as two of its adjacent neighbors. Consider the cycle $C$ formed by any shortest path from $c$ to $e$ in $P_1$ together with $a$. Let $S_1$ be the component of $P_1 \setminus c$ not containing any vertices of $C$. Because $e \in C$ and $e \in \clei$, then all of $\clei \cap P_1$ must be in the same component of $P_1 \setminus c$ as $e$, meaning $\cc_1 \notin S_1$.   We note $P_2 \setminus a$ has exactly two components, one of which is inside $C$ and one of which is outside $C$. If $S_1$ is inside $C$, we let $S_2$ be the component of $P_2 \setminus a$ outside $C$; if $S_1$ is outside $C$, we let $S_2$ be the component of $P_2 \setminus a$ inside $C$. When chosen this way, $S_1$ and $S_2$ do not have any adjacent vertices.  Furthermore, $S_2 \cap bd(T) = \emptyset$ because $P_2 \cap bd(T) = \emptyset$.  By Lemma~\ref{lem:s1s2}, there exists a sequence of moves after which (1) the partition is balanced, (2) all vertices in $S_1$ have been added to $P_2$, or (3) all vertices of $S_2$ have been added to $P_1$. In these moves only vertices in $S_1 \subseteq \cgi$ and $S_2$ have been reassigned, and in outcomes (2) and (3) the resulting partition has $|P_1 | = k_1 + 1$ and $|P_3| = k_3 -1$. In the first outcome, we are done; in the second outcome, as $c$'s 1-neighborhood is now connected we can add $c$ directly to $P_3$; and in the third outcome, as $a$'s 2-neighborhood is now connected and $a$'s  3-neighborhood remains connected, Lemma~\ref{lem:23int_a23con} applies. In any case, we have reached a balanced partition, as desired.

\underline{Case: $e \notin \clei$ and $e \in bd(T)$}: We note $e$ must have two adjacent neighbors in $P_2$: $a$ and $d$. Because $P_2 \cap bd(T) = \emptyset$, these neighbors must not be in $bd(T)$, and $e$ cannot be a corner vertex of $T$ because $a \notin bd(T)$. It follows that $e$'s two neighbors in $bd(T)$ must not be in $P_2$. If one of these neighbors is in $P_1$ and the other is in $P_3$, we can add $e$ to $P_3$ and have reached a balanced partition. As $e$ must have at least one neighbor in $P_1$ it's not possible for both of $e$'s neighbors in $bd(T)$ to be in $P_3$, so the only remaining case to consider is when both of $e$'s neighbors in $bd(T)$ are in $P_1$. 
  We note this means that in $N(a)$ we have $|E| = 2$, so $|B| = 1$, $|C| = 1$, $|D| = 1$, $|F| = 1$, and $|G| = 0$.
We let $e$ and $a$'s common neighbor in $bd(T)$ be $g \in P_1$, and we let $e$ and $d$'s common neighbor in $bd(T)$ be $h \in P_1$; see Figure~\ref{fig:2int_ebdry}(a) for an example, when $a$'s neighbors $b$,$c$,$d$, $e$, are oriented clockwise around $a$. 
Note $e$ must be  a cut vertex of $P_1$; we let $S_1$ be the component of $P_1 \setminus e$ not containing $\cc_1$. If $S_1$ contains $g$, we let $S_2$ be the component of $P_2 \setminus a$ containing $d$; see Figure~\ref{fig:2int_ebdry}(b).  If $S_1$ contains $h$, we let $S_2$ be the component of $P_2 \setminus a$ not containing $d$, that is, the component of $P_2 \setminus a$ containing $f$; see Figure~\ref{fig:2int_ebdry}(c).  This ensures that in $N(a) \cup N(e)$, $S_1$ and $S_2$ are not adjacent, but we additionally need to argue that they are not adjacent elsewhere.

\begin{figure}
	\centering
	\begin{subfigure}[b]{0.3\textwidth}
		\centering
		\includegraphics[scale = 0.85]{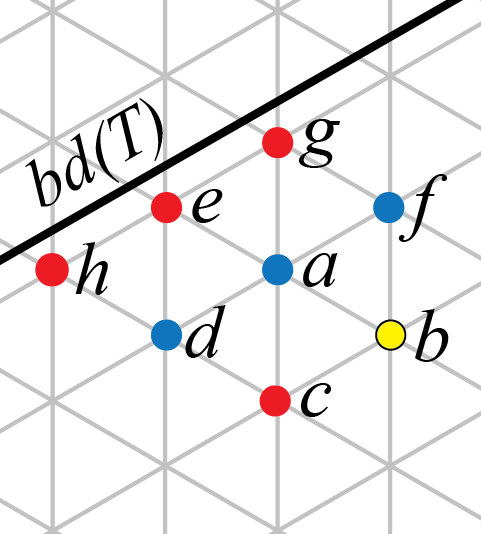}
		\caption{}
		\label{fig:2int_ebdry_labelling}
	\end{subfigure}
	\hfill
	\begin{subfigure}[b]{0.3\textwidth}
		\centering
		\includegraphics[scale = 0.85]{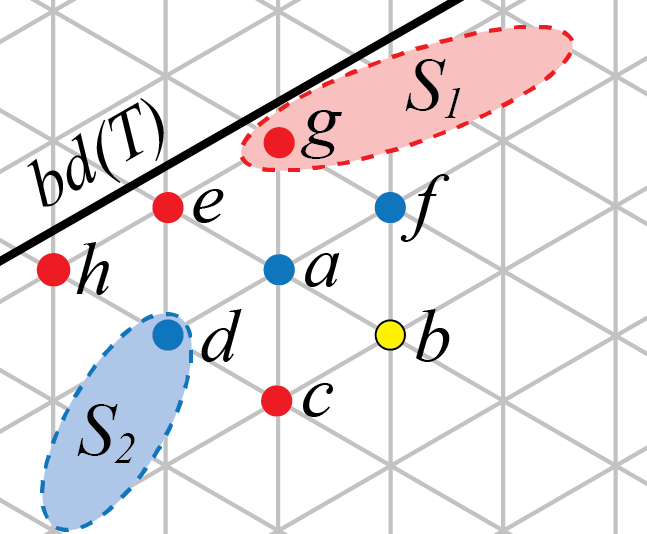}
		\caption{}
		\label{fig:2int_ebdry_g}
	\end{subfigure}
	\hfill
	\begin{subfigure}[b]{0.3\textwidth}
		\centering
		\includegraphics[scale = 0.85]{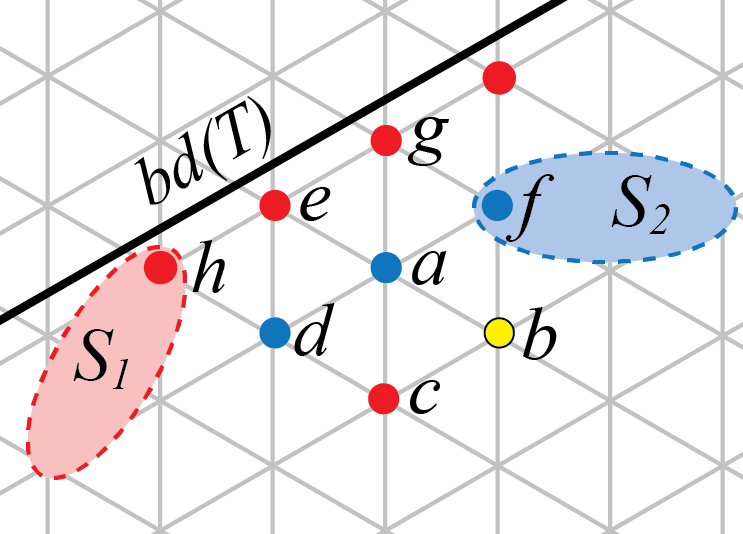}
		\caption{}
		\label{fig:2int_ebdry_h}
	\end{subfigure}
	\caption{In Lemma~\ref{lem:23int_a23discon_3bdry}, in the case where $e \notin \clei$ and $e \in bd(T)$. (a) What the partition must look like near $a$ and $e$ when $a$'s neighbors $b$, $c$, $d$, $e$, and $f$ are oriented clockwise around $a$; e's two neighbors in $bd(T)$ are labeled as $g$ (next to $f$) and $h$ (next to $d$). (b) When the component $S_1$ of $P_1 \setminus e$ not containing $\cc_1$ contains $g$, we let $S_2$ be the component of $P_2 \setminus a$ containing $d$. (c) When the component $S_1$ of $P_1 \setminus e$ not containing $\cc_1$ contains $h$, we let $S_2$ be the component of $P_2 \setminus a$ containing $f$.  \todoo{put a figure like this earlier, in Case A? Seems weird to have one here but not there. } }
	\label{fig:2int_ebdry}
\end{figure}

Suppose, for the sake of contradiction, that $S_1$ and $S_2$ are adjacent outside of $N(a)$, with vertex $x_1 \in S_1$ next to $x_2 \in S_2$. Let $C$ be the cycle formed by any path from $e$ to $x_1$ going through $S_1$, together with any path from $x_2$ to $a$ going through $S_2$.  Not that because of the way $S_2$ was chosen, this cycle $C$ must encircle $b \in P_3$, and therefore encircle all of $P_3$.  This is a contradiction, as we know $P_3 \cap bd(T) \neq \emptyset$.  We conclude that $S_1$ and $S_2$ do not have any adjacent vertices. 
 
We now apply Lemma~\ref{lem:s1s2}; we know $S_2 \cap bd(T) = \emptyset$ because $P_2 \cap bd(T) = \emptyset$.  This lemma implies there exists a sequence of moves after which (1) the partition is balanced, (2) all vertices in $S_1$ have been added to $P_2$, or (3) all vertices of $S_2$ have been added to $P_1$. In these moves only vertices in $S_1 \subseteq \cgi$ and $S_2$ have been reassigned, and in outcomes (2) and (3) the resulting partition has $|P_1 | = k_1 + 1$ and $|P_3| = k_3 -1$. In the first outcome, we are done.  In the second outcome, one of $g$ or $h$ has been added to $S_2$.  This means that it is no longer true that $P_2 \cap bd(T) = \emptyset$.  In fact, because no vertices of $P_3$ have been reassigned, it is now the case that $P_2 \cap bd(T) \neq \emptyset$, $P_3 \cap bd(T) \neq \emptyset$, and $P_2$ and $P_3$ are adjacent, for example because $a \in P_2$ remains adjacent to $b \in P_3$. By Lemma~\ref{lem:23bd}, this means there exists adjacent vertices $a' \in P_2 \cap bd(T)$ and $b' \in P_3 \cap bd(T)$. Applying Corollary~\ref{cor:23bdryadj} gives a sequence of valid steps resulting in a balanced partition. Finally, in outcome (3), now $a$ has a connected 2-neighborhood and 3-neighborhood, and applying Lemma~\ref{lem:23int_a23con} completes the proof in this case.

\underline{Case: $e \notin \clei$, $e \notin bd(T)$, $|E| = 1$}: At this point we split the argument into two further cases: when $|E| = 1$ and $|E| = 2$.  We consider the $|E| = 1$ case first, and then show how the $|E| = 2$ case can be reduced to the $|E| = 1$ case. 

If $|E| = 1$, then $E$ consists of the single vertex $e$, which is adjacent to both $d$ and $f$. If $e$ can be added to $P_2$, we do so; this results in $a$ having a connected 2-neighborhood, so we can subsequently add $a$ to $P_3$, reaching a balanced partition. If $e$ cannot be added to $P_2$, then it must have a disconnected 2-neighborhood or a disconnected 1-neighborhood.

Suppose $e$ has a connected 1-neighborhood. Because $e \notin bd(T)$, by the Alternation Lemma (Lemma~\ref{lem:alternation}) $e$ must have a connected 2-neighborhood or a connected 3-neighborhood.  If $e$ has a connected 2-neighborhood, we add $e$ to $P_2$, after which $a$ has a connected 2-neighborhood and can be added to $P_3$. If $e$ has a connected 3-neighborhood, we add $e$ directly to $P_3$ and are done. Therefore, we assume $e$ has a disconnected 1-neighborhood. This means $P_1 \setminus e$ has two connected components. Because $e$ only has six neighbors and we already know its neighbors $a$, $d$, and $f$ are in $P_2$, it must be that $e$'s neighbor adjacent to $f$, which we will call $g$, is in $P_1$ and $e$'s  neighbor adjacent to $d$, which we will call $h$, is in $P_1$ while $e$'s sixth neighbor, between $g$ and $h$, is not in $P_1$; we will not care whether this neighbor is in $P_2$ or $P_3$. See Figure~\ref{fig:2int_elabel}. 

\begin{figure}
\centering
\hfill
\begin{subfigure}[b]{0.45\textwidth}
	\centering
	\includegraphics[scale = 1]{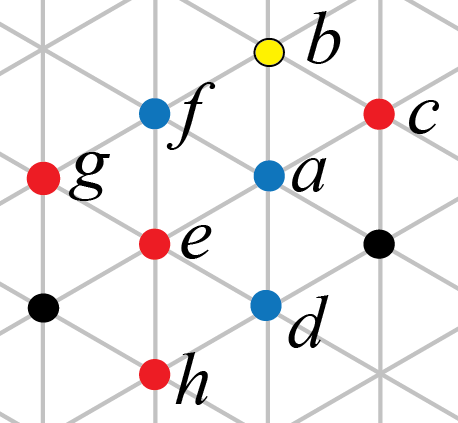}
	\caption{}
	\label{fig:2int_elabel-bf}
\end{subfigure}
\hfill
\begin{subfigure}[b]{0.45\textwidth}
	\centering
	\includegraphics[scale = 1]{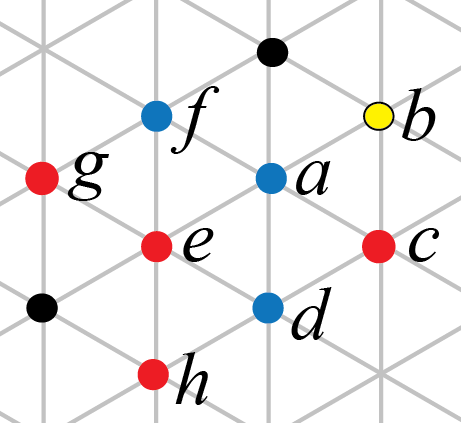}
	\caption{}
	\label{fig:2int_elabel-cd}
\end{subfigure}
\hfill
\caption{In Lemma~\ref{lem:23int_a23discon_3bdry}, in the case where $e \notin \clei$, $e \notin bd(T)$, and $|E| = 1$. These figures assume (without loss of generality) that triangle $a-b-c$ is oriented clockwise and show what the partition must look like near $a$ and $e$. In both cases, $e$'s common neighbor with $f$, which we call $g$, and $e$'s common neighbor with $d$, which we call $h$, must both be in $P_1$. $e$'s unlabeled sixth neighbor, shown in black, can be in $P_2$ or $P_3$.  (a) When $b$ is adjacent to $f$; the black vertex on the right can be in $P_1$ or $P_2$.  (b) When $c$ is adjacent to $d$ the black vertex on the top can be in any district.}
\label{fig:2int_elabel}
\end{figure}

We now show how to find two components $S_1 \subseteq P_1 \setminus e$ and $S_2 \subseteq P_2 \setminus a$ that do not share any adjacent vertices, so that we will then be able to apply Lemma~\ref{lem:s1s2x} to unwind them (Lemma~\ref{lem:s1s2} is insufficient here, so we need Lemma~\ref{lem:s1s2x} instead).

First, suppose $c$ and $\cc_1$ are in the same connected component of $P_1 \setminus e$. 
Let $C$ be the cycle formed by any path from $c$ to $e$ in $P_1$ together with $a$. Note this cycle will pass through exactly one of $g$ or $h$, and will encircle exactly one of $d$ or $f$; all four combinations are possible. Regardless, we let $S_1$ be the component of $P_1 \setminus e$ not containing $c$ and $\cc_1$, and note this means $S_1$ doesn't contain any vertices of $C$. If $S_1$ is outside $C$, we let $S_2$ be the component of $P_2 \setminus a$ inside $C$, and if $S_1$ is inside $C$, we let $S_2$ be the component of $P_2 \setminus a$ outside $C$. In either case, $S_1$ and $S_2$ are not adjacent because they are separated by $C$. This means $S_1$ contains $g$ and $S_2$ contains $d$, or $S_1$ contains $h$ and $S_2$ contains $f$.

Alternately, suppose $c$ and $\cc_1$ are in different components of $P_1 \setminus e$. Note this can only happen if $c \notin \clei$. Let $C$ be the cycle formed by any path from $e$ to $c$ in $P_1$ together with $a$.  Note that this cycle cannot encircle $b \in P_3$ because then it would separate $P_3$ from $bd(T)$ and we know $P_3 \cap bd(T) \neq \emptyset$. Therefore this cycle cannot encircle $f$ and must encircle $d$.  Let $S_1$ consist of the component of $P_1 \setminus e$ containing $c$, which also contains $C$. 
If $C$ were to contain $g$ and encircle $d$, not $f$, that would imply the component of $P_1 \setminus e$ containing $h$ must be inside $C$.  This is a contradiction, as this component must contain $\cc_1 \in bd(T)$. Therefore we conclude $C$ (and therefore, $S_1$) must contain $h$, not $g$. 
We let $S_2$ be the component of $P_2 \setminus a$ containing $f$.  We note $S_1$ and $S_2$ do not have any adjacent vertices within $N(a)$ or $N(e)$. Suppose, for the sake of contradiction, that $S_1$ and $S_2$ are adjacent outside of $N(a)$ and $N(e)$, with vertex $x_1 \in S_1$ next to $x_2 \in S_2$. Let $C'$ be the cycle formed by any path from $e$ to $x_1$ going through $S_1$, together with any path from $x_2$ to $a$ going through $S_2$.  Note that because of the way $S_1$ and $S_2$ were chosen and because $\cc_1$ and $h$ are both in $P_1 \setminus S_1$, this cycle $C'$ must encircle $b \in P_3$, and therefore encircle all of $P_3$.  This is a contradiction, as $P_3 \cap bd(T) \neq \emptyset$.  We conclude $S_1$ and $S_2$ do not have any adjacent vertices.

In either case, we have found component $S_1$ of $P_1 \setminus e$ not containing $\cc_1$ and component $S_2$ of $P_2 \setminus a$ such that $S_1$ and $S_2$ do not have any adjacent vertices.  Further, we know $S_2 \cap bd(T) = \emptyset$ because $P_2 \cap bd(T) = \emptyset$.
We do not apply Lemma~\ref{lem:s1s2}, because we need to additionally ensure that whichever of $d$ and $f$ is in $S_2$ and adjacent to $e$ is the last vertex of $S_2$ added to $P_2$. We use Lemma~\ref{lem:s1s2x} instead. Let $x$ denote whichever of $d$ or $f$ is in $S_2$.  If $S_2$ includes vertices in addition to $x$, this lemma tells us there exists a sequence of moves through balanced or nearly balanced partitions after which (1) the partition is balanced, (2) all vertices in $S_1$ have been added to $P_2$, or (3) all vertices of $S_2$ except $x$ have been added to $P_1$. In these moves only vertices in $S_1$ and $S_2\setminus x$ have been reassigned and in outcomes (2) and (3) the resulting partition has $|P_1| = k_1 + 1$ and $|P_3| = k_3 - 1$.

If outcome (1) is achieved, we are done.  If outcome (2) is achieved, then $e$ now has a connected 1-neighborhood and 2-neighborhood: All of $a$, $d$, and $f$ remain in $P_2$, because one of $d$ or $f$ is not in $S_2$ and the other is $x$; exactly one of $g$ and $h$, whichever was in $S_1$, is now also in $P_2$; and the other of $g$ and $h$ remains in $P_1$. No matter the district of $e$'s sixth neighbor, $e$ can be removed from $P_1$ and added to $P_2$.  This results in $a$ having a connected 2-neighborhood, consisting of $d$, $e$, $f$, and possibly an additional vertex adjacent to $d$ or $f$ (the black vertex in Figure~\ref{fig:2int_elabel}). Ensuring $a$ has a connected 2-neighborhood at this point is exactly why we need Lemma~\ref{lem:s1s2x} rather than Lemma~\ref{lem:s1s2} so that $d$, $e$, and $f$ are all in $P_2$. As no vertices in $P_3$ have been reassigned, $a$ still has a connected 3-neighborhood, so we add $a$ to $P_3$ to reach a balanced partition. If outcome (3) is achieved but outcome (2) is not achieved, it is now the case that the component of $P_2 \setminus a$ containing $x$ only has a single vertex, $x$. This is the last case we must consider. 

Suppose that $S_2 = \{x\}$. This means $x$ has a connected 2-neighborhood, consisting only of $a$. We note $x \notin bd(T)$, because $x \in P_2$ and $P_2 \cap bd(T) = \emptyset$.  Therefore $x$ must have a connected 1-neighborhood or 3-neighborhood by Lemma~\ref{lem:alternation}. We can also find $v_1$ in the component of $P_1 \setminus e$ not containing $\cc_1$ (what remains of $S_1$) that can be removed and added to another district by Condition~\ref{item:cut_corner} of Lemma~\ref{lem:shrinkable}. If $v_1$ can be added to $P_3$ we do so and are done, so we assume that $v_1$ can be added to $P_2$. Note that because of how $S_1$, $S_2$, and $x$ were chosen, $v_1$ is not adjacent to $x$. 
If $x$ has a connected 3-neighborhood, we add $x$ to $P_3$ and add $v_1$ to $P_2$ to reach a balanced partition. If $x$ has a connected 1-neighborhood, we add $v_1$ to $P_2$ and add $x$ to $P_1$.  As $x$ was the only vertex in its component of $P_2 \setminus a$, there is now only a single component in $P_2 \setminus a$, meaning $a$ now has a connected 2-neighborhood.  As no vertices in $P_3$ have been reassigned, $a$ still has a connected 3-neighborhood as well, so we are done by Lemma~\ref{lem:23int_a23con}.   
This concludes the case where $|E| = 1$. 

\underline{Case: $e \notin \clei$, $e \notin bd(T)$, $|E| = 2$}: 
Recall $E$ is the component of $N(a) \cap P_1$ containing $e$. We now consider the case where $|E| = 2$; see Figure~\ref{fig:23int_anbhd-E} or Figure~\ref{fig:2int_E2_label}. The approach will be similar to the case where $|E| = 1$, but rather than showing we can reach a balanced partition directly, we show we can reach a balanced partition or reach a configuration where $|E| = 1$, $|P_1| = k_1 $, and $|P_2| = k_2 + 1$; later, we show how this can be transformed into a balanced partition or the previous case where $|E| = 1$, $|P_1| = k_1 + 1$, and $|P_2| = k_2$, which we have already resolved. 

\begin{figure}
\begin{subfigure}[b]{0.3\textwidth}
	\centering
	\includegraphics[scale = 0.9]{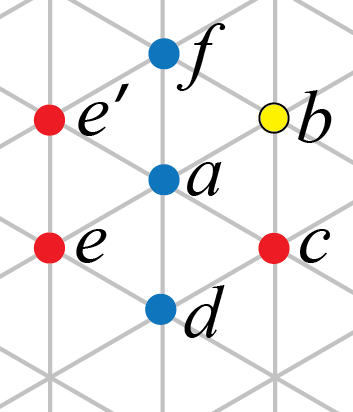}
	\caption{}
	\label{fig:2int_E2_label}
\end{subfigure}
\hfill
\begin{subfigure}[b]{0.3\textwidth}
	\centering
	\includegraphics[scale = 0.9]{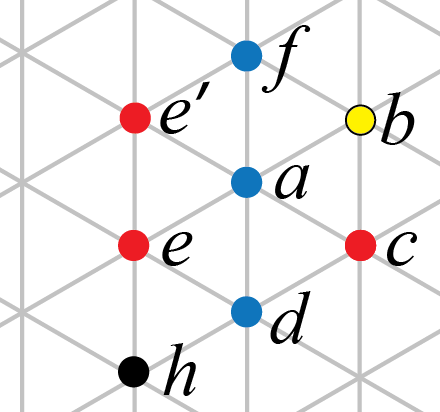}
	\caption{}
	\label{fig:2int_E2_h1}
\end{subfigure}
\hfill
\begin{subfigure}[b]{0.3\textwidth}
	\centering
	\includegraphics[scale = 0.9]{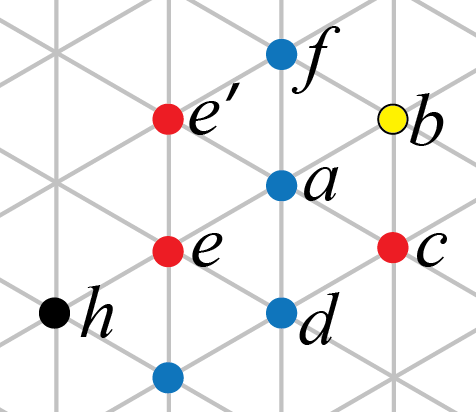}
	\caption{}
	\label{fig:2int_E2_h2}
\end{subfigure}
\caption{Images from the case where $|E| = 2$ in the proof of Lemma~\ref{lem:23int_a23discon_3bdry}. (a) The labeling of $N(a)$ we use when $|E| = 2$. (b,c) We let $h$ be the first vertex of $N(e)$ not in $P_2$, when $N(e)$ is traversed beginning at $a$ followed by $d$, and two examples are shown here.  We show $h \in P_1$, so when $e$ has a disconnected 1-neighborhood these are the only two possibilities for $h$. 	 
}
\label{fig:2int_E2}
\end{figure}

Let $e$ and $e'$ be the two vertices in $E$, with $e$ adjacent to $d$ and $e'$ adjacent to $f$; see Figure~\ref{fig:2int_E2_label}. 
Note $e'$ has three consecutive neighbors that are not in $bd(T)$, $e$, $a$, and $f$; the first because it's an assumption of this case, and the latter two because $P_2 \cap bd(T) = \emptyset$. It follows that $e' \notin bd(T)$. It remains possible that $e' \in \ci \subseteq \clei$. 

If $e$ or $e'$ can be added to $P_2$ without changing any vertices in $\clei$, we do so and have met our first objective, where now $|E| = 1$, $|P_1| = k_1 $, and $|P_2| = k_2 + 1$. Therefore, we assume neither $e$ nor $e'$ can be removed from $P_1$ and added to $P_2$.  For $e$, this means $e$'s 1-neighborhood or 2-neighborhood is disconnected; for  $e'$, this means $e' \in \clei$ or $e'$ has a disconnected 1-neighborhood or 2-neighborhood. Consider $N(e)$; traversing it beginning with $a$ then $d$, let $h$ be the first vertex not in $P_2$.  Two examples of $h$ are shown in Figure~\ref{fig:2int_E2}(b,c). 
Suppose, for the sake of contradiction, that $h \in P_3$. Consider the cycle formed by any path from $b$ to $h$ in $P_3$ together with the path from $h$ to $a$ in $N(e)$ passing through $d$. This cycle contains only vertices in $P_2$ and $P_3$. One of $e$ and $c$ is inside this cycle, and one of $e$ and $c$ is outside this cycle.  As both $e$ and $c$ are in $P_1$, this is a contradiction. Therefore $h \in P_3$ is not possible. Furthermore, $h \in P_2$ is not possible because of how $h$ was chosen and $h \notin T$ is not possible because $e, e' \notin bd(T)$, so we conclude $h \in P_1$. Note because we are assuming $e$ has a disconnected 1-neighborhood, $h$ cannot be adjacent to or equal to $e'$, so the only two possibilities for $h$ are the two shown in Figure~\ref{fig:2int_E2}(b,c).

Suppose first $e$ has a connected 1-neighborhood. Because $e \notin bd(T)$, it must be that all vertices in $N(e)$ on the path from $h \in P_1$ to $e' \in P_1$ that does not pass through $a$ are in $P_1$. This means $e$ has only neighbors in $P_1$ and $P_2$, and $e$'s 1-neighborhood and 2-neighborhood are both connected.  We add $e$ to $P_2$, after which we have met our first objective, where now $|E| = 1$, $|P_1| = k_1 $, and $|P_2| = k_2 + 1$.

Suppose instead $e$ has a disconnected 1-neighborhood. This means $P_1 \setminus e$ has two components; it cannot have three because $e$ already has two adjacent neighbors not in $P_1$, $a$ and $d$. One of the connected components of $P_1 \setminus e$ must contain $e'$, and the other must contain $h$. 
Similar to the previous case above, we show how to find two components $S_1 \subseteq P_1 \setminus e$ and $S_2 \subseteq P_2 \setminus a$ that do not share any adjacent vertices, so that we will then be able to apply Lemma~\ref{lem:s1s2} to unwind them. (We do not need Lemma~\ref{lem:s1s2x} as above because when $|E| = 2$ then $|D| = 1$ and $|F| = 1$, so the problem case addressed by that lemma cannot occur). 

First, suppose $c$ and $\cc_1$ are in the same connected component of $P_1 \setminus e$. We let $S_1$ be the component of $P_1 \setminus e$ not containing $c$ and $\cc_1$, and note this means $\cc_1 \in P_1 \setminus S_1$, one of the hypotheses of Lemma~\ref{lem:s1s2}. 
Let $C$ be the cycle formed by any path from $c$ to $e$ in $P_1$ together with $a$. 
Because $c \notin S_1$, this means $C \cap S_1 = \emptyset$. However, $S_1$ may be outside $C$ or may be inside $C$. 
If $S_1$ is outside $C$, we let $S_2$ be the component of $P_2 \setminus a$ inside $C$, and if $S_1$ is inside $C$, we let $S_2$ be the component of $P_2 \setminus a$ outside $C$.  In either case, $S_1$ and $S_2$ are not adjacent because they are separated by $C$.

Alternately, suppose $c$ and $\cc_1$ are in different components of $P_1 \setminus e$. Note this can only happen if $c \notin \clei$. Let $S_1$ be the component of $P_1 \setminus e$ containing $c$ but not $\cc_1$. 
Let $C$ be the cycle formed by any path from $e$ to $c$ in $P_1$ together with $a$; note that, except for $a$ and $e$, all vertices of $C$ are in $S_1$. Note $C$ cannot encircle $b \in P_3$ because then it would separate $P_3$ from $bd(T)$ and we know $P_3 \cap bd(T) \neq \emptyset$. Therefore this cycle must instead encircle $d \in P_2$.
Suppose for the sake of contradiction that $C$ does not pass through the component of $P_1 \cap N(e)$ containing $h$, which means it must pass through the component of $N(e) \cap P_1$ containing $e'$. This means the component of $N(e) \cap P_1$ containing $e'$ is in $S_1$, and so the component of $N(e) \cap P_1$ containing $h$ must be in the other component of $P_1 \setminus e$, which contains $\cc_1$.  However, as $d$ is encircled by $C$ then $h$ must also be encircled by $C$.  This implies the component of $P_1 \setminus e$ containing $h$ is encircled by $C$, a contradiction as it contains boundary vertex $\cc_1$.  Therefore, we conclude $C$ must pass through the component of $P_1 \cap N(e)$ containing $h$. 
This means $S_1$ must be the component of $P_1 \setminus e$ containing $h$ while the component of $P_1 \setminus e$ containing $\cc_1$ contains $e'$.

We let $S_2$ be the component of $P_2 \setminus a$ containing $f$.  We note $S_1$ and $S_2$ do not have any adjacent vertices within $N(a)$ or $N(e)$. Suppose, for the sake of contradiction, that $S_1$ and $S_2$ are adjacent outside of $N(a)$ and $N(e)$, with vertex $x_1 \in S_1$ next to $x_2 \in S_2$. Let $C'$ be the cycle formed by any path from $e$ to $x_1$ going through $S_1$, together with any path from $x_2$ to $a$ going through $S_2$.  Note because $S_1$ and $S_2$ contain $h$ and $f$, respectively, $C'$ must encircle $e'$ or $b$.  The first is impossible because $e'$ is in the same component of $P_1 \setminus e$ as $\cc_1 \in bd(T)$, and the second is impossible because $b \in P_3$ and $P_3 \cap bd(T) \neq \emptyset$.  We conclude that $S_1$ and $S_2$ do not have any adjacent vertices. 
  

In either case, we have found component $S_1$ of $P_1 \setminus e$ not containing $\cc_1$ and component $S_2$ of $P_2 \setminus a$ such that $S_1$ and $S_2$ do not have any adjacent vertices. We know $S_2 \cap bd(T) = \emptyset$ because $P_2 \cap bd(T) = \emptyset$. Before applying Lemma~\ref{lem:s1s2}, we additionally show $S_1 \cap bd(T) = \emptyset$, which will be necessary to know the application of this lemma maintains the main condition defining our current Case B, that $P_2 \cap bd(T) = \emptyset$. 
Suppose, for the sake of contradiction, that $S_1 \cap bd(T) \neq \emptyset$, that is, that there exists $x \in S_1 \cap bd(T)$. Because $e$ has a disconnected 1-neighborhood, both paths from $e'$ to $h$ in $N(e)$ must contain a vertex not in $S_1$.  One contains $a$ and $d$, and let $y \notin P_1$ be any vertex not in $P_1$ on the other path from $e'$ to $h$ in $N(e)$. Consider the cycle formed by any path from $\cc_1$ via $e$ to $x$ in $P_1$, together with any path from $x$ to $\cc_1$ in $\Gtri$ where all vertices except $x$ and $\cc_1$ are outside $T$. This cycle must separate $y$ from both $a \in P_2$ and its neighbor $b \in P_3$. Whether $y \in P_2$ or $y \in P_3$, this implies a district is disconnected, a contradiction.  We conclude that $S_1 \cap bd(T) = \emptyset$.

We now apply Lemma~\ref{lem:s1s2} to $S_1$ and $S_2$, which tells us there exists a a sequence of moves through balanced or nearly balanced partitions after which (1) the partition is balanced, (2) all vertices in $S_1$ have been added to $P_2$, or (3) all vertices of $S_2$ have been added to $P_1$. In these moves only vertices in $S_1$ and $S_2$ have been reassigned and in outcomes (2) and (3) the resulting partition has $|P_1| = k_1 + 1$ and $|P_3| = k_3 - 1$. If outcome (1) occurs, we are done.  

If outcome (3) is achieved, all vertices in $N(a) \cap S_2$ have been added to $P_1$. However, before concluding that $a$ now has a connected 2-neighborhood, we must also consider whether vertices in $S_1 \cap N(a)$ may have been added to $a$'s 2-neighborhood during this process. 
Because $S_1$ is always a component of $P_1 \setminus e$, it holds that $e \notin S_1$, so $e$ has not been added to $P_2$. The remaining vertices in $N(a) \cap P_1$ are $e'$ and $c$. If $e' \in S_1$, then because $S_1$ and $S_2$ are not adjacent, it must be that $d \in S_2$. This means that $f \notin S_2$, so $f$ remains in $P_2$ even after Lemma~\ref{lem:s1s2} is applied. Because $e'$ is adjacent to $f$ in $N(a)$, even if $e'$ is added to $P_2$, this cannot disconnect $a$'s 2-neighborhood. Alternately, if $c \in S_1$, then from the arguments above we also know $h \in S_1$, so $f \in S_2$. In particular this means $d \notin S_2$, so $d$ remains in $P_2$ even after Lemma~\ref{lem:s1s2} is applied. Because $c$ is adjacent to $d$ in $N(a)$, even if $c$ is added to $P_2$, this cannot disconnect $a$'s 2-neighborhood. We conclude that if outcome (3) of Lemma~\ref{lem:s1s2} occurs, then afterwards $a$ has a connected 2-neighborhood. 
Because $a$ originally had a connected 3-neighborhood and no vertices in $P_3$ have been reassigned, $a$ still has a connected 3-neighborhood. This means $a$ can be added to $P_3$, reaching a balanced partition. 

If outcome (2) is achieved but outcome (3) is not achieved, $e$ now has a connected 1-neighborhood because one connected component of $P_1 \setminus e$ has been added in its entirety to $P_2$. By Lemma~\ref{lem:alternation}, $e$ must have a connected 2-neighborhood or 3-neighborhood. If $e$ has a connected 3-neighborhood, we add $e$ to $P_3$ and have reached a balanced partition. If $e$ has a connected 2-neighborhood, we add $e$ to $P_2$.  If $e'$ was in $S_1$, then $e'$ has already been added to $P_2$, meaning that $a$ now has a connected 2-neighborhood (in this case $c \notin S_1$ so we do not have to consider the case where $c$ has been also added to $P_2$). Because $a$ also has a connected 3-neighborhood then $a$ can be added to $P_3$, producing a balanced partition. If $e'$ was not in $S_1$, then $e'$ remains in $P_1$, and we have  met our first objective, where now $|E| = 1$ ($E$ now only contains $e'$), $|P_1| = k_1$, and $|P_2| = k_2 + 1$. Additionally, as $S_1 \cap bd(T) = \emptyset$, no vertices in $bd(T)$ have been added to $P_2$, so $P_2 \cap bd(T) = \emptyset$ remains true. We now show what to do once our first objective has been met.


All that remains is to show that if we have added $e$ or $e'$ to $P_2$, resulting in a configuration where $|E| = 1$, $|P_2| = k_2 + 1$, and $|P_1 | = k_1$, that there exists a move transforming this into a balanced partition or an earlier case where $|E| = 1$, $|P_1| = k_1 + 1$, and $|P_2| = k_2$, a case which we have already resolved. As $P_2$ has greater than four vertices (because $k_2 \geq n \geq 5$), there must be at least one vertex of $P_2 \setminus a$ that is not in $N(a)$. Let $S_2$ be non-empty component of $P_2 \setminus \{a \cup N(a)\}$; by Condition~\ref{item:nobd} of Lemma~\ref{lem:shrinkable}, $S_2$ contains a vertex $v_2$ that can be removed from $P_2$ and added to another district. If $v_2$ can be added to $P_3$, we do so and have reached a balanced partition.  Otherwise, $v_2$ can be added to $P_1$.  Doing so does not change $N(a)$, so we are now in the case where $|E| = 1$, $|P_1| = k_1 + 1$, and $|P_2| = k_2$.  This case was already resolved earlier in this proof, so we know  there exists a sequence of moves resulting in a balanced partition, completing this proof. \end{proof}

\begin{cor}\label{cor:p2nobd}
		Let $P$ be a partition such that $\cli \subseteq P_1$, $P_1 \cap \cgi \neq \emptyset$, $|P_1| = k_1 + 1$, and $|P_3| = k_3 - 1$. Suppose $P_2 \cap bd(T) = \emptyset$.  There exists a sequence of moves resulting in a balanced partition that does not reassign any vertices in $P_1 \cap \clei$.
\end{cor}
\begin{proof} Because $P_2 \cap bd(T) = \emptyset$, by Lemma~\ref{lem:int_adj} we know that $P_2$ must be adjacent to both $P_1$ and $P_3$. By Lemma~\ref{lem:2tricolortri}, there exists a tricolor triangle whose three vertices are in three different districts. Let $a \in P_2$, $b \in P_3$, and $c \in P_1$ be three such vertices incident on a common triangular face.  Because $P_2 \cap bd(T) = \emptyset$, $a \notin bd(T)$. If $a$'s 2-neighborhood and 3-neighborhood are both connected, then Lemma~\ref{lem:23int_a23con} completes the proof.  If $a$'s 3-neighborhood is not connected, Lemma~\ref{lem:p2_3nbhd_discon} completes the proof.  If $a$'s 3-neighborhood is connected but $a$'s 2-neighborhood is disconnected, Lemma~\ref{lem:23int_a23discon_3bdry} completes the proof. 
\end{proof}

This concludes Case B, where $P_2 \cap bd(T) = \emptyset$.  We proceed to our remaining two cases. 

\subsection{Case C: $P_3 \cap bd(T) = \emptyset$ }

In this case, we must be much more careful.  Because $P_2 \cap bd(T) \neq \emptyset$, it's not always true that for a cut vertex $a$, any connected component of $P_2\setminus a$ is shrinkable, but we will show how we can always find such a component. 

However, we begin with an easier subcase: when there is a tricolor triangle whose vertex in $P_2$ is in $bd(T)$. Recall that if $P_3 \cap bd(T) = \emptyset$, then there must exist a tricolor triangle, as we assume in this lemma. Note this lemma focuses on the case where $a$'s 2-neighborhood is not connected because Lemma~\ref{lem:23int_a23con} applies when $a$'s 2-neighborhood is connected.  

\begin{lem}\label{lem:23int_abd}
Let $P$ be a partition such that $\cli \subseteq P_1$, $P_1 \cap \cgi \neq \emptyset$, $|P_1| = k_1 + 1$, and $|P_3| = k_3 - 1$. Suppose there exists vertices $a \in P_2$, $b \in P_3$, and $c \in P_1$ that are incident on a common triangular face of $T$, where $a \in bd(T)$ and $b \notin bd(T)$. If $a$'s 2-neighborhood is not connected, there exists a sequence of moves resulting in a balanced partition that does not reassign any vertices in $P_1 \cap \clei$. 
\end{lem}
\begin{proof}
	Let $a \in bd(T)$, and let $a \in P_2$, $b \in P_3$, and $c \in P_1$ be incident on a common triangular face $F$.  If $a$'s 2-neighborhood is disconnected, it must be that $a$'s two neighbors in $bd(T)$ are in $P_2$, while its two interior neighbors are $b$ and $c$. Let $d$ be $a$'s neighbor in $bd(T)$ that is adjacent to $c$, and let $e$ be $a$'s neighbor in $bd(T)$ that is adjacent to $b$; see Figure~\ref{fig:23int_abd}(a). 
		
	\begin{figure}
		\centering
		\begin{subfigure}[b]{0.45\textwidth}
			\centering
			\includegraphics[scale = 0.9]{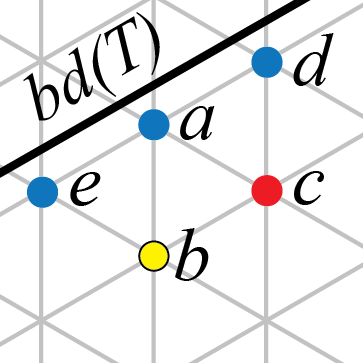}
			\caption{}
			\label{fig:23int_abd-labeling}
		\end{subfigure}
		\hfill
		\begin{subfigure}[b]{0.45\textwidth}
			\centering
			\includegraphics[scale = 0.9]{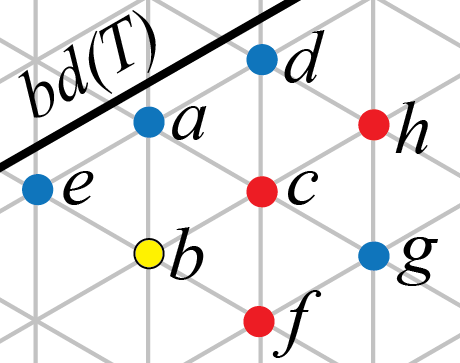}
			\caption{}
			\label{fig:23int_abd-fgh}
		\end{subfigure}

		\caption{Figures from the proof of Lemma~\ref{lem:23int_abd}, where tricolor triangle $abc$ has $a \in bd(T)$ and $a$'s 2-neighborhood is disconnected. (a) The labeling of the vertices in $a$' neighborhood that we use and the district assignment of these vertices resulting from our assumptions. (b) The labeling of $c$'s neighborhood that we use and the district assignment of these vertices resulting from our assumptions. }\label{fig:23int_abd}
	\end{figure}
	
	First, note that $c$'s 3-neighborhood must be connected: if it is not, by Lemma~\ref{lem:p1_3nbhd_discon_bdry} we are done. Note also that $c \in \cgi$: If $a$ is along the top or bottom of $T$, then $a \in P_2$ or $d \in P_2$ is left of $c$, while if $a$ is in $T$'s right boundary, then $c \in \cc_{n-1}$ which, because $i \leq n-2$ by Lemma~\ref{lem:i_n-2}, means $c \in \cgi$. If $c$'s 1-neighborhood is connected, we remove $c$ from $P_1$ and add it to $P_3$, completing the lemma. 
	
	If $c$'s 1-neighborhood is disconnected, it must be that $c$ and $b$'s common neighbor $f \in P_1$; $c$ and $f$'s common neighbor $g \in P_2$ ($g$ cannot be in $P_3$ because $c$ already has neighbor $B \in P_3$ and $c$'s 3-neighborhood is connected); and $c$ and $g$'s common neighbor $h \in P_1$; see Figure~\ref{fig:23int_abd}(b). 
	Consider any shortest path $Q$ from $g$ to $a$ in $P_2$; together with vertex $c$, this forms a cycle $C$.  We let $S_1$ be the component of $P_1 \setminus c$ inside this cycle, and note $S_1$ cannot contain $\cc_1$ because $\cc_1 \in bd(T)$. We let $S_2$ be the component of $P_2 \setminus \{a\}$ not containing any vertices of $Q$, which must exist because $a$ is a cut vertex of $P_2$ and, because $Q$ is a shortest path, will only contain vertices in one component of $P_2 \setminus a$.  We also note that $(P_2 \setminus S_2) \cap bd(T) \neq \emptyset$ because it contains exactly one of $d$ or $e$, and $S_1$ and $S_2$ have no adjacent vertices because they are separated by cycle $C$. This means we can apply Lemma~\ref{lem:s1s2}, meaning there exists a sequence of moves through balanced or nearly balanced partitions after which either (1) the partition is balanced, (2) all vertices in $S_1$ have been added to $S_2$, or (3) all vertices of $S_2$ have been added to $S_3$. In these moves only vertices in $S_1$ and $S_2$ have been reassigned, and in outcomes (2) and (3) the resulting partition has $|P_1| = k_1 + 1$ and $|P_3| = k_3 - 1$. If outcome (1) occurs we are done. If outcome (2) occurs, now $c$'s 1-neighborhood is connected, so $c$ can be added to $P_3$. If  outcome (3) occurs, now $a$'s 2-neighborhood and 3-neighborhood are both connected, and we are done by Lemma~\ref{lem:23int_a23con}. In all three cases, we have reached a balanced partition without reassigning any vertices in $P_1 \cap \clei$.
\end{proof}

\noindent When we do not have a tricolor triangle whose vertex in $P_2$ is in $bd(T)$, it is much harder to find a component of $S_2$ that is shrinkable.  However, it is still possible to do so, using Lemma~\ref{lem:2tricolortri}, which states when $P_3 \cap bd(T) = \emptyset$ there must be two tricolor triangles with opposite chiralities. This is what we do next. 

\begin{lem}\label{lem:p3nobd_S2}
	Let $P$ be a partition such that 
	$\cli \subseteq P_1$ and 
	$P_3 \cap bd(T) = \emptyset$.
	Let $a'\in P_2$ be in one tricolor triangle and let $a'' \in P_2$ be in the other tricolor triangle (note $a' = a''$ is possible). Suppose $a'$ and $a''$ are both cut vertices of $P_2$, neither is in $bd(T)$, and each has a connected 3-neighborhood. 	
	Then there exists a tricolor triangle with vertices $a \in P_2$, $b \in P_3$, and $c \in P_1$ where, if $d$ is the first vertex in $N(a) \cap P_2$ on the longer path from $c$ to $b$ in $N(a)$, the component $S_2$ of $P_2 \setminus a$ containing $d$ has $S_2 \cap bd(T) = \emptyset$.  Furthermore, $a$ must have another neighbor $e \in P_1$ in a different connected component of $N(a) \cap P_1$ than $c$, and the cycle consisting of any path from $c$ to $e$ in $P_1$ together with $a$ encircles $S_2$. 
\end{lem}

\begin{proof}
	By Lemma~\ref{lem:2tricolortri}, there exists two tricolor triangles in this partition with opposite chiralities.  Let $a' \in P_2$, $b' \in P_3$, and $c' \in P_1$ be the vertices of the tricolor triangle for which these vertices are oriented clockwise, and let $a'' \in P_2$, $b'' \in P_3$, and $c'' \in P_1$ be the vertices of the tricolor triangle for which these vertices are oriented counterclockwise.  Because $P_3 \cap bd(T) = \emptyset$, we also know from Lemma~\ref{lem:p2p3int} that $P_2 \cap bd(T) \neq \emptyset$.  
	
	First, we consider the case where $a' \neq a''$, though one or both of $b' = b''$ or $c' = c''$ may be true. 
	Because $a'$ is a cut vertex of $P_2$ and not in $bd(T)$, in $N(a')$ there must exist vertices $d'$ and $f'$ in separate components of $P_2 \cap N(a')$; we label $d'$ and $f'$ such that the longer (clockwise) path from $c'$ to $b'$ in $N(a')$ encounters $d'$ before $f'$. One path from $d'$ to $f'$ in $N(a')$ contains $b'$ and $c'$, and the other path from $d'$ to $f'$ in $N(a')$ must contain some $e' \notin P_2$; $e'$ cannot be in $P_3$ because we assume $a'$ has a connected 3-neighborhood, so we must have $e' \in P_1$. Because $c'$ is clockwise from $b'$ in $N(a')$ (by our original assumption), this means $N(a')$ must include in clockwise order (but not necessarily consecutively): $b',c',d',e',f'$. See Figure~\ref{fig:3int_aa'-Q3} for an example.   Similarly, because $a''$ is also a cut vertex of $P_2$, it must have neighbors $d''$ and $f''$ in separate components of $P_2 \cap N(a)$, labeled so that the longer (counterclockwise) path from $c''$ to $b''$ in $N(a'')$ encounters $d''$ before $f''$.  
	As above for $a'$,  $N(a'')$ must include in counterclockwise order (but not necessarily consecutively): $b'', c'', d'' \in P_2, e'' \in P_1, f'' \in P_2$.  Both $N(a')$ and $N(a'')$ contain one additional vertex we have not assigned a label to.  An example is shown in Figure~\ref{fig:3int_aa'}(a), though note this is not the only possibility for what $N(a')$ and $N(a'')$ can look like, especially as $b' = b''$ or $c' = c''$ is possible. It is also possible $a'$ and $a''$ are adjacent, meaning some vertices have been given two names, one based on its position relative to $a'$ and one based on its position relative to $a''$.  
	Regardless, we will show one of $d'$ or $d''$ satisfies the conclusions of the lemma. 
	
	\begin{figure}
		\centering
		\begin{subfigure}[b]{0.45\textwidth}
			\centering
			\includegraphics[scale = 0.8]{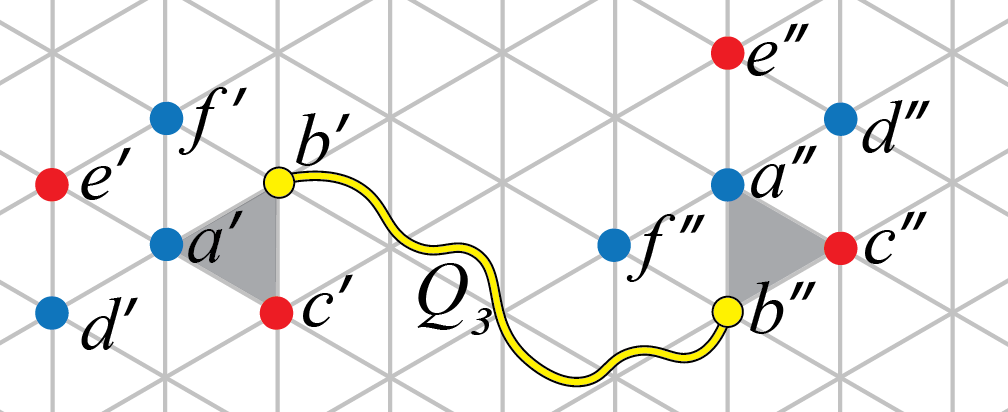}
			\caption{}
			\label{fig:3int_aa'-Q3}
		\end{subfigure}
		\hfill
		\begin{subfigure}[b]{0.45\textwidth}
			\centering
			\includegraphics[scale = 0.8]{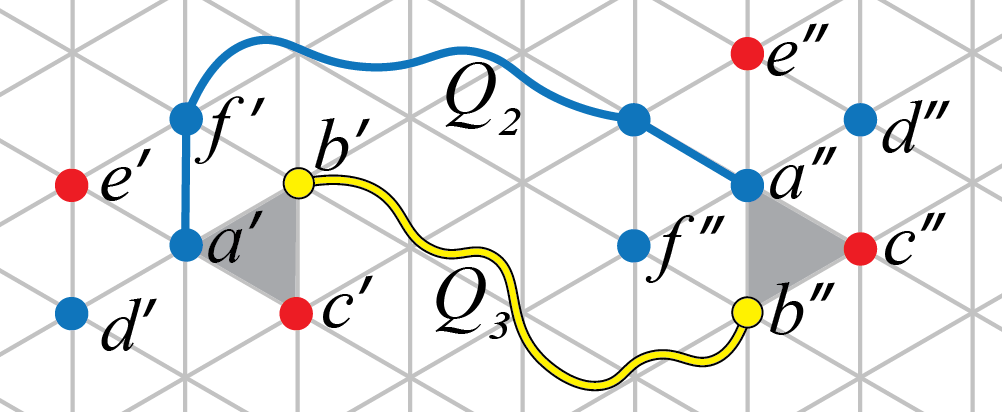}
			\caption{}
			\label{fig:3int_aa'-Q2}
		\end{subfigure}
		\caption{Figures from the proof of Lemma~\ref{lem:p3nobd_S2}, where there exist two tricolor triangles (shaded) of different chiralities. (a) An example of what the partition in the proof of Lemma~\ref{lem:p3nobd_S2} might look like, though this is not the only possibility as $N(a')$ and $N(a'')$, which must have these five labeled vertices in the order given, may not have these five vertices in these exact locations. Yellow path $Q_3$ represents any path in $P_3$ from $b' \in P_3$ to $b'' \in P_3$.  (b) Any path $Q_2$ from $a'$ to $a''$ in $P_2$ must use the component of $N(a') \cap P_2$ containing $f'$ and must use the component of $N(a'') \cap P_2$ containing $f''$, otherwise $P_1$ will be disconnected, a contradiction. }
		\label{fig:3int_aa'}
	\end{figure}


	Let $Q_3$ be any path from $b'$ to $b''$ in $P_3$; if $b' = b''$, then $Q_3$ is the length 0 path $\{b'\}$. Consider the shortest path $Q_2$ from $a'$ to $a''$ in $P_2$. This path can use a vertex in the component of $N(a') \cap P_2$ containing $d'$, or a vertex in the component of $N(a') \cap P_2$ containing $f'$, but not both or it would not be a shortest path. Suppose, for the sake of contradiction, that this path uses a vertex in the component of $P_2 \cap N(a')$ containing $d'$. Consider the cycle $C$ formed by $Q_2$ and $Q_3$ together with the edge $a'b'$ and the edge $a''b''$. Note all vertices in $C$ are in $P_2$ or $P_3$.  Because of the ordering of the vertices in $N(a')$, exactly one of $c'$ or $e'$ will be inside this cycle and the other will be outside it.  This contradicts that $P_1$ is connected. 
	Therefore the shortest path from $a'$ to $a''$ in $P_2$ must begin with a vertex in the component of $P_2 \cap N(a')$ containing $f'$.  Following the same reasoning, its penultimate vertex must be in the component of $P_2 \cap N(a'')$ containing $f''$. 
	See Figure~\ref{fig:3int_aa'-Q2} for an example. 
	

	Let $S_2'$ be the component of $P_2 \setminus a'$ containing $d'$, and let $S_2''$ be the component of $P_2 \setminus a''$ containing $d''$. Both of these components are disjoint from $Q_2$. Suppose for the sake of contradiction that both $S_2'$ and $S_2''$ contain vertices of $bd(T)$. Consider a cycle formed by: Path $Q_2$ from $a'$ to $a''$; any path in $S_2''$ from $a''$ to a vertex of $S_2'' \cap bd(T)$; any path outside of $T$ from that vertex of $S_2'' \cap bd(T)$ to any vertex of $S_2' \cap bd(T)$; and any path from that vertex of $S_2' \cap bd(T)$ to $a'$.  Just as above, this cycle contains one of $c'$ or $e'$ inside it and the other outside it, but itself contains no vertices of $P_1$, contradicting that $P_1$ is simply connected.  Therefore at least one of $S_2'$ or $S_2''$ contains no vertices of $bd(T)$. If $S_2' \cap bd(T) = \emptyset$, we choose $S_2 = S_2'$, $a = a'$, $b = b'$, etc.; Otherwise, we choose $S_2 = S_2''$, $a = a''$, $b = b''$, etc. In either case, the first conclusion of the lemma is satisfied: $S_2 \cap bd(T) = \emptyset$. 
	To see the second conclusion of the lemma follows, consider the cycle $C$ formed by any path from $c$ to $e$ in $P_1$ together with $a$.  Suppose, for the sake of contradiction, this cycle did not encircle $S_2$.  This means it must encircle the other component of $P_2 \setminus a$, containing $f$. However, this means the component of $P_2 \setminus a$ containing $f$ does not have any vertices in $bd(T)$.  As we already know $S_2 \cap bd(T) = \emptyset$ because of how $S_2$ was chosen, this implies $P_2 \cap bd(T) = \emptyset$, a contradiction.  We conclude $C$ must encircle $S_2$, as claimed. 
%

\begin{figure}
	\centering

	\begin{subfigure}[b]{0.22\textwidth}
		\centering
		\includegraphics[scale = 0.8]{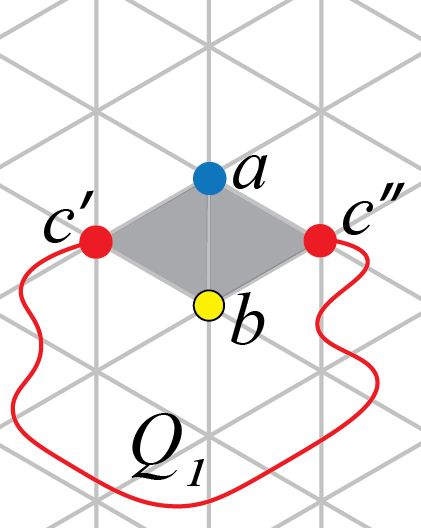}
		\caption{}
		\label{fig:3int_same-label}
	\end{subfigure}
	\hfill
	\begin{subfigure}[b]{0.22\textwidth}
		\centering
		\includegraphics[scale = 0.8]{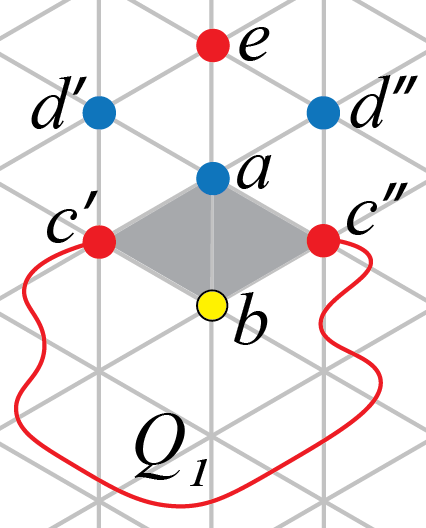}
		\caption{}
		\label{fig:3int_same-s2}
	\end{subfigure}
	\hfill
	\begin{subfigure}[b]{0.22\textwidth}
		\centering
		\includegraphics[scale = 0.8]{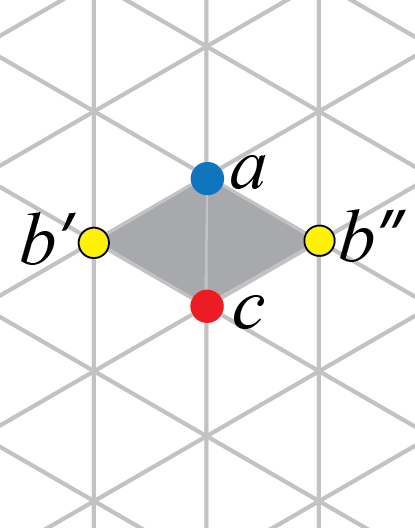}
		\caption{}
		\label{fig:3int_same-c}
	\end{subfigure}
	\hfill
	\begin{subfigure}[b]{0.22\textwidth}
		\centering
		\includegraphics[scale = 0.8]{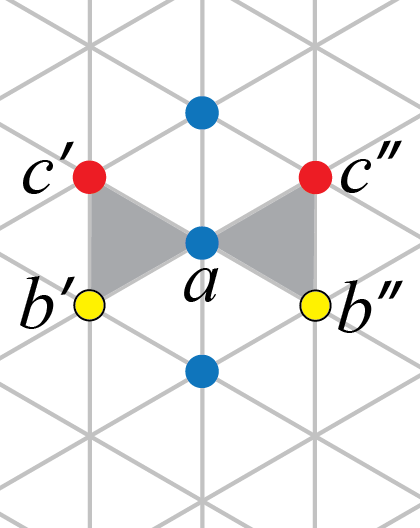}
		\caption{}
		\label{fig:3int_same-distinct}
	\end{subfigure}

	\caption{The three cases from the proof of Lemma~\ref{lem:p3nobd_S2} that occur when $a' = a''$; this single vertex is simply labeled as $a$. (a,b) $b' = b''$ and $c' \neq c''$; (c) $b' \neq b''$ and $c' = c''$; (d) $b' \neq b''$ and $c' \neq c''$.  The first case can be shown to satisfy the conclusions of the lemma by letting $S_2$ be either the component of $P_2 \setminus a$ containing $d'$ or the component of $P_2 \setminus a $ containing $d''$, while the second and third cases can be shown to be impossible. } 
	\label{fig:3int_same}
\end{figure}

	Finally, we consider when $a' = a''$. For simplicity, we refer to this vertex as $a$. There are three cases: $b' = b''$ and $c' \neq c''$; $b' \neq b''$ and $c' = c''$; and  $b' \neq b''$ $c' \neq c''$.  Because the two tricolor triangles are distinct by Lemma~\ref{lem:2tricolortri}, we do not need to consider the case $a' = a''$, $b' = b''$, $c' = c''$. When $b' = b''$ we will refer to this vertex as $b$, and when $c' = c''$ we will refer to this vertex as $c$. 
	
	 First, suppose $b' = b''= b$, and $c' \neq c''$; see Figure~\ref{fig:3int_same}(a). Note that any cycle formed by a path from $c'$ to $c''$ in $P_1$ must encircle $b$, as we know $P_3 \cap bd(T) = \emptyset$ and $P_2 \cap bd(T) \neq \emptyset$. We know by assumption that $a$ has a disconnected 2-neighborhood.  Since three of $a$'s neighbors are, in rotational order, $c'$, $b$, $c''$, it must be that $a$'s common neighbor with $c'$, which we'll call $d'$, and $a$'s common neighbor with $c''$, which we'll call $d''$, are both in $P_2$, while the last remaining neighbor of $a$ is not in $P_2$.  This neighbor cannot be in $P_3$ because $P_3$ is enclosed by the cycle described above, so it must be that this final neighbor $e \in P_1$; see Figure~\ref{fig:3int_same}(b). We know that $a$ is a cut vertex of $P_2$, so label the two components of $P_2 \setminus a$ as $S_2'$ (containing $d'$) and $S_2''$ (containing $d''$). Both $S_2'$ and $S_2''$ cannot contain a vertex of $bd(T)$: if both did, we could easily find a cycle consisting of a path in $P_2$ and vertices outside $T$ encircling exactly one of $e$ or $c'$, contradicting the connectivity of $P_1$.  We therefore pick whichever of $S_2'$ or $S_2''$ doesn't contain a vertex of $bd(T)$, and this component together with either $c'$ and $d'$ or $c''$ and $d''$, as appropriate, satisfies the first conclusion of the lemma. The second conclusion of the lemma follows from the fact that the cycle $C$ formed by the shortest path from $c$ to $e$ in $P_1$ together with $a$ must encircle $S_2$ (note this cycle may or may not also include the other of $c'$ or $c''$ that was not chosen as $c$).  If it didn't encircle $S_2$, then $S_2$ would necessarily include boundary vertices or $P_1$ would not be simply connected, neither of which is possible. This concludes the proof of the lemma when $a' = a''$, $b' = b''$, and $c' \neq c''$.

	Next, suppose  $b' \neq b''$ and $c' = c'' = c$; we show that this case is in fact impossible under the hypotheses of the lemma.  See Figure~\ref{fig:3int_same}(c).  Any path from $b'$ to $b''$ in $P_3$ will necessarily separate either $P_1$ or $P_2$ from $bd(T)$, which we know is impossible because both must intersect $bd(T)$. 
	
	Finally, suppose $b' \neq b''$ and $c \neq c''$. Because $a$ must have at least two non-adjacent neighbors in $P_2$ as it is a cut vertex of $P_2$, and its neighbors $b'$ and $c'$ are adjacent and its neighbors $b''$ and $c''$ are adjacent, $a$ must have one neighbor in $P_2$ between $b'$ and $b''$, and one neighbor in $P_2$ between $c'$ and $c''$; see Figure~\ref{fig:3int_same}(d). Note this means $a$ must have a disconnected 3-neighborhood as $b'$ and $b''$ are in different connected components of $N(a) \cap P_3$. This is a contradiction to the hypotheses of the lemma, so this case cannot occur.


	In all cases that are possible under the hypotheses of our lemma, we have shown how to find a tricolor triangle $a$, $b$, $c$ and a component $S_2$ of $P_2 \setminus a$ satisfying the conclusions of the lemma.

\end{proof}

Now that we know how to find a shrinkable component of $P_2 \setminus a$, we can use this to eventually reach a balanced partition. The following is the main result of this case, and incorporates all of of earlier results.

\begin{lem}\label{lem:3int_disjoint}
Let $P$ be a partition such that $\cli \subseteq P_1$, $P_1 \cap \cgi \neq \emptyset$, $|P_1| = k_1 + 1$, and $|P_3| = k_3 - 1$. Suppose $P_3 \cap bd(T) = \emptyset$. 
Then there exists a sequence of moves resulting in a balanced partition that does not reassign any vertices in $P_1 \cap \clei$. 	
\end{lem}
\begin{proof}
	By Lemma~\ref{lem:2tricolortri}, there exist exactly two triangular faces $F'$ and $F''$ in $T$ whose three vertices are in three different districts.  Exactly one of these triangles has its vertices in $c' \in P_1$, $a' \in P_2$, and $b' \in P_3$ in clockwise order and the other must have its vertices $c'' \in P_1$, $a'' \in P_2$, and $b'' \in P_3$ in counterclockwise order. Additionally, $P_2 \cap bd(T) \neq \emptyset$.

	
	First, note that $b', b'' \notin bd(T)$ because $P_3 \cap bd(T) = \emptyset$.  If $a'$ or $a''$ has a connected 2-neighborhood and a connected 3-neighborhood, we are done by Lemma~\ref{lem:23int_a23con}. If $a'$ or $a''$ has a disconnected 3-neighborhood, we are done by Lemma~\ref{lem:p2_3nbhd_discon}. Therefore, we conclude that $a'$ and $a''$ must both have a connected 3-neighborhood and a disconnected 2-neighborhood. This means $a'$ and $a''$ are both cut vertices in $P_2$. If $a'\in bd(T)$ or $a'' \in bd(T)$, we are done by Lemma~\ref{lem:23int_abd}, so we assume $a, a' \notin bd(T)$. 
	
	We can now apply Lemma~\ref{lem:p3nobd_S2}, and see there exists a tricolor triangle with vertices $a \in P_1$ (where $a = a'$ or $a = a''$), $b \in P_2$, and $c \in P_3$ where, if $d$ is the first vertex in $N(a) \cap P_2$ on the longer path from $c$ to $b$ in $N(a)$, the component $S_2$ of $P_2 \setminus a$ containing $d$ has $S_2 \cap bd(T) = \emptyset$.  Furthermore, $a$ must have another neighbor $e \in P_1$, and the cycle consisting of any path from $c$ to $e$ in $P_1$ together with $a$ encircles $S_2$. Finally, because $a$ must have a disconnected 2-neighborhood, $a$ must have another neighbor $f$ in a different component of $P_2 \cap N(a)$ than $d$; because $b$ and $c$ are adjacent in $N(a)$, just as in the proof of Lemma~\ref{lem:p3nobd_S2} it must be that $f$ is on the path from $e$ to $b$ that doesn't include $c$ and $d$. This means $N(a)$ must include in rotational order (clockwise or counterclockwise, as appropriate) but not necessarily consecutively: $b, c, d, e, f$.  Also $N(a)$ contains one additional vertex we have not assigned a label to.

	If the component of $N(a) \cap P_2$ containing $d$ is of size 2, we will need to distinguish between the vertex $d$, which as we've defined it is adjacent to $c$, and the other vertex in this component, which we call $d'$ and is adjacent to $e$. If the component of $N(a) \cap P_2$ containing $d$ is of size 1, we simply let $d' = d$.

	Our remaining argument will focus on the vertices $c$ and $e$ and show there exists a sequence of moves, not reassigning any vertices in $P_1 \cap \clei$,  that results in (1) a balanced partition or (2) $e$ being added to $P_2$ such that $|P_2| = k_2 + 1$, $|P_3| = k_3 - 1$, and $d'$ remains in $P_2$. 
	First we will explain why such a sequence of moves must exist, and then we will explain why a sequence of moves resulting in (2) where $e$ is added to $P_2$ is sufficient to prove the lemma. 
	
	We begin with the following two claims, which apply in the cases where $c$ or $e$ is a cut vertex of $P_1$; our later explanations include the (much easier) cases when $c$ or $e$ is not a cut vertex of $P_1$. 
	
	\begin{claim}\label{claim:c-cut}
		If $c$ is a cut-vertex of $P_1$, $c$ has a connected 3-neighborhood, and one component $S_1$ of $P_1 \setminus c$ contains neither $e$ nor $\cc_1$, then there exists a sequence of moves, not reassigning any vertices of $P_1 \cap \clei$,  resulting in a balanced partition.
	\end{claim}

	\noindent {\it Proof of Claim~\ref{claim:c-cut}}: Let $C$ be the cycle formed by any path from $e$ to $c$ in $P_1$ together with $a$; we already know $C$ encircles $S_2$. We will do two cases, based on whether $S_1$, the component of $P_1 \setminus c$ containing neither $e$ nor $\cc_1$, is inside $C$ or outside $C$. 
	
	If $S_1$ is outside of cycle $C$, we know $S_1$ and $S_2$ are not adjacent because they are separated by cycle $C$, so we apply Lemma~\ref{lem:s1s2}, the Unwinding Lemma. This means there exists a sequence of moves through balanced or nearly balanced partitions after which (1) the partition is balanced, (2) all vertices of $S_1$ have been added to $P_2$, or (3) all vertices of $S_2$ have been added to $P_1$. Only vertices of $S_1$ and $S_2$ have been reassigned, and in outcomes (2) and (3) the resulting partition has $|P_1| = k_1 + 1$ and $|P_3| = k_3 - 1$.  If we reach outcome (1) we are done.  If we reach outcome (2), then $c$'s 1-neighborhood is now connected, and because we also know that $c$'s 3-neighborhood is connected, $c$ can be added to $P_3$ producing a balanced partition. If we reach outcome (3), then $a$'s 2-neighborhood is now connected, and we are done by Lemma~\ref{lem:23int_a23con}.

	If $S_1$ is inside cycle $C$, we will instead apply Lemma~\ref{lem:cycle-recom}, the Cycle Recombination Lemma, with $x = a$, and $y = c$. See Figure~\ref{fig:3int-cycle-recom-c} for an example that is not meant to be exhaustive of all possibilities. Because of the ordering of vertices in $N(a)$, in which $d \in S_2$ is between $c$ and $e$ but $b$ is not, vertex $b \in P_3$ is outside of $C$ so it follows that all of $P_3$ is outside of $C$. 
	 By Lemma~\ref{lem:cycle-recom}, there exists a recombination step for $P_1$ and $P_2$ after which $N_C(c) \cap P_1$ is connected. This means that along the path in $N(c)$ and inside $C$ from $a$ to $c$'s other neighbor in $C$, which we will call $g$,  there is a sequence of vertices in $P_2$ followed by a sequence of vertices in $P_1$. We know that $N(c) \cap P_1$ previously had two connected components (it could not have three because $c$ has two adjacent neighbors, $a$ and $b$, that are not in $P_1$), one containing vertices of $S_1$ and the other containing $g$. This rearrangement has made it so that all of $e$'s neighbors in $P_1$ are now in the same connected component of $P_1 \cap N(c)$ as $g$. Therefore $c$ now has a connected 1-neighborhood, and because we already knew it had a connected 3-neighborhood, we add $c$ to $P_3$, producing a balanced partition. 
	This completes the proof of Claim~\ref{claim:c-cut}.\qed

	\begin{figure}
	\begin{subfigure}[b]{0.45\textwidth}
		\centering
		\includegraphics[scale = 0.8]{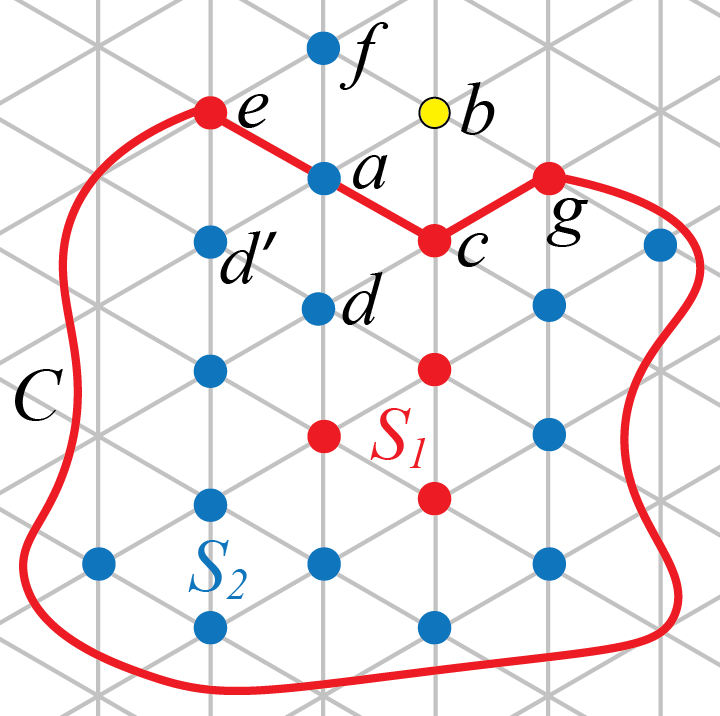}
		\caption{}
		\label{fig:3int-cycle-recom-c}
	\end{subfigure}
	\hfill
	\begin{subfigure}[b]{0.45\textwidth}
		\centering
		\includegraphics[scale = 0.8]{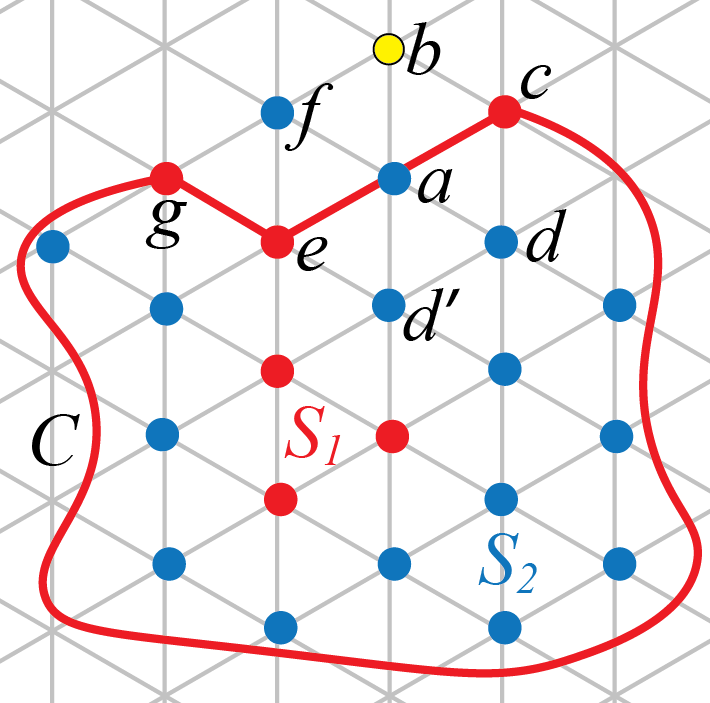}
		\caption{}
		\label{fig:3int-cycle-recom-e}
	\end{subfigure}
	\caption{(a) An example of a situation from the proof of Claim~\ref{claim:c-cut} where we  apply the Cycle Recombination Lemma (Lemma~\ref{lem:cycle-recom}). Cycle $C$ (red) encircles both a component $S_2$ of $P_2 \setminus a$ (blue vertices inside $C$) and a component $S_1$ of $P_1 \setminus c$ (the three red vertices inside $C$). Applying the Cycle Recombination Lemma with $a = x$ and $c = y$ will rearrange all vertices inside $C$ such that afterwards $c$ has a connected 1-neighborhood. 	
(b) An example of a situation from the proof of Claim~\ref{claim:e-cut} where we apply a corollary of the Cycle Recombination Lemma (Lemma~\ref{lem:cycle-recom-vtx}). Cycle $C$ (red) encircles both a component $S_2$ of $P_2 \setminus a$ (blue vertices inside $C$) and a component $S_1$ of $P_1 \setminus e$ (the three red vertices inside $C$). Applying this corollary with $a = x$, $e = y$, and $d' = z$  will rearrange all vertices inside $C$ such that afterwards $e$ has a connected 1-neighborhood and $d'$ remains in $P_2$.}
	\label{fig:3int-cycle-recom}
\end{figure}

\begin{claim}\label{claim:e-cut}
	If $e$ is a cut-vertex of $P_1$ and one component $S_1$ of $P_1 \setminus e$ contains neither $c$ nor $\cc_1$, then there exists a sequence of moves, not reassigning any vertices of $P_1 \cap \clei$, resulting in a balanced partition or in $e$ being added to $P_2$ such that $|P_2| = k_2 + 1$, $|P_3| = k_3 - 1$, and $d'$ remains in $P_2$.
\end{claim}

\noindent{\it Proof of Claim~\ref{claim:e-cut}}:
Let $C$ be the cycle formed by any path from $e$ to $c$ in $P_1$ together with $a$; we already know $C$ encircles $S_2$ and $d' \in S_2$ because it is equal to or adjacent to $d \in S_2$. As above, we do two cases, based on whether $S_1$, the component of $P_1 \setminus e$ containing neither $c$ nor $\cc_1$, is inside $C$ or outside $C$.  


If $S_1$ is outside $C$, we wish to apply Lemma~\ref{lem:s1s2}, the Unwinding Lemma, to $S_1$ and $S_2$ with cut vertices $e$ and $a$, respectively.  However, this lemma is not quite enough, because it cannot guarantee (in the case where $S_1$ is added to $P_2$) that $d' \in S_2$ is not removed from $P_2$, a condition we will critically need later on. If $S_2$ contains at least one vertex in addition to $d'$ (the case where $S_2$ contains only a single vertex is handled below), we instead use the closely related Lemma~\ref{lem:s1s2x}, with $x = d'$.  There exists a sequence of moves through balanced or nearly balanced partitions after which either (1) the partition is balanced, (2) all vertices of $S_1$ have been added to $S_2$, or (3) all vertices of $S_2$ except $d'$ have been added to $S_1$. In these moves only vertices in $S_1$ and $S_2 \setminus d'$ have been reassigned. In the second two outcomes, the resulting partition has $|P_1| = k_1 + 1$ and $|P_3| = k_3 - 1$.  If outcome (1) occurs, we are done. More work is required if outcome (2) or outcome (3) occurs. 

If outcome (2) occurs, then because $P_1 \setminus e$ has only two components and one of them was $S_1$, $e$ now has a connected 1-neighborhood. 
First, suppose $e \notin bd(T)$; in this case, by Lemma~\ref{lem:alternation}, $e$ can be added to $P_2$ or $P_3$. In the latter case we immediately reach a balanced partition, and in the former case we reach a partition where $|P_2| = k_2 + 1$, $|P_3| = k_3 - 1$, and (because of our use of Lemma~\ref{lem:s1s2x}) $d'$ remains in $P_2$. 
Next, suppose $e \in bd(T)$. Before applying Lemma~\ref{lem:s1s2x}, $e$ necessarily had two adjacent neighbors in $P_2$ -- one of $e$ and $a$'s common neighbors must be in $P_2$, and $a \in P_2$ as well -- and a disconnected 1-neighborhood.  It follows that $e$'s two neighbors in $bd(T)$ were in $P_1$ while $e$'s two neighbors not in $bd(T)$ were in $P_2$. Exactly one of $e$'s neighbors in $bd(T)$ was in $S_1$, so after applying Lemma~\ref{lem:s1s2x}, $e$ now has three neighbors in $P_2$ and one neighbor in $P_1$, with the neighbor in $P_1$ necessarily in $bd(T)$.  
This implies that $e$ has a connected 1-neighborhood and a connected 2-neighborhood, meaning $e$ can be added to $P_2$, resulting in a partition with the desired district sizes and (because of our use of Lemma~\ref{lem:s1s2x}) $d'$ remains in $P_2$.

If outcome (3) occurs but outcome (2) does not occur, it is now the case that the component of $P_2 \setminus a$ containing $d'$ has only a single vertex, $d'$. This means $d'$ has a connected 2-neighborhood, consisting only of $a$. We note $d' \notin bd(T)$ because $d' \in S_2$ and $S_2 \cap bd(T) = \emptyset$. Therefore $d'$ must have a connected 1-neighborhood or a connected 3-neighborhood by Lemma~\ref{lem:alternation}. We can also find $v_1$ in the component of $P-1 \setminus e$ not containing $\cc_1$ (what remains of $S_1$) that can be removed and added to another district by Condition~\ref{item:cut_corner} of Lemma~\ref{lem:shrinkable}.  If $v_1$ can be added to $P_3$ we are done so we assume that $v_1$ can be added to $P_2$. If $d'$ has a connected 3-neighborhood, we add $v_1$ to $P_2$ and add $d'$ to $P_3$, reaching a balanced partition. If $d'$ has a connected 1-neighborhood, we add $v_1$ to $P_2$ and add $d'$ to $P_1$. As $d'$ was the only vertex in its component of $P_2 \setminus a$, there is now only a single component of $P_2 \setminus a$, meaning $a$ now has a connected 2-neighborhood. We already know $a$'s 3-neighborhood is connected and no vertices in $P_3$ have been reassigned, so we are done by Lemma~\ref{lem:23int_a23con}.



If instead $S_1$ is inside $C$, we will apply Lemma~\ref{lem:cycle-recom-vtx}, a corollary of the Cycle Recombination Lemma (Lemma~\ref{lem:cycle-recom}) with $x = a$ and $y = e$.  This corollary allows for the additional conclusion that one particular neighbor of $a$ in $S_2$ remains in $P_2$ even after this recombination; we use this to ensure that $d'$ remains in $P_2$ after the recombination. See Figure~\ref{fig:3int-cycle-recom-e} for an example that is not meant to be exhaustive. 
First, we check the remaining hypothesis of Lemma~\ref{lem:cycle-recom-vtx}. Because of the ordering of vertices in $N(a)$, in which $d \in S_2$ is between $c$ and $e$ but $b$ is not, vertex $b \in P_3$ is outside of $C$ so it follows that all of $P_3$ is outside of $C$. 
We now apply this lemma with $x = a$, $y = e$, and $z = d'$. 
We conclude there exists a recombination step for $P_1$ and $P_2$ which only changes the district assignment of vertices enclosed by $C$ and leaves $|P_1|$ and $|P_2|$ unchanged, after which $e$'s 1-neighborhood of vertices enclosed by or in $C$ is connected and $d'$ remains in $P_2$. 
This means within $N(e)$ inside $C$ from $a$ to $e$'s other neighbor $g$ in $C$, there is a sequence of vertices in $P_2$ followed by a sequence of vertices in $P_1$, with no interleaving.  $N(e) \cap P_1$ previously had two connected components: one containing a vertex (or vertices) of $S_1$ and one containing $e$'s other neighbor in $C$ that is not $a$; this rearrangement has made it so that all of $e$'s neighbors in $P_1$ are now in the same connected component of $P_1 \cap N(e)$ as $e$'s other neighbor in $C$, meaning $e$'s 1-neighborhood is connected.

We note in this case it is impossible to have $e \in bd(T)$ as $e$ must have at least five neighbors: two neighbors in $C$; at least two neighbors inside $C$, necessary for satisfying the hypotheses of Lemma~\ref{lem:cycle-recom-vtx}; and $e$'s other common neighbor with $a$, which must be in $T$ because $a \notin bd(T)$. By Lemma~\ref{lem:alternation}, because $e \notin bd(T)$ and $e$ has a connected 1-neighborhood, $e$ must have a connected 2-neighborhood or 3-neighborhood. In the latter case, we add $e$ to $P_3$ and reach a balanced partition.  In the former case, we add $e$ to $P_2$, producing  a partition with the desired district sizes.
%
This completes the proof of Claim~\ref{claim:e-cut}.\qed

\vspace{3mm}
We will now show that in all cases, we either reach a balanced partition; add $e$ to $P_2$ reaching a partition where $|P_2| = k_2 + 1$, $|P_3| = k_3 - 1$, and $d' \in P_2$; or satisfy the hypotheses of Claim~\ref{claim:c-cut} or Claim~\ref{claim:e-cut}, either of which then implies one of the first two conclusions is reached. 
We do cases based on whether $e$ is in $\clei$.

\underline{\textit{Case: $e \in \clei$}}. Note it is impossible to have both $e \in \clei$ and $c \in \clei$, as $e$ and $c$ have a common neighbor $a$ but are not adjacent, and some of the vertices between them in $N(a)$ are not in $P_1$ and so can't be in $\cli$. Therefore in this case it must hold that $c \notin \clei$.  If $c$ has a connected 1-neighborhood and a connected 3-neighborhood, we add $c$ to $P_3$ and are done.  If $c$'s 3-neighborhood is disconnected, as $P_3 \cap bd(T) = \emptyset$ implies $P_2 \cap bd(T) \neq \emptyset$, we are done by Lemma~\ref{lem:p1_3nbhd_discon_bdry}. Therefore we assume $c$'s 3-neighborhood is connected and $c$'s 1-neighborhood is disconnected. Because $c$ is a cut vertex of $P_1$ and $e \in \clei$, $e$ and $\cc_1$ must be in the same component of $P_1 \setminus c$. The other component $S_1$ of $P_1 \setminus c$ satisfies the hypotheses of Claim~\ref{claim:c-cut}, so there exists a sequence of moves resulting in a balanced partition.

\underline{\textit{Case: $e \notin \clei$}}. If $e$ can be immediately added to $P_2$ we do so and have reached one of our goal states (because $d'$ has not been reassigned, it remains in $P_2$).  Therefore we suppose this is not the case, meaning $e$'s 1-neighborhood is disconnected or $e$'s 2-neighborhood is disconnected. 

First, we suppose $e$'s 1-neighborhood is connected, which means its 2-neighborhood must be disconnected. We do cases for $e \in bd(T)$ and $e \notin bd(T)$. Suppose $e \in bd(T)$, meaning $e$ only has four neighbors in $T$. We know $e$ is adjacent to $a$, and at least one of $e$ and $a$'s common neighbors must be $d'$ or $f$, so $e$ and $a$ have a common neighbor in $P_2$. 
Because $e$'s 2-neighborhood is disconnected and $a \notin bd(T)$, there is only one possibility for what $N(e)$ looks like (see Figure~\ref{fig:p3nobd_ebd}): $e$ and $a$'s common neighbor in $P_2$, which we call $x$ because it might be $d'$ or might be $f$, is in $bd(T)$; next to $x$ in $N(e) \cap T$ is $a$; $e$ and $a$'s other common neighbor, which we will call $e'$, is in $P_1$; and $e$'s other neighbor in $bd(T)$, which we call $g$, is in $P_2$. 
%
In this case, consider the cycle formed by any path in $P_2$ between $x$ and $g$, together with $e$'s two neighbors outside of $T$.  This cycle entirely encircles $P_1$. In the interior of this cycle, $P_1$'s only vertex in $bd(T)$ is $e$, a contradiction as $e \notin \clei$ and we know $\cc_1 \in P_1 \cap bd(T)$ as well. As we have found a contradiction, it cannot be the case that both $e$'s 1-neighborhood is connected and $e \in bd(T)$. 


\begin{figure}
	\centering

	\includegraphics[scale = 0.9]{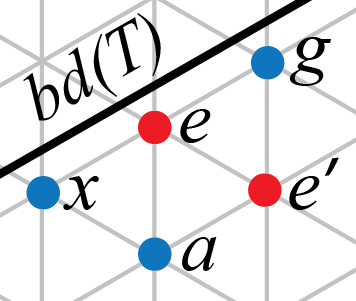}
	\caption{A case from Lemma~\ref{lem:3int_disjoint}, where $e \in bd(T)$, $e$ has a connected 1-neighborhood, and $e$ has a disconnected 2-neighborhood. When these conditions are met, this is what the partition in $N(e)$ must look like; this leads to a contradiction, as any path from $x$ to $g$ in $P_2$ separates $P_1$ from all corner vertices including $\cc_1 \in P_1$.}
	\label{fig:p3nobd_ebd}
\end{figure}

Suppose $e$'s 1-neighborhood is connected and $e \notin bd(T)$. 
By Lemma~\ref{lem:alternation}, $e$ must have a connected 2-neighborhood or 3-neighborhood. It cannot have a connected 2-neighborhood because we have assumed $e$ cannot be added to $P_2$, so  $e$ must have a connected 3-neighborhood. We add $e$ to $P_3$ and have reached a balanced partition. 


We now suppose for the remainder of this proof that $e$'s 1-neighborhood is disconnected, that is, $e$ is a cut vertex of $P_1$. 
If $c \in \clei$, $c$ and $\cc_1$ must be in the same component of $P_1 \setminus e$. The other component $S_1$ of $P_1 \setminus e$ satisfies the hypotheses of Claim~\ref{claim:e-cut}, so there exists a sequence of moves resulting in a balanced partition or in $e$ being added to $P_2$ while $d'$ remains in $P_2$.  If $c \notin \clei$, more work is needed. If adding $c$ to $P_3$ is valid, we do so and are done.  If $c$'s 3-neighborhood is disconnected, as $P_3 \cap bd(T) = \emptyset$ implies $P_2 \cap bd(T) \neq \emptyset$, we are done by Lemma~\ref{lem:p1_3nbhd_discon_bdry}. Therefore we assume $c$'s 3-neighborhood is connected and $c$'s 1-neighborhood is disconnected, that is, $c$ and $e$ are both cut vertices of $P_1$ and  $P_1 \setminus c$ and $P_1 \setminus e$ each has at  least two components.  Additionally, neither can have three components: for this to be true every other neighbor would need to be in $P_1$. However, as already mentioned above, $e$ has two adjacent neighbors in $P_2$, $a$ and one of $d'$ or $f$, so $P_1 \setminus e$ must have exactly two components. Similarly, $c$ has two adjacent neighbors $a$ and $b$ that are not in $P_1$, so $P_1 \setminus c$ must have exactly two components.
Look at any shortest path $Q$ from $c$ to $e$ in $P_1$, which together with $a$ forms a cycle $C$. Let $S^c_1$ be the component of $P_1 \setminus c$ not containing $e$, and let $S^e_1$ be the component of $P_1 \setminus e$ not containing $c$.  These components are disjoint, so at least one of them must not contain $\cc_1$.  Let $S_1$ be whichever of $S^c_1$ and $S^e_1$ doesn't contain $\cc_1$, and let $y$ denote whichever of $c$ and $e$ is the cut vertex separating this component from the remainder of $P_1$.  If $y = c$, we apply Claim~\ref{claim:c-cut} and are done; if $y = e$, we apply Claim~\ref{claim:e-cut} and are done.

We have shown there always exists a sequence of moves that does not reassign any vertices in $P_1 \cap \clei$ and results in either (1) a balanced partition or  (2) $e$ being added to $P_2$ such that $|P_2| = k_2 + 1$, $|P_3| = k_3 - 1$, and $d' \in P_2$. 
If (1) occurs the conclusion of this lemma is satisfied, so we now show why this is the case for (2) as well. 
 In outcome (2), the net result is that the size of the component of $P_1 \cap N(a)$ containing $e$ has decreased by one as follows. If that component was originally size 1, its neighbors in $N(a)$ were $d'$ and some vertex $f'$ in the same component of $P_2 \cap N(a)$ as $f$ (possibly $f' = f$). Vertex $f'$ remains in $P_2$ because it is not in $S_2$ and no moves described here changed $P_2 \setminus \{S_2 \cup a\}$, while $d'$ remains in $P_2$ because we proved it does. Therefore, once $e$ has been added to $P_2$, it connects the component of $N(a)$ containing $d$ with the component of $N(a)$ containing $f$, resulting in $a$ having a connected 2-neighborhood.\footnote{this is the reason we need to be careful about $d'$ remaining in $P_2$: If some moves along the way had added $d'$ to $P_1$, $a$ would not have a connected 2-neighborhood at this point} We know $a$ has a connected 3-neighborhood because no moves we described changed $P_3$, so we add $a$ to $P_3$ and reach a balanced partition. 

If instead the component of $N(a) \cap P_1$ containing $e$ was originally size 2, we have decreased the size of that component by adding $e$ to $P_2$ while ensuring $d'$ and $f' = f$ remain in $P_2$, just as above. If this component's size has decreased by two (perhaps because the other vertex in the component of $N(a) \cap P_1$ containing $e$ was in $S_1$ and so has already been added to $P_2$), we are in the same case as the previous paragraph: $a$ now has a connected 2-neighborhood so can be added to $P_3$, proving the lemma. If the size of  the component of $N(a) \cap P_1$ containing $e$ has only decreased 
by one,  it now contains a single vertex which we call $\overline{e}$. We show how we can find a vertex $v$ that can be removed from $P_2$, resulting in either a balanced partition or other progress toward reaching one.


Let $\overline{d}$ be the vertex in $N(a)$ adjacent to $\overline{e}$ and in the same connected component of $N(a) \cap P_2$ as $d$ (possibly $\overline{d} = d' = d$ or $\overline{d} = e$, as $e$ has just been added to $P_2$).
If we continue to define $S_2$ to be the component of $P_2 \setminus a$ containing $d$, then $\overline{d} \in S_2$. None of the sequences of moves above that result in adding $e$ to $P_2$ include any changes to the component of $P_2 \setminus a$ not containing $d$. This component $S_2'$ of $P_2 \setminus a$ not containing $d$ must have originally contained a vertex of $bd(T)$ -- because $P_2 \cap bd(T) \neq \emptyset$ by Lemma~\ref{lem:p2p3int} and $S_2 \cap bd(T) = \emptyset$ -- and so $S_2'$ must still contain at least one vertex of $bd(T)$ even after this sequence of moves. It cannot be the case that both components of $P_2 \setminus a$ contain a vertex of $bd(T)$: if they did, you could easily construct a cycle consisting only of vertices in $P_2$ or outside $T$ that separates $c \in P_1$ from $e \in P_1$. Therefore it still holds that $S_2 \cap bd(T) = \emptyset$.

First suppose $S_2$ contains vertices other than $\overline{d}$. This means there is at least one nonempty component of $P_2 \setminus \{a, \overline{d} \}$ that is a subset of  $S_2 \setminus \overline{d}$; let $\overline{S_2}$ be any such component. Because it must hold that $\overline{S_2} \cap bd(T) = \emptyset$, by Condition~\ref{item:nobd} of Lemma~\ref{lem:shrinkable}, $\overline{S_2}$ contains a vertex $v$ that can be added to another district. If $v$ can be added to $P_3$, we do so and are done. If $v$ can be added to $P_1$, we do so, resulting in a partition that once again has $|P_1| = k_1 + 1$ and $|P_3| = k_3 - 1$.  It still holds that vertex $a$ and its tricolor triangle satisfy the conclusions of Lemma~\ref{lem:p3nobd_S2}.  Because the component of $P_1 \cap N(a) $ containing $\overline{e}$ is now of size 1, we repeat the same argument as above with $\overline{e}$ replacing $e$.  Ultimately we reach a balanced partition or $\overline{e}$ is added to $P_2$: as described above, the latter results in $a$ having a connected 2-neighborhood so it can be added to $P_3$, and doing so produces a balanced partition.

The last case to consider is when $\overline{d}$ is the only vertex of $S_2$. In this case, $\overline{d}$'s only neighbor in $P_2$ is $a$, so it has a connected 2-neighborhood. We know $\overline{d} \notin bd(T)$ because $\overline{d} \subseteq S_2$, so by Lemma~\ref{lem:alternation} it must be that $\overline{d}$ has a connected 1-neighborhood or 3-neighborhood. In the latter case we add $\overline{d}$ to $P_3$ and are done; in the former case we add it to $P_1$, once again reaching a partition where $|P_1| = k_1 + 1$ and $|P_3| = k_3 - 1$.  This means, because $P_2 \setminus a$ originally had two components, that now $P_2 \setminus a$ only has one component, so $a$ has a connected 2-neighborhood. We already knew $a$ had a connected 3-neighborhood and that has not changed, so we can now apply Lemma~\ref{lem:23int_a23con} to produce a balanced partition. 

We have shown that in any situation, there exists a sequence of moves that does not reassign any vertices in $P_1 \cap \clei$ and results in a balanced partition. This proves the lemma.

\end{proof}

\noindent As in earlier cases, we sum up our results in a single corollary. 

\begin{cor}\label{cor:p3nobd}
Let $P$ be a partition such that $\cli \subseteq P_1$, $P_1 \cap \cgi \neq \emptyset$, $|P_1| = k_1 + 1$, and $|P_3| = k_3 - 1$, and $P_3 \cap bd(T) = \emptyset$. The there exists a sequence of moves resulting in a balanced partition that does not reassign any vertices in $P_1 \cap \clei$. 
\end{cor}
\begin{proof}
	This is exactly Lemma~\ref{lem:3int_disjoint}; no additional cases need to be considered. 
\end{proof}

\subsection{Case D: $P_2$ and $P_3$ are not adjacent}

In this last case, we assume that $P_2$ and $P_3$ are not adjacent, that is, there is no vertex in $P_2$ that is adjacent to any vertex in $P_3$.  This implies that $P_1$ must be adjacent to both $P_2$ and $P_3$ and, by Lemma~\ref{lem:23notadj_touchbdry} and Lemma~\ref{lem:23notadj_bdryconn}, $P_2 \cap bd(T)$ and $P_3\cap bd(T)$ are both nonempty and connected. We can use this structure to find $a \in P_1 \cap \cgi \cap bd(T)$ and $b \in P_3 \cap bd(T) $ that are adjacent; this will be the key we use to reach a balanced partition.

\begin{lem}\label{lem:23notadj_j>i}
	Let $P$ be a partition such that  $\cli \subseteq P_1$ and $P_1 \cap \cgi \neq \emptyset $. Suppose further that $P_2$ and $P_3$ are not adjacent. Then there exists adjacent vertices $a \in bd(T) \cap P_1$ and $b \in bd(T) \cap P_3$ such that $a \notin \clei$.  
\end{lem}
\begin{proof}
	Because $P_3 \cap bd(T)$ is connected and $P_2$ and $P_3$ are not adjacent, there exist two possible locations in $T$ where $a \in P_1 \cap bd(T)$ and $b \in P_3 \cap bd(T)$ are adjacent: we will call these pairs $(a', b')$ and $(a'', b'')$.  Because $\cli \subseteq P_1$, note $a' \neq a''$ and both are in $\cc_{\geq i-1}$ (it is possible to have $b' = b''$).

	Examine the path $P'$ from $a'$ to $\cc_{i-1}\cap bd(T)$ in $bd(T)$ that avoids $P_3$ and the path $P''$ from $a''$ to $\cc_{i-1} \cap bd(T)$ in $bd(T)$ that avoids $P_3$.  At least one of these paths must pass through $P_2$, because these paths together with $P_3 \cap bd(T)$ and $\cli \cap bd(T)$ comprise all of $bd(T)$ and $P_2 \cap bd(T)$ is nonempty by Lemma~\ref{lem:23notadj_touchbdry}. If $P'$ passes through $P_2$, then $a'$ must be in $\cgi$; if $P''$ passes through $P_2$, then $a''$ must be 	in $\cgi$. In either case we have found a vertex in $(P_1 \cap bd(T) ) \setminus \clei$ adjacent to a vertex in $P_3 \cap bd(T)$, as desired.
\end{proof}

The following is the main result of this section; it focuses on the adjacent pair of boundary vertices found in the previous lemma. Note that moves made in an attempt to balance a partition may result in configurations that no longer fall under this case, that is, that have $P_2$ adjacent to $P_3$. In these cases, we simply apply one of the previous cases (Corollary~\ref{cor:23bdryadj}, Corollary~\ref{cor:p2nobd}, or Corollary~\ref{cor:p3nobd}) and are done. 

\begin{lem}\label{lem:23notadj}
	Let $P$ be a partition such that $\cli \subseteq P_1$, $P_1 \cap \cgi \neq \emptyset$, $|P_1| = k_1 + 1$, and $|P_3| = k_3 - 1$, and no vertex in $P_2$ is adjacent to any vertex in $P_3$. Then there exists a sequence of moves resulting in a balanced partition that does not reassign any vertices in $P_1 \cap \clei$. 
\end{lem}
\begin{proof}
	Let $a \in (P_1 \cap bd(T)) \setminus \clei$ and $b \in P_3 \cap bd(T)$ be adjacent; such vertices exist by Lemma~\ref{lem:23notadj_j>i}. 
	
	If $a$'s 1-neighborhood is connected and $a$'s 3-neighborhood is connected, then removing $a$ from $P_1$ and adding it to $P_3$ is a valid move. We do so and are done. 
	If $a$'s 3-neighborhood is disconnected, by Lemma~\ref{lem:p1_3nbhd_discon_bdry} (applicable because $P_2 \cap bd(T)$ is nonempty by Lemma~\ref{lem:23notadj_touchbdry}), there exists a move resulting in a balanced partition that does not reassign any vertices in $P_1 \cap \clei$, satisfying the lemma. 
	
	Therefore, the main case to consider is when $a$'s 3-neighborhood is connected and $a$'s 1-neighborhood is disconnected. 
	 Label $a$ and $b$'s common neighbor in $T$ as $c$; label $a$ and $c$'s common neighbor that is not $b$ as $d$; and label $a$ and $d$'s common neighbor that is not $c$ as $e$; see Figure~\ref{fig:23notadj_setup-labeling}. Because $a$ only has three neighbors in $T$ other than $b$, it must be that $c \in P_1$ and $e \in P_1$. Because $a$'s 3-neighborhood is connected, it must be that $d \in P_2$; see Figure~\ref{fig:23notadj_setup-districts}. We will now do additional cases based on $d$'s neighborhood. 
	
	\begin{figure}
		\begin{subfigure}[b]{0.45\textwidth}
			\centering
			\includegraphics[scale = 0.8]{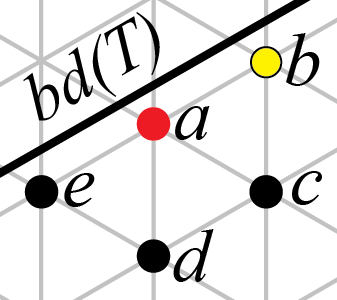}
			\caption{}
			\label{fig:23notadj_setup-labeling}
		\end{subfigure}
		\hfill
		\begin{subfigure}[b]{0.45\textwidth}
			\centering
			\includegraphics[scale = 0.8]{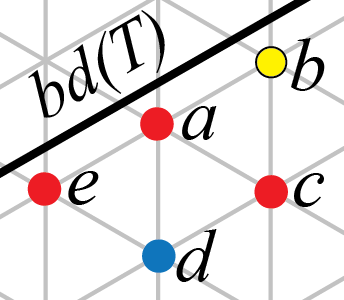}
			\caption{}
			\label{fig:23notadj_setup-districts}
		\end{subfigure}
		
		\caption{The main case of Lemma~\ref{lem:23notadj} where $P_2$ and $P_3$ are not adjacent.  Here $a \in (P_1 \setminus \clei) \cap bd(T)$ and $b \in P_3 \cap bd(T)$. (a) The labeling used for the other vertices in $N(a)$. (b) The districts of the vertices in $N(a)$ if $a$'s 3-neighborhood is connected and $a$'s 1-neighborhood is disconnected: $P_1$ is red, $P_2$ is blue, and $P_3$ is yellow. }\label{fig:23notadj_setup}
	\end{figure}

	\underline{\it Case 1: $d$'s 1-neighborhood and 2-neighborhood are connected}: In this case, we want to remove $d$ from $P_2$ and add it to $P_1$, so that $a$'s 1-neighborhood becomes connected and $a$ can be subsequently added to $P_3$.  However, this would create a partition that is not nearly balanced as we already have $|P_1| = k_1 + 1$. We need to find another vertex to remove from $P_1$ before we can do this. 	
	
	Let $S_1$ be the component of $P_1 \setminus a$ not containing $\cc_1$.  $S_1$ must contain exactly one of $e$ or $c$, and we will refer to whichever of $e$ or $c$ is in $S_1$ as $x$. By Lemma~\ref{lem:s1_c1}, $S_1$ does not contain any vertices of $\clei$. We consider $N(d)$, which is a 6-cycle, and note it must have at least one vertex in $P_2$ and has no vertices in $P_3$ (as $P_2$ and $P_3$ are not adjacent). Traverse $N(d)$ beginning with $a$ followed by $x$, and let $y$ be the last vertex in $P_1$ before encountering a vertex in $P_2$ (possibly $y = x$).  Let $z$ be the next vertex in $N(d)$ after $y$, where $z \in P_2$. See Figure~\ref{fig:23notadj-yz} for examples. 
	We do two further cases: when $y$ has a connected 1-neighborhood and when $y$ has a disconnected 1-neighborhood. 
	
	\begin{figure}
			\begin{subfigure}[b]{0.16\textwidth}
			\centering
			\includegraphics[scale = 0.7]{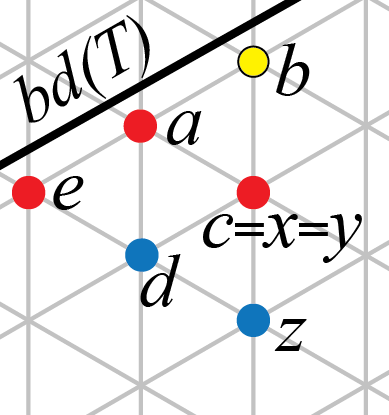}
			\caption{}
			\label{fig:23notadj-yz-c1}
		\end{subfigure}
		\hfill
		\begin{subfigure}[b]{0.16\textwidth}
			\centering
			\includegraphics[scale =  0.7]{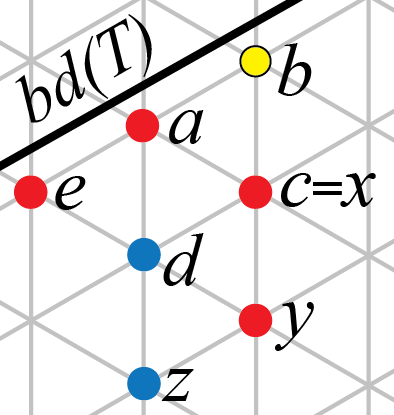}
			\caption{}
			\label{fig:23notadj-yz-c2}
		\end{subfigure}
		\hfill
		\begin{subfigure}[b]{0.16\textwidth}
			\centering
			\includegraphics[scale =  0.7]{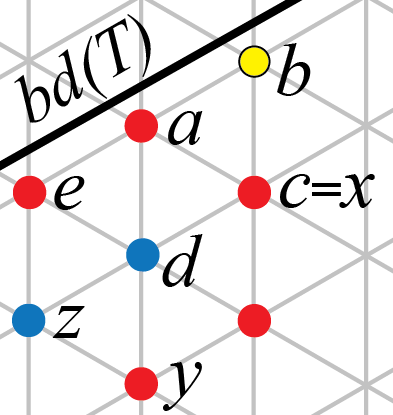}
			\caption{}
			\label{fig:23notadj-yz-c3}
		\end{subfigure}
		\hfill
		\begin{subfigure}[b]{0.16\textwidth}
			\centering
			\includegraphics[scale =  0.7]{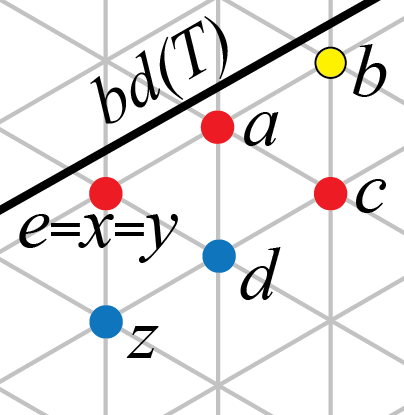}
			\caption{}
			\label{fig:23notadj-yz-e1}
		\end{subfigure}
		\hfill
		\begin{subfigure}[b]{0.16\textwidth}
			\centering
			\includegraphics[scale =  0.7]{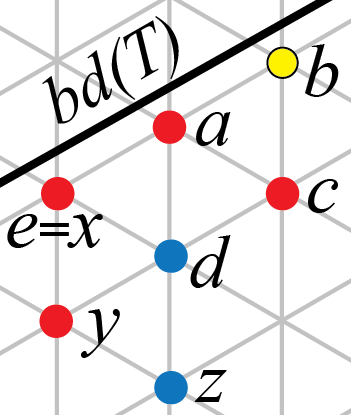}
			\caption{}
			\label{fig:23notadj-yz-e2}
		\end{subfigure}
		\hfill
		\begin{subfigure}[b]{0.16\textwidth}
			\centering
			\includegraphics[scale =  0.7]{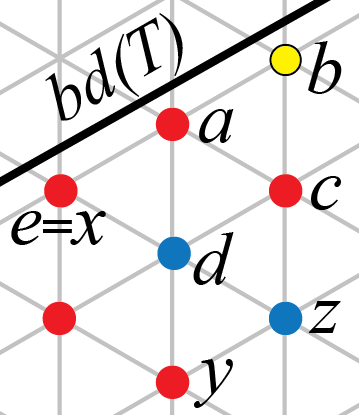}
			\caption{}
			\label{fig:23notadj-yz-e3}
		\end{subfigure}
		\caption{Images from the proof of Lemma~\ref{lem:23notadj} from the case where $d$'s 1-neighborhood and 2-neighborhood are connected. (a,b,c) The three possibilities for $y$ and $z$ when $c \in S_1$, so $x = c$.  (d,e,f) The three possibilities for $y$ and $z$ when $e \in S_1$, so $x = e$. }
		\label{fig:23notadj-yz}
	\end{figure}

	First, suppose $y$ has a connected 1-neighborhood. If $y \notin bd(T)$, then by Lemma~\ref{lem:alternation}, $y$ has a connected 2-neighborhood or 3-neighborhood.
	If $y$ has a connected 3-neighborhood, we add $y$ to $P_3$ and are done, so we assume $y$ has a connected 2-neighborhood.  We now show that when $y \in bd(T)$, $y$ also must have a connected 2-neighborhood. Note $y$ cannot be a corner of $T$ because it is part of the 6-cycle $N(d)$ where $d \notin bd(T)$.  Suppose for the sake of contradiction that $y$'s 2-neighborhood is not connected. Because no vertex of $P_3$ is adjacent to any vertex of $P_2$ and $y$ has a connected 1-neighborhood, the only way for this to happen is for both of $y$'s neighbors in $bd(T)$ to be in $P_2$, while one or both of its neighbors not in $bd(T)$ are in $P_1$. Consider the cycle formed by any path from $y$ to $\cc_1$ in $P_1$ and then any path from $\cc_1$ to $y$ outside of $T$.  This cycle contains no vertices of $P_2$ but must encircle exactly one of $y$'s neighbors in $bd(T) \cap P_2$, a contradiction.  Therefore even when $y \in bd(T)$, $y$ must have a connected 2-neighborhood. Whether $y \in bd(T)$ or $y \notin bd(T)$, we remove $y$ from $P_1$ and add it to $P_2$. Because $y$ is adjacent to $z \in P_2 \cap N(d)$, doing so does not change the fact that $d$ has a connected 1-neighborhood and 2-neighborhood.  We add $d$ to $P_1$ and subsequently add $a$ to $P_3$, producing a balanced partition, satisfying the lemma.

	Next, suppose that $y$ does not have a connected 1-neighborhood.  This means that $y$ is a cut vertex of $P_1$.  Consider any component $\overline{S_1}$ of $P_1 \setminus y$ that does not contain $a$; it follows that $\overline{S_1} \subseteq S_1$ and so $\overline{S_1}$ does not contain any vertices of $\clei$.  $\overline{S_1}$ also cannot contain any vertices of $P_1 \cap N(d)$, as all such vertices are in the same component of $P_1 \setminus y$ as $a$ and $\cc_1$.  By Condition~\ref{item:cut_corner} of Lemma~\ref{lem:shrinkable}, $\overline{S_1}$ contains a vertex $v$ that can be removed from $\overline{S_1}$ and added to another district. If $v$ can be added to $P_3$, we do so and reach a balanced partition.  If $v$ can be added to $P_2$, we do so and note this has not changed $N(d)$. Therefore $d$ still has a connected 1-neighborhood and 2-neighborhood, so we add $d$ to $P_1$ and subsequently add $a$ to $P_3$, producing a valid partition.

	We have thus resolved the case when $d$'s 1-neighborhood is connected and $d$'s 2-neighborhood is connected.

	\underline{\it Case 2: $d$'s 1-neighborhood or 2-neighborhood is not connected}:
	If it is not the case that $d$'s 1-neighborhood and 2-neighborhood are both connected, then because $d$ cannot have any neighbors in $P_3$, there is only one possibility for $N(d)$: $d$ and $c$'s common neighbor $f$ is in $P_2$; $d$ and $f$'s common neighbor $g$ is in $P_1$; and $d$ and $g$'s common neighbor $h$ is in $P_2$.  See Figure~\ref{fig:23notadj-d-disconn-labels}. We do further cases based on $g$'s neighborhood.
	
	\begin{figure}
		\centering
		\begin{subfigure}[b]{0.3\textwidth}
			\centering
			\includegraphics[scale = 0.9]{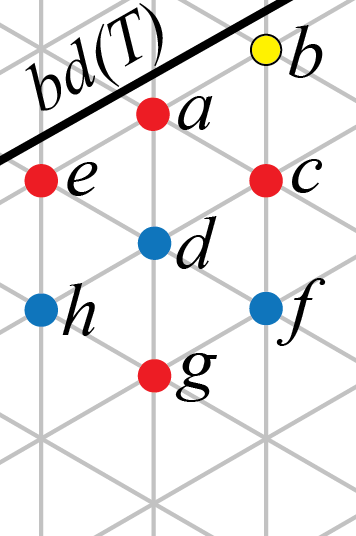}
			\caption{}
			\label{fig:23notadj-d-disconn-labels}
		\end{subfigure}
			\hfill
	\begin{subfigure}[b]{0.3\textwidth}
		\centering
		\includegraphics[scale = 0.9]{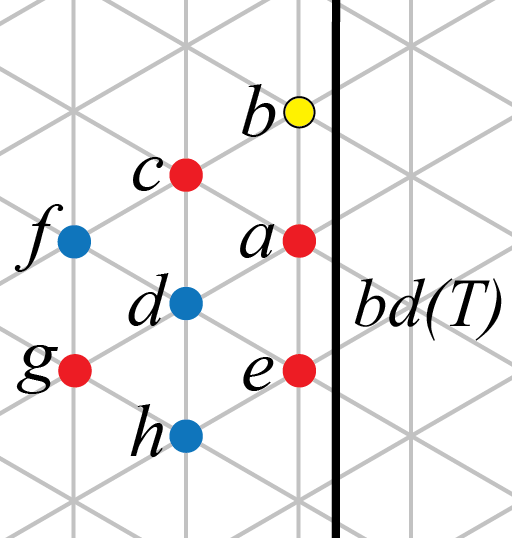}
		\caption{}
		\label{fig:23notadj-d-disconn-rotate}
	\end{subfigure}
		\hfill
		\begin{subfigure}[b]{0.3\textwidth}
			\centering
			\includegraphics[scale = 0.9]{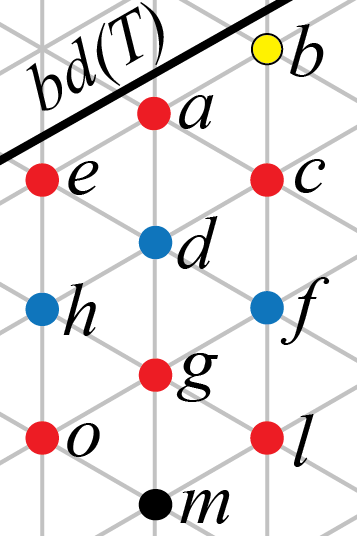}
			\caption{}
			\label{fig:23notadj-d-disconn-gnbhd}
		\end{subfigure}
		
	\caption{Images from Case 2 of the proof of Lemma~\ref{lem:23notadj}. (a) When $d$'s 1-neighborhood and 2-neighborhood are not both connected, this is what the district assignment must be in $N(d)$, where $f \in P_2$, $g \in P_1$, and $h \in P_2$. (b) In Case 2(a), when $g$'s 1-neighborhood and 2-neighborhood are both connected, and $a$ and $b$ are adjacent to $T$'s vertical right boundary, we must be more careful to check $g \not\in \clei$. 		
		(c) From Case 2(b), when $g$'s 1-neighborhood and 2-neighborhood are not both connected, this what the district assignment must be in $N(g)$, where $l \in P_1$, $o \in P_1$ and $m \notin P_1$: $m$ may be in $P_2$ or $P_3$. } \label{fig:23notadj-d-disconn}
	\end{figure}

	\underline{\it Case 2(a): $g$'s 1-neighborhood and 2-neighborhood are connected}:
	First, suppose $g$'s 1-neighborhood and 2-neighborhood are connected. If $g \in \cgi$, we add $g$ to $P_2$, after which $d$'s 1-neighborhood and 2-neighborhood are connected; we add $d$ to $P_1$, after which $a$'s 1-neighborhood and 3-neighborhood are connected; and finally we add $a$ to $P_3$, resulting in a balanced partition. 
	 Note that if $a$ is adjacent to the top or bottom boundary of $T$, then because $h \in P_2$ or $f \in P_2$ is in a column to the left of $g$, $g \in \cgi$. 
	 
	If $a$ is adjacent to $T$'s right boundary, we must be more careful; see Figure~\ref{fig:23notadj-d-disconn-rotate} for an example. Note $g$ cannot be in $\cc_{i-1}$ because $f \in P_2$ is necessarily in the same column as $g$, but $g \in \cc_i$ is possible. In this case, consider any path $Q$ from $a$ to $g$ in $P_1$, which uses $e$ or $c$ but not both.  Together with $d$, this forms a cycle $C$. Let $S_1$ be the component of $P_1 \setminus a$ not containing any vertices of $Q$, which implies $\cc_1 \notin S_1$.  	
	Because $S_1$ contains a vertex in $bd(T)$ (if $e \in S_1$) or a vertex adjacent to $b \in bd(T)$ (if $c \in S_1$), $S_1$ cannot be inside $C$ and so must be entirely outside $C$.  Let $S_2$ be the component of $P_2 \setminus d$ that is inside $C$. Note this means $S_1$ and $S_2$ do not have any adjacent vertices and $S_2 \cap bd(T) = \emptyset$. By Lemma~\ref{lem:s1s2}, there exists a sequence of moves, involving only vertices in $S_1 \subseteq P_1 \cap \cgi$ and $S_2 \subseteq P_2$, after which (1) the partition is balanced, (2) all vertices in $S_1$ have been added to $P_2$, or (3) all vertices in $S_2$ have been added to $P_1$. 
	In the first outcome we are immediately done.  In outcome (2) or (3), it remains true after this sequence that $|P_1| = k_1 + 1$ and $|P_3| = k_3 - 1$. In outcome (2), now $a$'s 1-neighborhood is connected and so we can add it to $P_3$. In outcome (3), now $d$'s 1-neighborhood and 2-neighborhood are connected. If this sequence of moves has resulted in there being a vertex of $P_2$ adjacent to a vertex of $P_3$, we apply Corollary~\ref{cor:23bdryadj} (Case A), Corollary~\ref{cor:p2nobd} (Case B), or Corollary~\ref{cor:p3nobd} (Case C), as appropriate, and are done. Otherwise, if $P_2$ and $P_3$ are still not adjacent, we apply Case 1 of this lemma. In all cases, the result is a balanced partition that has not changed any vertices of $P_1 \cap \clei$.

	\underline{\it Case 2(b): $g$'s 1-neighborhood or 2-neighborhood is not connected}:
	We now suppose that $g$'s 1-neighborhood or 2-neighborhood is disconnected. Because $g$ already has three neighbors in $P_2$, this is impossible if $g \in bd(T)$, so it must be that $g \notin bd(T)$. Note it is impossible for $g$'s 2-neighborhood to be disconnected while its 1-neighborhood remains connected, as no vertex of $P_2$ can be next to a vertex of $P_3$. It follows that $g$'s 1-neighborhood must be disconnected. Let $l\neq d$ be $g$ and $f$'s other common neighbor; let $m \neq f$ be $g$ and $l$'s other common neighbor; and let $o \neq l$ be $g$ and $j$'s other common neighbor.  For $g$'s 1-neighborhood to be disconnected, it must be that $l, o \in P_1$ and $m \notin P_1$. See Figure~\ref{fig:23notadj-d-disconn-gnbhd}, where we allow for both possibilities $m \in P_2$ and $m \in P_3$. 

Consider the shortest path from $a$ to $g$ in $P_1$, which together with $d$ forms a cycle $C$.  We will let $S_1^a$ be the component of $P_1 \setminus a$ not containing any vertices of $C$, and we will let $S_1^g$ be the component of $P_1 \setminus g$ not any vertices of $C$. Because $S_1^a$ and $S_1^g$ are disjoint, at most one of them can contain $\cc_1$.  We will let $S_1$ denote whichever of $S_1^a$ or $S_1^g$ does not contain $\cc_1$. 
We will do cases based on whether $S_1$ is inside or outside of $C$. 

If $S_1$ is outside $C$, let $S_2$ be the component of $P_2 \setminus d$ that is inside $C$; it will necessarily be the case that $S_2 \cap bd(T) = \emptyset$ because it is enclosed by $C$. Note $S_1$ and $S_2$ have no adjacent vertices.  By Lemma~\ref{lem:s1s2}, there exists a sequence of moves, involving only vertices in $S_1$ and $S_2$, after which (1) the partition is balanced, (2) all vertices in $S_1$ have been added to $P_2$, or (3) all vertices in $S_2$ have been added to $P_1$. In outcome (2) or (3), it remains true after this sequence of moves that $|P_1| = k_1 + 1$ and $|P_3| = k_3 - 1$. If the first outcome occurs we are immediately done. If in the second or third outcome the sequence of moves results in a vertex of $P_2$ being adjacent to a vertex of $P_3$, we apply Corollary~\ref{cor:23bdryadj} (Case A), Corollary~\ref{cor:p2nobd} (Case B), or Corollary~\ref{cor:p3nobd} (Case C), as appropriate, to complete the proof. Otherwise we assume $P_2$ and $P_3$ remain not adjacent. Suppose outcome (2) occurred.  If $S_1 = S_1^a$, then now $a$'s 1-neighborhood is connected and we can add $a$ to $P_3$.  If $S_1 = S_1^g$, $g$ now has a connected 1-neighborhood consisting of exactly one vertex (whichever of $l$ or $o$ was not in $S_1^g$). If $g$ has a neighbor in $P_3$, this neighbor must be adjacent to at least one other neighbor of $g$ that is in $P_2$, meaning we have already applied Corollary~\ref{cor:23bdryadj}, Corollary~\ref{cor:p2nobd}, or Corollary~\ref{cor:p3nobd} to complete the proof. Therefore $g$'s other five neighbors must all be in $P_2$, and so $g$ now also has a connected 2-neighborhood.  Because $g$'s 1-neighborhood and 2-neighborhood are both connected, we can apply Case 2(a) to complete the proof. 
If outcome (3) occurs, then either $f$ or $h$ has been added to $P_1$, and $d$ now has a connected 1-neighborhood and 2-neighborhood. This is Case 1 above, which we have already seen how to resolve. In all three outcomes, one can reach a balanced partition via a sequence of moves that has not changed any vertices of $P_1 \cap \clei$. This completes the proof when $S_1$ is outside $C$.


Instead, suppose $S_1$ is inside $C$. Because $a \in bd(T)$, and for both of $a$'s neighbors $c$ and $e$ there exists a path to $bd(T)$ not using any vertices of $C$, a component of $P_1$ containing either of these vertices could not be encircled by $C$.
Therefore in this case it's impossible to have $S_1 = S_1^a$, that is, it must be that $S_1$ is a component of $P_1 \setminus g$ and $S_1$ contains $l$ or $o$. We will apply Lemma~\ref{lem:cycle-recom} (the Cycle Recombination Lemma) with $d = x$ and $g = y$. We note because $b \in P_3 \cap bd(T)$ and $C$ contains no vertices of $P_3$, all vertices enclosed by $C$ must be in $P_1$ or $P_2$. Furthermore, because component $S_1$ of $P_1 \setminus g$ is inside $C$, $N_C(g) \cap P_1$ is disconnected. Since its hypotheses are satisfied, Lemma~\ref{lem:cycle-recom} implies there exists a  recombination step for $P_1$ and $P_2$, only changing the district assignment of vertices enclosed by $C$ and leaving $|P_1|$ unchanged, after which $N_C(g) \cap P_1$ is connected. Because $g$'s neighbors in $C$ are $d$ and whichever of $l$ and $o$ is not in $S_1$, $g$'s only neighbor outside $C$ is one of $f \in P_2$ or $h \in P_2$.  Because $g$ has no neighbors in $P_1$ outside of $C$ and $N_C(g) \cap P_1$ is now connected, it follows that $N(g) \cap P_1$ is also now connected. 
Because Lemma~\ref{lem:cycle-recom} only reassigned vertices inside $C$ and $P_3$ is outside of $C$, it remains true that $P_2$ and $P_3$ are not adjacent. Furthermore, because $g$ had no neighbors in $P_3$ before this recombination step, $g$ still has no neighbors in $P_3$. As $g \notin bd(T)$, $g$ has a connected 1-neighborhood, and $g$ has no neighbors in $P_3$, it follows by Lemma~\ref{lem:alternation} that $g$ has a connected 2-neighborhood. Applying the reasoning from Case 2(a), above, completes the proof.

In all cases, the result is a balanced partition that has not changed any vertices of $P_1 \cap \clei$.
\end{proof}

\subsection{Putting it all together} 

In this section, we pull together all the previous results to show that the sweep line process terminates, and then give the final steps necessary to show a ground state is reached. 


\begin{lem}\label{lem:sweepline}
	From any balanced partition where $\cc_1 \in P_1$, there exists a sequence of moves through balanced and nearly balanced partitions after which, for some $2 \leq i \leq n-2$, $\cli \subseteq P_1$ and $P_1 \subseteq \clei$. 
\end{lem}
\begin{proof}
	We prove this by induction on the columns $i$ of $T$.  Initially, we know $\cc_1 \in P_1$. Let $i$ be the first column of $T$ such that $\clei \not\subseteq P_1$.  We will give a sequence of moves after which either we have proved the lemma or $\clei \subseteq P_1$.

	 If $P_1 \subseteq \clei$, we are done, so we suppose this is not the case. This means $\cgi \cap P_1 \neq \emptyset$. Because of how $i$ was chosen, we also know $P_1 \cap \cc_i \neq \cc_i$. 
	 We can apply Lemma~\ref{lem:increase-Ci} and see there exists a sequence of moves, though balanced and nearly balanced partitions, that maintains $\cli \subseteq P_1$ and increases $|\cc_i\cap P_1|$. The resulting partition is balanced or has $|P_1| = k_1 + 1$. If it is not balanced, by Lemma~\ref{lem:4cases}, this partition must fall into Case A, Case B, Case C, or Case D. By Corollary~\ref{cor:23bdryadj}, Corollary~\ref{cor:p2nobd}, Corollary~\ref{cor:p3nobd}, or Lemma~\ref{lem:23notadj}, as appropriate, there exists a sequence of moves resulting in a balanced partition that does not reassign any vertex in $P_1 \cap \clei$. This has increased the number of vertices in $\cc_i \cap P_1$ and produced a balanced partition. If now it holds that $P_1 \subseteq \cc_i$, we are done.  Otherwise, we repeat this process, until either $P_1 \subseteq \clei$ or $P_1 \cap \cc_i = \cc_i$. In the former case, we're done; in the later case, we now know that $\clei \subseteq P_1$, so we move on to $i+1$ and repeat. 
	 
	 Because the number of vertices in $P_1$ never exceeds $k_1 + 1$, we cannot increase the $i$ such that $\clei \subseteq P_1$ indefinitely, so eventually we must reach a state where $P_1 \subseteq \clei$, proving the lemma. 
\end{proof}

We now show how to reach a ground state: first, we recombine $P_2$ and $P_3$ to untangle them.  Then, we recombine $P_1$ and $P_2$ so that $P_1$ is as it should be in $\sigma_{123}$. Recall that for an ordering of the vertices from left to right and, within each column, from top to bottom, $\sigma_{123}$ is the ground state whose first $k_1$ vertices in this ordering are in $P_1$, whose next $k_2$ vertices are in $P_2$, and whose final $k_3$ vertices are in $P_3$.

\begin{lem}\label{lem:groundstate}
Let $P$ be a balanced partition where for some $i \leq n-2$, $\cli \subseteq P_1 \subseteq \clei$. There exists a sequence of recombination steps through balanced partitions producing ground state $\sigma_{123}$.
\end{lem}
\begin{proof}
	If we can show the the first $k_1$ vertices of the left-to-right, top-to-bottom ordering of the vertices are in $P_1$, we are done because one additional recombination step for $P_2$ and $P_3$ produces $\sigma_{123}$. This already holds for $P_1$ except for its vertices in $\cc_i$. If $P_1$ has $m$ vertices in $\cc_i$, we wish those vertices to be the top $m$ vertices in the column, but initially they may be anywhere in the column. 
	
	We prove by induction that there exists a sequence of steps after which the top $m$ vertices in $\cc_i$ are in $P_1$.
	If the number of vertices in $(\cc_i \setminus P_1) \cup \cc_{i+1}$ is less than or equal to $k_2$, this is trivial: recombine $P_2$ and $P_3$ so that $P_3$ occupies the last $k_3$ vertices in the left-to-right, top-to-bottom ordering, which means $P_2$ will necessarily occupy all of $(\cc_i \setminus P_1) \cup \cc_{i+1}$. We can then recombine $P_1$ and $P_2$ to reach $\sigma_{123}$, a step which only reassigns the vertices in $\cc_i$. However, if $k_2$ is less than $|(\cc_i \setminus P_1) \cup \cc_{i+1}|$, we need to be more careful. This is the case we now consider.

	Suppose that the vertices of $P_1$ occupy only the top $j$  positions of $\cc_i$ for some $j < m$ (possibly $j = 0$); see Figure~\ref{fig:groundstate_recom-ex} for an example where $j = 1$.  We give a sequence of steps after which at least the top $j+1$ positions of $\cc_i$ are occupied by $P_1$. Let $v$ be the topmost vertex in $\cc_i$ that is not in $P_1$, and let $Q$ be the component of $\cc_i \setminus P_1$ containing $v$.	
	This component must have $v$ as its topmost vertex, and if $P_1$ does not yet occupy the top $m$ position of $\cc_i$, the vertex $w$ in $\cc_i$ immediately below $Q$ must be in $P_1$. Let $u$ be $v$'s neighbor in $\cc_{i+1}$ that is below-left of $v$. See Figure~\ref{fig:groundstate_recom-labels} for what $v$, $w$, $u$, and $Q$ are for the example from Figure~\ref{fig:groundstate_recom-ex}.

	 We recombine $P_2$ and $P_3$ as follows. First, we erase all district assignments for $P_2$ and $P_3$ and begin constructing them from scratch. We add $Q\ \cup\ u$ to $P_2$, possible because $|Q| + 1 \leq i \leq n -2 < k_2$.  Until $P_2$ has $k_2$ vertices, we add additional vertices in $\cc_i \setminus P_1$ or $\cc_{i+1}$ according to their distance from $u$ in $\Gtri$; if a vertex in $\cc_i$ and a vertex in $\cc_{i+1}$ are at the same distance from $u$ in $\Gtri$, the one in $\cc_{i}$ is added first. See Figure~\ref{fig:groundstate_recom-result} for the result of this process on an example. We know from above that $k_2 < |(\cc_i \setminus P_1) \cup \cc_{i+1}|$, so we will not need to add any vertices in $\cc_{>i+1}$ to $P_2$. 
	 All vertices not added to $P_2$ are subsequently added to $P_3$. 
	 Constructing $P_2$ and $P_3$ this way ensures that $P_2$ is simply connected, as $Q \cup u$ is simply connected and all other vertices added to $P_2$ are connected to $u$ via a path in $\cc_{i+1} \cap P_2$.  We also know that $v$ now has a connected 2-neighborhood, consisting of $u$ and one or both of $u$ and $v$'s common neighbors. It also holds that $w$ has a nonempty 2-neighborhood because it is adjacent to a vertex of $Q$. Furthermore, $w$'s 2-neighborhood is connected because $w$'s remaining neighbors not in $P_1$, beginning with its common neighbor with $Q$, are in increasing distance from $u$ and so are added to $P_2$ in this order. Because $i \leq n-2$,  then $i+2 \leq n$ and so all of $\cc_{i+2}$ is added to $P_3$. Because a vertex in $\cc_i \setminus P_1$ is added to $P_2$ before a vertex in $\cc_{i+1}$ at an equal distance from $u$, this ensures all vertices in $\cc_i \setminus P_1$  that are not added to $P_2$ also have a neighbor in $\cc_{i+1}$ that was not added to $P_2$, meaning it has a path in $P_3$ to $\cc_{i+2} \subseteq P_3$.  This ensures $P_3$ is simply connected, as required.

	 After this recombination of $P_2$ and $P_3$, we do a recombination step for $P_1$ and $P_2$ that swaps $v$ and $w$, adding $v$ to $P_1$ and adding $w$ to $P_2$: when viewed as two flips, because $v$ and $w$ both have connected 1-neighborhoods and connected 2-neighborhoods, this produces a valid partition. Now the top $j+1$ vertices of $\cc_i$ are in $P_1$, as desired; see Figure~\ref{fig:groundstate_recom-progress}.

	 \begin{figure}
	 	
	 	\begin{subfigure}[b]{0.23\textwidth}
	 		\centering
	 		\includegraphics[scale = 0.7]{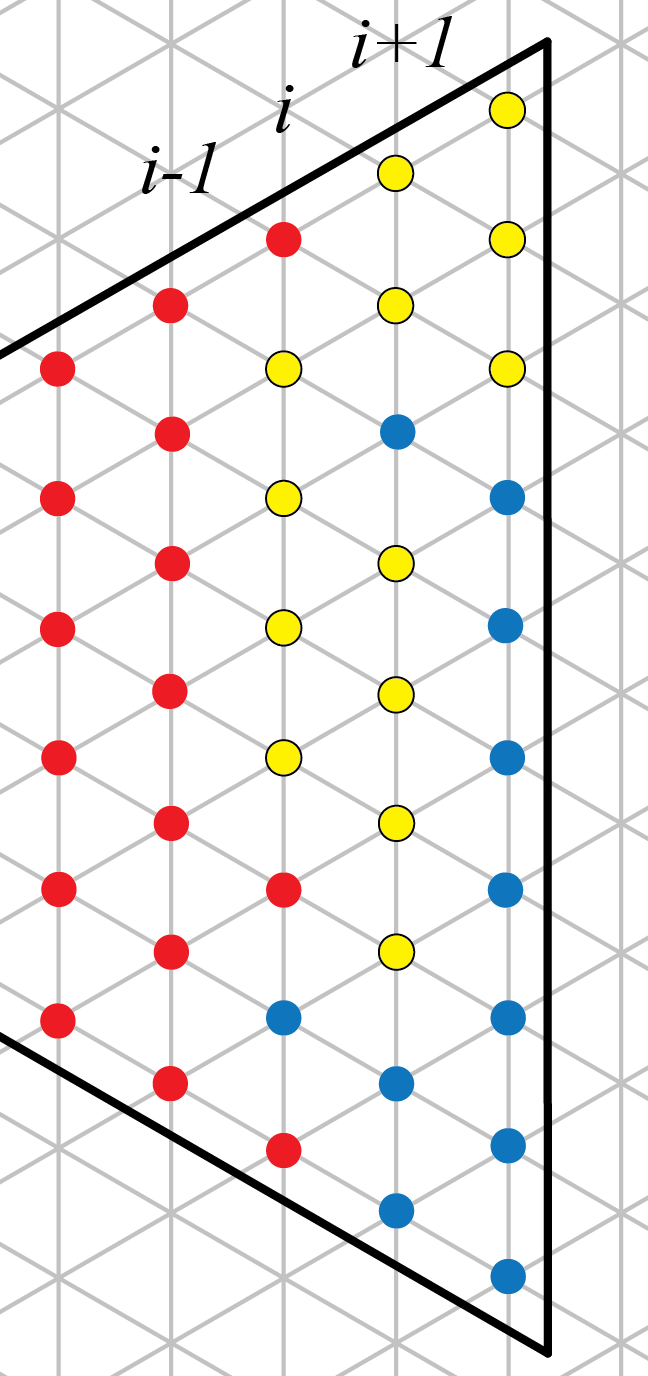}
	 		\caption{}
	 		\label{fig:groundstate_recom-ex}
	 	\end{subfigure}
	 	\hfill
	 	\begin{subfigure}[b]{0.23\textwidth}
	 		\centering
	 		\includegraphics[scale = 0.7]{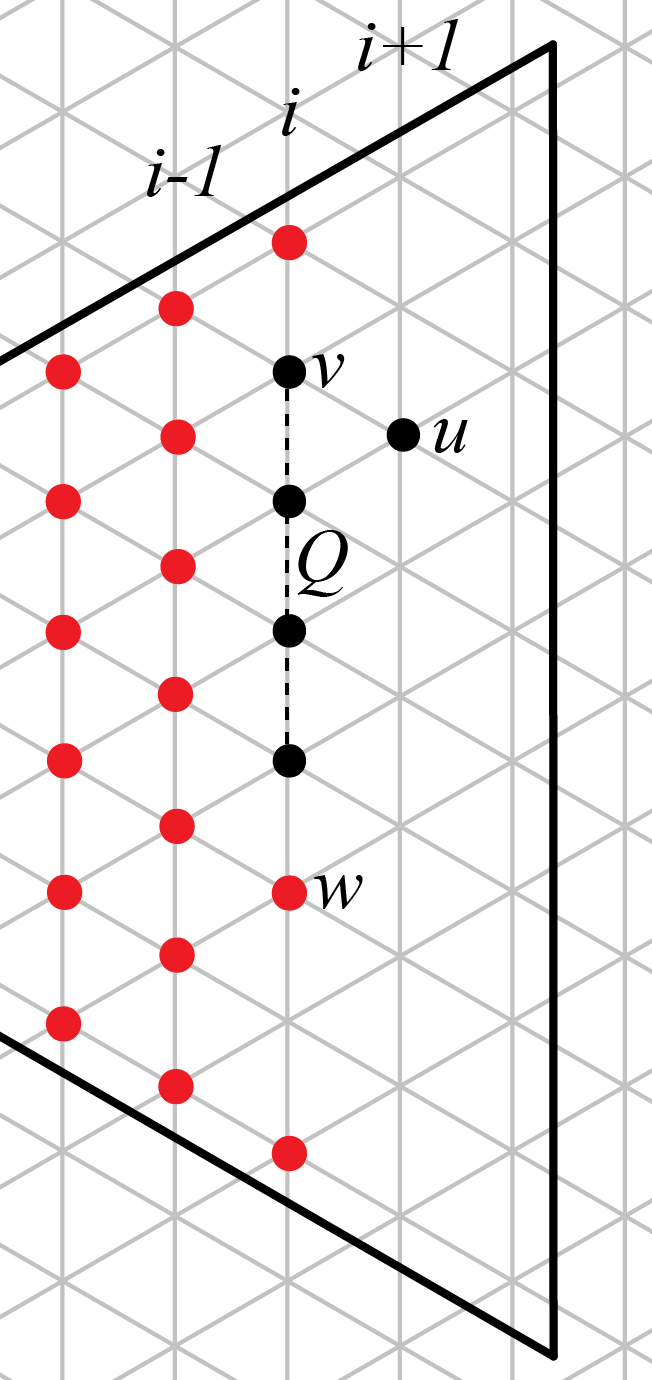}
	 		\caption{}
	 		\label{fig:groundstate_recom-labels}
	 	\end{subfigure}
	 	\hfill
	 	\begin{subfigure}[b]{0.23\textwidth}
	 		\centering
	 		\includegraphics[scale = 0.7]{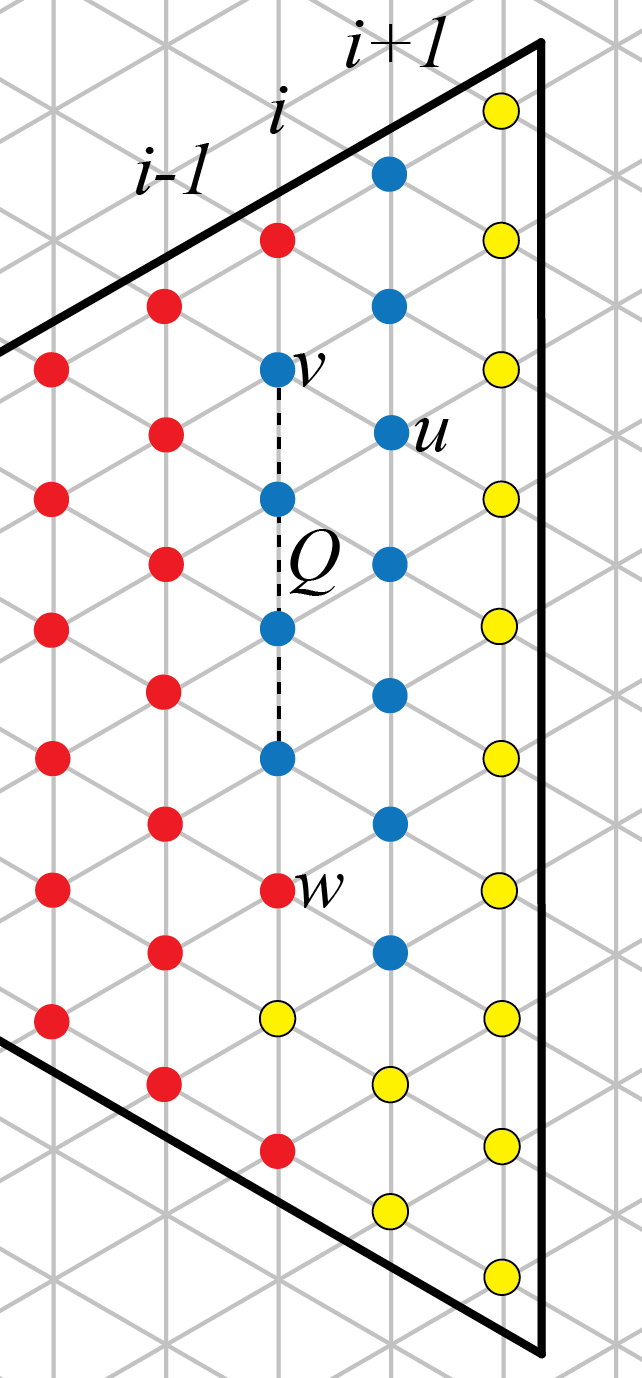}
	 		\caption{}
	 		\label{fig:groundstate_recom-result}
	 	\end{subfigure}
 		 	\hfill
 	\begin{subfigure}[b]{0.23\textwidth}
 		\centering
 		\includegraphics[scale = 0.7]{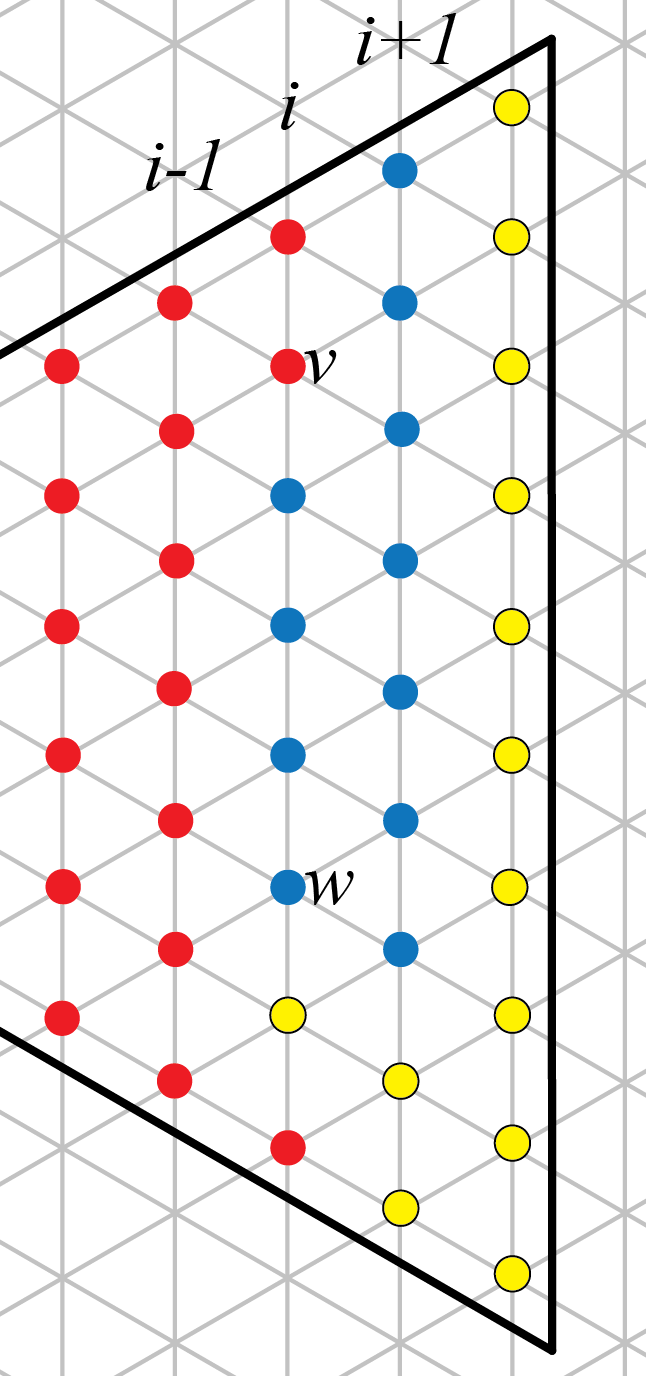}
 		\caption{}
 		\label{fig:groundstate_recom-progress}
 	\end{subfigure}
	 	\caption{An example illustrating the proof of Lemma~\ref{lem:groundstate}. (a) An example where $\cli \subseteq P_1 \subseteq \clei$ and $P_1$ has $m = 3$ vertices in $\cc_i$. The top $j = 1$ vertices of $\cc_i$ are currently in $P_1$. (b) When the district assignments of $P_2$ and $P_3$ are erased, $v$ is the topmost vertex in $\cc_i$ not in $P_1$; $Q$ is the component of $\cc_i \setminus P_1$ containing $v$; $w$ is the vertex immediately below $Q$ in $\cc_i$; and $u$ is $v$'s lower right neighbor in $\cc_{i+1}$. (c) $P_2$ and $P_3$ are recombined so that $Q \cup u$ is in $P_2$; additional vertices of $\cc_i$ and $\cc_{i+1}$ are added to $P_2$ in order of increasing distance from $u$ until $|P_2| = k_2$; and all remaining vertices are added to $P_3$. (d) We recombine $P_1$ and $P_2$ so that $v$ is added to $P_1$ and $w$ is added to $P_2$. 
	 		Now the top $j + 1 = 2$ vertices of $\cc_i$ are in $P_1$. }
	 	\label{fig:groundstate_recom}
	 \end{figure}

	Repeating this process (at most $m$ times) produces a partition where all of $\cli$ and the top $m$ positions of $\cc_i$ are occupied by $P_1$. One additional recombination step of $P_2$ and $P_3$ produces $\sigma_{123}$, completing the proof. 
\end{proof}

\noindent This concludes our proof that from every balanced partition, there exists a sequence of moves through balanced or nearly balanced partitions that reaches a ground state.

\subsection{From a Nearly Balanced Partition to a Balanced Partition}
\label{sec:getbalanced}

We have shown that from any balanced partition there exists a sequence of moves transforming that partition into a ground state. To prove irreducibility, it only remains to show that from any nearly balanced partition there exists a sequence of moves transforming that nearly balanced partition into a balanced partition. This can be done utilizing the many lemmas and approaches we have learned so far.  It is nearly trivial, using prior results, when the district whose size is one larger than ideal contains a corner vertex of $T$, and fairly straightforward when the district whose size is already equal to its ideal size contains a corner of $T$ as well. 
 However, much more work is needed when all corner vertices are in the district whose size is one smaller than ideal.

\begin{lem}\label{lem:nearlybalanced}
	Let $P$ be a nearly balanced partition, with $P_i = k_i + 1$, $P_j = k_j$, and $P_l = k_l - 1$.  There exists a sequence of steps producing a balanced partition. 
\end{lem}
\begin{proof}
	First, we suppose $P_i$ contains a corner vertex of $T$. In this case, we simply apply the same rebalancing lemmas from earlier, with $P_i$ replacing $P_1$ and whichever corner is in $P_i$ replacing $\cc_1$. By Lemma~\ref{lem:4cases}, at least one of the four cases described must apply to $P_j$ and $P_l$, replacing $P_2$ and $P_3$, respectively.  Corollary~\ref{cor:23bdryadj} (Case A), Corollary~\ref{cor:p2nobd} (Case B), Corollary~\ref{cor:p3nobd} (Case C), or Lemma~\ref{lem:23notadj} (Case D), as appropriate, proves the lemma.

Suppose $P_i$ does not contain a corner vertex of $T$.  Now the same lemmas/cases from earlier can't be applied directly, but we use similar ideas and approaches.	
	First, we will show that $P_i$ has a shrinkable component, satisfying one of the Conditions of Lemma~\ref{lem:shrinkable}.  We do cases based on the size of $P_i \cap bd(T)$. 
	If $P_i$ has exactly one vertex in $bd(T)$, let $w$ be that vertex. Let $S$ be the (only) component of $P_i \setminus w$, and note $S$ is shrinkable by Condition~\ref{item:nobd} of Lemma~\ref{lem:shrinkable} because $S \cap bd(T) = \emptyset$.  Otherwise, note that $P_i$ must have at least two exposed vertices, as it's impossible for $P_i$ to only have one vertex adjacent to a vertex of a different district. Let $w$ be any exposed vertex of $P_i$ ($w$ does not necessarily need to be a cut vertex of $P_i$), and let $S$ be any component of $P_i \setminus w$ containing an exposed vertex; at least one component of $P_i\setminus w$ must contain an exposed vertex because $P_i$ has at least two exposed vertices. Then $S$ is shrinkable by Lemma~\ref{lem:shrinkable}, using Condition~\ref{item:nobd} if $P_i \cap bd(T) = \emptyset$ and Condition~\ref{item:exp_2bd} if $|P_i \cap bd(T)| \geq 2$. Regardless, $P_i$ has a shrinkable component $S$, meaning $S$ contains a vertex $v$ that can be removed from $P_i$ and added to another district, producing a valid partition.

If any vertex that can be removed from $P_i$ can be added to $P_l$, we do so, reaching a balanced partition. 
Therefore we assume all vertices that can be removed from $P_i$ and added to another district can only be added to $P_j$. 
If $P_j$ contains any corner vertices of $T$, we remove vertex $v$ from $P_i$ and add it to $P_j$, after which $|P_i| = k_i$ and $|P_j| = k_j + 1$. We then apply the same argument as above, replacing $P_i$ with $P_j$, to prove the lemma. 
All that remains is the case where all three corner vertices of $T$ are in $P_l$, the district whose size is currently one less than its ideal size.  Note we cannot simply apply our tower lemma (Lemma~\ref{lem:increase-Ci}), as doing so may involve decreasing the size of $P_l$ before increasing it, which we cannot do as we already have $|P_l| = k_l - 1$.

 Instead, because it's not possible to have $bd(T) \subseteq P_l$ as $P_l$ must be simply connected, there must be somewhere in $bd(T)$ where $P_l$ is adjacent to $P_i$ or $P_j$, and we focus here. Let  $a \in bd(T) \setminus P_l$ and $b \in P_l \cap bd(T)$ be adjacent. We would like $a$ to be in whichever district has size one larger than its ideal size, but this may not always be the case. However, if $a \in P_j$, we remove $v$ from $P_i$ and add it to $P_j$.  Now $P_j$ has size one larger than its ideal size, and $a$ is in this district.  For the sake of consistency with the lemma statement, we will relabel $P_i$ and $P_j$ so that $a \in P_i$ where $|P_i| = k_i + 1$.  Therefore we can always assume that $P_l$ with $|P_l| = k_l - 1$  and $P_i$ with $|P_i| = k_i + 1$ are adjacent in~$bd(T)$. 
 
 
 Let $a \in P_i\cap bd(T)$ be adjacent to $b \in P_l \cap bd(T)$. We know $a$ can't be removed from $P_i$ and added to $P_l$ (otherwise we would be done), so $a$ must have a disconnected $i$-neighborhood or $l$-neighborhood.

 \textit{\underline{Case: $a$ has a disconnected $i$-neighborhood and a disconnected $l$-neighborhood}}: First, suppose both are true: $a$ has a disconnected $i$-neighborhood and a disconnected $l$-neighborhood. This means, in order around $a \in P_i$, we must have $b \in P_l$, $c \in P_i$, $d \in P_l$, and $e \in P_i$.  See Figure~\ref{fig:nearlybalanced-anbhd} for an example, where we use red to denote $P_i$ because $|P_i| = k_i + 1$ and use yellow to denote $P_l$ because $|P_l| = k_l - 1$; however, recall that now all three corners of $T$ are in $P_l$, not $P_i$. 
  Note the cycle $C$ consisting of any path in $P_l$ from $b$ to $d$ together with $a$ contains no vertices of $P_j$, so $P_j$ must be entirely inside or entirely outside this cycle.  This means at least one of the components of $P_i \setminus a$ - the one containing $c$ or the one containing $e$ - is not adjacent to any vertices of $P_\ell$.  Use $S$ to denote this component of $P_i \setminus a$ not adjacent to any vertices of $P_\ell$. Because $|P_i \cap bd(T)| \geq 2$ as $a,e \in P_i \cap bd(T)$, by Condition~\ref{item:cut_2bd} of Lemma~\ref{lem:shrinkable} $S$ has a vertex that can be removed from $P_i$ and added to another district. Because $S$ is not adjacent to any vertices of $P_j$, this vertex can be added to $P_l$.  Doing so produces a balanced partition, as desired. 

\begin{figure}
	\begin{subfigure}[b]{0.22\textwidth}
		\centering
		\includegraphics[scale = 0.9]{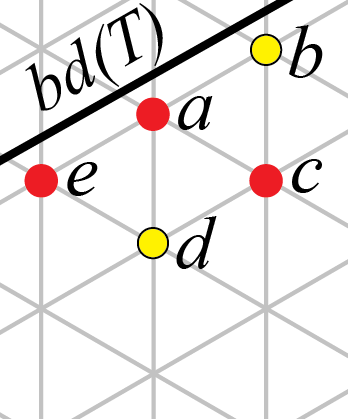}
		\caption{}
		\label{fig:nearlybalanced-anbhd}
	\end{subfigure}
	\hfill
	\begin{subfigure}[b]{0.22\textwidth}
		\centering
		\includegraphics[scale = 0.9]{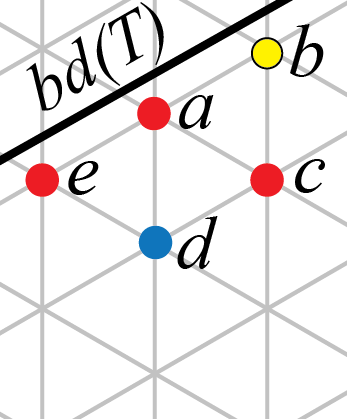}
		\caption{}
		\label{fig:nearlybalanced-anbhd2}
	\end{subfigure}
	\hfill
	\begin{subfigure}[b]{0.22\textwidth}
		\centering
		\includegraphics[scale = 0.9]{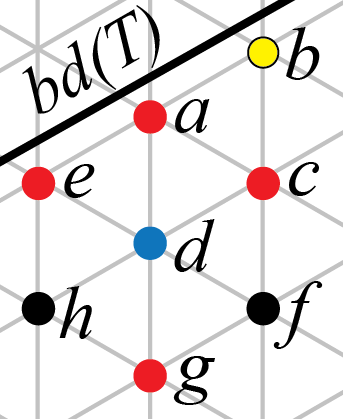}
		\caption{}
		\label{fig:nearlybalanced-dnbhd}
	\end{subfigure}
	\hfill
	\begin{subfigure}[b]{0.22\textwidth}
		\centering
		\includegraphics[scale = 0.9]{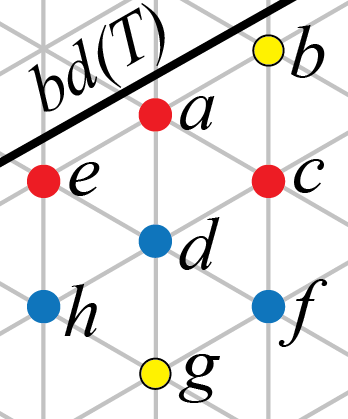}
		\caption{}
		\label{fig:nearlybalanced-dnbhd2}
	\end{subfigure}

	\caption{Images from the proof of Lemma~\ref{lem:nearlybalanced} when $a$ has a disconnected $i$-neighborhood. $P_i$ is shown in red, $P_j$ in blue, $P_l$ in yellow, and vertices whose district is not uniquely determined in black. (a) When $a \in P_i \cap bd(T)$ has a disconnected $i$-neighborhood and a disconnected $j$-neighborhood, its neighbors must be $b \in P_l$,$c \in P_i$, $d \in P_l$, and $e \in P_i$, as shown.  (b)  When $a \in P_i \cap bd(T)$ has a disconnected $i$-neighborhood but a connected $j$-neighborhood, its neighbors must be $b \in P_l$, $c \in P_i$, $d \in P_j$, and $e \in P_i$, as shown. (c) If $d \in P_j$ has a disconnected $i$-neighborhood, it must be that $g \in P_i$ but $f,h \notin P_i$; at least one of $f$ or $h$ must be in $P_j$, but the other may be in $P_l$. (d) If $d \in P_j$ has a connected $i$-neighborhood but a disconnected $j$-neighborhood, it must be that $f \in P_j$, $g \in P_l$, and $h \in P_j$.
}
	\label{fig:nearlybalanced}
\end{figure}

\textit{\underline{Case: $a$ has a disconnected $i$-neighborhood and a connected $l$-neighborhood}}: Next, suppose $a$ has a disconnected $i$-neighborhood but a connected $l$-neighborhood. For this to be true, it must the that $a$'s neighbors, in order, are $b \in P_l$, $c \in P_i$, $d \in P_j$, and $e \in P_i$; see Figure~\ref{fig:nearlybalanced-anbhd2}. First, suppose $d$'s $i$-neighborhood and $j$-neighborhood are connected. In this case we would like to add $d$ to $P_i$, but can't do so immediately because $|P_i| = k_i + 1$.  Instead, look at any component of $P_i \setminus N(d)$, of which there must exist at least one as $|P_i| = k_i + 1 > 5$.  As $P_i$ does not contain any corner vertices, at least one component $S$ of $P_i \setminus N(d)$ must contain an exposed vertex. By Condition~\ref{item:exp_2bd} of Lemma~\ref{lem:shrinkable}, there exists a vertex $v \in S$ that can be removed from $S$ and added to a different district.  If $v$ can be added to $P_l$ we do so and reach a balanced partition. Otherwise, we add $v$ to $P_j$, add $d$ to $P_i$, and, now that $a$ has a connected $i$-neighborhood, add $a$ to $P_l$, reaching a balanced partition.

 If $d$ does not have a connected $i$-neighborhood and $j$-neighborhood, we need to consider additional cases.  Suppose first that $d$ has a disconnected $i$-neighborhood. Let $f \neq a$ be $d$ and $c$'s other common neighbor; let $g$ be $f$ and $d$'s other common neighbor; and let $h$ be $d$'s last neighbor, between $e$ and $g$.  If $d$'s $i$-neighborhood is disconnected, it must be that $g \in P_i$ but $f,h \notin P_i$; at least one of $f$ or $h$ must be in $P_j$, but the other may be in $P_l$. See Figure~\ref{fig:nearlybalanced-dnbhd}.  Look at the cycle $C$ formed by any path from $a$ to $g$ in $P_i$ together with $d$.  This cycle encircles exactly one of $h$ or $f$. It also must hold that $P_l$ is entirely outside $C$, since $b \in P_l$ is outside $C$. That means that whichever of $f$ or $h$ is inside $C$ must be in $P_j$. Let $S_j$ be the component of $P_j \setminus d$ which is inside $C$. Let $S_i$ be the component of $P_i \setminus a$ that does not contain any vertices of $C$. We know $S_j \cap bd(T) = \emptyset$ because it is inside $C$, and we know $(P_i \setminus S_i) \cap bd(T) \neq \emptyset$ because it contains $a$. We  apply the Unwinding Lemma, Lemma~\ref{lem:s1s2}, with $P_1$ playing the role of $P_i$, $P_2$ playing the role of $P_j$, and $P_3$ playing the role of $P_l$.  We conclude there exists a sequence of moves, through balanced and nearly balanced partitions, after which (1) the partition is balanced, (2), all vertices in $S_i$ have been added to $P_j$, or (3) all vertices of $S_j$ have been added to $P_i$. In these moves only vertices in $S_i$ and $S_j$ have been reassigned, and in outcomes (2) and (3) the resulting partition has $|P_i| = k_i + 1$ and $|P_l| = k_l - 1$. If outcome (1) occurs, we are immediately done.  If outcome (2) occurs, then $a$ now has a connected $i$-neighborhood and so can be added to $P_l$, reaching a balanced partition. If outcome (3) occurs, $d$ now has a connected $i$-neighborhood and $j$-neighborhood, a case we already showed how to resolve above. In all cases, we eventually reach a balanced partition, as desired.

 The above paragraph resolves the case where $d$ has a disconnected $i$-neighborhood, so next we consider the case where $d$ has a connected $i$-neighborhood but a disconnected $j$-neighborhood.  It must be the case that $f \in P_j$, $g \in P_l$, and $h \in P_j$; see Figure~\ref{fig:nearlybalanced-dnbhd2}
 	 Let $S_i$ be the component of $P_i \setminus a$ containing $e$, and let $S_j$ be the component of $P_j \setminus d$ containing $f$.  We note  $(P_i \setminus S_i) \cap bd(T) \neq \emptyset$ because it contains $a$. There also must be a path from $b \in P_l$ to $g \in P_l$, which together with $a$ and $d$ forms a cycle $C$.  Because $S_j$ is contained inside this cycle, it must be that $S_j \cap bd(T) = \emptyset$.  Because $S_i$ is outside $C$ and $S_j$ is inside $C$, $S_i$ and $S_j$ do not have any adjacent vertices. Therefore we can apply Lemma~\ref{lem:s1s2} with $P_1 = P_i$, $P_2 = P_j$, and $P_3 = P_l$, and conclude there exists a sequence of moves after which (1) the partition is balanced, (2), all vertices in $S_i$ have been added to $P_j$, or (3) all vertices of $S_j$ have been added to $P_i$. In these moves only vertices in $S_i$ and $S_j$ have been reassigned, and in outcomes (2) and (3) the resulting partition has $|P_i| = k_i + 1$ and $|P_l| = k_l - 1$. If outcome (1) occurs, we are immediately done.  If outcome (2) occurs, then $a$ now has a connected $i$-neighborhood and so can be added to $P_l$, reaching a balanced partition. If outcome (3) occurs, $d$ now has a connected $j$-neighborhood and still has a connected $i$-neighborhood, a case we already showed how to resolve above. In all cases, we eventually reach a balanced partition, as desired. This completes the case when $a$ has a disconnected $i$-neighborhood and a connected $l$-neighborhood.

    \begin{figure}
    		\begin{subfigure}[b]{0.13\textwidth}
    		\centering
    		\includegraphics[scale = 0.7]{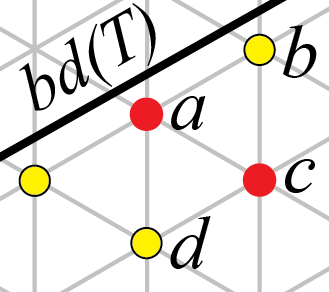}
    		\caption{}
    		\label{fig:nearlybalanced_aiconn}
    	\end{subfigure}
    	\hfill
    	\begin{subfigure}[b]{0.13\textwidth}
    		\centering
    		\includegraphics[scale = 0.7]{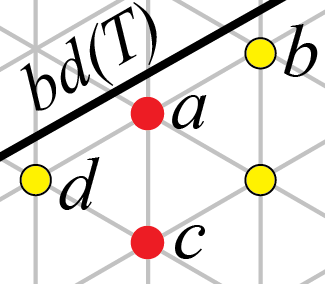}
    		\caption{}
    		\label{fig:nearlybalanced_aiconn}
    	\end{subfigure}
    	\hfill
    	\begin{subfigure}[b]{0.13\textwidth}
    		\centering
    		\includegraphics[scale = 0.7]{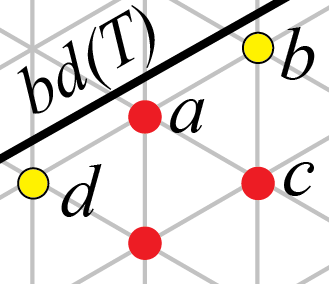}
    		\caption{}
    		\label{fig:nearlybalanced_aiconn}
    	\end{subfigure}
    	\hfill
    	\begin{subfigure}[b]{0.13\textwidth}
    		\centering
    		\includegraphics[scale = 0.7]{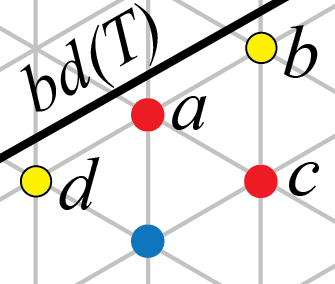}
    		\caption{}
    		\label{fig:nearlybalanced_aiconn}
    	\end{subfigure}
    	\hfill
    	\begin{subfigure}[b]{0.13\textwidth}
    		\centering
    		\includegraphics[scale = 0.7]{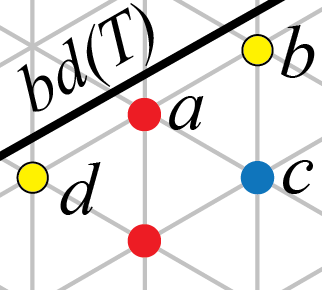}
    		\caption{}
    		\label{fig:nearlybalanced_aiconn}
    	\end{subfigure}
    	\hfill
    	\begin{subfigure}[b]{0.13\textwidth}
    		\centering
    		\includegraphics[scale = 0.7]{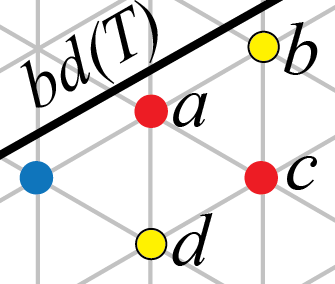}
    		\caption{}
    		\label{fig:nearlybalanced_aiconn}
    	\end{subfigure}
    	\hfill
	    \begin{subfigure}[b]{0.13\textwidth}
	    	\centering
	    	\includegraphics[scale = 0.7]{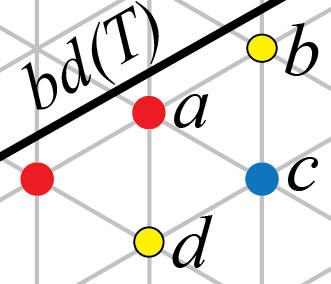}
	    	\caption{}
	    	\label{fig:nearlybalanced_aiconn}
	    \end{subfigure}
    \caption{Images from the proof of Lemma~\ref{lem:nearlybalanced}. $P_i$ is shown in red, $P_j$ in blue, and $P_l$ in yellow.    	
    	When $a \in P_i \cap bd(T)$ and $b \in P_l \cap bd(T)$ are adjacent, $a$ has a connected $i$-neighborhood, and $a$ has a disconnected $l$-neighborhood, these are all seven possibilities for $N(a)$.  In each, $d$ is a vertex in a different component of $N(a) \cap P_l$ than $b$, and $c$ is a vertex not in $P_l$ on the path from $b$ to $d$ in $N(a)$ that does not leave $T$; where there are multiple possibilities for $c$ and $d$, just one is labeled.}
    \label{fig:nearlybalanced_aiconn}
    \end{figure}
 
\textit{\underline{Case: $a$ has a connected $i$-neighborhood and a disconnected $l$-neighborhood}}: The final, and most difficult, case is when $a$ has a connected $i$-neighborhood and a disconnected $l$-neighborhood. Let $d$ be in a different component of $N(a) \cap P_l$ than $b$, and let $c \notin P_l$ be a vertex in $N(a)$ on the path in $N(a)$ from $b$ to $d$ that does not leave $T$. All seven possibilities for $N(a)$ in this case are shown in Figure~\ref{fig:nearlybalanced_aiconn}; when there are multiple options for how to label $c$ or $d$, it does not matter which is chosen. 
Consider the cycle $C$ formed by any path from $b$ to $d$ in $P_l$ together with $a$.  This path contains no vertices of $P_j$, so $P_j$ must be entirely inside or entirely outside this cycle.  Because $a$ has a connected $i$-neighborhood, $P_i \setminus a$ cannot have two components, so $P_i \setminus a$ must be entirely inside or entirely outside $C$.  If $P_2$ and $P_1 \setminus a$ are on opposite sides of $C$ (such as in (f,g) of Figure~\ref{fig:nearlybalanced_aiconn}, or possibly in (a,b,c) if $P_j$ is outside of $C$), we are done: $P_1 \setminus a$ must have a removable vertex (Condition~\ref{item:nobd} or Condition~\ref{item:cut_2bd} of Lemma~\ref{lem:shrinkable}), and this vertex is not adjacent to $P_j$ so we must be able to add it to $P_l$. We also note that $P_i \setminus a$ and $P_j$ both being outside $C$ is impossible, because vertex $c$, which is inside $C$, must be in $P_i$ or $P_j$.

This leaves the last remaining case, when $P_i \setminus a$ and $P_j$ are both inside $C$.  Note this means that $P_j \cap bd(T) = \emptyset$ and $P_i \cap bd(T) = \{a\}$. Because $P_j \cap bd(T) = \emptyset$, by Lemma~\ref{lem:2tricolortri}, there exists exactly two distinct triangular faces $F'$ and $F''$ whose three vertices are in three different districts. Exactly one of these triangular faces has its vertices in $P_1$, $P_2$, and $P_3$ in clockwise order, and the other must have its vertices of $P_1$, $P_2$, and $P_3$ in counterclockwise order. We label the vertices around $F'$ as $a' \in P_i$, $b' \in P_l$, and $c' \in P_j$ and assume they are in clockwise order, and label the vertices around $F''$ as $a'' \in P_i$, $b'' \in P_l$, and $c'' \in P_j$ and assume they are in counterclockwise order. If adding $a'$ or $a''$ to $P_l$ is valid we do so and are done, so we assume $a'$ and $a''$ cannot be added to $P_l$. This means $a'$ and $a''$ must have a disconnected $i$-neighborhood or a disconnected $l$-neighborhood.

Suppose first that $a'$ has a disconnected $l$-neighborhood but a connected $i$-neighborhood, and $a' \notin bd(T)$. This must mean that around $a'$, in order but not necessarily consecutive, are $b' \in P_l$, $c' \in P_j$, $d' \in P_l$, and $e' \in P_i$, with no other connected components of $N(a') \cap P_i$ except for the one containing $e'$. See Figure~\ref{fig:nearlybalanced_a'-ldiscon} for an example, though this is not the only possibility for $N(a')$. $N(a')$ also cannot contain more components of $P_j$ or $P_l$ as this would require more alternation of $P_j$ and $P_l$ in $N(a')$ than is allowed by Lemma~\ref{lem:alternation}.  Consider the cycle $C$ formed by any path from $b'$ to $d'$ in $P_l$ together with $a'$. This cycle cannot encircle $e' \in P_i$ because then it would encircle all of $P_i \setminus a'$ but  $a \in P_i \cap bd(T)$, so instead it must encircle $c' \in P_j$ and thus encircle all of $P_j$. This means that the only place $P_i$ is adjacent to $P_j$ is at $a'$. By Condition~\ref{item:nobd} of Lemma~\ref{lem:shrinkable}, in $P_j \setminus N(a')$ there is $v \in P_j$  that can be removed and added to another district; because $a'$ is the only vertex of $P_i$ adjacent to $P_j$, it must be that $v$ can be added to $P_l$. We therefore add $v$ to $P_l$ and add $a$ to $P_j$, possible because $a'$ has a connected $i$-neighborhood and a connected $j$-neighborhood. This resolves the case when $a' \notin bd(T)$ has a connected $i$-neighborhood and, by the same arguments, the case where $a'' \notin bd(T)$ has a connected $i$-neighborhood.  We therefore assume that $a'$ and $a''$ both have a disconnected $i$-neighborhood or are in $bd(T)$. 

\begin{figure}
	\begin{subfigure}[b]{0.45\textwidth}
		\centering
		\includegraphics[scale = 0.9]{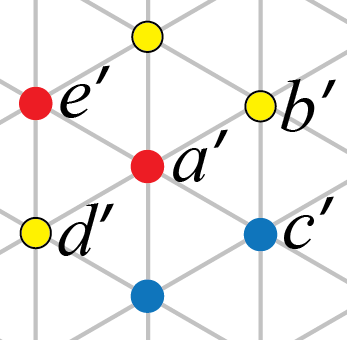}
		\caption{}
		\label{fig:nearlybalanced_a'-ldiscon}
	\end{subfigure}
	\hfill
	\begin{subfigure}[b]{0.45\textwidth}
		\centering
		\includegraphics[scale = 0.9]{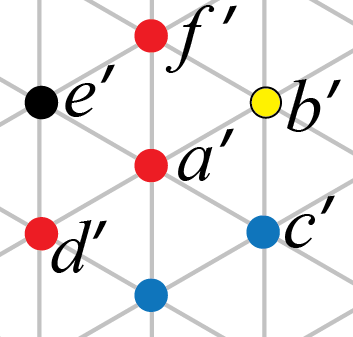}
		\caption{}
		\label{fig:nearlybalanced_a'-e'}
	\end{subfigure}
	
	\caption{Figures from the proof of Lemma~\ref{lem:nearlybalanced}. In these images, $P_i$ is red, $P_j$ is blue, $P_l$ is yellow, and vertices whose district assignment is not yet determined are black. (a) When $a'$, $b'$, $c'$ is a tricolor triangle and $a' \notin bd(T)$, an example of when $a'$ has a connected $i$-neighborhood and a disconnected $l$-neighborhood. (b) When $a'$, $b'$, $c'$ is a tricolor triangle and $a' \notin bd(T)$, an example of when $a'$ has a disconnected $i$-neighborhood; $e'$ may be in $P_j$ or $P_l$.}
	\label{fig:nearlybalanced_a'}
\end{figure}

Suppose $a'$ has a disconnected $i$-neighborhood. Let $d'$ and $f'$ be two vertices in different connected components of $N(a') \cap P_i$ such that, when traversing $N(a')$ beginning with $b'$ followed by $c'$, $d'$ is encountered before $f'$. When $a' \notin bd(T)$, one path from $d'$ to $f'$ in $N(a)$ must contain $b'$ and $c'$, and the other must contain some vertex $e' \in T$ where $e' \notin P_i$. See Figure~\ref{fig:nearlybalanced_a'-e'} for an example, though note this not the only possibility for $N(a')$. 
We first consider the case where $a' \notin bd(T)$ and $e' \in P_j$. Then, we show how the case $a' \in bd(T)$ can be resolved by looking at $a''$ instead; much of this argument resolves reducing to previously solved cases with $a''$ replacing $a'$. Our final case then considers when neither $a'$ nor $a''$ is in $bd(T)$ and $e'$ and the similarly-defined $e''$ are both in $P_l$.

\underline{\it Subcase: $a' \notin bd(T)$, $e' \in P_j$}.  We first consider the case where $a' \notin bd(T)$ and $e' \in P_j$. Look at any path from $c'$ to $e'$ in $P_j$; together with $a'$, this path forms a cycle $C$. Note $P_l$ cannot be inside this cycle because $P_l$ contains boundary vertices, notably all three corners of $T$.
Let $S_1$ be the component of $P_i \setminus a$ inside this cycle, and note $S_1 \cap bd(T) = \emptyset$. By Condition~\ref{item:nobd} of Lemma~\ref{lem:shrinkable}, there exists a vertex $v \in S_1$ that can be removed from $S_1$ and added to another district; because $P_l$ is outside $C$, we must be able to add $v$ to $P_j$.

We next focus on $c' \in P_j$.  First, suppose removing $c'$ from $P_j$ and adding it to $P_l$ is valid; if it is, we do so.  This will not change the fact that $v$ can be removed from $S_1$, but it may change which district $v$ can be added to.  If $v$ can now be added to $P_j$, we do so and reach a balanced partition.  If $v$ can now be added to $P_l$ we do so; in this case, $P_l$ now has one more than its ideal number of vertices.  Now the district with one too many vertices contains a corner vertex (in fact, three corner vertices) of $T$, and we apply a rebalancing step as above. In either case, we can reach a balanced partition if $c'$ can be added to $P_l$.

Next, we assume $c'$ cannot be added to $P_l$, meaning $c'$ must have a disconnected $j$-neighborhood or a disconnected $l$-neighborhood. 
 Note because we have assumed $P_j \cap bd(T) = \emptyset$, $c' \notin bd(T)$.
First, suppose $c'$ has a disconnected $l$-neighborhood, and show it must also have a disconnected $j$-neighborhood   
We know $b'$ is in one component of $N(c') \cap P_l$, and let $g'$ be in a different component of $N(c') \cap P_l$.  Consider the cycle $C'$ formed by any path in $P_l$ from $b'$ to $g'$ together with $c'$.  Note $P_i$ must be entirely outside $C'$, as $P_i \cap bd(T)$ is not empty (it contains vertex $a$). Consider the path in $N(c')$ from $b'$ to $g'$: This path must contain a vertex not in $P_l$ because $b'$ and $g'$ are in different components of $P_l \cap N(c')$; can't contain a vertex outside $T$ because $c' \notin bd(T)$; can't contain a vertex of $P_i$ because $P_i$ is outside $C'$; so must contain a vertex of $P_j$. Therefore $c'$ has a neighbor in $P_j$ inside $C'$.  However $c'$ must also have a neighbor in $P_j$ outside $C$, namely it's neighbor in cycle $C$; recall cycle $C$ consists of any path from $c'$ to $e'$ in $P_j$ together with $a$.  Note $e'$ cannot be inside $C'$ because it's adjacent to $a' \in P_i$, so any path from $c'$ to $e'$ in $P_j$ must begin with a neighbor of $c'$ in $P_j$ that is outside $C'$. Because $c'$ has a neighbor in $P_j$ inside $C'$ and a neighbor in $P_j$ outside $C'$, $c'$ must have a disconnected $j$-neighborhood. Because we have shown whenever $c'$ has a disconnected $l$-neighborhood it also has a disconnected $j$-neighborhood, it suffices to resolve the case when $c'$ has a disconnected $j$-neighborhood, which we now do.



Suppose $c'$ has a disconnected $j$-neighborhood.  We let $S_2$ be the component of $P_j \setminus c'$ that does not contain any vertices of $C$, and do cases for whether $S_2$ is inside or outside of $C$.

If $S_2$ is outside $C$, we use the Unwinding Lemma, Lemma~\ref{lem:s1s2}, with $P_i$ in place of $P_1$ and $P_j$ in place of $P_2$. This lemmas applies to $P_i$ and $P_j$  because $S_1 \cap bd(T) = \emptyset$ because $S_1$ is inside $C$; $S_2 \cap bd(T) = \emptyset$ because $P_j \cap bd(T) = \emptyset$; and $S_1$ and $S_2$ are not adjacent because one is inside $C$ and the other is outside $C$. There exists a sequence of moves through nearly balanced partitions after which (1) the partition is balanced, (2) all vertices in $S_1$ have been added to $P_j$, or (3) all vertices of $S_2$ have been added to $P_i$. In these moves only vertices in $S_1$ and $S_2$ have been reassigned, and in outcomes (2) and (3) the resulting partition has $|P_i| = k_i + 1$ and $|P_j| = k_j$. If outcome (1) occurs, we are done.  If outcome (2) occurs, now $a'$ has a connected $i$-neighborhood, a case we have already considered and resolved above. If outcome (3) occurs and outcome (2) does not occur, now $c'$ has a connected $j$-neighborhood.  As we already saw it's impossible for $c'$ to have a connected $j$-neighborhood but a disconnected $l$-neighborhood, this means $c'$ must now have a connected $l$-neighborhood as well. This means $c'$ can be added to $P_l$. Additionally, because what remains of component $S_1$ is not empty and therefore still contains a removable vertex, we can remove that vertex from $P_i$ and add it to $P_j$, reaching a balanced partition.

If $S_2$ is inside $C$, we instead use Lemma~\ref{lem:cycle-recom} applied to cycle $C$ and district $P_j$; cycle $C$ consists entirely of vertices in district $j$ except for one vertex, $a' \in P_i$. All vertices enclosed by $C$ are in $P_i$ or $P_j$. We focus on the vertex $c'$ which is adjacent to $a'$ in $C$, and note that $N_C(c') \cap P_j$ must be disconnected because component $S_2$ of $P_j \setminus c'$ is inside $C$. By Lemma~\ref{lem:cycle-recom}, there exists a recombination step for $P_i$ and $P_j$, only changing the district assignment of vertices enclosed by $C$ and leaving $|P_j|$ and $|P_i|$ unchanged, after which $N_C(c') \cap P_j$ is connected. This means that $c'$ now has a connected $j$-neighborhood, and just as above there exists a sequence of two additional steps leading to a balanced partition, adding $c'$ to $P_l$ and adding a vertex of the new component of $P_i \setminus a'$ inside $C$ to $P_j$.  

If $a''$ and $e''$, with $e''$ defined similarly to $e'$, satisfy $a'' \notin bd(T)$ and $e'' \in P_j$, there also exists a sequence of moves reaching a balanced partition, replacing $c'$ in the arguments with $c''$, etc.

\underline{\it Subcase: $a' \in bd(T)$}.  Suppose $a' \in bd(T)$, meaning one of the two tricolor triangles in $T$ contains a vertex of $bd(T)$ (in fact, two vertices of the tricolor triangle must be in $bd(T)$, though we will not need this fact). In this case, we look at the other tricolor triangle, consisting of vertices $a'' \in P_i$, $b''\in P_l$, and $c'' \in P_j$. Because of how $N(a')$ must look, shown in Figure~\ref{fig:nearlybalanced_aiconn}(d,e), it's impossible to have $a' = a''$.  Because $P_i$ contains exactly one boundary vertex, $a'$, we therefore know $a'' \notin bd(T)$. If $a''$ has a connected $i$-neighborhood, this case is resolved above (for $a'$, but the same arguments apply), so we assume $a''$ has a disconnected $i$-neighborhood. Let $d''$ and $f''$ be two vertices in different connected components of $N(a'') \cap P_i$. Then one path from $d''$ to $f''$ in $N(a)$ must contain $b''$ and $c''$, and the other must contain some vertex $e'' \in T$ where $e'' \notin P_i$. If $e'' \in P_j$, this subcase is resolved above (for $a'$ and $e'$, but the same arguments apply). Therefore we assume $e'' \in P_l$.  In this case, look at any path from $b''$ to $e''$ in $P_l$, which together with $a''$ forms a cycle $C$. Let $S_1$ be the component of $P_i \setminus a''$ inside this cycle. There exists a vertex $v \in S_1$ that can be removed from $S_1$ and added to another district by Condition~\ref{item:nobd} of Lemma~\ref{lem:shrinkable}. Note that $P_j$ must be outside of $C$, because there exists a path from $c' \in P_j$ to $bd(T)$ (via $a'$) not using any vertices of $P_l$. Therefore it must be that $v$ can be added to $P_l$, so we do so and reach a balanced partition.

\underline{\it Subcase: $a' \notin bd(T)$, $a'' \notin bd(T)$, $e' \in P_l$, $e'' \in P_l$}. This is the final case that must be considered. Let $Q$ be any path from $a'$ to $a''$ in $P_i$. Let $S_1'$ be the component of $P_i \setminus a'$ not containing $Q$, and let $S_1''$ be the component of $P_i \setminus a''$ not containing $Q$. At least one of $S_1'$ and $S_1''$ must contain no vertices of $bd(T)$ because they are disjoint and $P_i$ only contains one vertex of $bd(T)$. Without loss of generality we assume it's $S_1'$ that contains no vertices of $bd(T)$, and we now call this component $S_1$. Consider any path from $b'$ to $e'$ in $P_l$ which together with $a'$ forms a cycle $C$.  Note that $P_j$ must be outside this cycle because there exists a path from $P_j$ to $bd(T)$ not containing any vertices of $C$, namely from $c''$ to $a''$ and then through $P_i$ to $a$. Additionally, $S_1$ must be inside this cycle, as $a \in bd(T)$ must be in the other component of $P_i \setminus a'$ that is not $S_1$.  We know $S_1$ must have a vertex that can be removed from $P_i$ and added to another district by Condition~\ref{item:nobd} of Lemma~\ref{lem:shrinkable}, and this vertex (which is inside $C$) must be able to be added to $P_l$, because $P_j$ is entirely outside $C$.  This produces a balanced partition, as desired. 

This concludes our proof that from any nearly balanced partition, there exists a sequence of moves producing a balanced partition.




\end{proof}

\noindent This is sufficient to prove our main theorem. 

\begin{thm}[Main Theorem, also stated as Theorem~\ref{thm:main}]
	$G_\Omega$ is connected. 
\end{thm}
\begin{proof}
	$G_\Omega$ is the graph whose vertices are the balanced and nearly balanced partitions of $T$ into three simply connected districts with ideal sizes $k_1$, $k_2$, and $k_3$. Two partitions are adjacent in $G_\Omega$ (connected by an undirected edge) if a recombination move can transform one into another. 	
	By Lemma~\ref{lem:nearlybalanced}, from any nearly balanced partition there exists a sequence of moves producing a balanced partition. By Lemma~\ref{lem:sweepline} and Lemma~\ref{lem:groundstate}, from any balanced partition there exists a sequence of moves producing one of the six ground states. Because it is straightforward to move among the six ground states with recombination moves, this proves that $G_\Omega$ is connected. 
\end{proof}

\section{Diameter Bound} 

Our proof that $G_\Omega$ is connected is constructive, meaning for any two partitions we find a path between them in $G_\Omega$. A natural next question is to ask about the lengths of these paths, which we answer here. This result is stated in terms of $n$,  the side length of the triangle $T$, and it is worth recalling that triangle $T$ contains $ n(n+1)/2 = O(n^2)$ vertices.  

\begin{lem}\label{lem:diameter} (Theorem~\ref{thm:dia_intro})
	For any two partitions $\sigma$ and $\tau$ in $\Omega$, the distance between $\sigma$ and $\tau$ in $G_{\Omega}$ is $O(n^3)$. 
\end{lem}
\begin{proof}
	We first analyze how many steps are required to transform any balanced partition $\sigma$ into a ground state partition. Without loss of generality, suppose in $\sigma$ the vertex $\cc_1 \in P_1$.  The number of steps in the sweep-line process is $O(k_1) = O(n^2)$.  In each step of the sweep line process, adding a vertex of $P_1$ to $\cc_i$ could require at most on tower move.  A tower move could be comprised of up to $O(n)$ individual recombination steps, one at each level of the tower. 
	
	We now consider how many recombination steps are required for the rebalancing step. 
	We note that, as written, the Unwinding Lemmas (Lemmas~\ref{lem:s1s2} and~\ref{lem:s1s2x}) require $O(k_1 + k_2)$ flip moves.  However, all of these flip moves (except possibly the last, which adds a vertex to $P_3$) only require moving vertices between $P_1$ and $P_2$: therefore, all of the flip steps described carefully in the lemma's proof can in fact be accomplished with one recombination step which jumps directly to either the final partition (if no vertex is added to $P_3$ during this process, which is Outcome (2) or (3)) or to the step before a single flip step adds a vertex tp $P_3$ (as in Outcome (1)).  This means that using an Unwinding Lemma requires only 1 or 2 recombination steps, or $O(1)$ recombination steps overall. The cycle recombination lemmas (Lemmas~\ref{lem:cycle-recom} and~\ref{lem:cycle-recom-vtx}) similarly only require $1$ recombination step.	
	Each of the four cases (Case A, Case B, Case C, Case D) requires at most $O(1)$ flip moves, $O(1)$ applications of an unwinding lemma, and/or $O(1)$ applications of a cycle recombination lemma. As each can be accomplished with $O(1)$ recombination moves, this means at most $O(1)$ recombination moves are required to rebalance a partition during the sweep line process. In total, this means the sweep line process accomplishing $\cli \subseteq P_1 \subseteq \clei$ (Lemma~\ref{lem:sweepline}) takes $O(n^3)$ steps: in each of $O(n^2)$ iterations, at most $O(n) + O(1) = O(n)$ recombination moves are performed.

	


	
	

	Once we have satisfied $\cli \subseteq P_1 \subseteq \clei$, reaching the ground state $\sigma_{123}$ (as described in Lemma~\ref{lem:groundstate}) requires at most $O(n)$ additional recombination steps: two for each vertex of $P_1 \cap \ci$ that must be moved up within column $i$, plus one final recombination of $P_2$ and $P_3$.  Careful analysis of Lemma~\ref{lem:nearlybalanced} shows that, similar to our rebalancing process, it takes at most $O(1)$ flip moves, $O(1)$ unwinding steps, and/or $O(1)$ cycle recombinations to transform an arbitrary nearly balanced partition into a balanced partition.  In all, this means that $O(1) + O(n^3)  + O(n) = O(n^3)$ recombination moves can be used to move from any partition - balanced or nearly balanced - to a ground state. 

	Altogether, for two arbitrary partitions $\sigma$ and $\tau$, it takes $O(n^3)$ recombination moves to transform $\sigma$ into a ground state; $O(1)$ recombination steps to move among ground states; and $O(n^3)$ recombination moves to transform that ground state into $\tau$. As $O(n^3) + O(1) + O(n^3) = O(n^3)$, this proves the theorem. 	
\end{proof}

\begin{thm}\label{thm:diameter}
	The diameter of $G_\Omega$ is $O(n^3)$.
\end{thm}
\begin{proof}
	This follows immediately from Lemma~\ref{lem:diameter} and the definition of diameter. 
\end{proof}

\section*{Acknowledgements} 

The author thanks Jamie Tucker-Foltz and Dana Randall for fruitful discussions during the preparation of this manuscript. 
S. Cannon is supported in part by NSF grants CCF-2104795 and DMS-1803325.

\bibliographystyle{plain}
\bibliography{ergod_bib}

\end{document}